\documentclass[12pt]{article}
\usepackage[margin=1in]{geometry}

\usepackage{subcaption}
\usepackage[ colorlinks = true,
             linkcolor = blue,
             urlcolor  = blue,
             citecolor = blue,
             anchorcolor = green,
]{hyperref}

\usepackage{tabularx}
\usepackage{amsfonts}
\usepackage{amssymb}
\usepackage{amsthm}
\usepackage{amsmath}
\usepackage{amsxtra}
\usepackage{float}

\usepackage{mathtools}
\mathtoolsset{showonlyrefs}

\usepackage[sc,osf]{mathpazo}
\usepackage{eulervm}

\usepackage{graphics}
\usepackage{enumerate}
\usepackage[linesnumbered,ruled,vlined]{algorithm2e}
\usepackage[auth-sc-lg,affil-sl]{authblk}

\newcommand{\appref}[1]{\hyperref[#1]{{Appendix~\ref*{#1}}}}

\theoremstyle{plain}
\newtheorem{theorem}{Theorem}[section]
\newtheorem{lemma}[theorem]{Lemma}
\newtheorem{proposition}[theorem]{Proposition}
\newtheorem{claim}[theorem]{Claim}
\newtheorem{corollary}[theorem]{Corollary}

\newtheorem{definition}[theorem]{Definition}
\newtheorem{observation}[theorem]{Observation}

\theoremstyle{definition}
\newtheorem{example}{Example}

\def\idty{{\leavevmode\mathbf{I}}} 
\def\Rl{{\mathbb R}}
\def\C{{\mathbb C}}     
\def\Nl{{\mathbb N}}     

\def\norm #1{\Vert #1\Vert}

\def\bra #1{\langle #1\vert}
\def\ket #1{\vert #1\rangle}

\def\ketbra #1#2{\vert #1\rangle \langle #2\vert}
\def\kettbra#1{\ketbra{#1}{#1}}
\def\tr{{\rm Tr}}
\newcommand\Scp[2]{\ensuremath{\, \langle #1 , #2 \rangle}}

\def\E{{\mathbb E}}

\newcommand*{\vphi}{\varphi}

\newcommand*{\half}{\frac{1}{2}}

\newcommand*{\cA}{\mathcal{A}}
\newcommand*{\cB}{\mathcal{B}}

\newcommand*{\cE}{\mathcal{E}}

\newcommand*{\cH}{\mathcal{H}}

\newcommand*{\cK}{\mathcal{K}}
\newcommand*{\cL}{\mathcal{L}}

\newcommand*{\cO}{\mathcal{O}}

\newcommand*{\cS}{\mathcal{S}}

\newcommand*{\cV}{\mathcal{V}}
\newcommand*{\cW}{\mathcal{W}}

\newcommand*{\cY}{\mathcal{Y}}

\newcommand*{\End}{\mathsf{End}}

\newcommand{\bc}{\begin{center}}
\newcommand{\ec}{\end{center}}

\def\Complex{\mathbb{C}}

\def\Z{\mathbb{Z}}

\def\01{\{0,1\}}

\newcommand{\floor}[1]{\lfloor{#1}\rfloor}
\newcommand{\eps}{\varepsilon}

\newcommand{\proj}[1]{|#1\rangle\langle#1|}

\newcommand{\rank}{\operatorname{rank}}

\renewcommand*{\Im}{\mathop{\mathsf{Im}}}
\renewcommand*{\Re}{\mathop{\mathsf{Re}}}

\newcommand*{\functionalstate}{\sigma}
\newcommand*{\zetap}{\zeta^\prime}

\newcommand*{\voalkzero}{{\mathsf{L}}_{k,0}} 
\newcommand*{\moduleklambda}{{\mathsf{L}}_{k,\lambda}} 
\newcommand*{\modulekflambda}[1]{\mathsf{L}_{k,#1}} 
\newcommand*{\hwgmodule}[1]{\mathsf{L}_{#1}} 

\newcommand*{\algoG}{{\bf\mathsf{H}}} 
\newcommand*{\anag}{{\bf\mathsf{h}}} 

\newcommand{\Y}{\mathsf{Y}}
\newcommand{\y}{\mathbf{y}}
\newcommand{\dimctwo}{C_{\cV}}

\newcommand{\iY}{\mathcal{Y}}
\newcommand{\iYtr}[1]{\mathcal{Y}^{[#1]}} 
\newcommand*{\cftY}{\mathbb{Y}} 

\newcommand{\iy}{\mathsf{y}}
\newcommand{\itype}[3]{\tiny \ensuremath{\begin{pmatrix}#1 \\ #2 \; #3\end{pmatrix}}}

\newcommand*{\pairL}{\left(}
\newcommand*{\pairR}{\right)}
\newcommand*{\torusq}{\mathsf{p}}

\newcommand*{\1}{\mathsf{1}}
\newcommand*{\wt}{\mathsf{wt}\ }
\newcommand*{\mpsA}{\mathsf{A}} 

\newcommand{\W}{\mathsf{W}}

\newcommand{\Wtrq}[2]{\mathsf{W}^{[#2]}_{#1}} 

\newcommand*{\g}{\mathfrak{g}}
\newcommand*{\ga}{\hat{\mathfrak{g}}}

\newcommand*{\Uga}{U(\hat{\mathfrak{g}})}
\newcommand*{\Ugam}{U(\hat{\mathfrak{g}}_{-})}
\newcommand*{\Ugap}{U(\hat{\mathfrak{g}}_{+})}

\newcommand*{\fa}{\mathfrak{a}}
\newcommand*{\fb}{\mathfrak{b}}
\newcommand*{\fc}{\mathfrak{c}}

\newcommand*{\fe}{\mathfrak{e}}
\newcommand*{\ff}{\mathfrak{f}}
\newcommand*{\fg}{\mathfrak{g}}
\newcommand*{\fh}{\mathfrak{h}}

\newcommand*{\Residuum}{\mathsf{Res}}
\newcommand*{\Shom}{S_{\textsf{hom}}}
\newcommand*{\cShom}{\mathcal{S}_{\textsf{hom}}}

\newcommand*{\Top}{{\mathsf{T}}} 

\newcommand*{\cftW}{\mathbb{W}} 
\newcommand*{\cftE}{\mathbb{D}} 
\newcommand*{\cftEcp}{\mathbb{E}} 
\newcommand*{\itw}[3]{\small \left[\begin{matrix}#1 \\ #2 \; #3\end{matrix}\right]}

\newcommand*{\bnd}{{\vartheta}} 
\newcommand*{\bndt}{{\vartheta^{\textsf{tr}}}} 
\newcommand*{\cftbnd}{{\Theta}} 
\newcommand*{\apperr}{{\mathsf{err}}} 

\newcommand*{\Endint}[2]{{\End(#1)\{\{#2\}\}}}

\begin{document}
\title{Matrix product approximations \\
to conformal field theories}

\author{Robert K\"onig}
\affil{Institute for Advanced Study \& Zentrum Mathematik, \\
Technische Universit\"at M\"unchen, 85748 Garching, Germany\thanks{robert.koenig@tum.de}}

\author{Volkher B.~Scholz}
\affil{Institute for Theoretical Physics,\\
 ETH Zurich, 8093 Z\"urich, Switzerland\thanks{scholz@phys.ethz.ch}}

\maketitle

\begin{abstract}
We establish rigorous error bounds  for approximating  correlation functions of conformal field theories (CFTs) by certain finite-dimensional tensor networks. For chiral CFTs, the approximation takes the form of a matrix product state. For full CFTs consisting of a chiral and an anti-chiral part, the approximation is given by a finitely correlated state.  We show that the bond dimension scales polynomially in the inverse of the approximation error and sub-exponentially in the ultraviolett cutoff.  We illustrate our findings using Wess-Zumino-Witten models, and show that
there is a one-to-one correspondence between group-covariant MPS and our approximation.
\end{abstract}

\newpage
\tableofcontents
\newpage

\section{Introduction\label{sec:intro}}

Quantum field theory is arguably one of the most versatile physical theories developed to date. Beyond its early and astounding successes in the modelling of fundamental interactions, its applications now span many different areas of physics across all scales, ranging from subatomic particles, to condensed matter, to cosmology. The language of quantum field theory provides a sophisticated, unifying conceptual framework for addressing a variety of questions of physical interest. It constitutes one of the  main pillars of modern physics.

Quantum field theory provides  significant insight into the mechanics of interacting quantum systems, yet calculations  often tend to be extremely tedious, or even intractable. One major obstacle is the fact that in many cases, methods for obtaining approximate answers are unknown. Even in settings where, e.g., systematic expansions exist, estimating the accuracy of
a computational scheme may be challenging or impossible. This difficulty of applying the variational principle to quantum field theories in order to find good approximate expressions was already noted by Feynman~\cite{feynman1987difficulties}.

Quantum field theory is also -- as a number of physical theories --  a rich source of inspiration and challenging problems in  mathematics. Indeed, putting  general quantum field theories on a firm axiomatic footing remains an important research topic. The special class of conformal field theories (CFTs)~\cite{francesco2012conformal} is an important exception in this regard: here the presence of conformal symmetries allows to provide a rigorous algebraic formulation. Fortunately, CFTs also turn out to be physically relevant, as they provide accurate descriptions of e.g.,~critical systems. Consequently, CFTs provide an ideal testbed for ideas related to general quantum field theories. In particular, it is natural to first investigate variational methods in the context of CFTs. This is the topic of this paper.

Contrary to the case of quantum field theories, variational methods  in non-relativistic quantum mechanics, as well as condensed matter physics, are well established: here approximation schemes with well-controlled error guarantees exist and are being applied successfully.  One prime example are tensor network contraction schemes. In one spatial dimension, an important member of this class is the \emph{density matrix renormalization group} method, in short DMRG~\cite{RevModPhys.77.259}. It varies over a certain class of Ansatz states, called matrix product states~\cite{Wolf:2007wt} (MPS) or finitely correlated states~\cite{FannesNachtergaeleWerner1994,fannesnachtergaelewerner1992} (FCS). This class of states has been successfully applied to variety of physical models, often yielding exceptionally good results. 
The suitability of such states for variational physics has been  rigorously explained by Hastings~\cite{Hastings:2007vw}, who showed that ground states of gapped Hamiltonians in one spatial dimensions can be arbitrarily well approximated by matrix product states in an efficient manner. His line of work culminated in~\cite{landau2015polynomial}, yielding a provably convergent and efficient algorithm to find the ground state to arbitrary accuracy.

In this paper, we argue that similar statements hold for
CFTs in two independent variables. To be precise, we show that correlation functions for surfaces of genus zero and one admit an approximate representation as a special class of matrix product states. We emphasize that our derivation provides rigorous error bounds for the validity of these approximations. Our results hold both for the chiral parts of such theories, as well as for the full theory consisting of both a chiral as well as an anti-chiral part. We illustrate our findings with examples from the family of Wess-Zumino-Witten models~\cite{wess1971consequences,witten1984non,novikov1981multivalued}. These are a natural family of CFTs which are rich enough to illustrate the main concepts, and, correspondingly, are often considered in the literature. For these examples, we find that the matrix product approximations originate from the special class of symmetric (`$G$-covariant') matrix product states. We provide an algorithm for constructing the approximation to the corresponding CFT. 

\subsubsection*{Prior work}
Our work is motivated by  a long line of previous results. 

First, Haegeman, Cirac, Osborne, Verschelde and Verstraete~\cite{2010PhRvL.104s0405V,2010PhRvL.105y1601H} and Osborne, Eisert and Verstraete~\cite{2010PhRvL.105z0401O} initiated a series of studies~\cite{2013PhRvB..88h5118H,2013PhRvL.110j0402H,jennings2015continuum} by constructing approximation schemes for certain quantum field theories. These can be seen as the continuum limit of tensor network schemes, and are thus aimed at addressing the difficulties that Feynman pointed out. However, so far, no error bounds exist for these approaches. While our approximations do not fall into this class, the fact that such continuum limits seem to work well in a variety of models is certainly one of the motivations for our studies. 

Second, several groups, including Nielsen
et al.~\cite{2013PhRvB..87p1112E,nielsen2014bosonic,nielsen2012laughlin},
Estienne et al.~\cite{Estienne:2013vf}, as well as Zaletel and Mong~\cite{zaletel2012exact}, examined the possibility of constructing matrix product states for quantum Hall systems from the field theoretic representation of the spatial electronic wave functions. Their central idea of suitably truncating the Hilbert spaces at a finite dimension is very much in our spirit. However, 
there are currently no error bounds for these methods and so far, only chiral correlation functions have been considered. This line of research is complemented by the work of Pollmann et al.~\cite{2009PhRvL.102y5701P}, Pirvu et al.~\cite{2012PhRvB..86g5117P} and Stojevic et al.~\cite{Stojevic:2014ul} on simulating critical quantum systems using matrix product states. As these systems are assumed to be described by CFTs, our results may be understood as an explanation of the empirical success of these studies. For a discussion of the  relationship of our work (which focuses on MPS) to the \emph{multi-scale renormalization Ansatz}, another method providing accurate descriptions of quantum critical systems introduced by Vidal~\cite{PhysRevLett.99.220405} (see also~\cite{PhysRevLett.102.180406,PhysRevB.79.144108}), we refer to our conclusions.

The challenge of simulating quantum field theories has, of course, been addressed on many levels. For example, it  has been envisioned as a potential field of application of a working quantum computer. Here we refer the reader to the work of Jordan et al.~\cite{jordan2012quantum}, where a quantum algorithm for computing relativistic scattering amplitudes in a quantum field theory with quartic interaction is presented. Such field theories are generally not conformal. We emphasize that contrary to~\cite{jordan2012quantum}, we are interested  in tensor networks which can be (ideally efficiently) contracted on a classical computer. Furthermore, our focus is on all correlation functions for a large class of CFTs.

From the mathematical side, our results rely strongly on the theory of vertex operator algebras (VOAs). These were introduced by Borcherds~\cite{Borcherds:1992bi} as well as Frenkel~\cite{frenkel1989vertex}, and further investigated in detail by Frenkel, Huang and Lepowsky~\cite{Frenkel:1993fv}.
Further key definitions and properties were established by
Huang~\cite{Huang:1995kz,Huang:2005gs}, Huang and Kong~\cite{HuangKong}, Zhu~\cite{Zhu:1996cp}, Dong~\cite{Dong:2014gx} and others. 
Of particular importance to this work are the contributions of Frenkel, Zhu and Huang. We point the interested reader to the books of Kac~\cite{Kac:1996}, Lepowsky and Li~\cite{LepowskyLi03}, as well as to the monograph~\cite{Frenkel:1993fv} for a thorough introduction to this theory.

\subsection{Conformal field theories (CFTs) in 1+1 dimensions\label{sec:cftintro}}
As any quantum field theory,  a CFT  is determined by its correlation functions. These  are physically interpreted as expectation values of products of basic observable quantities, or quantum fields. They depend continuously on certain parameters, specifying the degrees of freedom of the theory such as position or time. The correlation functions are postulated to transform in a simple manner under symmetry transformations of the theory. CFTs are special examples of quantum field theories possessing a much richer symmetry group than that of relativistic or non-relativistic quantum fields. We begin with a short discussion of relevant background material, focusing on CFTs in two variables. 

\subsubsection*{Historical development of CFTs}
We  point the interested reader to~\cite{2010JMP51a5210F} for an extensive review of the historical development of CFTs. Here we merely summarize key advances with regards to the concepts relevant to our work. We emphasize that we neither claim to nor attempt to provide a complete discussion of this vast subject.

CFTs in two variables (often called CFTs in $1+1$ dimensions) where axiomatized in pioneering works of Belavin, Polyakov and Zamolodchikov~\cite{BPZ84}, Fridan and Shenker~\cite{FRIEDAN:1987ix} as well as by Segal~\cite{segal1988definition}. The first three authors concentrated on analytic properties, starting from symmetry properties and the operator product expansion. The work of the last three authors concentrated on a geometric definition of CFTs. In the work~\cite{segal1988definition} of Segal, such theories were characterized as a functor between two-dimensional surfaces and certain types of trace-class operators on a Hilbert space carrying a unitary representation of the conformal symmetry. Among the first examples were \emph{minimal models}, 
which possess only conformal symmetries and no additional ones~\cite{BPZ84}. The second important class concerns Wess-Zumino-Witten or Wess-Zumino-Witten-Novikov models (WZW)~\cite{wess1971consequences,witten1984non,novikov1981multivalued}, which in addition to the conformal symmetry, also  possess an internal one given by a simple Lie group. Heuristically, it may be defined as a particle moving on the manifold given by the Lie group. Because of their prominence in the literature, and their suitability for this purpose, WZW models will serve as our prime example to illustrate our findings.

After this initial period, Moore and Seiberg~\cite{MooreSeiberg88,Moore:1989cm,MS89} as well as Felder et~al.~\cite{Felder:1989tk,Felder:1990tw} further advanced the understanding of these objects by showing that analytic properties of correlation functions imply many additional identities beyond those given by the axioms. These identities are called braiding and fusion relations and are constrained by the decomposition of the CFT into irreducible representations. These questions were also studied from an operator algebraic point of view by Fr\"ohlich and Gabbiani~\cite{Gabbiani:1993wb}, as well as follow-up work. The operator algebraic language was also applied to study WZW models by Wassermann~\cite{Wassermann:1998cs}, some of whose estimates we will make use of here.

From a purely mathematical side, CFTs were studied in terms of \emph{vertex operator algebras} (VOAs). This algebraic structure was introduced by Borcherds~\cite{Borcherds:1992bi}, as well as Frenkel et al.~\cite{frenkel1989vertex,Frenkel:1993fv}  in their study of the representations of the Monster group. This approach is very algebraic, and shares many similarities with Lie algebra theory. Apart from the Monster group, VOAs have also been used to study minimal models~\cite{Dong:1994um} and WZW theories~\cite{Frenkel:1992jt}. In a long series of contributions, Dong, Frenkel, Huang, Kong, Lepowsky, Zhu   and others
nailed down the properties of VOAs and showed that they are consistent with physically expected properties of a CFT.  This makes VOAs the up-to-date tool for rigorously studying CFTs. VOAs also appear in the axiomatization of CFTs by Gaberdiel and Goddard, which puts the emphasis entirely on the properties of correlation functions~\cite{Gaberdiel:2000br,Gaberdiel:2000tb}. 

Connections between the different approaches to CFTs only lately became clearer, thanks to the work of Carpi et al.~\cite{Carpi:2015fk}, as well as Dong and Lin~\cite{Dong:2014gx}.  In~\cite{Carpi:2015fk}, the authors provided a connection between VOAs and the algebraic quantum field theoretic picture. In~\cite{Dong:2014gx}, the focus was on VOAs and modules which are also Hilbert spaces, and where the Hilbert space structure is compatible with the VOA structure (see Section~\ref{sec:prelim} for precise definitions). Our work is restricted to this case. In the following, we give a short but high-level outline of the setup considered. The exact definitions can  be found in Section~\ref{sec:prelim} of this paper.

\subsubsection*{Setup}
We are concerned with CFTs in two variables. That is, the basic quantum fields~$\cftY(\vphi,z,\bar{z})$ depend on two complex parameters $(z,\bar{z})$. The argument~$\vphi$ serves as a label for the field in question and it is assumed that $\cftY$ is linear in it. Depending on the setting, the variables $(z,\bar{z})$ are elements of different subsets of the complex plane.

It is worth pausing for a moment to explain the physical interpretation of the complex parameters~$(z,\bar{z})$.  Here we distinguish between two settings where such CFTs arise: the first one is relativistic quantum field theory in (1+1)-dimensional Minkowski space. Here the parameters are usually the light-cone variables $(z,\bar{z})=(t-x,t+x)$ and are both real. The second setting of interest (arising in statistical mechanics and  most relevant to condensed matter physics) is the Euclidean setting, where a point $(x,y)\in\mathbb{R}^2$ is identified with the complex number~$z = x + i y$, and where we set $\bar{z} = z^*$ equal to the complex conjugate of~$z$.

The CFT is defined by its correlation functions
\begin{align}
  \left\langle \cftY(\vphi_1,z_1,\bar{z}_1) \cftY(\vphi_2,z_2,\bar{z}_2) \cdots \cftY(\vphi_n,z_n,\bar{z}_n ) \right\rangle \,,
\end{align}
where $\langle \cdot \rangle$ denotes the expectation value in a fixed state of interest (often the so-called `vacuum state').  In physical parlance, a term $\cftY(\vphi_j,z_j,\bar{z}_j)$ is an `insertion' of the field~$\vphi_j$ at the `insertion point' $(z_j,\bar{z}_j)$. The complex parameters $(z_1,\bar{z}_1,\ldots,z_n,\bar{z}_n)$  belong to a certain domain
of a two-dimensional surface parametrized by a subset of complex numbers. We will be concerned with surfaces of genus zero or one, that is, either the compactified complex plane $\C \cup \{\infty\}$ --- the Riemann sphere --- or the torus. The latter can be identified with a subset of the complex plane with periodic boundary conditions.  In this paper, we are mostly concerned with correlation functions evaluated on the real line, with insertion points $z_1,\ldots,z_n$, $\bar{z}_1,\ldots,\bar{z}_n$ separated by a minimal distance.  We note that using appropriate conformal transformations, most configurations of insertion points can be brought into this standard form.

The transition from $2D$-Minkowski or Euclidean space to a compactified two-di\-men\-sional surface is necessary to have a proper, globally defined notion of conformal invariance. Indeed, this is the defining property of a CFT: The correlation functions are invariant with respect to conformal reparametrizations of the insertion points. That is, let $w:\mathbb{C}\cup \{\infty\}\rightarrow\mathbb{C}\cup\{\infty\}$ be a holomorphic map. Then conformal invariance is defined as the existence of real numbers $h_i, \overline{h_i}$ such that 
\begin{align}
  & \left\langle \cftY(\vphi_1,z_1,\bar{z}_1) \cdots \cftY(\vphi_n,z_n,\bar{z}_n ) \right\rangle = 
\\
  &\quad\quad\prod_i \left.\left(\frac{dw}{dz}\right)^{h_i}\right|_{z_i}  \prod_i \left.\left(\frac{dw}{d\bar{z}}\right)^{\overline{h_i}}\right|_{\bar{z}_i}  \left\langle \cftY(\vphi_1,w(z_1),w(\bar{z}_1)) \cdots \cftY(\vphi_n,w(z_n),w(\bar{z}_n) ) \right\rangle\,. 
\end{align}
Stated differently, the correlations functions of certain fields are assumed to transform under conformal reparametrizations in a simple manner. Fields for which this assumptions holds are called \emph{primary}, and they will be of main interest to us. 

The underlying physical reasoning motivating this transformation property of correlation functions depends on the setting under consideration. In the relativistic case, this invariance property can be derived from the Wightman axioms (axiomatizing relativistic quantum field theories) by restricting to one 
spatial dimension and demanding the  existence of conformal symmetries.  We point the interested reader to the book of Kac~\cite{Kac:1996} or the paper by Furlan~et~al.~\cite{furlan1989two} for an explanation of this derivation. In the statistical mechanics (Euclidean) setting, a similar derivation exists, starting from the  Osterwalder-Schrader axioms (axiomatizing Euclidean field theories), see for example the papers by Felder et al.~\cite{Felder:1989tk,Felder:1990tw}.

It is customary to assume that the two variables $z$ and $\bar{z}$ can be decoupled, which can be proven for many CFTs of interest. More precisely, in this case it is sufficient to study basic fields~$\iY(\vphi,z)$ only depending on either $z$ or $\bar{z}$ and the corresponding correlation functions,
\begin{align}
  \langle \iY(\vphi_1,z_1) \cdots \iY(\vphi_n,z_n)\rangle \,.
\end{align}
Such correlations functions are called \emph{chiral} (if depending on~$z$) or \emph{anti-chiral} (if depending on~$\bar{z}$). Again, the same invariance under conformal reparametrizations is assumed. The decoupling of chiral and anti-chiral part is the starting point of a formal axiomatization of chiral CFTs (i.e., those describing chiral correlation functions). In this axiomatization, chiral CFTs are determined by a symmetry algebra $\cV$, which includes the conformal one but may be considerably larger. The exact definition can be found in Section~\ref{sec:prelim}, but from a general viewpoint, $\cV$ can be thought of as a complex vector space, together with a $z$-dependent multiplication rule. The symmetry algebra is called a {\em vertex operator algebra} (VOA), and will be defined below in detail. Correlation functions are specified by maps $\iY $ intertwining three irreducible representations of this symmetry algebra. That is, depending on the variable~$z$, $\iY$~maps an element $a\in A$  in the module~$A$ of~$\cV$ to an endomorphisms $\iY(a,z):B\rightarrow C$ from a module~$B$ into a third module~$C$. Composing these maps gives rise to correlation functions as before, evaluated at elements $a_1,\dots, a_n$ of modules for $\cV$. In terms of VOAs, a full CFT depending on both $z$ and $\bar{z}$ consists of two VOAs, the chiral as well as the anti-chiral part. The corresponding algebraic object is called a conformal full field algebra.  The dependence on both variables is recovered by considering the tensor product of the chiral VOA with the anti-chiral VOA. 

\subsection{Matrix product states and finitely correlated states}
Before stating our main result, we briefly recall basic facts  about matrix product and finitely correlated states. We would like to emphasize that while usually a \emph{state} of a quantum system is normalized, here we are relaxing this requirement due to the fact that correlation functions are a priori not normalized. For the purpose of this paper, a \emph{matrix product state} (MPS), or also a \emph{matrix product tensor network}, is a linear functional $\functionalstate_{MPS}$ on $(\C^d)^{\otimes n}$, defined in terms of $d+1$ many  $D \times D$-matrices~$X,\Gamma_1,\ldots \Gamma_d$ (with integers $d, n, D$), of the form
\begin{align}
  \functionalstate_{MPS}(\ket{k_1} \otimes \ket{k_2} \otimes \ldots \otimes \ket{k_n}) = \tr[\Gamma_{k_1} \Gamma_{k_2} \cdots \Gamma_{k_n} X] \,,
\end{align}
where $\ket{k}$, $k \in \{1,\ldots,d\}$ is an orthonormal basis of $\C^d$. The number $D$ is called the bond dimension. Similarly, a {\em finitely correlated state} (FCS) --- or more precisely a \emph{finitely correlated functional} --- is a functional $\functionalstate_{FCS}$ on $\mathsf{Mat}(\C^d)^{\otimes n}$, where $\mathsf{Mat}(\C^d)$ are the $d\times d$ matrices with complex entries, defined in terms of $d+2$ many $D \times D$-matrices $\rho,e,\Gamma_1,\ldots \Gamma_d$ by
\begin{align}\label{eq:introfcsdef}
  \functionalstate_{FCS}(\ket{k_1}\bra{l_1} \otimes \ket{k_2}\bra{l_2} \otimes \ldots \otimes \ket{k_n}\bra{l_n}) = \tr[\rho \Gamma_{k_{1}} \Gamma_{k_{2}} \cdots \Gamma_{k_n} e \Gamma_{l_n}^* \Gamma_{l_{n-1}}^* \cdots \Gamma_{l_1}^*]\ ,
\end{align}
where $\Gamma^*$ is the adjoint of $\Gamma$.  In our work, the spaces $\C^d$ or $M_d$ will be the linear span of primary vectors.
The function $\functionalstate_{MPS}$ will encode chiral correlation functions, whereas $\functionalstate_{FCS}$ encodes correlation functions of a full CFT. Intuitively, if $D\ll d^n$, then functionals of the form~$\functionalstate_{MPS}$ and~$\functionalstate_{FCS}$ have significantly reduced complexity compared to general functionals on~$(\C^d)^{\otimes n}$ or~$\mathsf{Mat}(\C^d)^{\otimes n}$, respectively. As stated earlier, matrix product and finitely correlated states are widely used in the analysis of quantum spin systems in one spatial dimension. As it turns out, they can also be used to provide finite-dimensional approximations to the correlation functions of CFTs.

\subsection{Results}

Before turning to the main statement of this work, let us summarize our main assumptions. We consider a CFT, either a chiral one or a full CFT consisting of both chiral and anti-chiral parts. The chiral CFT is defined in terms of a VOA $\cV$ and its modules, whereas the full CFT is defined in terms of a pair of VOAs. For the latter, we only consider theories where the chiral and anti-chiral VOAs are isomorphic --- such theories are called \emph{diagonal} in the physics literature~\cite{MS89}. In addition, we are only concerned with VOAs $\cV$ that have finitely many non-isomorphic irreducible modules. Such theories are called \emph{rational}. Moreover, we require the existence of a unique vacuum vector (the exact definition of this statement is found in Section~\ref{sec:prelim}). We need another technical assumption, which is called $C_2$-co-finiteness (to be defined in Section~\ref{sec:prelim}). We further restrict ourselves to the case of \emph{unitary} VOAs and modules, those possessing a scalar product turning them into a Hilbert space (see Section~\ref{sec:unitarity}).  While this seems to be an impressive list of assumptions, we stress that most examples of physical interests are known to satisfy these. In particular, minimal models as well as WZW models do so (see pointers to the literature below).

In a CFT specified by a VOA~$\cV$ satisfying these assumptions, correlation functions are defined in terms of  intertwining maps between modules of~$\cV$. As mentioned earlier, in this paper we are only concerned with correlation functions involving primary fields. This allows us to exploit the  transformation properties under conformal mappings. We point out, however, that more general correlation functions can be obtained from correlation functions involving primary fields by the so-called Ward identities, see e.g.,~\cite{francesco2012conformal}. 
\begin{theorem}[Main result, informal version]\label{thm:main1}
  Consider
\begin{enumerate}[(i)]
\item
a chiral CFT, specified by a VOA~$\cV$ satisfying  our assumptions, or
\item
a full CFT of diagonal kind, containing both chiral and anti-chiral parts, specified by a VOA~$\cV$ satisfying our assumptions. 
\end{enumerate}
Then the 
 genus-0 and genus-1 correlation functions, with equispaced (on $\mathbb{C}$, respectively the torus) insertions of primary fields, can be approximated arbitrarily well by
\begin{enumerate}[(i)]
\item
 matrix-product states in case (i).
\item
finitely correlated states in case (ii).
\end{enumerate}
In both cases, the approximating expression is a certain contraction of a corresponding tensor network. More precisely, in case~(i), it is the value~$\functionalstate_{MPS}(\ket{k_1}\otimes\cdots\otimes\ket{k_n})$ of $\functionalstate_{MPS}$  evaluated on a product input, whereas in case~(ii), it is given by an expression of the form
$\functionalstate_{FCS}(\proj{k_1}\otimes\cdots\otimes\proj{k_n})$, respectively a finite linear combination thereof. 
\end{theorem}
We emphasize that our results provide exact error bounds for the accuracy of these approximations. They are expressed in terms of the number~$n$ of insertion points as well as the minimal distance between the points. These error bounds imply that the bond dimension~$D$ scales polynomially in the inverse of the error in the approximation, as well as sub-exponentially in the inverse of the minimal distance between the insertion points (the so-called {\em ultraviolet cutoff}). In the case of a full CFT, we prove the existence of a transfer operator which determines the long-distance behavior of correlation functions. We illustrate our findings with examples based on WZW models, and also provide an algorithm to compute the matrices $X,\Gamma_1,\ldots, \Gamma_d$  for this case.

\subsection*{Outline}
Let us briefly summarize the structure of this contribution. Section~\ref{sec:prelim} provides a (short) introduction to MPS, FCS and VOAs. In Section~\ref{sec:corfunc}, we study correlation functions of chiral theories, and show how to express these in terms of our main technical tool, \emph{scaled intertwiners}. In  Section~\ref{sec:boundedint}, we show that these objects define bounded operators on the Hilbert space of the chiral theory. This is the main technical ingredient for Section~\ref{sec:approx}, where we present our main arguments providing our approximation estimates in the chiral case. Section~\ref{sec:fullcft} then  studies the case of a CFT consisting of both chiral as well as anti-chiral parts. We end with some open questions and remarks in Section~\ref{sec:outlook}. Some technical lemmas are deferred to Appendix~\ref{app:bounds}, while Appendix~\ref{app:algorithmwzw} provides a description of the algorithm for the computation of the approximation for WZW models.

\section{Preliminaries\label{sec:prelim}}
The purpose of this section is to introduce the necessary terminology as well as notation. Section~\ref{sec:finitelycorrelatedandmatrix} is devoted to finitely correlated and matrix product states. Section~\ref{sec:voasmodulesandintertwiners} contains the necessary definitions related to chiral CFTs. The discussion of the  definitions related to full CFTs is postponed to Section~\ref{sec:voafullfieldalgebras}.

\subsection{Finitely correlated states and matrix product states\label{sec:finitelycorrelatedandmatrix}}
This section is devoted to a brief review of basic definitions related to finitely correlated (FCS)  and matrix product states (MPS). We begin with a discussion of finitely correlated states for translation-invariant systems, and then proceed to formally introduce general finitely correlated states as well as matrix product states. The discussion of how CFT correlation functions of full/chiral CFTs can be written as/approximated by MPS/FCS will be deferred to subsequent sections.
\subsubsection{Finitely correlated states for translation-invariant systems}
Finitely correlated states, introduced by Fannes, Nachtergaele and Werner~\cite{fannesnachtergaelewerner1992,FannesNachtergaeleWerner1994},
 describe trans\-lation-invariant states on a one-dimensional
lattice~$\mathbb{Z}$ with associated Hilbert space~$\cH_{\mathbb{Z}}=\bigoplus_{x\in\mathbb{Z}}\cH_x$. We briefly review some of the relevant facts, following~\cite{Nachtergaelelectures}, to which we refer for more details.
 Here the Hilbert spaces~$\cH_x\cong\cH$
are isomorphic for different sites~$x$, and are
usually assumed to be finite-dimensional. In particular, the
 algebra~$\cA_x\cong \cA=\cB(\cH)$ of single-site observables
is the set $\mathsf{Mat}(\C^d)$ of $d\times d$-matrices. A ($C^*$-)finitely correlated state (FCS)
 with bond-system~$\cB$ (usually a direct sum of matrix algebras, but see the remarks on \cite[p. 7]{Nachtergaelelectures})
then is a triple~$(\cE,\rho,e)$,
where $\cE:\cA\otimes\cB\rightarrow\cB$ is completely positive, $e\in \cB$ is a positive element, and $\rho$ a positive linear functional on~$\cB$.
Setting $\E_A(B)=\cE(A\otimes B)$ for $A\in\cA$ and $B\in\cB$, these objects satisfy the two conditions
\begin{align}
\mathbb{E}_{\idty_\cA}(e)=e\qquad \textrm{ and }\quad  \rho(\mathbb{E}_{\idty_\cA}(B))=\rho(B)\quad\textrm{ for all }B\in\cB\ ,\label{eq:fcsconditions}
\end{align}
where $\idty_\cA$ is the identity element in the algebra $\cA$. In what follows, we will also use the notation $\idty_\cH$ for the identity operator on a Hilbert space $\cH$. In terms of these objects, the finitely correlated state~$\functionalstate_{FCS}$ then is defined by the expression
\begin{align}
\functionalstate_{FCS}(A_1\otimes\cdots\otimes A_n)&=\rho(\mathbb{E}_{A_1}\circ\cdots\circ\mathbb{E}_{A_n}(e))\ .\label{eq:fcsexpectationvalues}
\end{align}
for local observables $A_1,\ldots,A_n$. The expression~\eqref{eq:introfcsdef} can be recovered if the completely positive map $\cE$ is written in Kraus operator form. It can be shown~\cite{AHK78} that if $\cE$ is a unital map, i.e.,
\begin{align}
\cE(\idty_\cA\otimes\idty_{\cB})=\idty_{\cB}\ ,\label{eq:unitalitycedef}
\end{align} then there is always exists
a state~$\rho$ satisfying the second condition in~\eqref{eq:fcsconditions} -- taking $e=\mathsf{id}_{\cB}$ then defines an FCS.  An example
is the case where $\cB=\cB(B)$ for some Hilbert space~$B$, and
the map~$\cE$ is of the form
\begin{align}
\cE(A\otimes B) &=V^*(A\otimes B)V\ \label{eq:purefcsdef}
\end{align}
for a linear map $V:B\rightarrow\cH\otimes B$ which is an isometry, $V^*V = \idty_B$.
Such FCS are called {\em purely generated}.

\subsubsection{General finitely correlated states}
The definition of finitely correlated states can be adapted in a straightforward manner to cover possibly non-translation-invariant states. Here we introduce the corresponding definitions.

We remark that the unitality~\eqref{eq:unitalitycedef} of~$\cE$ is not a necessary condition: in fact, existence of a positive~$e\in\cB$ satisfying the first condition in~\eqref{eq:fcsconditions} is sufficient to guarantee existence of a suitable state~$\rho$. At any rate, conditions~\eqref{eq:fcsconditions} are only necessary to provide a translation-invariant state. Our focus is rather on the functional form of expression~\eqref{eq:fcsexpectationvalues} (and we do not need conditions~\eqref{eq:fcsconditions},~\eqref{eq:unitalitycedef} or the fact that the FCS is purely generated explicitly). Since we do not need the normalization condition, we sometimes also refer to functionals of the form~\eqref{eq:fcsexpectationvalues} as \emph{finitely correlated functionals}.

 More specifically, we  consider non-translation-invariant states on a  system consisting of $n$ sites. For bond systems $\cB^{(0)},\ldots,\cB^{(n)}$
and completely positive operators $\cE^{(j)}:\cA^{(j)}\otimes \cB^{(j)}\rightarrow\cB^{(j-1)}$,
a positive element~$e\in \cB^{(n)}$, and a linear functional~$\rho$ on $\cB^{(0)}$, we may define a linear functional~$\functionalstate_{FCS}$ by
\begin{align}
\functionalstate_{FCS}(A_1\otimes\cdots\otimes A_n)=\rho\left(\mathbb{E}^{(1)}_{A_1}\circ\cdots\circ\mathbb{E}^{(n)}_{A_n}(e)\right)\label{eq:fcsnontranslationinvariant}
\end{align} 
where $A_j\in\cA^{(j)}$ and $\mathbb{E}^{(j)}_{A_j}(B_j)=\cE^{(j)}(A_j\otimes B_j)$  for $j=1,\ldots,n$. Note that we do not impose
constraints on positivity in this definition. As before, in the case where $\cA^{(j)}=\cB(\cH^{(j)})$, $\cB^{(j)}=\cB(\cK^{(j)})$ for Hilbert spaces $\cH^{(j)}$, $\cK^{(j)}$, we call the FCS purely generated
if the maps~$\cE^{(j)}$ have the form~\eqref{eq:purefcsdef}
for linear maps~$V^{(j)}:\cK^{(j-1)}\rightarrow\cH^{(j)}\otimes\cK^{(j)}$.
\subsubsection{Matrix product states}
Matrix product states are obtained by considering a vector-analog of~\eqref{eq:fcsnontranslationinvariant}, as follows. Let $A^{(j)}, B^{(j)}$ be Hilbert spaces and $\cW^{(j)}:A^{(j)}\otimes B^{(j)}\rightarrow B^{(j-1)}$ be linear operators. Define $\mpsA^{(j)}_{a_j}(b_j)=\cW^{(j)}(a_j\otimes b_j)$ for $a_j\in A^{(j)}$ and $b_j\in B^{(j)}$.
Denoting by $\pairL\cdot,\cdot \pairR$ the pairing between $\cB^{(0)}$ and its dual space, we consider functionals on $\otimes_{j=1}^n A_j$ of the form
\begin{align}
\functionalstate_{MPS_0}(a_1\otimes\cdots\otimes a_n)&=\pairL (v^{(0)})',\mpsA^{(1)}_{a_1}\circ\cdots\cdots\circ \mpsA^{(n)}_{a_n}(v^{(n)})\pairR\label{eq:MPSstatex}
\end{align}
for $(v^{(0)})'\in (\cB^{(0)})'$, $v^{(n)}\in \cB^{(n)}$ as well as (in the case where $\cB^{(0)}=\cB^{(n)}$)
\begin{align}
\functionalstate_{MPS_1}(a_1\otimes\cdots\otimes a_n)&=\tr\left(\mpsA^{(1)}_{a_1}\circ\cdots\cdots\circ \mpsA^{(n)}_{a_n} X\right)\label{eq:genusoneMPS}
\end{align}
for suitable operators~$X$.

As a familiar special case of~\eqref{eq:genusoneMPS}, consider the case where
$A^{(j)}\cong \mathbb{C}^d$, $B^{(j)}\cong \mathbb{C}^{b}$ for all~$j$. Fixing an orthonormal basis $\{\ket{k}\}_{k=1}^d$ of $\cA^{(j)}$,
each operator~$\cW^{(j)}$ has the form
\begin{align*}
\cW^{(j)}(a_j\otimes b_j)&=\sum_{k=1}^d \langle k,a_j\rangle \mpsA^{(j)}_{k} b_j\ .
\end{align*}
for a family~$\{\mpsA^{(j)}_k\}^d_{k=1}$ of linear operators $\mpsA^{(j)}_k:\C^b \to \C^b$. Then~\eqref{eq:genusoneMPS} becomes
\begin{align}
\functionalstate_{MPS_1}(\ket{k_1}\otimes\cdots\otimes\ket{k_n})&=\tr\left(\mpsA_{k_1}^{(1)}\cdots\mpsA_{k_n}^{(n)}X\right)\label{eq:mpsdef}
\end{align}
for all $1\leq k_1,\ldots,k_n\leq d$. An element in $\C^{dn} = (\C^{d})^{\otimes n}$ corresponding to the linear form~\eqref{eq:mpsdef} on $(\C^d)^{\otimes n}$ is called a matrix-product state (MPS) with bond dimension equal to~$b$. Note that because we did not require a proper normalization, the term ``state'' should be considered literally. More precisely, we call functionals of the form~\eqref{eq:mpsdef} \emph{matrix product tensor networks}. 
\subsection{$G$-invariant MPS\label{sec:ginvariantmpsdescription}}
Of particular interest to us will be MPS/FCS with additional local symmetries.
Given  a linear map $\cW:\mathbb{C}^{d_3}\otimes \mathbb{C}^{d_2}\rightarrow\mathbb{C}^{d_1}$ and unitary representations
$U_j:G\rightarrow\mathsf{GL}(\mathbb{C}^{d_j})$, \(j=1,2,3\), of a compact Lie group $G$, we say that $\cW$ is a $G$-intertwining map if 
\begin{align}
  \cW \,(U_3(g) \otimes U_2(g)) = U_1(g) \,\cW\,,\quad \textrm{ for } g\in G\,.
\end{align}
 A $G$-invariant MPS is defined by operators $\cW^{(j)}:A^{(j)}\otimes B^{(j)}\rightarrow B^{(j-1)}$ that are $G$-intertwining with respect to fixed unitary representations of $G$ on the spaces $A^{(j)}$, $B^{(j)}$.

Consider the translationally invariant case, $\cW^{(j)} = \cW$ for all $j = 1,\ldots,n$. Here we assume $A^{(j)}\cong\mathbb{C}^d$, $B^{(j)}=\mathbb{C}^d$, and that we are given two unitary representations $U_1:G\to\mathsf{GL}(\C^d)$, $U_2:G\to\mathsf{GL}(\C^b)$ of some compact Lie group $G$.
The translation-invariant MPS is then given by a linear map $\cW$ satisfying
\begin{align}
  \cW \,(U_1(g) \otimes U_2(g)) = U_2(g) \,\cW\,,\quad \textrm{ for } g\in G\,.
\end{align}
It then follows from this intertwining property, the unitarity of the representation and Eq.~\eqref{eq:mpsdef} that the MPS possesses a global $G$-invariance in the sense that
\begin{align}
  \functionalstate_{MPS_1}(U_1(g)\ket{k_1}\otimes\cdots\otimes U_1(g)\ket{k_n}) = \functionalstate_{MPS_1}(\ket{k_1}\otimes\cdots\otimes\ket{k_n})\,,\quad \textrm{ for all } \,g\in G\,.
\end{align}
Similarily, if we have
\begin{align}
  \E_{U_1(g)^*A U_1(g)}(U_2(g)^* B U_2(g)) = U_2(g)^* \E_A(B) U_2(g) \,,\textrm{ for } g\in G\,
\end{align}
and the positive functional $\rho$ is invariant under the representation $U_2$, it follows that the corresponding FCS satisfies
\begin{align}
  \functionalstate_{FCS}(\textrm{ad}_{U_1(g)}(A_1)\otimes\cdots\otimes \textrm{ad}_{U_1(g)}(A_n)) = \functionalstate_{FCS}(A_1\otimes\cdots\otimes A_n)\,,
\end{align}
for all local observables $A_1,\ldots,A_n$, where $\textrm{ad}_{U_1(g)}(A) = U_1(g)^* A U_1(g)$. We call such functionals $G$-invariant MPS/FCS. Starting from the work by~\cite{FannesNachtergaeleWerner1994}, this class of MPS/FCS was extensively studied, see in particular~\cite{schuch2011classifying}. It provides a very efficient way to construct trial states for quantum spin chains with built-in physical symmetries.

In the following, we will show that genus-zero correlation functions
of chiral CFTs can be written in the form~\eqref{eq:MPSstatex}
for a functional $\functionalstate_{MPS_0}$,
whereas genus-one correlation functions of such CFTs take the MPS form~\eqref{eq:genusoneMPS} for a functional $\functionalstate_{MPS_1}$. Finally, genus-$0$ correlation functions of full CFTs are given by functionals $\functionalstate_{FCS}$ of the form~\eqref{eq:fcsnontranslationinvariant},
and genus-$1$ correlation functions will be certain linear combinations thereof. We will ilustrate our findings by WZW models, and the corresponding approximations will turn out to be closely related to $G$-invariant MPS/FCS.


\subsection{Vertex operator algebras, modules and intertwiners\label{sec:voasmodulesandintertwiners}}
We will work in the language of vertex operator algebras (VOAs), their modules and intertwiners. We will introduce these concepts only to the extent necessary for our purposes and refer to~\cite{Frenkel:1993fv} for details. For a discussion of how VOAs axiomatize~$2$-dimensional CFT, and in particular, the relationship to the Wightman axioms, see e.g.,~\cite{Kac:1996}.

\subsubsection{Vertex operator algebras (VOAs)}
A vertex operator algebra~$\cV$ is a tuple $(\cV,\Y,\1,\omega)$ consisting of an $\mathbb{N}_0$-graded vector space $\cV=\bigoplus_{n\in\mathbb{N}_0}\cV_n$, a linear map $\Y(\cdot,z):\cV\rightarrow\mathsf{End}(\cV)[[z,z^{-1}]]$ into the space of formal Laurent series with coefficients in $\mathsf{End}(\cV)$, and two distinguished vectors $\1\in\cV_0$ and $\omega\in\cV_2$. To state the conditions obeyed by these objects, the following terminology is convenient: The vector~$\1$ is called the {\em vacuum}, whereas $\omega$ is referred to as the {\em conformal} or {\em Virasoro vector}. Each space~$\cV_n$ is called a {\em weight space}. A vector $v\in\cV_n\subset \cV$ belonging to a weight space~$\cV_n$ is {\em homogeneous} of  {\em weight} (or {\em level}) $\wt v =n$.

By definition, the {\em vertex operator} $\Y(v,z)$ associated with a vector $v\in\cV$ can be written as
\begin{align}
\Y(v,z)=\sum_{n\in\mathbb{Z}}v_nz^{-n-1}\ , \label{eq:firstconventionmodeops}
\end{align} where $v_n\in \End(\cV)$ is referred to as a {\em mode operator}  of~$v$. For all $u,v\in \cV$, these satisfy
\begin{align}
v_n u&=0\qquad\textrm{ for } n \textrm{ sufficiently large.}
\end{align}
For a homogeneous vector~$v$, we use the notation
\begin{align}
\Y(v,z)=\sum_{n\in\mathbb{Z}} \y(v)_n z^{-n-\wt v}\label{eq:yvzweightdef}
\end{align}
instead of~\eqref{eq:firstconventionmodeops}, i.e., we index mode operators as~$\y(v)_n=v_{n+\wt v-1}$. This convention is motivated by Eq.~\eqref{eq:weightchangehomogeneous} below.
The vacuum vector satisfies
\begin{align}
\Y(\1,z)&=\mathsf{id}_{\cV}\ ,\label{eq:vacuumidentityoperatoraxiom}
\end{align}
where $\mathsf{id}_{\cV}$ is the identity operator on~$\cV$, and  the
 {\em creativity} property
\begin{align}
 \Y(u,z)\1=u+\sum_{n\in \mathbb{N}} \tilde{u}_n z^n\qquad\textrm{ for some }\tilde{u}_n\in\cV\ , \label{it:creativity}
\end{align}
 which is sometimes written $\lim_{z\rightarrow 0}\Y(u,z)\1=u$
 and colloquially known as the operator-state correspondence. We remark that in the latter limit, the formal indeterminate~$z$ is replaced by a complex number~$z\in\mathbb{C}$, a procedure we will discuss in more detail and use extensively below when considering correlation functions. However, in the present section, such substitutions are not necessary and every identity is to be understood as 
an identity between formal Laurent series.

  For the conformal vector~$\omega$, which
is homogeneous of weight~$\wt \omega=2$, the mode operators~$\y(\omega)_n$ are denoted by $L_n$, i.e.,
\begin{align}
\Y(\omega,z)&=\sum_{n\in\mathbb{Z}}L_nz^{-n-2}\ .
\end{align}
The operators~$\{L_n\}_{n\in\mathbb{Z}}$ are sometimes called the {\em Virasoro operators}. 

Every weight space $\cV_n$ is finite-dimensional, and the grading of~$\cV$ is given by the spectral decomposition of the operator  $L_0$: for every $n\in\mathbb{N}_0$,  $\cV_n$ is the eigenspace of $L_0$ with eigenvalue~$n$.   A homogeneous vector (or ``field'') $v\in\cV_n$ is {\em quasi-primary} if $L_1u=0$ and {\em primary}  if $L_nv=0$ for all $n>0$.

The operators $L_n$ satisfy the Virasoro algebra relations
\begin{align}
[L_n,L_m]=(n-m)L_{n+m}+\frac{c}{12}n(n^2-1)\delta_{n+m,0}\cdot \mathsf{id}_{\cV}\qquad\textrm{ for all }m,n\in\mathbb{N}_0\  ,\label{eq:virasoroalgebracond}
\end{align}
where the constant~$c$ is the central charge (sometimes denoted $\rank \cV$).

Using additional indeterminate variables, one can define  products of vertex operators in terms of formal series.
Interpreted as such, a product such as~$Y(u,z_1)Y(u,z_2)$ is an element of~$\End(\cV)[[z_1,z_1^{-1},z_2,z_2^{-1}]]$.  A VOA satisfies the following  {\em locality} or {\em weak commutativity} property with respect to such products:
  for all $u,v\in\cV$, there is a non-negative integer $k$ such that
\begin{align}
(z_1-z_2)^k[\Y(u,z_1),\Y(v,z_2)]=0\label{it:locality}\ .
\end{align}
In fact, this condition implies  the so-called {\em Jacobi identity},
which is often used instead as it is more explicit (see e.g.,~\cite[Section 1.4]{LepowskyLi03} for a discussion of the relationship between different definitions). For completeness, we include the latter:
it states that (see e.g.,~\cite[Eq. (1.2.14)]{Frenkel:1992jt})
\begin{align}\label{eq:jacobiidentity}
&\Residuum_{z_1-z_2}\left(\Y(\Y(u,z_1-z_2)v,z_2)(z_1-z_2)^m\iota_{z_2,z_1-z_2}(z_2+(z_1-z_2))^n\right)=\\
&\quad\Residuum_{z_1}\left(\Y(u,z_1)\Y(v,z_2)\iota_{z_1,z_2}(z_1-z_2)^mz_1^n\right)
-\Residuum_{z_1}\left(\Y(v,z_2)\Y(u,z_1)\iota_{z_2,z_1}(z_1-z_2)^mz_1^n\right)\,,
\end{align}
for all $m,n\in\mathbb{Z}$, $u,v\in\cV$. In this expression,
$\Residuum_z f(z)$  is the residue of~$f(z)$, i.e., the coefficient of $z^{-1}$ in $f(z)$, and $\iota_{z_1,z_2}f(z_1,z_2)$ is the series expansion of the function $f(z_1,z_2)$ in the domain $|z_1|>|z_2|$. In addition, a VOA has the {\em translation property}: for any $v\in\cV$, we have
\begin{align}
\frac{d}{dz}\Y(v,z)=\Y(L_{-1}v,z)\ , \label{it:translation}
\end{align}
where the lhs.~is to be understood as the Laurent series obtained by termwise differentiation. 
This concludes the definition of a VOA. Important consequences of these axioms are e.g., the associativity property
\begin{align}
(z_1+z_2)^k \Y(u,z_1+z_2)\Y(v,z_2)w&= (z_1+z_2)^k \Y(\Y(u,z_1)v,z_2)w\ \label{eq:associativitydef}
\end{align}
for large enough $k$, which can be seen as the VOA-version of the operator product expansion of fields, as well as the commutator formula
\begin{align}
[\Y(u,z_1),\Y(v,z_2)]&=\Y((\Y(u,z_1-z_2)-\Y(u,-z_2+z_1))v,z_2)\ .\label{eq:commutatorformula}
\end{align}
We also remark that the subalgebra generated by $\{L_0,L_1,L_{-1}\}$ generates an action of $SL(2,\mathbb{C})$  on the formal variable~$z$ by M\"obius transformations. This is a consequence of the Virasoro algebra relations~\eqref{eq:virasoroalgebracond} and the translation property~\eqref{it:translation}. Explicitly, we have
\begin{align}
q^{L_0} \Y(u,z) q^{-L_0}&=\Y(q^{L_0}u,qz)\label{eq:dilationintegrated}\\
q_\lambda^{L_{-1}} \Y(u,z) q_\lambda^{-L_{-1}}&=\Y(u,z+\lambda)\qquad\textrm{ where }q_\lambda=e^\lambda\\
q_\lambda^{L_1}\Y(u,z)q_\lambda^{-L_1}&=\Y\left(q^{L_1}_{\lambda(1-\lambda z)}(1-\lambda z)^{-2L_0} u,\frac{z}{1-\lambda z}\right)\  .\label{eq:l1integrated}
\end{align}
More generally, an element
\begin{align*}
A=\left(\begin{matrix}
a & b\\
c & d
\end{matrix}\right)\in SL(2,\mathbb{C})\ ,
\end{align*}
corresponding to the M\"obius transformation
\begin{align}
\gamma(z)&=\frac{az+b}{cz+d}\ ,
\end{align}
acts by
\begin{align}
  D_\gamma \Y(v,z)D_\gamma^{-1}&=\Y\left(\left(\frac{d\gamma}{dz}\right)^{L_0}\exp\left(\frac{\gamma''(z)}{2\gamma'(z)}L_1\right)v,\gamma(z)\right)\ ,\label{eq:generalmoebiustrsfaction}
\end{align}
and $\gamma\mapsto D_\gamma$ defines a representation of $SL(2,\mathbb{C})$ on $\cV$.  The maps~\eqref{eq:dilationintegrated}--\eqref{eq:l1integrated} are generated by the elements
\begin{align*}
\cL_{0}=\left(
\begin{matrix}
1/2 & 0\\
0 & -1/2
\end{matrix}\right)\qquad
\cL_{-1}=\left(
\begin{matrix}
0 & 1\\
0 & 0
\end{matrix}\right)\qquad
\cL_{1}=\left(
\begin{matrix}
0 & 0\\
1 & 0
\end{matrix}\right)\ ,
\end{align*}
of $\mathfrak{sl}(2,\mathbb{C})$. We also point out that
the vacuum is invariant under M\"obius transformations, i.e.,
\begin{align}
D_\gamma \1&=\1\qquad\textrm{ for all }\gamma\ ,\label{eq:vacuuminvariancemoeb}
\end{align}
since $L_0\1=0$ (as $\1\in\cV_0$), $L_1\1=0$ since by~\eqref{eq:weightchangehomogeneous}
\begin{align}
\wt(L_1\1)=\wt \1-1<0
\end{align}
hence $L_1\1\in \cV_{-1}=0$ by the assumption $\cV_n=0$ for $n<0$,
and $L_{-1}\1=0$. The latter is an immediate consequence of~\eqref{eq:vacuumidentityoperatoraxiom} and the translation property~\eqref{it:translation}. Infinitesimally, the  relationships~\eqref{eq:dilationintegrated}--\eqref{eq:l1integrated} read
\begin{align}
[L_0,\Y(v,z)]&=\Y(L_0v,z)+z\Y(L_{-1}v,z)\label{eq:hzerocommutator}\\
[L_{-1},\Y(v,z)]&=\Y(L_{-1}v,z)\label{eq:hmonecommutator}\\
[L_1,\Y(v,z)]&=\Y(L_1v,z)+2zY(L_0v,z)+z^2\Y(L_{-1}v,z)\ .\label{eq:honecommutator}
\end{align}
A consequence of~\eqref{eq:hzerocommutator}, which is particularly relevant for our purposes is the following. The mode operators associated with a homogeneous vector $w\in\cV$ map weight spaces to weight spaces, and change the weight according to
\begin{align}
\wt (\y(w)_n v)&=\wt v-n\ \label{eq:weightchangehomogeneous}
\end{align}
for any homogeneous vector $v\in\cV$.

We will require a few additional technical assumptions on the VOAs to derive our results. These are expressed by  the following definitions: A VOA~$\cV$ is called of {\em CFT-type} if the weight space $\cV_0=\mathbb{C}\1$ is one-dimensional, i.e., spanned by the vacuum vector.  It is {\em rational} if every admissible $\cV$-module (as defined below) is completely reducible, i.e., a direct sum of irreducible admissible~$\cV$-modules. It is $C_2$-co-finite if the space~$C_2 = \mathsf{span} \{u_{-2}v\ |\ u,v\in\cV\}$ has finite co-dimension in~$\cV$, \(\dim \cV/C_2 <\infty\). We will also assume that the weight spaces are finite-dimensional, i.e., $\dim \cV_n<\infty$. We defer the discussion of the additional property of unitarity of VOAs to Section~\ref{sec:unitarity}, where we also discuss some consequences. Many examples of physical interest satisfy all these conditions. In particular, a large class of minimal models as well as lattice models and WZW models are of CFT-type, rational, unitary, and $C_2$-co-finite (see~\cite{Dong:2014gx} and the references therein).

\begin{example}[\bf Wess-Zumino-Witten (WZW) models]
Our prime example in this paper will be CFTs  of Wess-Zumino-Witten type~\cite{wess1971consequences,witten1984non}, which are built upon a local symmetry action by a compact Lie group. However, the CFT is most easily understood in the differential picture, that is, from a Lie algebra viewpoint. 

Following a first rigorous treatment of these theories by Kanie and Tsuchiya~\cite{tsuchiya1987vertex} in the case of the Lie algebra $\mathfrak{sl}(2,\C )$, and the work by Wassermann~\cite{Wassermann:1998cs} for general elements of the A-series, the corresponding VOA structure, the modules and the relations between these have been very nicely constructed and studied by Frenkel and Zhu~\cite{Frenkel:1992jt}. We will follow their approach, and try to illustrate our statements for these physically relevant examples.

The basic building block is a simple Lie algebra $\g $ over the complex numbers (from now on, all Lie algebras are considered to be complex, unless otherwise stated), with normalized\footnote{We usually take the normalization to be $(\theta,\theta) = 2$, where $\theta$ is the maximal root of $\g $.} Killing form~$(\cdot,\cdot):\g \times \g \to \C $. We then consider its affinization, which is the Lie algebra $\ga = \g \otimes \C[t,t^{-1}] \oplus \mathsf{k}\C $, where $\mathsf{k}$ is an element of the center of $\ga$, and $\C[t,t^{-1}]$ is the algebra of Laurent polynomials in the formal variable~$t$. The Lie bracket ($ \fa,\fb \in \g $, $n,m \in \Z$) is
\begin{align}\label{eq:wzwcommutator}
  [\fa \otimes t^n,\fb \otimes t^m] = [\fa,\fb]\otimes t^{n+m}+ \delta_{n+m,0}\,n\,(\fa,\fb)\mathsf{k} \,.
\end{align}
In the following, we adopt the standard notation $\fa \otimes t^n = \fa(n)$, and consider the decomposition of $\ga$ into the subalgebras
\begin{align}
  \ga_{+} = \g \otimes \C[t]t\,,\quad \ga_{-} = \g \otimes \C[t^{-1}]t^{-1}\,,\quad \ga = \ga_{+} \oplus \ga_{-} \oplus \g \oplus \C \mathsf{k}\  .
\end{align}
Let $\Uga$ be the universal enveloping algebra of $\ga$. It follows that the quotient of~$\Uga$ by~$\ga_{+} $ and~$\mathsf{k} - k \idty_{\Uga}$, $k \in \mathbb{N}$ is a natural $\ga$-module. Its elements are of the form
\begin{align}
  \fa_1(-n_1)\fa_2(-n_2) \cdots \fa_l(-n_\ell)\idty_{\Uga} \,, \quad \fa_j(-n_j) \in \ga_{-}\,,\quad \textrm{ for } \,j=1,\ldots,\ell
\end{align}
and we define the vacuum vector by $\1 = \idty_{\Uga}$. The space has a natural grading defined by
\begin{align}
  \wt(\fa_1(-n_1)\fa_2(-n_2) \cdots \fa_\ell(-n_\ell) \1) = \sum_{j=1}^\ell n_j \,,\label{eq:wzwgradingdef}
\end{align}
Moreover, this module (which is isomorphic to $\Ugam$ as a vector space) has a maximal proper submodule. This allows us to again  take the quotient with respect to that submodule. The resulting $\ga$-module is denoted by~$\voalkzero$. In the following, we routinely identify elements in~$\voalkzero$ with corresponding representatives in~$\Ugam$.

As shown by Frenkel and Zhu~\cite{Frenkel:1992jt}, the space $\voalkzero$ has a VOA structure if~$-k$ is not equal to the dual Coxeter number $\mathsf{g}$ of~$\g $. We will assume this in the following. The Virasoro operators~$\{L_n\}_{n\in\mathbb{Z}}$ are given by the Sugawara-Segal construction  and are compatible with the grading~\eqref{eq:wzwgradingdef}, see~\cite{Frenkel:1992jt} or~\cite[Chapter 15.2]{francesco2012conformal}.  For $\fa \in \g \simeq \g \otimes t^{-1} $, the associated vertex operator is 
\begin{align}\label{eq:vertexopwzw}
  \Y(\fa,z) = \sum_{n \in \Z} \fa(n) \,z^{-n-1} \,.
\end{align}
Given an element of the form $v = \fa_1(-n_1) \fa_2(-n_2) \dots \fa_\ell(-n_\ell)\1 \in \voalkzero$, $\fa_j \in \g $,  the associated vertex operator is defined via the Jacobi identity, which can be used  to successively reduce the definition to elements in~$\g $, see~\cite{Frenkel:1992jt} for details. This definition is then linearly extended to all of~$\voalkzero$.  A more physical way of thinking of the  object~$\voalkzero$ is as the Fock space over the space~$\g \otimes \C[t^{-1}]$ of polynomials in $t^{-1}$ with values in the Lie algebra $\g $, modulo the relations imposed by the commutator rule~\eqref{eq:wzwcommutator}, and the identity playing the role of the vacuum. 

Frenkel and Zhu~\cite{Frenkel:1992jt} proved that, if we choose $k$ (which is often referred to as the {\em level}) to be a positive integer --- which we will do from this point onwards --- then the VOA~$\voalkzero$ is simple, $C_2$-co-finite and rational, that is, satisfies our basic assumptions and only possesses finitely many irreducible representations. This will be our prime example to illustrate our findings.   

\end{example}

\subsubsection{Modules\label{sec:introVOAmodules}}
A {\em module} of a VOA is a vector space carrying a structure satisfying almost all defining properties of a VOA, as well as certain compatibility properties with the VOA. For a VOA~$(\cV,\Y,\1,\omega)$, a module $(A,\Y_A)$ is again a graded vector space~$A=\bigoplus_{n\in\mathbb{N}_0}A_n$, together with a linear map
\begin{align}
  \Y_A(\cdot,z):\cV\rightarrow\End(A)[[z,z^{-1}]]\,,\quad\,\Y_A(v,z) = \sum_{n \in \Z} v_n^A z^{-n-1}\,,
\end{align}
where $v_n^A \in \End(A)$ is again called the \emph{mode operator} of $v \in \cV$. We call $A_0$ the \emph{top level}, and $A_n$ the $n$-th level of the module~$A$. (Here we follow the convention of~\cite{Frenkel:1992jt} and assume that the grading is over the non-negative integers~$\mathbb{N}_0=\mathbb{N}\cup\{0\}$. Such $\mathbb{N}_0$-gradable modules are also called {\em admissible}, see e.g.,~\cite[Section~6]{MasonTuite:2010}.) 

Homogeneous vectors and mode operators are  defined analogously for modules as for VOAs. Weights are defined as eigenvalues of $L_{A,0}$, which, in contrast to the case of VOAs do not need to be integers: they are of the form $\alpha+n$, where $\alpha\in I_A$ for some finite  set $I_A\subset\mathbb{C}$ and $n\in\mathbb{N}_0$. More precisely,  for every $n\in\mathbb{N}_0$, we have for all $a\in A_n$
\begin{align}
L_{A,0}a&=(\alpha+n)a\qquad\textrm{ for some } \alpha\in I_A\ .\label{eq:formofweightsmodule}
\end{align}
According to~\eqref{eq:formofweightsmodule}, we can refine the grading
of the module to
\begin{align*}
A&=\bigoplus_{\alpha\in I_A, n\in\mathbb{N}_0}A_{n,\alpha}\ \qquad\textrm{ where }\qquad A_n=\bigoplus_{\alpha\in I_A} A_{n,\alpha}
\end{align*}
and elements of $A_{n,\alpha}$ are weight vectors~$a$ with weight $\wt a=\alpha+n$.
In other words, the grading again coincides with the spectral decomposition, i.e., the eigenspaces of $L_{A,0}$, but the latter provides more detail. For an irreducible module $A=\bigoplus_{n\in\mathbb{N}_0}A_n$, the set ~$I_A=\{h\}$ consists of a single scalar (called the {\em conformal weight} or {\em highest weight} of $A$), and thus
\begin{align}
L_{A,0}|_{A_n}=(h+n)\mathsf{id}_{A_n}\qquad\textrm{ for all }n\in\mathbb{N}_0\ .\label{eq:irreduciblemodule}
\end{align}
For a homogeneous vector $v \in A_n$ we again use the notation
\begin{align}
  \Y_A(v,z) = \sum_{m \in \Z} \y_A(v)_m z^{-m-n}\,,
\end{align}
such that $\y_A(v)_m = v^A_{m+n-1}$ and $\wt(\y_A(v)_m w) = \wt w - m$, for any homogeneous vector $w \in A$.

As shown by Gaberdiel and Nietzke~\cite[Proposition 10 and subsequent comment]{GaberdielNeitzke}, as well as by Karel and Li~\cite{karel1999certain}, the weight spaces $A_n$ of an irreducible module of a VOA satisfying the $C_2$-co-finiteness condition are finite-dimensional. More precisely, their dimension is bounded by (see~\cite{buhl2008ordered} for a nice discussion of results of this kind)
\begin{align}
\dim A_n \leq (\dim A_0)\cdot P(n,\dimctwo)\,, \label{eq:dimweightspaces}
\end{align} 
where $\dimctwo = \dim \cV/C_2$, $C_2=\mathsf{span}\{u_{-2}v\ |\ u,v\in\cV\}$ as before, and $P(n,\dimctwo)$ is the number of  $\dimctwo$-component multi-partitions\footnote{A partition~$\mu$ of an integer $n$ is a set of integers $\mu = \{\mu^1,\ldots,\mu^l\}$ which sum up to $n$, $|\mu| = \sum_{i=1}^l \mu^i = n$. An $m$-component multi-partition of an integer $n$ is a set $\{\mu_1,\ldots,\mu_m\}$ of $m$ partitions such that $\sum_{i=1}^m |\mu_i| = n$.} of the integer $n$.  As this number will later determine the size of the bond dimension of our approximation, we give a bound on its growth behavior in Lemma~\ref{lem:app:boundpartitionfunction}. 

The vertex operators $\Y_A(u,z)$ (and the mode operators $L_{A,n}$ of $\Y_A(\omega,z)$) satisfy all axioms of a VOA with the exception of the creativity property~\eqref{it:creativity}.  Various consequences derived for VOAs continue to hold for modules, with suitable replacements when vertex operators are applied inside arguments: for example, the commutator formula~\eqref{eq:commutatorformula} becomes
\begin{align}
[\Y_A(u,z_1),\Y_A(v,z_2)]&=\Y_A((\Y(u,z_1-z_2)-\Y(u,-z_2+z_1))v,z_2)\ .\label{eq:commutatorformulamodules}
\end{align}
Following the literature, we will often suppress the index~$A$ when it is clear from the context.

\paragraph{Special modules.}
 Given a VOA $(\cV,\Y,\1,\omega)$, it is clear that $(\cV,\Y)$ is a module for $\cV$. This is called the {\em adjoint module}. Given a module $A=\bigoplus_{n\in\mathbb{N}_0}A_n$ for a VOA~$\cV$, the restricted dual space $A^\prime=\bigoplus_{n\in\mathbb{N}_0}A_n^\prime$  (i.e., the space of linear functionals on~$A$ vanishing on all but finitely many~$A_n$) can be given a  $\cV$-module structure as well~\cite{Frenkel:1992jt}. The module operator is defined by 
\begin{align}
  \Y_{A^\prime} \,:\,\cV \to \End(A^\prime)[[z,z^{-1}]]\,,\quad \Y_{A^\prime}(v,z)(a')(a) = a'(\Y_A(e^{L_1}(-z^{-2})^{L_0} v,z^{-1})a)\,,
\end{align}
for $v \in \cV$, $a' \in A^{\prime}$ and $a \in A$. This is called the \emph{contragredient} module.

\begin{example}[\bf WZW modules]\label{exmp:wzwmodules}
As before, fix a simple Lie algebra $\g $ and an integer~$k\in\mathbb{N}$. Let~$V$ be a module for $\g $. Then there is a canonical module $V_{k}$ for the affinization~$\ga$ extending $V $ by setting $\ga_{+} V = 0$, $\mathsf{k} = k \idty_{V}$ and defining
\begin{align}
  V_k = \Uga \otimes_{\Ugap \oplus \g \oplus k\idty} V \,.
\end{align}
(This notation indicates that~$V_k$ is the quotient of $\Uga\otimes V$ 
by~$\Ugap \oplus \g \oplus k\idty$.)
 Elements of this space are spanned by vectors of the form
\begin{align}
  \fa_1(-n_1)\fa_2(-n_2) \cdots \fa_\ell(-n_\ell) \vphi \,, \quad \fa_j(-n_j) \in \ga_{-}\,,\;\vphi \in V\, .
\end{align}
We can define a natural grading on $V_k$ by setting
\begin{align}
  \wt(\fa_1(-n_1)\fa_2(-n_2) \cdots \fa_\ell(-n_\ell) \vphi) = \sum_{j=1}^\ell n_j + h_\vphi\,,\label{eq:vkgradingdef}
\end{align}
where $h_\vphi\vphi = L_0 \vphi$, and the Virasoro operators are again given by the Sugawara-Segal construction and are compatible with the grading~\eqref{eq:vkgradingdef}. If $V$ is an irreducible $\g $-module of highest weight $\lambda$, the associated $\ga$-module~$V_k$ has a maximal proper submodule, and we denote the quotient of~$V_k$ by this submodule by~$\moduleklambda$. The space~$\moduleklambda$ turns out to be  an irreducible module for the VOA~$\voalkzero$ (for certain weights $\lambda$ as specified below). To avoid an accumulation of indices, we will denote the $\mathbb{N}_0$-grading of the module as
\begin{align*}
\moduleklambda=\bigoplus_{n\in\mathbb{N}_0}\moduleklambda(n)\ .
\end{align*}
Again, elements in $\moduleklambda(n)$ are said to belong to the level $n$, or the top level if $n=0$.
 
   Eq.~\eqref{eq:vertexopwzw}  defines a module vertex operator by identifying~$\fa(n)$ with its action on~$\moduleklambda$, i.e. $\y_{\moduleklambda}(\fa)_n = \fa(n)$. The action of module operators on elements composed of products of the building blocks $\fa(n)$ is defined via the Jacobi identity. The corresponding Fock space analogy is the same as before, but now the vacuum is replaced by the highest weight vector corresponding to $\lambda$. Accordingly, the top level~$\moduleklambda(0)$ of the $\voalkzero$-module $\moduleklambda$ (consisting of vectors with smallest eigenvalues of~$L_0$) is isomorphic to the $\g $-module~$V$. For the irreducible modules $\moduleklambda$, the action of the grading operator $L_0$ on the top level $\moduleklambda(0)$ is a multiple of the Casimir operator, more specifically, we have for~\cite[chapter 15]{francesco2012conformal}
\begin{align}
  h_\lambda := \frac{\Scp{\lambda}{\lambda+\rho}}{k+\mathsf{g}}\,,
\end{align}
that $h_\vphi = h_\lambda$ for $\vphi \in \moduleklambda(0)$, which, as discussed above, is an irreducible $\g$-module. Here, $\rho$ is the Weyl vector of $\g $ and $\mathsf{g}$ is its dual Coxeter number. It was shown in~\cite{Frenkel:1992jt} that the irreducible modules of~$\voalkzero$ are all derived from irreducible modules of~$\g $, and are hence of the form $\moduleklambda$, with the additional constraint that the weight $\lambda$ is integrable and satisfies $\lambda(\theta) \leq k$, where $\theta $ is  the maximal root of~$\g $.
\end{example}

\subsubsection{Unitary modules\label{sec:unitarity}}

A key tool in our analysis is the existence of certain positive definite  Hermitian forms. This is the assumption of unitarity: A VOA $\cV=(\cV,\Y,\1,\omega)$ is called {\em unitary} if there is an anti-linear involution $\phi:\cV\rightarrow\cV$ of $\cV$ with 
\begin{align}
\phi(\1)=\1\ ,\quad \phi(\omega)=\omega\ ,\quad\textrm{  and  }\quad \phi(v_nw)=\phi(v)_n\phi(w)\quad \textrm{ for all }v,w\in\cV\ ,\label{eq:automorphismdef}
\end{align} together with a positive definite Hermitian form $\langle\cdot,\cdot\rangle_{\cV}:\cV\times\cV\rightarrow\mathbb{C}$ 
which is $\mathbb{C}$-linear in the second argument and anti-$\mathbb{C}$-linear in the first argument
such that the invariance condition
\begin{align}\label{eq:unitarityvoadef}
  \langle w_1,\Y(e^{zL_1}(-z^{-2})^{L_0}v,z^{-1})w_2\rangle_\cV &=\langle\Y(\phi(v),z)w_1,w_2\rangle_\cV\  
\end{align}
holds for all $v,w_1,w_2\in\cV$.  Adopting the convention used in physics (but in contrast to~\cite{Dong:2014gx}), we assume that $\langle\cdot,\cdot\rangle_\cV$ is anti-linear in the first and linear in the second argument.  
The map $v\mapsto \phi(v)$ is called an anti-linear automorphism of the  VOA~$\cV$. (Note that we do not complex conjugate $z$, which is a formal variable.) In the following, we will also assume (as e.g., in~\cite[Theorem 3.3]{Dong:2014gx}) that
\begin{align}
\langle\phi(v_1),\phi(v_2)\rangle_\cV=\overline{\langle v_1,v_2\rangle_\cV}\qquad\textrm{ for all }v_1,v_2\in\cV\ .\label{eq:unitaryvertexopandinvolution}
\end{align}

Unitarity for modules can be defined in an analogous way (see~\cite{Dong:2014gx}). More precisely, consider a  unitary VOA~$\cV$ with anti-linear automorphism $\phi:\cV\rightarrow\cV$. A module $(A,\Y_A)$ of $\cV$ is called unitary if
there is a positive definite Hermitian form $\langle\cdot,\cdot\rangle_A:A\times A\rightarrow\mathbb{C}$ which is $\mathbb{C}$-linear in the second and $\mathbb{C}$-anti-linear in the first argument such that
\begin{align}\label{eq:unitaritymoduledefinition}
\langle a_1,\Y_A(e^{zL_1}(-z^{-2})^{L_0}v,z^{-1})a_2\rangle_A&=\langle \Y_A(\phi(v),z)a_1,a_2\rangle_A\ ,
\end{align}
for all $v\in\cV$ and $a_1,a_2\in A$. As for the case of VOAs, we assume in addition that there exists an anti-linear involution $\eta: A \to A$, which is compatible with  $\phi$ in the sense that
\begin{align}
  \eta(a_n b) = \phi(a)_n \eta(b)\,,\quad a \in \cV\,,\; b \in A\,\label{eq:involutionanbcompatibility}
\end{align}
as well as with the scalar product,
\begin{align}
  \Scp{\eta(a)}{\eta(b)} = \overline{\Scp{a}{b}}\,,\quad a,b \in A\,.
\end{align}
Note that the invariance property applied to the Virasoro vector $\omega$ implies that the adjoint operator of $L_n$ for $n \in \Z$ is $L_{-n}$,
\begin{align}
  \Scp{L_n a}{b} = \Scp{a}{L_{-n}b}\,,\quad a,b \in A\,.
\end{align}
Following Gaberdiel~\cite{Gaberdiel:2000tb}, it is also easy to see that the operator $L_0$ has to be positive for a unitary module, since its expectation value on the top level $A_0$ satisfies 
\begin{align}
  \Scp{a}{L_0 a} = \Scp{a}{L_1 L_{-1}a} = \Scp{L_{-1}a}{L_{-1}a}\,,\quad \textrm{for } a \in A_0\,,
\end{align}
which is necessarily positive. 

Due to the Frechet-Riesz representation theorem, we note that for unitary modules we have a linear vector space isomorphism $\tilde{\eta}$ from the unitary module to its restricted dual space $A^\prime=\bigoplus_{n\in\mathbb{N}_0}A_n^\prime$, given by the formula
\begin{align}\label{eq:isodualcomplconj}
  \tilde{\eta}(a)(b) = \Scp{\eta(a)}{b} \,.
\end{align}
 Put differently, 
\begin{align}
  A \times A \ni a_1 \times a_2 &\mapsto \pairL a_1,a_2\pairR_A:=\Scp{\eta_A(a_1)}{a_2}\label{eq:bilinearaformdefinition}
\end{align}
defines a bilinear form on $A$. Moreover, this bilinear form inherits the invariance property from Eq.~\eqref{eq:unitaritymoduledefinition}. That is,
we have
\begin{align}
\pairL \Y_A(v,z)a_1,a_2 \pairR_A&=
\pairL a_1,\Y_A(e^{zL_1}(-z^{-2})^{L_0}v,z^{-1})a_2 \pairR_A\ \textrm{ for all }a_1,a_2\in A\textrm{ and }v\in\cV\ . \label{eq:invariancepairingclaim}
\end{align}
For completeness, we give  the proof of this statement.
\begin{proof}
By definition, we have
\begin{align*}
\pairL \Y_A(v,z)a_1,a_2\pairR_A
&= \langle \eta_A(\Y_A(v,z)a_1),a_2\rangle\\
&=\overline{\langle \Y_A(v,z)a_1,\eta_A(a_2)\rangle}\\
&=\overline{\langle a_1,\Y_A(e^{zL_1}(-z^{-2})^{L_0}\phi(v),z^{-1})\eta_A(a_2)\rangle}\\
&=\langle \eta_A(a_1),\eta_A\left(\Y_A(e^{zL_1}(-z^{-2})^{L_0}\phi(v),z^{-1})\eta_A(a_2)\right)\rangle\\
&=\pairL a_1,\eta_A\left(\Y_A(e^{zL_1}(-z^{-2})^{L_0}\phi(v),z^{-1})\eta_A(a_2)\right)\pairR_A\ . 
\end{align*}
The claim then follows since
\begin{align*}
\eta_A\left(\Y_A(e^{zL_1}(-z^{-2})^{L_0}\phi(v),z^{-1})\eta_A(a_2)\right)&=\Y_A(e^{zL_1}(-z^{-2})^{L_0}v,z^{-1})a_2
\end{align*}
because of~\eqref{eq:involutionanbcompatibility}, the commutation relations between $L_0$ and $L_1$, and  the fact that $\phi(L_n)=L_n$ because $\omega$ is invariant under~$\phi$. 
\end{proof}

Unitary modules are in fact special pre-Hilbert spaces, as they are linear spaces equipped with a positive definite sesquilinear form $\Scp{\cdot}{\cdot}$. This form can be used to define a \emph{norm}, by setting $\norm{a}_A = \Scp{a}{a}^\half$, which turns the unitary module $A$ into a graded, normed space. We can then complete this space with respect to this norm, obtaining a Hilbert space, with inner product given by the unique extension of the sesquilinear form to the completion of $A$. In the following, we do not differentiate between unitary modules and their completion, as our arguments do not rely on this distinction. However, we will  frequently use facts related to Hilbert spaces, such as the Cauchy-Schwarz inequality, and thus find it more appealing to work with the completed spaces.

\begin{example}[\bf WZW modules are unitary]\label{exmp:wzwunitary}
Let again $\g $ be a complex simple Lie algebra, with associated VOA~$\voalkzero$. If we denote by $\{\fe_j,\ff_j,\fh_j\} $ the Chevalley generators of $\g $, then the associated Chevalley involution
\begin{align}
  \eta_0(\fe_j) = - \ff_j\,,\quad \eta_0(\ff_j) = - \fe_j\,,\quad \eta_0(\fh_j) = - \fh_j\,,
\end{align}
defines an involution on the real span of these generators. It can be extended to the associated real subspace (spanned by $\{\fe_j,\ff_j,\fh_j\}\otimes \{t^m\ | m\in\mathbb{Z}\}$) of~$\ga$ by setting 
\begin{align}
\eta_0(\fa\otimes t^n)=\eta_0(\fa)\otimes t^n\qquad\textrm{ and }\qquad 
\eta_0(\mathsf{k})=-\mathsf{k} 
\end{align}
or, succintly, $\eta_0(\fa(n))=\eta_0(\fa)(-n)$.  Moreover, $(-\eta_0)$ can be extended to an anti-linear anti-automorphism on the whole of $\ga$. This anti-automorphism~$\hat{\eta}$ satisfies the additional property that there exists a positive definite sesquilinear form~\cite[Chapter 11]{Kacbook} on $\Uga$ for which the adjoint of $\fa(n) \in \ga$ is given by $\hat{\eta}(\fa)(-n)$,
\begin{align}
  \Scp{\fa(n) v}{w} = \Scp{v}{\hat{\eta}(\fa)(-n) w} \,,\quad v,w \in \Uga\,,\;\fa \in \g\,,
\end{align}
Conversely, as shown by Dong and Lin~\cite{Dong:2014gx}, setting 
\begin{align}
  \phi(\fa_1(n_1) \cdots \fa_1(n_l)) = (-\hat{\eta}(\fa_1))(n_1) \cdots (-\hat{\eta}(\fa_l))(n_l)
\end{align}
defines an anti-linear involution on $\Uga$, which leaves the maximal proper submodule invariant and hence induces an anti-linear involution on~$\voalkzero$. The scalar product on $\voalkzero$ induced by the one on $\Uga$ inherits the adjoint relations for elements $\fa(n) \in \ga$, and moreover satisfies the required invariance property (Eq.~\eqref{eq:unitarityvoadef}) as well as $\Scp{\phi(v)}{\phi(w)} = \overline{\Scp{v}{w}}$. Thus, $\voalkzero$ is a unitary VOA. 

If $V_k = \Uga \otimes_{\ga_{+} \oplus \g \oplus k\idty} V$ is a $\ga$-module, Kac~\cite[Chapter 11]{Kacbook} shows that the map induced by $\hat{\eta}$ is again an anti-linear involution on the whole of $V_k$. Moreover, if $V_k$ stems from an irreducible highest weight representation with weight $\lambda$ satisfying $(\lambda,\theta) \leq k$, then there exists a positive definite Hermitian form such that
\begin{align}\label{eq:involutionvbgmodule}
  \Scp{\fa(n) \chi_1}{\chi_2} = \Scp{\chi_1}{\hat{\eta}(\fa)(-n) \chi_2} \,,\quad \chi_1, \chi_2 \in V_k\,,\;\fa \in \g\,.
\end{align}
Again following the ideas of Dong and Lin~\cite{Dong:2014gx}, we can show that the associated irreducible $\voalkzero$-module $\moduleklambda$ carries an anti-linear involution $\eta$ satisfying $\Scp{\eta(\chi_1)}{\eta(\chi_2)} = \overline{\Scp{\chi_1}{\chi_2}}$, for $\chi_1,\chi_2 \in \moduleklambda$ and the induced scalar product on  $\moduleklambda$ is invariant as required in~\eqref{eq:unitaritymoduledefinition}, and thus a unitary module. For a more operator algebraic discussion of the same structures, leading to the same conclusions, see~\cite{Wassermann:1998cs}.
\end{example}

\subsubsection{Intertwiners\label{sec:intertwiners}}
Let $A,B,C$ be modules of a rational VOA~$\cV$.
An {\em intertwining operator~$\iY$ of type $\itype{C}{A}{B}$}
is a family of linear maps~$\iY(\cdot,z)$ from $A$ to certain
Laurent-like series with coefficients in $\End(B,C)$, i.e., it associates to every $a\in A$ a series
\begin{align}
\iY(a,z)=\sum_{\tau \in I, m\in \mathbb{Z}}a_{\tau,m}z^{-\tau-m}\ ,\label{eq:intertwineroriginaldef}
\end{align}
where $I=I^C_{AB}=\{\tau_1,\ldots,\tau_d\}$ is a finite collection of  complex numbers (depending only on $A$, $B$ and $C$) and $a_{\tau,m}\in\End(B,C)$ for $\tau\in I$ and $m\in\mathbb{Z}$. For all $b\in B$, the mode operators satisfy
\begin{align*} 
a_{\tau,m}b=0\qquad\textrm{ for sufficiently large }m\  .
\end{align*}
Consider the case where $A$, $B$, $C$ are irreducible with highest weights $h_A$, $h_B$, $h_C$, respectively.
 It can be shown (see~\cite[Remark 5.4.4]{Frenkel:1993fv})
that in this case,  
\begin{align}
I^C_{AB}=\{h_A+h_B-h_C\}\label{eq:iabcirreducible}
\end{align}
consists of a single element. In particular,
up to a factor $z^{h_C-(h_A+h_B)}$, an intertwiner of type  $\itype{C}{A}{B}$
is given by a formal Laurent series. In the general case an intertwiner is thus a finite linear combination of Laurent-like series  which upon multiplication with a fixed power in the indeterminate $z$ are formal Laurent series. That is, we can write
\begin{align}
  \iY(a,z) = \sum_{\tau \in I} \iY_\tau(a,z)\,\quad \textrm{ and }\qquad \iY_\tau(a,z) \in z^{\tau} \End(B,C)[[z,z^{-1}]] \,.
\end{align}
In the following, we will use the notation $\Endint{B,C}{z,z^{-1}}$ for such a formal series and thus write $\iY(a,z) \in \Endint{B,C}{z,z^{-1}}$.
As in~\eqref{eq:yvzweightdef}, it will be convenient to use the notation
\begin{align}
\iY(a,z)=\sum_{\tau\in I, m\in \mathbb{Z}} \iy(a)_{\tau,m} z^{-m-\tau-n}\label{eq:modeexpansionintertwiner}
\end{align}
for homogeneous vectors $a\in A_{n}$. The operators $\iY(a,z)$
satisfy the translation property
\begin{align}
\frac{d}{dz}\iY(a,z)&=\iY(L_{A,-1}a,z)\qquad\textrm{ for all }a\in A\  .\label{eq:translationpropertyintertwiner}
\end{align}
The compatibility of an intertwiner with modules is expressed by a corresponding Jacobi identity: we have
\begin{align}\label{eq:jacobiidentity}
&\Residuum_{z_1-z_2}\left(\iY(\Y_A(v,z_1-z_2)a,z_2) (z_1-z_2)^m \iota_{z_2,z_1-z_2}((z_1-z_2)+z_2)^n)\right)=\\
&\qquad \Residuum_{z_1}\left(\Y_C(v,z_1)\iY(a,z_2)\iota_{z_1,z_2}(z_1-z_2)^mz_1^n\right)-
\Residuum_{z_1}\left(\iY(a,z_2)\Y_B(v,z_1)\iota_{z_2,z_1}(z_1-z_2)^mz_1^n\right)
\end{align}
for all $m,n\in\mathbb{Z}$, $v\in\cV$ and $a\in A$. A useful consequence results if we set $m=0$ and evaluate the residues for a homogeneous vector $v \in \cV$. We find by expansion of $((z_1-z_2)+z_2)^n$ that
\begin{align}\label{eq:commutationmoduleintertwiner}
  \y_C(v)_{n - \wt v + 1} \iY(a,z) - \iY(a,z) \y_B(v)_{n - \wt v + 1} = \sum_{j=0}^\infty \binom{n}{j} \,\iY(\y_A(v)_{j-\wt v + 1}a,z)\,z^{n-j} 
\end{align}
for $a \in A$. This identity will be used below in our algorithm for WZW models. As before, Eq.~\eqref{eq:translationpropertyintertwiner} and
the analog of~\eqref{eq:hzerocommutator} (which follows from Eq.~\eqref{eq:commutationmoduleintertwiner}) imply that the mode operators defined by~\eqref{eq:modeexpansionintertwiner} preserve homogeneity and satisfy for $a\in A_{n,\alpha}$ the identity
\begin{align}
\wt(\iy(a)_{\tau,m}b)=\wt b-m-\tau+\alpha\qquad\textrm{ for any }\tau\in I \textrm{ and }m\in \mathbb{Z}\ , \label{eq:weightchangehomogeneousintertwiner}
\end{align}
for any homogeneous vector~$b\in B$. 
Indeed, we have
\begin{align*}
L_{C,0} \iY(a,z)b&=(L_{C,0}\iY(a,z)-\iY(a,z)L_{B,0}) b+\iY(a,z)L_{B,0}b\\
&=(\iY(L_{A,0}a,z)+z\iY(L_{A,-1}a,z))b+\wt b\cdot \iY(a,z)b\\
&=(\wt a+z\frac{d}{dz}+\wt b)\iY(a,z)b\ .
\end{align*}
Inserting~\eqref{eq:modeexpansionintertwiner} and comparing coefficients gives the claim~\eqref{eq:weightchangehomogeneousintertwiner}. 

Eq.~\eqref{eq:weightchangehomogeneousintertwiner} has a simple consequence in terms of levels. Let $A,B,C$ be irreducible modules of heighest weights $h_A,h_B$ and $h_C$, respectively. Recall that a homogeneous vector $b\in B_n$ in the $n$-th level~$B_n$ of the module~$B$ has weight of the form $\wt b=h_B+n$. For such a vector, and~$a\in A_{n'}$ (with arbitrary $n'\in\mathbb{N}_0$), we get
with $\tau=h_A+h_B-h_C\in I^C_{AB}$ (cf.~\eqref{eq:iabcirreducible})
\begin{align*} 
  \wt(\iy(a)_{\tau,m}b)&=\wt b-m+h_A-\tau=h_A + h_B - (h_A+h_B-h_C) +(n-m) = h_C + (n-m)\,, 
\end{align*}
according to~\eqref{eq:weightchangehomogeneousintertwiner}. We see that~$\iy(a)_{\tau,m}b\in C_{n-m}$ belongs to the level~$n-m$ of the module~$C$. That is, we have for homogeneous vectors $a\in A$
\begin{align}
\iy(a)_{\tau,m}B_n\subset C_{n-m}\ .
\end{align}
In particular, the same inclusion holds for all $a\in A$, as it is independent of the weight of $a$.

Now consider three arbitrary unitary modules $A,B,C$, which  can be decomposed into finite direct sums of irreducible modules,
\begin{align}
  A = \bigoplus_{\alpha} A[\alpha]\,,\quad B = \bigoplus_\beta B[\beta]\,,\quad C=\bigoplus_\gamma C[\gamma] \ ,
\end{align}
with the sums including multiplicities. 
Let $Q_{D[\delta]}$ denote the projector which maps the module $D = \{A,B,C\}$ into the irreducible component $\delta = \{\alpha,\beta,\gamma\}$, e.g., $Q_{A[\alpha]} A = A[\alpha]$. If $\iY$ is  an intertwiner of type $\itype{C}{A}{B}$, then each expression $Q_{C[\gamma]} \iY(Q_{A[\alpha]} \cdot,z)Q_{B[\beta]}$ is an intertwiner between irreducible modules and hence the previous discussion applies to the modes $Q_{C[\gamma]} \iy(Q_{A[\alpha]} \cdot)_m Q_{B[\beta]}$. But since $\iY(\cdot,z) = \sum_{\alpha,\beta,\gamma} Q_{C[\gamma]} \iY(Q_{A[\alpha]} \cdot,z)Q_{B[\beta]}$, we find that 
\begin{align}
\iy(a)_{\tau,m}B_n\subset C_{n-m}\ .\label{eq:containementinequalitylevelsintertw}
\end{align}
holds for general modules as well.

\paragraph{Intertwiners between unitary modules.}

As our discussion focuses on  unitary modules, intertwiners between those deserve special attention. Equipped with positive definite Hermitian form, unitary modules are Hilbert spaces, and intertwiners are thus maps between different Hilbert spaces. A frequently used  operation on operators is the adjoint, and hence the question arises whether the adjoint of an intertwiner is again an intertwiner. The next lemma is the first step in providing a positive answer.

\begin{lemma}\label{lem:abccomplexconjug}
  Let $A,B,C$ be three unitary modules of a VOA, and denote by $\eta_A, \eta_B, \eta_C$ the corresponding anti-linear automorphisms. Let $\iY$ be an intertwiner of type
$\itype{C}{A}{B}$, and define
\begin{align*}
\bar{\iY}(a,z)b&=\eta_C(\iY(\eta_A(a),z)\eta_B(b))
\end{align*}
for all $a\in A$, $b\in B$ and $c\in C$. Then $\bar{\iY}$ is an intertwiner of type $\itype{C}{A}{B}$.
\end{lemma}
\begin{proof}
  We first check the translation property of $\bar{\iY}$. For this purpose, note that $\eta_A(L_{-1} a) = L_{-1} \eta_A(a)$ since the involution $\phi$ leaves the Virasoro vector invariant, $\phi(\omega)=\omega$. This implies that
  \begin{align}
    \bar{\iY}(L_{-1}a,z) &= \eta_C \circ \iY(\eta_A(L_{-1}a),z)\circ \eta_B\nonumber\\
&= \eta_C \circ \iY(L_{-1}\eta_A(a),z)\circ\eta_B\nonumber\\
& = \eta_C \circ \frac{d}{dz}\iY(\eta_A(a),z)\circ\eta_B \nonumber\\
    &= \frac{d}{dz} \eta_C \circ \iY(\eta_A(a),z)\circ\eta_B\\
& = \frac{d}{dz}  \bar{\iY}(a,z)
  \end{align}
  since $z$ is a formal variable and  differentiation is defined as an operation on the formal Laurent series. Finally, we have to check the Jacobi identity~\eqref{eq:jacobiidentity}
for $\bar{\iY}$, 
\begin{align}\label{eq:jacobiconjugate}
&\Residuum_{z_1-z_2}\left(\bar{\iY}(\Y_A(v,z_1-z_2)a,z_2) (z_1-z_2)^m \iota_{z_2,z_1-z_2}((z_1-z_2)+z_2)^n)\right)=\\
&\qquad \Residuum_{z_1}\left(\Y_C(v,z_1)\bar{\iY}(a,z_2)\iota_{z_1,z_2}(z_1-z_2)^mz_1^n\right)-
\Residuum_{z_1}\left(\bar{\iY}(a,z_2)\Y_B(v,z_1)\iota_{z_2,z_1}(z_1-z_2)^mz_1^n\right)\ . 
\end{align}
But for $z=z_1-z_2$ we have
\begin{align*}
\bar{\iY}(\Y_A(v,z)a,z_2)
&=\eta_C\circ\iY (\eta_A(\Y_A(v,z)a),z_2)\circ \eta_B\\
&=\eta_C\circ\iY (\Y_A(\phi(v),z)\varphi_A(a),z_2)\circ\eta_B\ .
\end{align*}
Similarly, we find
\begin{align*}
\Y_C(v,z_1)\overline{\iY}(a,z_2)&=
\Y_C(v,z_1)\eta_C\circ\iY(\eta_A(a),z_2)\circ\eta_B\\
&=\eta_C\circ\Y_C(\phi(v),z_1) \iY(\eta_A(a),z_2)\circ\eta_B\,,
\end{align*}
as well as
\begin{align*}
\bar{\iY}(a,z_2)\Y_B(v,z_1)&=
\eta_C\circ \iY(\eta_A(a),z_2)\circ\eta_B\circ \Y_B(v,z_1)\\
&=\eta_C\circ \iY(\eta_A(a),z_2)\circ \Y_B(\phi(v),z_1)\circ\eta_B\,,
\end{align*}
Applying both sides to $b\in B$, we conclude that~\eqref{eq:jacobiconjugate}
is just the Jacobi identity~\eqref{eq:jacobiidentity} evaluated at $\bar{v}=\eta(v)$, $\bar{a}=\eta_A(a)$, $b=\eta_B(b)$, and $\eta_C$ applied from the left.
\end{proof}

The following lemma gives an expression for the adjoint operator of an intertwiner. 
\begin{lemma}\label{lem:adjointoperatorintertw}
  Let $A,B,C$ be as in Lemma~\ref{lem:abccomplexconjug} and let $\iY$ be an intertwiner of type $\itype{C}{A}{B}$. 
Then there is an intertwiner $\iY_1$ of type $\itype{B}{A}{C}$ such that 
\begin{align*}
\langle\iY(a,z)b,c\rangle &=\langle b,\iY_1(e^{zL_1}(-z^{-2})^{L_0}\eta_A(a),z^{-1})c\rangle
\end{align*}
for all $a\in A$, $b\in B$ and $c\in C$. 
\end{lemma}
\begin{proof}
The proof uses similar reasoning as the proof of~\eqref{eq:invariancepairingclaim}.  Let $\pairL\cdot,\cdot \pairR_B$, 
$\pairL\cdot,\cdot \pairR_C$,  be the invariant bilinear forms defined by $\eta_B$ and $\eta_C$, respectively (cf.~\eqref{eq:bilinearaformdefinition}). Because of the invariance property, we can apply~\cite[Corollary 5.5.3]{Frenkel:1993fv}, which states that  there is an intertwiner~$\tilde{\iY}$ of type $\itype{B}{A}{C}$ such that 
\begin{align*}
\pairL c,\iY(a,z)b\pairR_C&=\pairL\tilde{\iY}(e^{zL_1}(-z^{-2})^{L_0}a,z^{-1})c,b\pairR_B\ .
\end{align*}
Hence we find that
\begin{align*}
\langle \iY(a,z)b,c\rangle &=\overline{\langle c,\iY(a,z)b\rangle}\\
&=\overline{\pairL(\eta_C(c),\iY(a,z)b\pairR_C}\\
&=\overline{\pairL\tilde{\iY}(e^{zL_1}(-z^{-2})^{L_0}a,z^{-1})\eta_C(c),b\pairR_B}\\
&=\overline{\langle \eta_B\circ \tilde{\iY}(e^{zL_1}(-z^{-2})^{L_0}a,z^{-1})\eta_C(c),b\rangle}\\
&=\langle b,\eta_B\circ\tilde{\iY}(e^{zL_1}(-z^{-2})^{L_0}a,z^{-1})\eta_C(c)\rangle
\end{align*}
By Lemma~\ref{lem:abccomplexconjug},
\begin{align*}
\iY_1(a,z)c&=\eta_B \circ \tilde{\iY}(\eta_A(a),z)\eta_C(c)
 \end{align*} 
defines an intertwiner of type $\itype{C}{A}{B}$.
The claim then follows since
\begin{align*}
\eta_B\circ\tilde{\iY}(e^{zL_1}(-z^{-2})^{L_0}a,z^{-1})\eta_C(c)&=\iY_1(
\eta_A(e^{zL_1}(-z^{-2})^{L_0}a),z^{-1})c\ 
\end{align*}
and $\eta_A(L_na)=L_n\eta_A(a)$.
\end{proof}

\begin{example}[\bf WZW intertwiners]\label{exmp:wzwintertwiner}
As all  irreducible modules of the VOA~$\voalkzero$ are derived from irreducible highest weight modules of the Lie algebra $\g $, it is sufficient to consider intertwiners  between three of these modules. Let us fix three weights $\lambda_1, \lambda_2, \lambda_3$, and consider the corresponding irreducible highest weight $\g $-modules~$\hwgmodule{\lambda_1}, \hwgmodule{\lambda_2}, \hwgmodule{\lambda_3}$, as well as the derived irreducible modules~$\modulekflambda{\lambda_1}, \modulekflambda{\lambda_2}, \modulekflambda{\lambda_3}$ of the VOA~$\voalkzero$. 

We now consider intertwiners between these three modules. Given that all objects of the VOA are derived from the corresponding objects of the Lie algebra~$\g $ (with additional constraints coming from the level $k$), it is not surprising that intertwiners can similarly be defined in terms of corresponding objects associated with~$\g$. As shown by Frenkel and Zhu~\cite{Frenkel:1992jt}, any intertwiner between three irreducible modules of the VOA~$\voalkzero$ is determined by an intertwiner between the three corresponding irreducible $\g $-modules, again with an additional assumption involving the level $k$ (see~\cite[Corollary 3.2.1]{Frenkel:1992jt} as well as~\cite{tsuchiya1987vertex} for a more analytic approach in the case of $\g = \mathfrak{sl}(2,\C )$).  
An application of Zhu's theory~\cite{Zhu:1996cp,Frenkel:1992jt} provides a procedure for reconstructing  an intertwiner of the VOA~$\voalkzero$ from such a Lie algebra intertwiner. We provide an explanation of that argument in terms of an implementable algorithm in Appendix~\ref{app:algorithmwzw}.

Conversely, given an intertwiner 
$\iY$  of type $\itype{\modulekflambda{\lambda_1}}{\modulekflambda{\lambda_3}}{\modulekflambda{\lambda_2}}$
for modules of the VOA~$\voalkzero$, we can obtain the associated $\g $-module intertwiner    as follows:  if we choose to evaluate Eq.~\eqref{eq:commutationmoduleintertwiner} for the zero modes of the module operators and the intertwiner operator, we find that
\begin{align}\label{eq:wzwintertwinerprop}
  \fa \iy(\vphi_3)_{\tau,0} \vphi_2 = \iy(\fa \vphi_3)_{\tau,0} \vphi_2 + \iy(\vphi_3)_{\tau,0} \fa \vphi_2 \,, 
\end{align}
for $\vphi_2 \in \modulekflambda{\lambda_2}(0)$, $\vphi_3 \in \modulekflambda{\lambda_3}(0)$, where $\tau=h_{\lambda_3}+h_{\lambda_2}-h_{\lambda_1}$. Here, $\fa \in \g $ is identified with its image in the corresponding irreducible representation of $\g $. This implies that the operator $\iy(\cdot)_{\tau,0}$ restricted to the top levels of the modules $\modulekflambda{\lambda_i}$, $i=1,2,3$ is an intertwiner from the tensor product  $\g $-module defined by the weights $\lambda_2$, $\lambda_3$ to the irreducible $\g $-module corresponding to the weight $\lambda_1$. This follows since the top levels are irreducible~$\g$-modules. We refer to Proposition~\ref{prop:wzwintertwiner} below for more details. 

\end{example}

\section{Correlation functions via transfer operators\label{sec:corfunc}}
Having introduced the necessary terminology, 
we continue to argue that
correlation functions can be expressed exactly in term of an MPS.
 This is the central point of this section; it is also the basis for
the finite-dimensional approximations discussed in Section~\ref{sec:approx}. 

\subsection{Correlation functions for modules and intertwiners}\label{sec:correlationfunctions}

Up to this point,  variables denoted by $z$ or similar were interpreted as \emph{formal variables}, and expressions involving powers in it as \emph{formal Laurent series}. All identities held term-by-term in different powers of $z$ (or $z^{-1}$), but no statement concerning convergence were made. In order to make contact with physical quantities, more precisely correlation functions, the formal Laurent series expansions have to be evaluated to yield finite numbers. 
That is,  the Laurent series are reinterpreted as sums of operators with complex coefficients (after substituting complex numbers for the   indeterminates), and these sums need to be shown to converge.

Matrix elements of products of vertex operators, or more generally module operators and intertwiners, thus become functions on the Riemann sphere. These can be shown to satisfy physical axioms of CFT correlation functions (e.g., modular invariance on the torus). Conversely, VOAs may be constructed from correlation functions, see~\cite{Huang:2010fv}. In the following, we briefly sketch the general theory, then focus on the special case of equidistant points.

Recall the pairing $\pairL\cdot,\cdot\pairR$ between the restricted dual space~$\cV'$ and $\cV$. For a VOA~$\cV$, the genus-zero $n$-point (vacuum-to-vacuum) correlation function is defined as
\begin{align}
F^{(0)}_{\cV} ((v_1,z_1),\ldots,(v_n,z_n))&=\pairL \1',\Y(v_1,z_1)\cdots\Y(v_n,z_n)\1\pairR\ .
\end{align}
(More generally, one considers general matrix elements by substituting arbitrary elements for $\1'\in\cV'$ and $\1\in\cV$.) It can be shown that this expression is the series expansion of a rational function~$f(z_1,\ldots,z_n)$, converging on a suitable subdomain of~$\mathbb{C}$ (see \cite[Proposition~3.5.1]{Frenkel:1993fv}). Again,  the formal variables $z_1,\ldots,z_n$ are to be interpreted as complex numbers. For vacuum-to-vacuum correlation functions, the function $f(z_1,\ldots,z_n)$ takes the form
\begin{align*}
f(z_1,\ldots,z_n)&=\frac{g(z_1,\ldots,z_n)}{\prod_{j<k}(z_j-z_k)^{s_{jk}}}\ .
\end{align*}
for a polynomial $g(z_1,\ldots,z_n)$.  Similarly, we can define the genus-one $n$-point functions by
\begin{align}
F^{(1)}_{\torusq,\cV}((v_1,z_1),\ldots,(v_n,z_n))&= \tr_{\cV}\left(\Y(z_1^{L_0}v_1,z_1)\cdots
\Y(z_n^{L_0}v_n,z_n)
\torusq^{L_0-c/24}\right)
\end{align}
Here $0<\torusq<1$ is the `diameter' of the torus, and the trace is calculated on each finite-dimensional level and hence is well-defined. The introduction of the factors~$z_j^{L_0}$ in the arguments is a convention.

In his seminal work~\cite{Zhu:1996cp}, Zhu identified sufficient conditions for the existence of genus-$1$-correlation functions: assuming rationality and $C_2$-cofiniteness of the VOA, the function $F^{(1)}_{\torusq,\cV}$ converges  on the domain   (see~\cite[Section 4.1]{Zhu:1996cp})
\begin{align}
1>|z_1|>\cdots > |z_n|>\torusq\ ,\label{eq:torusconvergencedomain}
\end{align} 
and can be continued to an analytical function, possibly with poles at $z_i=1$, $z_i=z_j$, $z_i=0$. This gives $n$-point correlation functions on the torus~$\mathbb{C}\backslash\{0\}/z\sim \{z\torusq^k\}$ regarded as the punctured plane $\mathbb{C}\backslash\{0\}$ modulo the relations
\begin{align}
z=z\torusq^k\qquad\textrm{ for }k\in\mathbb{Z}\  . \label{eq:torusrelationq}
\end{align}

The latter reasoning was extended to the case of intertwiners by Huang~\cite{Huang:2005gs,Huang:2005gsb}. For the genus-$0$-case, consider
$\cV$-modules $A^{(i)}$, $B^{(i)}$, $i=1,\ldots,n$, $B^{(0)}$, and intertwiners $\iY_i$ of type $\itype{B^{(i-1)}}{A^{(i)}}{B^{(i)}}$. For any $(v^{(0)})'\in (B^{(0)})'$, $v^{(n)}\in B^{(n)}$, the  genus-$0$ correlation function
$F_{(v^{(0)})',v^{(n)}}^{(0)}=F^{(0)}_{\iY_1,\ldots,\iY_n,(v^{(0)})',v^{(n)}}$ is defined by
\begin{align} 
F_{(v^{(0)})',v^{(n)}}^{(0)}((a_1,z_1),\ldots,(a_n,z_n))&=\pairL (v^{(0)})',\iY_1(a_1,z_1)\cdots \iY_n(a_n,z_n)v^{(n)}\pairR\label{eq:definitiongenuszeronpoint}
\end{align}
for all $a_i\in A^{(i)}$, $i=1,\ldots,n$,
where $\pairL\cdot,\cdot\pairR$ is the canonical pairing between the restricted dual space $(B^{(0)})'$ and $B^{(0)}$. For unitary modules, this pairing can be written in terms of the scalar product, as discussed (cf.~\eqref{eq:isodualcomplconj}). We then write 
\begin{align} 
F_{v^{(0)},v^{(n)}}^{(0)}((a_1,z_1),\ldots,(a_n,z_n))&=\langle v^{(0)},\iY_1(a_1,z_1)\cdots \iY_n(a_n,z_n)v^{(n)}\rangle\label{eq:definitiongenuszeronpointspr}
\end{align}
for the corresponding correlation function, where now $v^{(0)}\in B^{(0)}$.  Huang~\cite{Huang:2005gs,Huang:2005gsb} has shown that the corresponding series is absolutely convergent in the domain
\begin{align}
|z_1|>|z_2|>\cdots >|z_n|\ .\label{eq:genuszerodomain}
\end{align}
Despite this general result, if we speak about genus-0 correlation functions in the following, we  are always considering the special case where $B^{(0)} = B^{(n)} = \cV$, the adjoint module, and $v^{(0)} = v^{(n)} = \1$, the vacuum element. 
 
Similarly, for $\cV$-modules $A^{(i)}$ and $B^{(i)}$, $i=1,\ldots,n$ with $B^{(0)}=B^{(n)}$ and intertwiners $\iY_i$ of type $\itype{B^{(i-1)}}{A^{(i)}}{B^{(i)}}$, we can define
the genus-$1$ correlation function $F_\torusq^{(1)}:=F^{(1)}_{\torusq,\iY_1,\ldots,\iY_n}$ by
\begin{align}
F^{(1)}_\torusq((a_1,z_1),\ldots,(a_n,z_n))&=\tr_{B^{(n)}}\left(\iY_1(z_1^{L_0}a_1,z_1)\cdots \iY_n(z_n^{L_0}a_n,z_n)\torusq^{L_0-c/24}\right)\label{eq:genusoneintertwinercorrelationfct}
\end{align}
for all $a_i\in A_i$.  Huang~\cite{Huang:2005gs,Huang:2005gsb} showed that~\eqref{eq:genusoneintertwinercorrelationfct} (respectively a certain geometrically transformed version thereof, see~\cite[Remark~3.5]{Huang:2005gsb}) gives an absolutely convergent power series on the domain~\eqref{eq:torusconvergencedomain}. His results  apply to rational, $C_2$-co-finite VOAs of CFT-type.  The genus-$1$-partition function without insertions,
\begin{align}\label{eq:partitionfunction}
Z_{B}(\torusq)=\tr_{B}(\torusq^{L_0-c/24})
\end{align}
is the {\em character} of the module~$B$. It converges for $0<\torusq<1$ (see e.g.,~\cite{GaberdielNeitzke}) for the VOAs and modules considered here.  Expression~\eqref{eq:partitionfunction} appears below as a multiplicative factor in our accuracy estimate for approximations to correlation functions on the torus.

We deliberately choose the same letter~$z$ for the formal variable and the corresponding complex number, in order to emphasize the similarity and to make the connection between the formal language of VOAs and physical quantities clear. However, we stress that the interpretation in terms of complex numbers only makes sense in terms of correlation functions, as it always involves a statement about convergence. In what follows, a variable $z$ appearing in an expression involving \emph{only} intertwining operators, which are not evaluated to give a number of physical significance, is interpreted as formal indeterminate. If it appears in a quantity of physical significance, such as a correlation function, it denotes a complex number which can be freely chosen within its domain of definition.

\subsection{Invariance properties of correlation functions of primary fields}

It follows from the commutators~\eqref{eq:hzerocommutator}-~\eqref{eq:honecommutator} and the corresponding exponentiated versions that the correlation functions inherit the invariance properties with respect to global conformal transformations. If the vector~$v$ is primary, i.e., if Virasoro operators with positive index map it to zero, $L_n v =0$, $n>0$, the relations~\eqref{eq:hzerocommutator}--\eqref{eq:honecommutator} generalize to higher order Virasoro elements,
\begin{align}
  [L_n,\Y(v,z)] = z^{n+1}\partial_z \Y(v,z) + (\wt v)\, (n+1)\,z^n \Y(v,z) \,.
\end{align}
These relations generalize to module operators and also intertwiners~\cite{Frenkel:1993fv}, if we again define primary elements in a module $A$ to be vectors $a \in A$ such that $L_n a=0$, $n>0$,
\begin{align}
  [L_n,\iY(a,z)] = z^{n+1}\partial_z \iY(a,z) + (\wt a) \,(n+1)\,z^n \iY(a,z) \,.
\end{align}

A basic assumption of CFT is that this infinitesimal symmetry lifts to a local action which gives rise to a change of variables. More precisely, the genus-0 and genus-1 correlation functions of intertwiners evaluated at homogeneous primary vectors $a_1, \ldots, a_n$ should be, up to an overall factor depending on the transformation as well as the conformal weight, invariant under conformal mappings. That is, for a conformal map $z\mapsto w(z)$ we have
\begin{align}
  F((a_1,w(z_1)),\ldots,(a_n,w(z_n))) = \prod_{i}^n \left(\left.\frac{dw}{dz}\right|_{z=z_i}\right)^{-\wt a_i} F((a_1,z_1),\ldots,(a_n,z_n)) \,,
\end{align} 
where \(F = F^{(0)}\) or \(F = F^{(1)}\). Here, we use the terms conformal and holomorphic interchangeably --- meaning a mapping specified by a holomorphic function.

We proceed to consider correlation functions evaluated on a set of $n$ complex coordinates $\zeta_1,\ldots, \zeta_n$ on the complex plane, with constant imaginary part, $\Im(\zeta_j) = \Im(\zeta_{j'}) = \theta$, for all $j,j' = 1,\ldots,n$ and equally spaced real part, $\Re(\zeta_j) = j \mathsf{d} + \mathsf{d}_0$. Here, $\mathsf{d}_0 \geq 0$, $\mathsf{d}>0$ are positive real numbers, the first one possibly being  zero. We call such a set of points \emph{equally spaced on a line}, $\mathsf{d}_0$ the offset, and $\mathsf{d}$ the minimal distance between the insertion points, or sometimes the \emph{ultraviolet cutoff}. The configuration of points is illustrated in Figure~\ref{fig:plane}. Applying the conformal invariance property to the conformal transformation
\begin{align}
  z&\mapsto w(z):=e^{-z}=e^{-x}e^{-i y}\,,\quad z = x + i y\,,\;x,y \in \Rl\,,\label{eq:cylindermap}
\end{align}
leads to the following reparametrization of correlation functions with insertion points given by the coordinates $\zeta_1,\ldots, \zeta_n$, since \(\zetap_j := w(\zeta_j) = z q^j\), for $0< q = e^{-\mathsf{d}} < 1$, $z = e^{-\mathsf{d}_0} e^{i \theta }$.

\begin{figure}[t]
  \centering
  \includegraphics[scale=1.1]{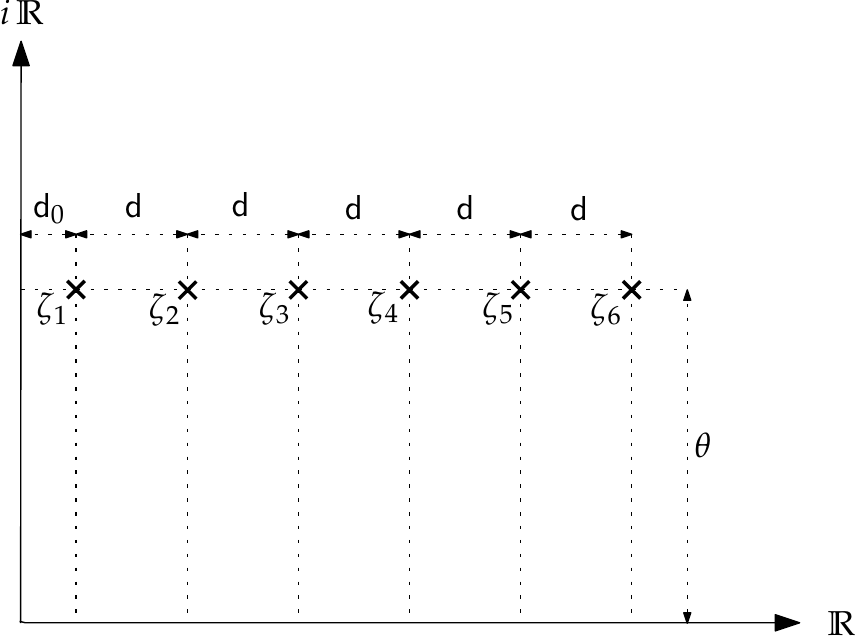}
  \caption{Configuration of insertion points on the plane before applying the conformal mapping. All points have a fixed imaginary part and equispaced real parts.\label{fig:plane}}
\end{figure}

\begin{observation}\label{obs:conformalmapcorrfuncplane}
  Let $A^{(1)}, \dots, A^{(n)}, B^{(0)},\dots,B^{(n)}$ be unitary modules of a VOA~$\cV$, with \(B^{(0)}=B^{(n)}=\cV\), the adjoint module. Assume that $S^{(j)} \subset A^{(j)}$ is a finite-dimensional linear subspace of primary vectors, $a_j \in S^{(j)}$ are homogeneous vectors therein and $\iY_1, \dots, \iY_n$ is a set of intertwiners, where $\iY_j$ is of type $\itype{B^{(j-1)}}{A^{(j)}}{B^{(j)}}$. Consider $n$ equally spaced points on a line $\zeta_1,\ldots,\zeta_n$ with offset $\mathsf{d}_0$ and minimal distance $\mathsf{d}$. Then the corresponding vacuum-to-vacuum correlation function can be written as (cf. Eq.~\eqref{eq:definitiongenuszeronpointspr})
  \begin{align}
    F^{(0)}((a_1,\zeta_1),\ldots,(a_n,\zeta_n))& = (-z)^{\sum_j \wt a_j} q^{\sum_j j \wt a_j}\,\Scp{\1}{\iY_1(a_1,\zetap_1) \cdots \iY_n(a_n,\zetap_n)\1 }\,,
  \end{align}
  for all $a_j \in S_j$ with the identifications \(\zetap_j = zq^j\) and $0< q = e^{-\mathsf{d}} < 1$, $z = e^{-\mathsf{d}_0} e^{i \theta }$.
\end{observation}

This shows that we need to consider correlation functions defined at powers of a positive variable, which can be chosen to be smaller than one. This is a key technical step in our discussion: it allows us to construct operators out of intertwiners, which are bounded in norm. This in turn is a necessary step in our approximation argument, and the error bound also depends on this norm.

Given this preview we clearly want a similar statement for torus correlation functions. Recall that we parametrized the torus as the punctured plane modulo the relations $z = z\torusq^k$, $k \in \Z$, $\torusq$ being the diameter of the torus.  The conformal mapping defined by the principal branch of the complex logarithm $z \mapsto \log z$ then maps the torus to the twice periodic strip (or patch)
\begin{align}
  \mathfrak{T}_\torusq = \C \,/\,(\log(1/\torusq)\Z + 2\pi i \Z) \,,
\end{align}
which is periodic both in the real as well as in the imaginary part. Again starting from a set of equally spaced insertion points $\zeta_1,\ldots,\zeta_n$, $\zeta_j = \mathsf{d}_0 + \log(1/\torusq) - \mathsf{d} j + i \theta$ on a line, now assumed to lie \emph{within} $\mathfrak{T}_\torusq$, we see that they are the image of points $zq^j$ under the complex logarithm, $z = e^{\mathsf{d}_0} e^{i \theta}$, $q=e^{-\mathsf{d}}$. This configuration of points is illustrated in Figure~\ref{fig:torusvm}. The requirement that all points $\zeta_j$ do not lie on the boundary of $\mathfrak{T}_\torusq$ translates into a non-zero offset $0<\mathsf{d}_0 < \mathsf{d}$, hence $q<1/|z|<1$, as well as $\torusq \leq q^n$. Applying again the conformal invariance of correlation functions leads to the following genus-1 version of Observation~\ref{obs:conformalmapcorrfuncplane}.

\begin{figure}[t]
  \centering
  \subcaptionbox{Configuration of insertion points on the torus, here pictured as the annulus with periodic boundary conditions.\label{fig:torusann}}[.45\linewidth]{\includegraphics[scale=0.5]{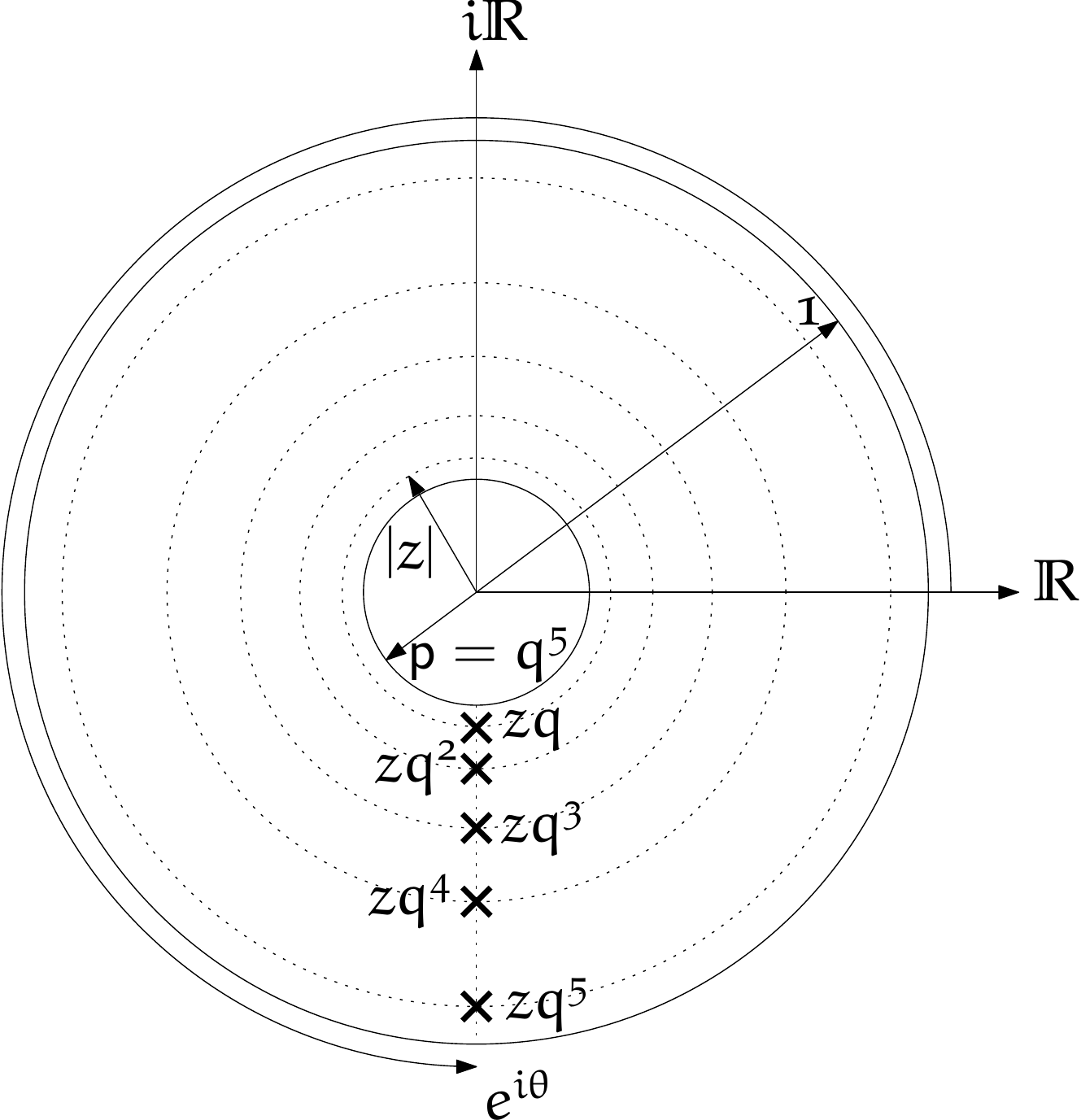}} \qquad%
  \subcaptionbox{Configuration of insertion points on the periodic strip.\label{fig:torusstrip}}[.45\linewidth]{\includegraphics[scale=0.5]{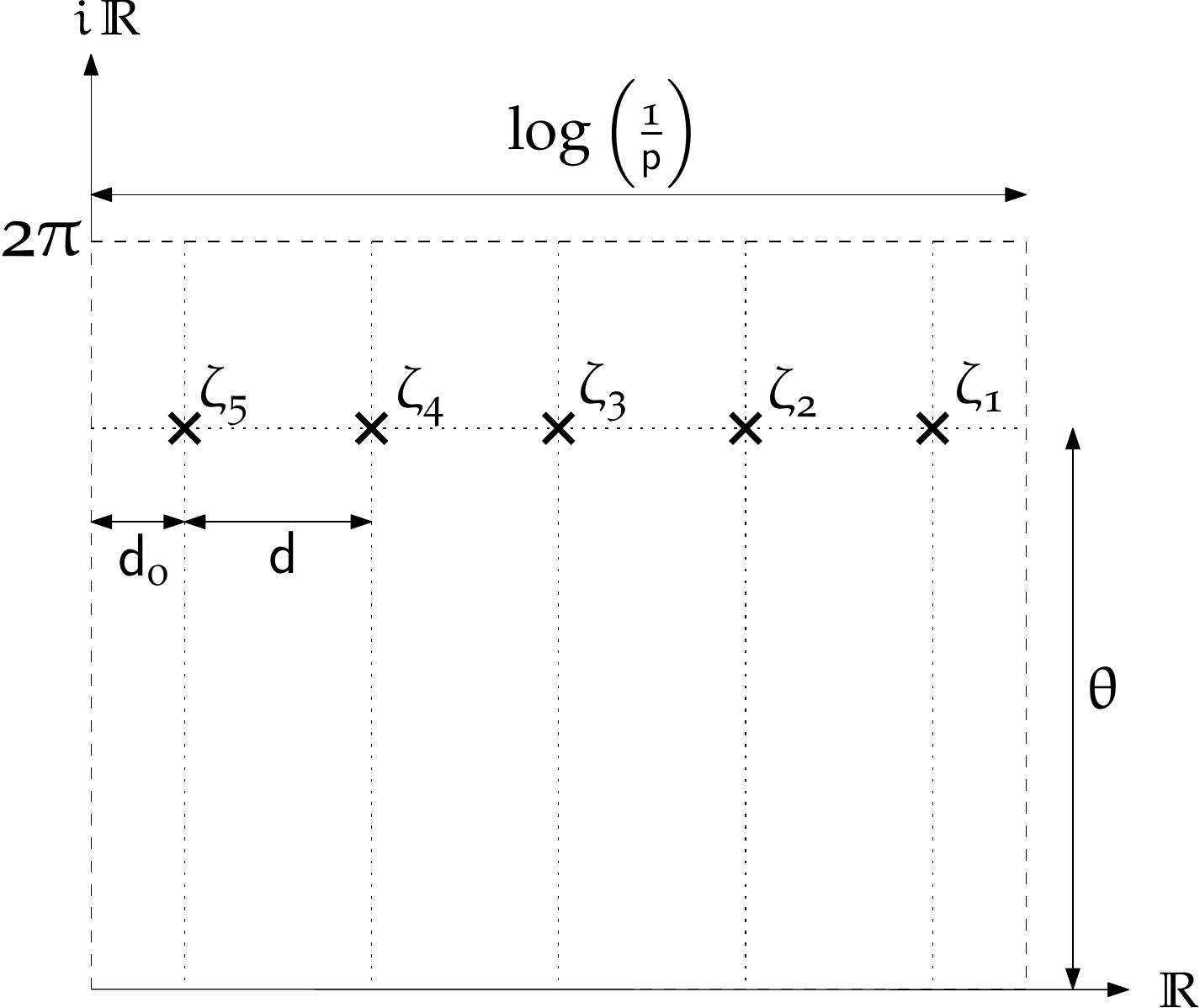}}%
    \caption{(a) Configuration of insertions on $\mathbb{C}$, $\torusq = q^5<zq^5<zq^{4}<\cdots<zq<1$, with the choice $n=5$ for illustration. (b) The image of (a) under the principal branch of the complex logarithm. Note that because the strip is periodic with  period $\log(1/\torusq)$, we can first take the logarithm and then translate the whole image by $-\log(1/\torusq) = -n\mathsf{d}$, the result of which is illustrated here.\label{fig:torusvm}}
\end{figure}

\begin{observation}\label{obs:conformalmapcorrfunctorus}
  Let $A^{(1)}, \dots, A^{(n)}, B^{(0)},\dots,B^{(n)}$ be unitary modules of a vertex operator algebra $\cV$, with \(B^{(0)} = B^{(n)}\). Assume that $S^{(j)} \subset A^{(j)}$ is a finite-dimensional linear subspace of primary vectors, $a_j \in S^{(j)}$ are homogeneous elements and $\iY_1, \dots, \iY_n$ is set of intertwiners, where $\iY_j$ is of type $\itype{B^{(j-1)}}{A^{(j)}}{B^{(j)}}$. Choose $n$ equispaced points $\zeta_1,\ldots,\zeta_n$ within the periodic strip of height $2\pi$ and length $\log(1/\torusq)$ as illustrated in Figure~\ref{fig:torusstrip}. The corresponding correlation function with insertions as these points is given by the torus correlation function
  \begin{align}
    F^{(1)}_{\torusq}((a_1,zq),\dots,(a_n,zq^n)) = z^{\sum_j \wt a_j} q^{\sum_j j \wt a_j}\, \tr_{B^{(n)}}\left(\iY_1(a_1,\zetap_1)\cdots \iY_n(a_n,\zetap_n)\torusq^{L_0-c/24}\right)
  \end{align}
  for all $a_j \in S_j$ under the identifications \(\zetap_j = zq^j\) and $z=e^{\mathsf{d}_0 + i \theta}$, $|z|>1$, $q=e^{-\mathsf{d}} < 1$.
\end{observation}

We see that in the genus-0 case, the parameter~$z$ can be chosen to be of modulus one, while in the genus-1 case, its modulus always has to be greater than one. In the following, the parameter $z$ appearing in the correlation functions is to be interpreted as characterizing the offset as well as the imaginary part of points on the line, as described in the previous observations. However, in view of the fact that correlation functions of VOAs are determined by setting a formal variable to some complex number, it is natural to consider the variable $z$ in the expression $\iY(a,\zeta_j) = \iY(a,zq^j)$ first as a formal variable, which is then set to a complex number if the intertwiner is evaluated within a correlation function. 

In order to motivate our further considerations, let us examine more closely the case of two-point correlation functions, that is, expressions involving two intertwiners. We note that due to the covariance property of intertwiners with respect to dilations we have for primary vectors \(a_1\), \(a_2\),
\begin{align}
  q^{\wt a_1} q^{2 \wt a_2} \iY(a_1,\zetap_1) \iY(a_2,\zetap_2) &=q^{\wt a_1} q^{2 \wt a_2} \iY(a_1,zq) \iY(a_2,zq^2) \\
  &= q^{L_0} \iY(a_1,z) q^{-L_0} q^{2L_0} \iY(a_2,z) q^{-2L_0} \notag\\
  &=  q^{L_0/2} \,q^{L_0/2} \iY(a_1,z) q^{L_0/2} q^{L_0/2} \iY(a_2,z) q^{L_0/2} q^{-5/2 \,L_0}\,.
\end{align}
However, the last power of the grading operator $L_0$ vanishes if applied to the vacuum, or can be absorbed with the power appearing in the expression of the torus correlation functions. We hence see that the formal operator $q^{L_0/2} \iY(a,z) q^{L_0/2}$ can be used to build up the correlation functions of interest. Moreover, we are mainly concerned with the situation where~$a$ lies in a linear subspace consisting of primary vectors. These considerations motivate the following definition. 

\begin{definition}[Scaled intertwiner]\label{def:scaledintertwiner}
  Let $\cV$ be a VOA, $A, B,C$ unitary modules of $\cV$, $S\subset A$ a linear subspace, $\iY$ a intertwiner operator of type $\itype{C}{A}{B}$. For $0<q<1$, we define the 
$q$-\emph{scaled intertwiner}~$\W_q$ of type $\itype{C}{A}{B}$, 
\begin{align}
\begin{matrix}
  \W_q(\cdot,z):&S & \rightarrow &  \Endint{B,C}{z,z^{-1}}\\
&a & \mapsto & \W_q(a,z)\, 
\end{matrix}
\end{align}
by
\begin{align}
\W_q(a,z) b :=
q^{L_0/2} \iY(q^{L_0/2}\,a,z)q^{L_0/2} b\label{eq:scaledintertwinerdef}
\end{align}
for $b \in B$.
\end{definition}

In the following, we often omit the exact dependence of $\W$ on $q$ or the subspace $S \subset A$ with the understanding that, whenever we write~$\W$, such a subspace~$S$ exists. In our applications, this subspace is a proper one. However, for most intermediate results this is not necessary, and~$S$ can be equal to the whole of $A$. The expression~$\W$ is defined for this $S$ and any $0<q<1$ as above.
\noindent
Scaled intertwiners with different scaling parameters \(q_1, q_2\) are related to each other as follows: for all $a\in S$ and $b\in B$, we have
\begin{align}
  \W_q(a,z)&=q_1^{L_0/2}\W_{q_2}(q_1^{L_0/2}a,z) q_1^{L_0/2}\qquad\textrm{ where }\qquad q=q_1q_2\ \label{eq:scaledintertwinerfactorization}
\end{align}
since
\begin{align*}
\W_{q_1q_2}(a,z)&=q_1^{L_0/2}q_2^{L_0/2}\iY(q_2^{L_0/2}q_1^{L_0/2}a,z)q_2^{L_0/2}q_1^{L_0/2}\\
&=q_1^{L_0/2}\W_{q_2}(q_1^{L_0/2}a,z)q_1^{L_0/2}\ 
\end{align*}

We see that scaling the original intertwiner makes its properties much nicer --- this will also be the key ingredient for our approximation result. As discussed, the composition of scaled intertwiners can be used to express correlation functions. The next definition and the following results make this intuitive statement clear.

\subsection{Reconstruction of correlation functions\label{sec:reconstructcorrfunct}}

\noindent With the definition of a scaled intertwiner at hand, we can show that genus-$0$ and genus-~$1$ correlation functions are reproduced exactly. For this purpose, we will introduce a certain operator~$\Top$ which we call the (formal) {\em transfer operator}, adopting the language used for MPS. Again, we emphasize that variables $z_j$ appearing in expressions involving intertwiners are to be interpreted as formal variables, which become complex valued only in the context of correlation functions.

\begin{definition}[Formal transfer operators]\label{def:transferoperators}
  Let $A^{(i)}$, $i=1,\ldots,n$ and $B^{(i)}$, $i=0,\ldots,n$ be unitary modules of a VOA~$\cV$. For $i=1,\ldots,n$, let $\iY_i$ a unitary intertwiner of type $\itype{B^{(i-1)}}{ A^{(i)}}{B^{(i)}}$, $S^{(i)}\subset A^{(i)}$ a subspace, and let 
\begin{align*}
  \W_i=\W_{i,q}(\cdot,z):S^{(i)} \rightarrow \Endint{B^{(i)},B^{(i-1)}}{z,z^{-1}}
\end{align*} 
be the associated scaled intertwiner for some $0<q<1$. For $a_i\in S^{(i)}$, $i=1,\ldots,n$, we define the formal transfer operator $\Top\in \Endint{B^{(n)},B^{(0)}}{z_1,z_1^{-1},\ldots,z_n,z_n^{-1}}$ by 
\begin{align}
  \Top&=\W_1(a_1,z_1) \circ \W_{2}(a_{2}, z_2 ) \circ \cdots \circ \W_n(a_n,z_n)  .\label{eq:operatorz1def}
\end{align}
We call this the {\em transfer operator with insertions $\{a_i\}_{i=1}^n$}. 
\end{definition}
For the genus-$1$-case, we will consider periodic boundary conditions, where $B^{(0)}=B^{(n)}=B$, such that the transfer operator is simply $\Top \in \Endint{B}{z_1,z_1^{-1},\ldots,z_n,z_n^{-1}}$. We will also occasionally specialize to translation-invariant systems, which are particularly natural for MPS: here $A^{(i)}=A$, \(S^{(i)}=S\), $B^{(i)}=B^{(0)}=B$, $\iY_i=\iY$ for $i=1,\ldots,n$ are all identical. We will argue below that the transfer operator encodes $n$-point correlation functions (both on the plane and the torus) for equidistant insertion points~$\{\zeta_j\}_{j=1}^n$ as expressed in Observations~\ref{obs:conformalmapcorrfuncplane} and~\ref{obs:conformalmapcorrfunctorus}. The parameter~$z$, which is chosen to be the same for each intertwiner, will determine the `offset' of the sequence of insertions.

\begin{lemma}\label{lem:trsfopexplicit}
Let $\Top$ be a formal transfer operator with insertions $\{a_i\}_{i=1}^n$ and parameter $q$ as in Definition~\ref{def:transferoperators}.  
Then we have
\begin{align}
\Top&=q^{-L_0/2}\left(\iY_1(\hat{a}_1,z_1 q) \iY_2(\hat{a}_2,z_2 q^2)\cdots \iY_n(\hat{a}_n,z_n q^n)\right)q^{(n+1/2)L_0}\ .\label{eq:zonexactexpanded}
\end{align}
as an identity for formal series in $\Endint{B^{(n)},B^{(0)}}{z_1,z_1^{-1},\ldots,z_n,z_n^{-1}}$, where
\begin{align}\label{eq:ajzjdef}
\hat{a}_j&=q^{(j+1/2)L_0}a_j \,.
\end{align}
\end{lemma}

\begin{proof}
  The formal transfer operator $\Top$ is defined recursively by $\Top=\Top_1$, where 
\begin{align*}
  \Top_n(b)&=\W_n(a_n)b\\
  \Top_{k}(b)&=\W_{k}(a_{k}) \Top_{k+1}(b)\qquad\textrm{ for }k=n-1,\ldots,1\
\end{align*}
for $b\in B$. To show~\eqref{eq:zonexactexpanded} we first argue that for every $k=1,\ldots,n-1$, we have the relation
\begin{align}
q^{(k-1/2)L_0}\Top_k(b)&=\iY_k(\hat{a}_k,z_k q^k)q^{((k+1)-1/2)L_0}\Top_{k+1}(b)\ .\label{eq:recursionrelationforproof}
\end{align}
Indeed, this follows from
\begin{align*}
q^{(k-1/2)L_0}\Top_k(b)&=
q^{(k-1/2)L_0} \W_k(a_k)\circ \Top_{k+1}(b)\\
&=q^{(k-1/2)L_0}q^{L_0/2}\iY_k(q^{L_0/2}a_k,z)q^{L_0/2}\Top_{k+1}(b)\\
&=q^{kL_0}\iY_k(q^{L_0/2}a_k,z)q^{-kL_0}q^{(k+1/2)L_0}\Top_{k+1}(b)\\
&=\iY_k(q^{(k+1/2)L_0}a_k,q^kz)q^{(k+1/2)L_0}\Top_{k+1}(b)\ ,
\end{align*}
where we used the dilation property~\eqref{eq:dilationintegrated} applied to intertwiners in the last step. For later use, we also point out that the same reasoning gives
\begin{align}
q^{(n-1/2)L_0}\Top_n(q^{-nL_0}b)&=\iY_n(\hat{a}_n,z_n q^n) q^{L_0/2}b\ .\label{eq:endpointinduction}
\end{align}
Applying~\eqref{eq:recursionrelationforproof} inductively gives
\begin{align*}
q^{L_0/2}\Top_1(b)&=
\iY_1(\hat{a}_1,z_1q)\iY_2(\hat{a}_2,z_2q^2)\cdots \iY_{n-1}(\hat{a}_{n-1},z_{n-1}q^{n-1})q^{(n-1/2)L_0}\Top_n(b)
\end{align*}
and combining this with~\eqref{eq:endpointinduction}, we conclude that
\begin{align*}
q^{L_0/2}\Top_1(q^{-nL_0}b)&=\iY_1(\hat{a}_1,z_1 q)\iY_2(\hat{a}_2,z_2 q^2)\cdots \iY_n(\hat{a}_n,z_n q^n)q^{L_0/2}b\qquad\textrm{ for all }b\in B\  .
\end{align*}
This implies
\begin{align*}
\Top_1&=q^{-L_0/2} q^{L_0/2}\Top_1q^{-nL_0}q^{nL_0}\nonumber\\
&=q^{-L_0/2}\left(\iY_1(\hat{a}_1,z_1 q)\iY_2(\hat{a}_2,z_2 q^2)\cdots \iY_n(\hat{a}_n,z_n q^n)\right)q^{(n+1/2)L_0}\
\end{align*}
which is the claim~\eqref{eq:zonexactexpanded}. 
\end{proof}

We stress again that the previous identity is defined in terms of equality of power series. Taking matrix elements and replacing the indeterminates by complex numbers then gives rise to correlation functions, if the series converges for the given choice of complex variables. In the following, we will choose all variables equal to a single complex number $z$, as motivated by Observations~\ref{obs:conformalmapcorrfuncplane} and~\ref{obs:conformalmapcorrfunctorus}. In the genus-1 case, the modulus of this number \(z\) has to be bigger than one. With this choice, an immediate consequence of Lemma~\ref{lem:trsfopexplicit} is the fact that transfer operators encode correlation functions in the following sense.

\begin{corollary}[Exact reproduction of genus-$0$ correlation functions]\label{cor:corrfctexactzero}
  Let $\Top:B^{(n)}\rightarrow B^{(0)}$ be a transfer operator  as in Definition~\ref{def:transferoperators} with homogeneous insertions~$\{a_i\}_{i=1}^n$, parameter~$0<q<1$ and the indeterminates~$\{z_j\}_{j=1}^n$ replaced by a complex number $z\in \C \setminus \{0\}$. Let $v^{(0)}\in B^{(0)}$ and $v^{(n)}\in B^{(n)}$ be homogeneous elements. Then
\begin{align}
  \langle v^{(0)}, \Top v^{(n)}\rangle_{B^{(0)}}=q^{(n+1/2)(\wt v^{(n)})+(\wt v^{(0)})/2+\sum_{j=1}^n j\,\wt a_j}\cdot F^{(0)}_{v^{(0)},v^{(n)}}((\tilde{a}_1,\zetap_1),\ldots,(\tilde{a}_n,\zetap_n))\ ,
\end{align}
where $F^{(0)}_{v^{(0)},v^{(n)}}$ is the genus-$0$ correlation function (cf.~\eqref{eq:definitiongenuszeronpointspr}), and
\begin{align}
\tilde{a}_j=q^{L_0/2}a_j\qquad\textrm{ and }\qquad \zetap_j=z q^j\qquad\textrm{ for }j=1,\ldots,n\, .\label{eq:ajzjdef}
\end{align}
Moreover, the sequence of points $(\zetap_1,\ldots,\zetap_n)$ belongs to the domain~\eqref{eq:genuszerodomain}. 
\end{corollary}

\begin{proof}
  According to~\eqref{eq:zonexactexpanded}, the expression~$\langle v^{(0)},\Top v^{(n)}\rangle$ has the form
\begin{align*}
  &\langle v^{(0)},\Top v^{(n)}\rangle_{B^{(0)}}= \\
  &\quad=\langle v^{(0)},q^{-L_0/2}\left(\iY_1(q^{(1+1/2)L_0}a_1,zq)\cdots \iY_n(q^{(n+1/2)L_0}a_n,zq^n)\right)q^{(n+1/2)L_0}v^{(n)}\rangle_{B^{(0)}}\\
  &\quad=\langle q^{L_0/2}v^{(0)},\left(\iY_1(q^{(1+1/2)L_0}a_1,zq)\cdots \iY_n(q^{(n+1/2)L_0}a_n,zq^n)\right)q^{(n+1/2)L_0}v^{(n)}\rangle_{B^{(0)}}\\
  &\quad=q^{(n+1/2)(\wt v^{(n)})+(\wt v^{(0)})/2+\sum_{j=1}^n j\,\wt a_j}\cdot \langle v^{(0)},\iY_1(q^{L_0/2}a_1,zq)\cdots \iY_n(q^{L_0/2}a_n,zq^n)v^{(n)}\rangle\ .
\end{align*}
where we used the definition of $\hat{a}_j$  and the insertion points~$z_j$, the fact that~$L_0$ is self-adjoint with respect to $\langle\cdot,\cdot\rangle_{B^{(0)}}$, as well as the assumption that all the vectors $\{a_i\}_{i=1}^n$ and $v^{(0)},v^{(n)}$ are homogeneous. Since we have $0<q<1$, it follows that $|zq|>|zq^2|>\ldots>|zq^n|>0$ and hence any matrix element of the transfer operator is a well-defined absolutely convergent power series (see Sect.~\ref{sec:correlationfunctions}). 
\end{proof}

Similarly, we obtain genus-$1$-correlation functions by taking the trace of the formal transfer operator~$\Top$ and again interpreting the formal variables as complex numbers. 

\begin{corollary}[Exact reproduction of genus-$1$ correlation functions]\label{cor:corrfctexactrepr}
  Assume periodic boundary conditions, i.e.,~$B^{(0)}=B^{(n)}=B$, and let $\Top:B\rightarrow B$ be a transfer operator with insertions $\{a_i\}_{i=1}^n$ as in Definition~\ref{def:transferoperators} and indeterminates replaced by a complex number $z$, with $|z|>1$ and  $0<q<1/|z|^2$. Then for any $0<r\leq 1$, the value $\tr_B \Top r^{L_0}$ is related to the genus-one correlation function
~$F^{(1)}_{\torusq}=F^{(1)}_{\torusq,\iY_1,\ldots,\iY_n}$ (Eq.~\eqref{eq:genusoneintertwinercorrelationfct}) by
\begin{align}
\tr_B \Top r^{L_0} &=\torusq^{c/24} F^{(1)}_{\torusq}\left((\tilde{a}_1,\zetap_1),\ldots,(\tilde{a}_n,\zetap_n)\right) \label{eq:partitionfunctiontorusregularized}
\end{align}
where
\begin{align}
\tilde{a}_j=q^{L_0/2}a_j\qquad\textrm{ and }\qquad \zetap_j=z q^j\qquad\textrm{ for }j=1,\ldots,n\ ,
\torusq&=rq^{n}\ .\label{eq:ajzjdef}
\end{align}
In particular, under the identification~\eqref{eq:torusrelationq}, the points $\{z_j\}_{j=1}^n$ are equidistant along one fundamental cycle of a torus of diameter~$\torusq=q^n$. The sequence of points $(\zetap_1,\ldots,\zetap_n,\torusq)$ belongs to the domain~\eqref{eq:torusconvergencedomain}, as can easily be verified.
\end{corollary}

\noindent 
The parameter~$r$ is introduced here  may appear somewhat arbitrary at this point: indeed, setting $r=1$, we recover the configuration~\eqref{eq:ajzjdef}  of equidistant insertions on the torus. However, it will have the effect of  ``regularizing'' the expression when we consider truncated intertwiners in the next section.

\begin{proof}
Clearly, the points~$(z_1,\ldots,z_n,\torusq)$ lie in the domain~\eqref{eq:torusconvergencedomain}, hence the rhs.~of  \eqref{eq:partitionfunctiontorusregularized} is well-defined.
The claim follows by taking the trace of the product $\Top r^{L_0}$ using expression~\eqref{eq:zonexactexpanded}.
\end{proof}


\subsection{Proof strategy\label{sec:proofstrategy}}

In Section~\ref{sec:reconstructcorrfunct}, we showed that correlation functions of interest can be expressed in terms of a transfer operator $\Top$. This object is itself a composition of scaled intertwiners~$\W$. The scaling will be essential for us, as it ensures that $\W$ is a bounded operator with respect to the operator norm induced by the scalar product, for finite-dimensional subspaces $S \subset A$. That is, for \(a \in S \subset A\), we will have
\begin{align}
  \norm{\W_q(a,z)} \leq \bnd(q,z)\, \norm{a}_A\,,
\end{align}
for some function \(\bnd(q,z)\), bounded for the parameter regime in \((q,z)\) we are interested in. Here, $\norm{a}_A$ denotes the norm in the unitary module $A$ inherited from the scalar product, and analogously,
\begin{align}
  \norm{\W_q(a,z)} = \sup\{\,|\Scp{c}{\W_q(a,z)b}| \,:\,c \in C\,,\,\norm{c}_C\leq 1\,,\;b \in B\,,\,\norm{b}_B \leq 1\,\}
\end{align}
denotes the operator norm of the linear and densely defined mapping \(\W_q(a,z)\). Of course, this statement only makes sense if the formal variable~\(z\) is replaced by a complex number. This will be the case throughout the next section if norm expression are present. We will choose the same complex number for all scaled intertwiners, although other choices are in principle possible. The complex parameter $z$ is always assumed to be non-zero, and to be bigger than one if appearing in genus-1 expressions. Applying the fact that scaled intertwiners have bounded norms recursively leads to a norm bound for $\Top$.

Section~\ref{sec:approx} starts with the observation that as a first step towards our approximation statement, we have to ensure that the image of a finite-dimensional subspace under the action of a scaled intertwiner is again contained in a fixed finite-dimensional space. However, this is not the case, since an intertwiner consists of terms allowing for arbitrary changes of the weight. The natural idea here is to truncate the Hilbert space with respect to the weight decomposition. This motivates the definition of a \emph{scaled truncated} intertwiner $\W_q^{[N]}$, where $N \in \Nl$ denotes the truncation parameter. This object will have the feature of changing the weight of a vector by at most $N$. The obvious question is how it compares to its original version, and we show that it fulfills
\begin{align}
  \norm{\W_q(a,z) - \W_q^{[N]}(a,z)} \leq \norm{a}_A \,\apperr(q,z) q^{N/4}\,,
\end{align}
where $\apperr(q,z)$ is a function independent of the truncation parameter $N$, and the statement is again with respect to the operator norm. This shows that for large enough $N$, we can safely replace $\W_q$ by its truncated version $\W_q^{[N]}$. Applying this argument recursively leads to an equivalent statement for the transfer operator $\Top$. 

In a last step, we have to ensure that we only have to apply the truncated transfer operator --- obtained by replacing the scaled intertwiners in its definition by their truncated versions --- to fixed finite-dimensional subspaces, so that the image is again a fixed finite-dimensional space which then can be chosen as the bond Hilbert space. In the case of genus-0 vacuum-to-vacuum correlation functions, this is immediate, since the transfer operators is applied to a fixed vector, the vacuum. In the genus-1 case, however, we also need to truncate the trace. After this is achieved, all expressions can be converted to an MPS picture. As in previous sections, we illustrate our findings with WZW models.

\section{Bounded intertwiners\label{sec:boundedint}}

As explained in Section~\ref{sec:proofstrategy}, 
we will argue that a scaled intertwiner defines a bounded operator if the formal variable is replaced by a non-zero complex number. We start with a motivating example, which shows that a special kind of intertwiner for WZW models is bounded, namely module operators. Although we will not use this fact later, the proof idea as well as the result serves as an illustration for the following arguments.  

\subsection{Motivation: energy bounds for WZW models}
Consider the WZW-type VOA~$\voalkzero$ and let us fix some $\fa\in\g\subset \voalkzero$. (Recall that elements of the Lie algebra $\g\cong\g\otimes t^{-1}$ are identified with vectors of the first level of $\voalkzero$ (see Eq.~\eqref{eq:wzwgradingdef}). 

Now consider the construction of  the module $\modulekflambda{\lambda}$ for WZW-type VOAs $\voalkzero$ as explained in Example~\ref{exmp:wzwmodules}, and the (module) mode operators $\fa(n)$ associated with $\fa$. These are defined by $ \Y_{\modulekflambda{\lambda}}(\fa,z)=\sum_{n\in\mathbb{Z}}\fa(n)z^{-n-1}$ -- here we slightly abuse notation by using the same lettter. From
the construction of the module, the mode operators $\{\fa(n)\}_{n\in\mathbb{Z}}$ also satisfy the commutator rule~\eqref{eq:wzwcommutator}. By the results in~\cite{zbMATH00008298} (see also~\cite{Wassermann:1998cs}), this implies that the  mode operators $\fa(n)$ associated with the unitary modules $\modulekflambda{\lambda}$ of the vertex operator algebra $\voalkzero$ in question satisfy \emph{linear energy bounds}. That is, the operators $\fa(n)$, $\fa \in \g$ satisfy the bound
\begin{align}
  \norm{\fa(n)\chi}_{\modulekflambda{\lambda}} \leq c\cdot (\hat{\eta}(\fa),\fa)^\half \,|n+1|\cdot \norm{(L_0 + \idty) \chi}_{\modulekflambda{\lambda}} \label{eq:faenergybound}
\end{align}
for any $\chi \in \modulekflambda{\lambda}$, where $c > 0$. This implies various continuity statements, for example that the \emph{scaled} module vertex operator for  $\fa\in\g$ and $z\in \C$ 
\begin{align}
  q^{L_0/2} \Y(\fa,z) q^{L_0/2} 
\end{align}
is bounded (and -- more fundamentally, the sum defining it converges in the norm topology).
\begin{proof}
We can decompose an arbitrary element $\chi\in \modulekflambda{\lambda}$ into weight spaces as
\begin{align*}
\chi&=\sum_{m\in\mathbb{N}_0}c_m\chi_m\qquad\textrm{ where }c_m\in\mathbb{C}\textrm{ and }\chi_m\in \modulekflambda{\lambda}(m)\ .
\end{align*}
Without loss of generality, we may assume that $\|\chi_m\|_{\modulekflambda{\lambda}}=1$ for all $m\in\mathbb{N}_0$. This implies that $\|\chi\|_{\modulekflambda{\lambda}}=(\sum_{m} |c_m|^2)^{1/2}$. In particular, of any operator $O$ on the module~$\modulekflambda{\lambda}$ and any $\chi\in \modulekflambda{\lambda}$, we get by the Cauchy-Schwarz inequality
\begin{align*}
\|O\chi\|_{\modulekflambda{\lambda}}\leq \sum_m |c_m|\cdot\|\chi_m\|_{\modulekflambda{\lambda}}\leq \|\chi\|_{\modulekflambda{\lambda}}\cdot \left(\sum_m \|O\chi_m\|^2\right)^{1/2}\ .
 \end{align*} 
 In particular, this implies that the operator norm of $O$ is bounded by
 \begin{align}
 \|O\|^2\leq \sum_{m\in\mathbb{N}_0} \sup_{\substack{\chi\in \modulekflambda{\lambda}(m)\\
  \|\chi\|_{\modulekflambda{\lambda}}\leq 1}} \|O\chi\|^2\ .\label{eq:upperboundoperatornormO}
 \end{align}
We now apply this to bound the operator norm of the  scaled intertwiner~$O=q^{L_0/2} \Y(\fa,z) q^{L_0/2}$. Observe that
for $\chi\in \modulekflambda{\lambda}(m)$, we have
\begin{align*}
q^{L_0/2} \Y(\fa,z) q^{L_0/2}\chi&=q^{(h_\lambda+m)/2}q^{L_0/2}\Y(\fa,z)\chi\\
&=q^{(h_\lambda+m)/2} q^{L_0/2}\sum_{n\in\mathbb{Z}}\fa(n)z^{-n-1}\chi\\
&=q^{h_\lambda+m}\sum_{n\in\mathbb{Z}, n\leq m}q^{-n/2}z^{-n-1}\fa(n)\chi
\end{align*}
where we used the fact that for $\chi\in \modulekflambda{\lambda}(m)$, we have  
\begin{align}
\fa(n)\chi\begin{cases}
\in \modulekflambda{\lambda}(m-n)\qquad &\textrm{  for }n\leq m\textrm{ and }\\
0 &\textrm{otherwise.}
\end{cases}\label{eq:fanannihilation}
\end{align}
In particular, these vectors are orthogonal, and we get
\begin{align*}
\|q^{L_0/2} \Y(\fa,z) q^{L_0/2}\chi\|^2_{\modulekflambda{\lambda}}&=q^{2(h_\lambda+m)}\sum_{n\in\mathbb{Z},n\leq m} q^{-n} |z|^{-2n-2} \|\fa(n)\chi\|^2_{\modulekflambda{\lambda}}\qquad\textrm{ for }\chi\in \modulekflambda{\lambda}(m)\ .
\end{align*}
Hence we obtain the operator norm bound (cf.~\eqref{eq:upperboundoperatornormO})
\begin{align*}
\|q^{L_0/2} \Y(\fa,z) q^{L_0/2}\|^2&\leq
q^{2h_\lambda}\sum_{m\in\mathbb{N}_0}q^{2m}
\sup_{\substack{\chi\in \modulekflambda{\lambda}(m)\\
  \|\chi\|_{\modulekflambda{\lambda}}\leq 1}} \sum_{n\in\mathbb{Z}, n\leq m} q^{-n} |z|^{-2n-2} \|\fa(n)\chi\|^2_{\modulekflambda{\lambda}}\\
  &\leq
  q^{2h_\lambda}\sum_{m\in\mathbb{N}_0}\sum_{n\in\mathbb{Z}, n\leq m}q^{2m-n}
 |z|^{-2n-2} c \cdot (\hat{\eta}(\fa),\fa) \,|n+1|^{2} |m+h_\lambda+1|^2\ 
\end{align*}
where we inserted the energy bound~\eqref{eq:faenergybound}. 
These sums can be bounded, yielding
\begin{align}
\|q^{L_0/2} \Y(\fa,z) q^{L_0/2}\|&\leq \norm{\fa}_{\modulekflambda{\lambda}}\, \bnd(q,z) 
\end{align}
where \(\bnd(q,z)\) is finite as long as \(0<q<\min\{|z|^2,|z|^{-2}\}\) \footnote{Evaluating the sums we arrive at 
  \begin{align}
    \ldots &\leq c^2 \left(\frac{q^{h_\lambda}}{|z|}\right)^{2}\,\norm{\fa}_{\modulekflambda{\lambda}}^2\, \left(\sum_{m\geq 0} \left(\frac{q}{|z|^2}\right)^m [(1+2h_\lambda)(m+1)^{5} + h_\lambda^2]\right. +\\
      &\qquad\qquad\qquad\qquad\qquad\left.\sum_{m\geq 0} q^{2m} [(1+2h_\lambda)(m+1)^2 + h_\lambda^2] \sum_{n \geq 0} (|z|^2\,q)^n (n+1)^{2} \right)\\
    &\leq c^2 \left(\frac{q^{h_\lambda}}{|z|}\right)^{2}\,\norm{\fa}_{\modulekflambda{\lambda}}^2\,\left[(1+2h_\lambda)\,\left(5! \frac{|z|^2}{q} \,\left[\log\left(\frac{|z|^2}{q}\right)\right]^{-6} + \frac{1}{2|z|^2 q^3} \left[\log\left(\frac{1}{q}\right)\,\log\left(\frac{1}{|z|^2 q}\right)\right]^{-3}\right) \right.+\\
      &\qquad\qquad\qquad\qquad\qquad\left. h_\lambda^2\left(\frac{1}{1-|z|^{-2}q} + \frac{1}{2|z|^2 q^3(1-q^2)(-\log(|z|^2q))^3}\right)\right] 
  \end{align}
  where we used Lemma~\ref{lem:estimategammasum} multiple times in the last bound.} and where we used that \((\hat{\eta}(\fa),\fa) = \Scp{\fa(-1)\1}{\fa(-1)\1}_{\modulekflambda{\lambda}} = \norm{\fa}_{\modulekflambda{\lambda}}^2\).   
\end{proof}
The conclusion that scaled module vertex operators are bounded
can in fact be extended to general elements $v\in \voalkzero$ (instead of merely $\fa\in \g$). This can be shown by identical arguments starting from energy bounds on the modes of the module operator~$\Y(v,z)$. The latter have the form
\begin{align}
  \norm{\y(v)_n \chi}_{\modulekflambda{\lambda}} \leq C_v |n+1|^{r_v} \norm{(L_0 + \idty)^{s_v}\chi}_{\modulekflambda{\lambda}}\,,
\end{align}
where $C_v, r_v, s_v$ are constants only depending on $v \in \voalkzero$. Following the discussion in~\cite[Section 6]{Carpi:2015fk}, such energy bounds for modes~$\y(v)_n$ of module operators~$\Y(v,z)$, for $v\in \voalkzero$ can be derived from those for the operators $\fa(n)$, $\fa\in\g$.
\subsection{$S$-Boundedness and implications}
Motivated by the energy bounds described in the context of WZW models, and their implication that they turn a ``scaled'' version of the module operator into a bounded operator, we now generalize these definitions and establish corresponding results. For this purpose, we  first introduce a certain form on the tensor product of two modules. We note that a similar object was already studied by Felder et. al.~\cite{Felder:1990tw}, where the boundedness was however \emph{assumed}. The principal idea is to use an intertwiner to construct a new Hilbert space out of the algebraic tensor product of two modules. Using the existence of genus-0 correlation functions, this then shows that for non-zero values of $z$ such that $0<q<1/|z|^2$, the scaled intertwiner at value $z$ is a densely defined operator. Next, we use the existence of genus-1 correlation functions to show that this operator is actually bounded.

\begin{lemma}\label{lem:scalarproductvertexop}
  Let $\cV$ be a rational and $C_2$-cofinite VOA, let $A,B,C$ be unitary modules of $\cV$, and let $\iY$ be an intertwiner operator of type $\itype{C}{A}{B}$. Let $z\in\mathbb{C}\backslash\{0\}$ be arbitrary and  $0<q< 1/|z|^2$. For $a_1,a_2\in A$ and $b_1,b_2\in B$, define
  \begin{align}
    \Scp{a_1 \otimes b_1}{a_2 \otimes b_2}_{\iY,q,z} := \Scp{q^{L_0/2}\,\iY(q^{L_0/2}a_1,z)q^{L_0/2}b_1}{q^{L_0/2}\,\iY(q^{L_0/2}a_2,z)q^{L_0/2}b_2}_C\label{eq:innerproductdefvertexop}
  \end{align}
Then the map $(a_1\otimes b_1,a_2\otimes b_2)\mapsto     \Scp{a_1 \otimes b_1}{a_2 \otimes b_2}_{\iY,q,z}$
  can be extended to a sesquilinear, densely defined and positive semi-definite form on $A\otimes B$.
\end{lemma}

\begin{proof}
Let us first verify that the expression~\eqref{eq:innerproductdefvertexop}
is well-defined, i.e.,  gives a finite value for any $z\in\mathbb{C}\backslash\{0\}$ and $q$ satisfying $0<q< \frac{1}{|z|^2}$. To do so, we rewrite
it as a genus-$0$-correlation function. 
 
Take $a_1,a_2\in A$ and $b\in B$ arbitrary. Then we have
\begin{align}\label{eq:homogenright}
    &\Scp{a_1 \otimes b_1}{a_2 \otimes b_2}_{\iY,q,z} =
 \langle\iY(q^{L_0/2}a_1,z)q^{L_0/2}b_1,q^{L_0}\iY(q^{L_0/2}a_2,z)q^{L_0/2}b_2\rangle_C\\
&\quad\quad=\langle q^{L_0/2} b_1,\iY_1(e^{zL_1}(-z^{-2})^{L_0}\eta_A(q^{L_0/2}a_1),z^{-1})q^{L_0}\iY(q^{L_0/2}a_2,z)q^{L_0/2}b_2\rangle_C \\
&\quad\quad=\langle q^{L_0/2} b_1,\iY_1(\tilde{a}_1,z^{-1})q^{L_0}\iY(q^{-L_0}\tilde{a}_2,z)q^{L_0/2}b_2\rangle_C
\ .
\end{align}
 In first identity, we used the fact that $L_0$ is self-adjoint with respect to $\langle \cdot,\cdot\rangle_C$, and in the second step, we used the intertwiner~$\iY_1$
introduced in Lemma~\ref{lem:adjointoperatorintertw}. Finally, in the third identity, we introduced the abbreviations
\begin{align*}
\tilde{a}_1:=e^{zL_1}(-z^{-2})^{L_0}\eta_A(q^{L_0/2}a_1)\qquad \tilde{a}_2:=q^{3L_0/2}\tilde{a}_2
\end{align*}
for convenience. 
Using the action of $L_0$ on mode operators, we have
\begin{align*}
q^{L_0}\iY(q^{-L_0}\tilde{a}_2,z)&=\iY(\tilde{a}_2,qz)q^{L_0}\ .
\end{align*}
This means that
\begin{align*}
  \Scp{a_1 \otimes b_1}{a_2 \otimes b_2}_{\iY,q,z} &=
 \langle\iY(q^{L_0/2}a_1,z)q^{L_0/2}b_1,q^{L_0}\iY(q^{L_0/2}a_2,z)q^{L_0/2}b_2\rangle_C\nonumber\\
&=\langle q^{L_0/2}b_1,\iY(\tilde{a}_1,z^{-1})\iY(\tilde{a}_2,qz)q^{3L_0/2}b_2\rangle_C\\
&=F^{(0)}_{v,w}\left((\tilde{a}_1,z_1),(\tilde{a}_2,z_2)\right)
\end{align*}
has the form of a genus-$0$ two-point correlation function  (cf.~\eqref{eq:definitiongenuszeronpointspr})
with 
$v:=q^{L_0/2}b_1$, $w:=q^{3L_0/2}b_2$, and insertions of~$\tilde{a}_1$, as well as~$\tilde{a}_2$ at
\begin{align*}
z_1 &:=z^{-1}\qquad\textrm{ and }\qquad z_2:=qz\ .
\end{align*}
By assumption on $z$ and $q$, we have $|z^{-1}|>|qz|$, hence $(z_1,z_2)$ lie in the domain~\eqref{eq:genuszerodomain}.
Since we assumed that the VOA~$\cV$
is rational and $C_2$-co-finite,
 it follows 
 from the results of Huang~\cite{Huang:2005gs} (see Section~\ref{sec:correlationfunctions})
  that the expression
  \begin{align}
    \langle\iY(q^{L_0/2}a_1,z)q^{L_0/2}b_1,q^{L_0}\iY(q^{L_0/2}a_2,z)q^{L_0/2}b_2\rangle_B
  \end{align}
  is finite for any $a_1,a_2\in A$ and $b\in B$. Extending the definition linearly to finite sums of the form $\sum_i a_i \otimes b_i \in A \otimes B$, and using the fact that the latter are dense in $A \otimes B$, it follows that $\Scp{\cdot}{\cdot}_{\iY,q,z}$ is indeed a sesquilinear, densely defined and positive semi-definite form on $A\otimes B$.
  
Let $\{a_\alpha\}_{\alpha}\subset A$
and  $\{b_\beta\}_{\beta}\subset B$    be finite families of  elements in~$A$ and $B$, respectively. To show that $\langle\cdot,\cdot\rangle_{\iY,q,z}$ is positive semi-definite, it suffices to check that for any such families, the matrix
\begin{align}
\left\{\quad \langle q^{L_0/2}\iY(q^{L_0/2}a_{\alpha_1},z)q^{L_0/2}b_{\beta_1},q^{L_0/2}\iY(q^{L_0/2}a_{\alpha_2},z)q^{L_0/2}b_{\beta_2}\rangle_C\quad \right\}_{(\alpha_1,\beta_1),(\alpha_2,\beta_2)}\label{eq:matrixApositiv}
\end{align}
is positive semi-definite. But this is the Gram matrix (with entries given by inner products) associated with the family of vectors
$\{q^{L_0/2}\iY(q^{L_0/2}a_{\alpha}z)q^{L_0/2}b_{\beta}\}_{(\alpha,\beta)}$, hence the claim follows.
\end{proof}

The following property of the sesquilinear form  constructed in Lemma~\ref{lem:scalarproductvertexop} constitutes our main technical step.

\begin{definition}\label{def:boundedintertwiner}
  Let  $\cV$ be a VOA, $A, B, C$ unitary modules of $\cV$,
  and let $S\subset A$ be a linear subspace in the module~$A$.
  Let $\iY$ be an intertwiner operator of type $\itype{C}{A}{B}$ and $ \Scp{\cdot}{\cdot}_{\iY,q,z}$ the sesquilinear form constructed in Lemma~\ref{lem:scalarproductvertexop}.  We call $\iY$ $S$-\emph{bounded} if for any 
$z\in\mathbb{C}\backslash\{0\}$ and 
$0<q<\min\{|z|^2,1/|z|^2\}$, there is a constant $\bnd(q,z)<\infty$ such that
  \begin{align}
    \Scp{a \otimes b}{a \otimes b}_{\iY,q,z} \leq \bnd(q,z)^2 \norm{a}_A^2 \norm{b}_B^2 \
  \end{align}
  for all $a\in S$ and $b\in B$. In this definition, 
  $\langle\cdot,\cdot\rangle_{\cY,q,z}$ is the form defined by Lemma~\ref{lem:scalarproductvertexop}, whereas  $\|\cdot\|_A$ and $\|\cdot\|_B$ are the norms induced by the non-degenerate forms on $A$ and~$B$, respectively.
\end{definition}
 We will verify that a large class of VOAs have the property that
for an appropriate subspace~$S$, all intertwiners   are $S$-bounded.
 We are especially interested in the case where the subspace~$S$ is
the  space of primary vectors, or the space of vectors of highest weight. We can use the known existence results for genus-$1$-correlation functions to establish the following:

\begin{proposition}\label{prop:Sboundedfinitedimensional}
  Let  $\cV$ be a rational and $C_2$-co-finite VOA, $A,B,C$ unitary modules of $\cV$ and $S\subset A$  a finite-dimensional subspace of $A$. Let  $\iY$ be an intertwiner operator of type~$\itype{C}{A}{B}$. Then~$\iY$ is~$S$-bounded.
\end{proposition}

\begin{proof}
Fix $a\in S \subset A$ and  $b\in B$ arbitrary (not necessarily homogeneous). By definition of the sesquilinear form~$\langle\cdot,\cdot\rangle_{\iY,q,z}$  we have
\begin{align}
    \Scp{a \otimes b}{a \otimes b}_{\iY,q,z}&=
\langle q^{L_0/2}\iY(q^{L_0/2}a,z)q^{L_0/2}b,q^{L_0/2}\cY(q^{L_0/2}a,z)q^{L_0/2}b\rangle_C\nonumber\\
&\leq  \|q^{L_0/2}\iY(q^{L_0/2}a,z)q^{L_0/2}\|^2\cdot \|b\|^2_B\ ,\label{eq:sesquilinearformbound}
\end{align}
where $\|\cdot\|$ is the operator norm.  Using an orthonormal basis~$\{a_j\}^{\dim S}_{j=1}$ of $S$, we get
\begin{align}
\|q^{L_0/2}\iY(q^{L_0/2}a,z)q^{L_0/2}\|\leq \sqrt{\dim S}\cdot \|a\|_A\cdot \max_{1\leq j\leq \dim S} \|q^{L_0/2}\iY(q^{L_0/2}a_j,z)q^{L_0/2}\|\ .\label{eq:operatornormbounda}
\end{align}
Combining~\eqref{eq:sesquilinearformbound} and~\eqref{eq:operatornormbounda}, we conclude that it suffices to show that each operator
\begin{align*} 
H(a_j):=q^{L_0/2}\iY(q^{L_0/2}a_j,z)q^{L_0/2}
\end{align*} is bounded. In fact, each of these operators is Hilbert-Schmidt, and this holds not just for basis elements, but for any $a\in S$: we have
using an orthonormal  basis~$\{b_k\}_k$ of~$B$  
\begin{align}
\tr_B(H(a)^*H(a))&=\sum_k \langle q^{L_0/2}\iY(q^{L_0/2}a,z)q^{L_0/2}b_k,q^{L_0/2}\iY(q^{L_0/2}a,z)q^{L_0/2}b_k\rangle_C\nonumber\\
&=\sum_k \langle \iY(q^{L_0/2}a,z)q^{L_0/2}b_k,q^{L_0}\iY(q^{L_0/2}a,z)q^{L_0/2}b_k\rangle_C\nonumber\\
&=\sum_k \langle q^{L_0/2}b_k,\iY_1(e^{zL_1}(-z^{-2})^{L_0}\eta_A(q^{L_0/2}a),z^{-1})q^{L_0}\iY(q^{L_0/2}a,z)q^{L_0/2}b_k \rangle_B\nonumber\\
&=
\sum_k \langle b_k,q^{L_0/2}\iY_1(e^{zL_1}(-z^{-2})^{L_0}\eta_A(q^{L_0/2}a),z^{-1})q^{L_0}\iY(q^{L_0/2}a,z)q^{L_0/2}b_k \rangle_B\nonumber\label{eq:xaxawrittenout}
\end{align}
where we used the fact that $L_0$ is self-adjoint with respect to $\langle\cdot,\cdot\rangle_C$ in the first and last step, and $\iY_1$ is the intertwiner introduced in  Lemma~\ref{lem:adjointoperatorintertw}.
Defining
\begin{align}
&a_1:=q^{L_0/2}e^{zL_1}(-z^{-2})^{L_0}\eta_A(q^{L_0/2}a)\,,
&a_2:=q^{2L_0}a\,,
\end{align}
 we get
\begin{align*}
q^{L_0/2}&\iY_1(e^{zL_1}(-z^{-2})^{L_0}\eta_A(q^{L_0/2}a),z^{-1})q^{L_0}\iY(q^{L_0/2}a,z)q^{L_0/2}\\
&=
q^{L_0/2}\iY_1(q^{-L_0/2}a_1,z^{-1})q^{L_0}\iY(q^{-3L_0/2}a_2,z)q^{L_0/2}
=
\iY_1(a_1,q^{1/2}z^{-1})\iY(a_2,q^{3/2}z)q^{2L_0}\ ,
\end{align*}
where we used the action of $L_0$ on the intertwiners (cf.~\eqref{eq:dilationintegrated}). 
Setting
  \begin{align*}
  z_1:=q^{1/2}z^{-1}\qquad z_2:=q^{3/2}z\qquad\torusq=q^2\ ,
    \end{align*}
    we conclude that
\begin{align}
  \tr_B(H(a)^*H(a))&=
  \tr_B\left(\iY_1(a_1,z_1)\iY(a_2,z_2)\torusq^{L_0}\right)\nonumber\\
  &=\torusq^{c/24}F^{(1)}_{\torusq,\iY_1,\iY}\left((z_1^{-L_0}a_1,z_1),(z_2^{-L_0}a_2,z_2)\right)
\end{align}
is proportional to a genus-$1$ two-point function.
By the assumption that the VOA~$\cV$
is rational and $C_2$-co-finite, this is finite 
since the triple $(z_1,z_2,\torusq)$ satisfies
$1>|z_1|>|z_2|>\torusq$
for any $0<q<\min\{|z|^2,1/|z|^2\}$
(see Section~\ref{sec:correlationfunctions} and Ref.~\cite{Huang:2005gs}).
\end{proof}

The notion of bounded intertwiners gives rise to a family of bounded operators, mapping the Hilbert subspace~$S$ into the bounded operators between the Hilbert spaces~$B$ and $C$. By inspection, we see that these are exactly given by scaled intertwiners~\eqref{eq:scaledintertwinerdef}, with the indeterminate replaced by a complex number. In fact, due to the special structure, these operators are also of trace class.

\begin{corollary}\label{cor:boundednessscaledintertwiner}
  Let $\cV$ be a rational and $C_2$-co-finite VOA, $A,B,C$ unitary modules of $\cV$, $S\subset A$ a subspace of $A$, \(\iY\) an \(S\)-bounded intertwiner of type~$\itype{C}{A}{B}$ and $z\in\mathbb{C}\backslash\{0\}$, $0<q<\min \{|z|^2,1/|z|^2\}$ be arbitrary. Then the associated scaled intertwiner $\W$ with the formal variable replaced by the complex number \(z\) is a bounded operator, with operator norm bounded by
  \begin{align}\label{eq:Sboundednessscaledintertwiner}
    \norm{\W_q(a,z)} \leq \bnd(q,z) \norm{a}_A\, \,,\quad \textrm{for all } \,a \in S\,,
  \end{align}
  where  $\|\cdot\|_A$ is the norm induced by the non-degenerate form on $A$. Furthermore, the operator $\W_q(a,z)$ is also of trace-class.
\end{corollary}

In the following, we refer to this operator as an \(S\)-bounded scaled intertwiner of type \(\itype{C}{A}{B}\), with the understanding that \(A,B,C\) are unitary modules and \(S\subset A\) is a linear subspace.

\begin{proof}
  Eq.~\eqref{eq:Sboundednessscaledintertwiner} follows immediately from the definitions since for \(b \in B\) arbitrary (cf. Lemma~\ref{lem:scalarproductvertexop} and Definition~\ref{def:boundedintertwiner})
\begin{align*}
\|\W_q(a,z)b\|_C^2&=\langle a\otimes b,a\otimes b\rangle_{\iY,q,z}\ .
\end{align*}
For the second assertion, note that by using~\eqref{eq:scaledintertwinerfactorization} with
$q_1=q^{1/2}$, $q_2=q^{1/2}$, we have
\begin{align*}
  \W_q(a,z)b&=q^{L_0/4} \W_{q^{1/2}}(q^{L_0/4}a,z) q^{L_0/4}b
\end{align*}
for all $a\in S$ and $b\in B$.  The H\"older inequality for operators on a Hilbert space implies that the product of a bounded operator and one of trace-class stays in the trace-class, and we find
  \begin{align}
    \norm{\W_q(a,z)}_1 \leq \norm{q^{L_0/4} \W_{q^{1/2}}(q^{L_0/4} a,z)} \cdot \norm{q^{L_0/4}}_1 < \infty\,,
  \end{align}
  since for any $0<q<1$, the operator $q^{L_0}$ is of trace-class, compare to Eq.~\eqref{eq:partitionfunction}. Here, we also used the first assertion as well as the fact that the operator norm of \(q^{L_0/4}\) is bounded by one, since the spectrum of \(L_0\) is positive and we have \(0<q<1\).
\end{proof}

Since we defined the formal transfer operator $\Top$ as the composition of several scaled intertwiners, it follows from the fact that the operator norm is sub-multiplicative that it itself defines a bounded operator.

\begin{lemma}\label{lem:boundedtransferop}
  Let $\Top$ be a transfer operator as in Definition~\ref{def:transferoperators} composed of \(S^{(i)}\)-bounded intertwiners, insertions $a_i \in S^{(i)}$, \(i = 1,\ldots,n\), and the formal variables replaced by a complex number $z\in\mathbb{C}\backslash\{0\}$ satisfying $0<q<\min \{|z|^2,1/|z|^2\}$.   
Then the operator norm of $\Top$ is bounded by
\begin{align}
  \|\Top\|\leq \prod_{j=1}^n \bnd_j(q,z) \|a_j\|_{A^{(j)}}\ .\label{eq:zonenormboundv}
\end{align}
\end{lemma}

\begin{proof}
The operator~\eqref{eq:operatorz1def} is defined recursively by $\Top=\Top_1$, where 
\begin{align*}
  \Top_n&=\W_n(a_n,z)\\
  \Top_{k}&=\W_{k}(a_{k},z) \circ \Top_{k+1}\qquad\textrm{ for }k=n-1,\ldots,1\
\end{align*}
for $b\in B$.  Inductively using~\eqref{eq:Sboundednessscaledintertwiner}, this immediately implies that  
\begin{align*}
  \|\Top_1\|^2 \leq  \prod_{j=1}^n \bnd^2_j(q,z) \|a_j\|^2_{A^{(j)}} \,,
\end{align*}
which is the claim~\eqref{eq:zonenormboundv}. 
\end{proof}
  
We close this section by again examining WZW models, but now intertwiners instead of module operators. These VOAs are rational and unitary, and hence the corresponding scaled intertwiners are, for suitable choices of \(S\), bounded by the results of this section. However, the exact analytic form of the parameter $\bnd(q,z)$ is in general unclear. Nevertheless, for certain choices of the Lie algebra and the irreducible modules, explicit bounds on the boundedness parameter $\bnd(q,z)$ can be obtained.

\begin{example}[\bf WZW intertwiners are \(S\)-bounded]\label{exmp:wzwbounded}
  Let us again fix a simple complex Lie algebra $\g $, which for concreteness is now assumed to be an element of the A series, hence $\g = \mathfrak{sl}(d,\C)$. This choice leads to models known in the physics literature as $\mathfrak{su}(2)$ at level~$k$, where the level is given by the choice of the value of the central extension. Let us now fix three irreducible highest weight modules of \(\g \) corresponding to the weights $\lambda_1$, $\lambda_2$, $\lambda_3$. As explained in Examples~\ref{exmp:wzwmodules} and~\ref{exmp:wzwintertwiner}, this also uniquely determines three irreducible modules \(\modulekflambda{\lambda_1},\modulekflambda{\lambda_2},\modulekflambda{\lambda_3}\) of the VOA  $\voalkzero$ as well as an intertwiner $\iY$ of type $\itype{\modulekflambda{\lambda_1}}{\modulekflambda{\lambda_3}}{\modulekflambda{\lambda_2}}$. We now choose the subspace $S \subset \modulekflambda{\lambda_3}$ equal to the top level $\modulekflambda{\lambda_3}(0)$ of $\modulekflambda{\lambda_3}$. As discussed, $S = \modulekflambda{\lambda_3}(0)$ is an irreducible $\mathfrak{sl}(d,\C )$ module. We now make the additional assumption that $S$ is the \emph{defining representation}, hence $S = \C^d$. For this choice, the work~\cite{Wassermann:1998cs} of Wasserman implies an explicit analytic bound on the boundedness parameter $\bnd(q,z)$ which we summarize as a corollary.

\begin{corollary}[of~\cite{Wassermann:1998cs}]
  Let $\g = \mathfrak{sl}(d,\C)$, $\lambda_1$, $\lambda_2$, $\lambda_3$ three heighest weights, with $\lambda_3$ corresponding to the defining representation on $\C^d$. Let  $\iY$ be an intertwiner of the VOA $\voalkzero$ of type $\itype{\modulekflambda{\lambda_1}}{\modulekflambda{\lambda_3}}{\modulekflambda{\lambda_2}}$, and let $\W_q$ be the associated scaled intertwiner for the subspace $S = \modulekflambda{\lambda_3}(0) \simeq \C^d$, the top level of the module $\modulekflambda{\lambda_3}$. Then we have for $0<q<1/|z|^2$, $|z|>1$,
  \begin{align}
    \norm{\W_q(\vphi,z)} \,\leq \, \norm{\vphi}_{\modulekflambda{\lambda_3}} \cdot q^{\half(h_{\lambda_1} + h_{\lambda_2} + h_{\lambda_3})} \cdot |z|^{-\tau}  \cdot \sqrt{\frac{1}{(1-q^2)(1-|z|^2 q)} + \frac{q}{(1-q)^2}} \,,\quad \vphi \in \C^d\,,
  \end{align}
  where \(\tau= h_{\lambda_2} + h_{\lambda_3} - h_{\lambda_1}\) and $\norm{\cdot}$ denotes the operator norm.
\end{corollary}

\begin{proof}
Due to the results of Wassermann~\cite[Section 25]{Wassermann:1998cs} the modes of any intertwiner between irreducible modules can be written as ($h$ a real number depending only on the three modules)
\begin{align}
  \iy(\vphi)_{\tau,n} = P \mathsf{a}[\vphi_n] Q\,,
\end{align}
where $\mathsf{a}[\cdot]$ is a fermionic creation operator with domain $S \otimes \C[t,t^{-1}]$, $P,Q$ are projections, and $\vphi_n = \vphi \otimes  t^n$. The detailed construction is not important for us, and we refer the interested reader to Wassermann's work, but what is important for us is that fermionic creation operators satisfy the estimate
\begin{align}\label{eq:normcreationopwzwint}
  \norm{\mathsf{a}[\vphi_n]} \leq \norm{\vphi}_{\modulekflambda{\lambda_3}}\,.
\end{align}
Since projectors do not increase the norm, it follows that bounds on the norm of the modes $\mathsf{a}[\vphi_n]$ imply norm bounds for the modes of the intertwiner. Indeed, considering an arbitrary homogeneous element $\chi \in \modulekflambda{\lambda_2}(m)$ belonging to the $m$-th level and of unit norm, we have --- since weight spaces of different weights are orthonormal --- that 
\begin{align}
  \norm{q^{L_0/2} \iY(\vphi,z) q^{L_0/2} \chi}_{\modulekflambda{\lambda_1}}^2 &= \sum_{n \in Z, n\leq m} q^{2m-n} |z|^{-2n-2\tau} \norm{P \mathsf{a}[\vphi_n] Q}^2 \notag\\
  &\leq \norm{\vphi}_{\modulekflambda{\lambda_3}}^2 \,q^{h_{\lambda_1}+h_{\lambda_2}}\,|z|^{-2\tau}\, \left( \sum_{n\geq 0} q^{2m+n} |z|^{2n} +  m q^{m}\right)\,,
\end{align}
where we used the norm estimate~\eqref{eq:normcreationopwzwint} and then separated the sum over $n$ into negative and positive terms. Note that we have~$n\leq m$, since weights in unitary modules have to be positive (cf.~\eqref{eq:fanannihilation}). We also applied Eq.~\eqref{eq:iabcirreducible}, since all three modules are irreducible. We are left with taking the sum over $m \in \Nl$ to get the final estimate, which leads to
\begin{align}
  \norm{q^{L_0/2} \iY(q^{L_0/2}\vphi,z) q^{L_0/2}\chi}_{\modulekflambda{\lambda_1}}^2 \leq \norm{\vphi}_{\modulekflambda{\lambda_3}}^2 \,q^{h_{\lambda_1}+h_{\lambda_2}+h_{\lambda_3}}\,|z|^{-2\tau}\, \left(\frac{1}{(1-q^2)(1-|z|^2 q)} + \frac{q}{(1-q)^2}\right) \,.
\end{align}
\end{proof}

The case of a general irreducible module $\modulekflambda{\lambda_3}$ with top level corresponding to an arbitrary highest weight representation of $\mathfrak{sl}(d,\C )$ can be handled as well, see~\cite{Wassermann:1995fk} for an exposition. This construction in the case of Spin groups is very nicely presented in the thesis of Laredo~\cite{Laredo}. However, representations other than the defining one generally lead to more complicated norm bounds.
\end{example}

\section{Approximating CFT correlation functions \label{sec:approx}}

In Section~\ref{sec:boundedint}, we showed that scaled intertwiners are bounded operators if the formal variable is replaced by a complex number. The next step in our considerations is the construction of a \emph{truncated} version of these operators, which only changes the weight by a finite number. This then ensures that the image of a finite direct sum of weight spaces under this operator is still contained in a fixed finite-dimensional space, which is obtained by truncating in the weight basis. This is the first main ingredient for our MPS construction.

\subsection{Truncated scaled intertwiners}

Recall the mode expansion~\eqref{eq:modeexpansionintertwiner} of an intertwiner
$\iY(a,z) = \sum_{\tau\in I, m\in \mathbb{Z}} \iy(a)_{\tau,m} z^{-m-\tau-n}$ for  a homogeneous vector~$a\in A_{n,\alpha}\subset A_n$.
According to~\eqref{eq:weightchangehomogeneousintertwiner},  the operator $\iy(a)_{\tau,m}$ changes the  weight of a vector by~$-\tau-m+\alpha$, for any $\tau\in I$ and $m\in\mathbb{Z}$.  This motivates the following definition:

\begin{definition}[Truncated intertwiner]\label{def:truncatedvertex}
  Let $A, B,C$ be  modules of a VOA~$\cV$ and~$\iY$ an intertwiner of type $\itype{C}{A}{B}$. Let $N>0$. The {\em truncated intertwiner}~$\iYtr{N}: A \to \Endint{B,C}{z,z^{-1}}$ is defined as
\begin{align*}
\iYtr{N}(a,z)&=\sum_{\substack{\tau\in I, m\in \mathbb{Z}\\
|m|\leq N}} \iy(a)_{\tau,m} z^{-m-\tau-n}
\end{align*}
for any homogeneous vector $a\in A_n$, and then linearly extended to the whole module~$A$.
\end{definition}
We usually choose $N$ to be a positive integer. We sometimes refer to it as the \emph{truncation level} or \emph{truncation parameter}. It is worth explicitly writing out what  condition~\eqref{eq:weightchangehomogeneousintertwiner} implies for truncated intertwiners. Indeed, this is the key to bounding the bond dimension of the resulting MPS (see Section~\ref{sec:finitebonddimensionerror}).
A truncated intertwiner $\iY^{[N]}(a,z)$ does not change the level by more than $N$ in the sense that 
\begin{align}
  \iY^{[N]}(a,z) B_{n}\subset\bigoplus_{\substack{m\in\mathbb{Z}\\|m|\leq N}} C_{n-m}\{\{z,z^{-1}\}\}
\qquad\textrm{ for all }n\in\mathbb{N}_0\ .\label{eq:containmenttruncatedintertwiner}
\end{align}
This is an immediate consequence of the definition and Eq.~\eqref{eq:containementinequalitylevelsintertw}. Analogously, we can introduce a similar notion of truncated scaled intertwiners. However, before doing so, we need to introduce the following notion for a finite-dimensional subspace $S\subset A$ of a module~$A$. We define 
\begin{align}\label{eq:defShom}
  \Shom:=\bigoplus_{n\in\mathbb{N}_0, \alpha\in I_A} (S\cap A_{n,\alpha})\,, 
\end{align}
where the sum is over all weights occurring in the module $A$. The main motivation for this definition is to ensure that $\Shom$ is spanned by homogeneous vectors, which is not a priori true for any finite-dimensional subspace $S\subset A$ (take for example a single, non-homogeneous vector). In the following, all our statements apply to $\Shom$, that is finite-dimensional subspaces which are spanned by homogeneous vectors. Note however that starting from any finite-dimensional $S\subset A$, we can suitably enlarge $S$ by homogeneous vectors to get a finite-dimensional $\tilde{S} \subset A$ with $S \subset \tilde{S} = \tilde{S}_{hom}$.

\begin{definition}[Scaled truncated intertwiner]\label{def:truncatedintertwiner}Let $A,B,C$ be unitary modules of a VOA~$\cV$. For a linear subspace~$S\subset A$, let $\Shom$ be defined by Eq.~\eqref{eq:defShom}. Let $\iY$ be an~$S$-bounded unitary intertwiner operator of type $\itype{C}{A}{B}$, and let $\W$ be the associated \(S\)-bounded scaled intertwiner of type \(\itype{C}{A}{B}\) (cf.~Definition~\ref{def:scaledintertwiner}).  Fix a truncation level $N>0$ and $0<q<1$. Define 
\begin{align*}
\begin{matrix}
\Wtrq{q}{N}:&\Shom &\rightarrow  &\Endint{B,C}{z,z^{-1}}\\
&a & \mapsto  &\Wtrq{q}{N}(a,z)
\end{matrix}
\end{align*}
  by
\begin{align*}
\Wtrq{q}{N}(a,z)b:=
q^{L_0/2} \iYtr{N}(q^{L_0/2}a,z)q^{L_0/2} b
\end{align*}
for homogeneous $a\in S$ and arbitrary $b\in B$, and extend this linearly to all of  $\Shom$. Then we call the family $\{\Wtrq{q}{N}:\Shom  \rightarrow  \Endint{B,C}{z,z^{-1}}\}_{0<q<1}$  the {\em $N$-th level truncation of~$\W$}, or simply a {\em truncated scaled intertwiner}.
\end{definition}
Of course, truncated scaled intertwiners also 
have a ``controlled'' behavior on the levels in the sense that
\begin{align}
  \W_q^{[N]}(a,z) B_{n}\subset\bigoplus_{\substack{m\in\mathbb{Z}\\|m|\leq N}}C_{n-m}\{\{z,z^{-1}\}\}
\qquad\textrm{ for all }n\in\mathbb{N}_0\ ,\label{eq:containmenttruncatedscaledintertwiner}
\end{align}
as follows immediately from~\eqref{eq:containmenttruncatedintertwiner} and the fact that~$L_0$ preserves levels. Starting from this formal power series, we can get a well-defined operator by replacing the formal variable by a complex number (the fact that it is well-defined follows since it is a finite sum of operators times a power of a complex variable). Observe that for $z\in\mathbb{C}$, the operator~$\Wtrq{q}{N}(a,z)$ acts on the infinite-dimensional module~$B$. This raises the question whether it is bounded, which is answered by the following lemma. 

\begin{lemma}\label{lem:boundednessparametertruncated}
  Let  $\Wtrq{q}{N}$ be the $N$-th level truncation of an $S$-bounded scaled intertwiner~$\W_q$ of type $\itype{C}{A}{B}$.
Let~$\bnd$ be the boundedness parameter of $\W_q$ introduced in Definition~\ref{def:boundedintertwiner}, $z\in\mathbb{C}\backslash\{0\}$, $0<q<\min\{|z|^2,1/|z|^2\}$ and $0<q_1,q_2<1$ satisfying $q_1\cdot q_2 =q$ be given. Then $\Wtrq{q}{N}$ is $S$-bounded with 
\begin{align}
  \norm{\Wtrq{q}{N}(a,z)} \leq \norm{a}_A \cdot \frac{\sqrt{|I_B|}\cdot \bnd(q_2,z)}{\sqrt{1-q_1}} \,,\quad \textrm{for all } \,a \in S\,.
\end{align}
Specifically, setting $q_1=q_2=\sqrt{q}$ we get
\begin{align}\label{eq:Sboundedtruncation}
  \norm{\Wtrq{q}{N}(a,z)} \leq \norm{a}_A \cdot \frac{\sqrt{|I_B|}\cdot\bnd(\sqrt{q},z)}{\sqrt{1-\sqrt{q}}} \,,\quad \textrm{for all } \,a \in S\,.
\end{align}
The bound is independent of the truncation parameter $N$, and in fact holds in the case $N=\infty$, that is, for the scaled intertwiner $\W_q$ itself.
\end{lemma}

\begin{proof}
  Let $b\in B_n$ be a homogeneous vector at level $n$. Since the weights spaces are orthogonal and $q^{L_0/2}$ does not change the weight, we find for the square of the norm
  \begin{align}
    \norm{\Wtrq{q}{N}(a,z)b}_C^2 = \sum_{\substack{\tau\in I, m\in \mathbb{Z}\\|m|\leq N}} \norm{q^{L_0/2} \iy(a)_{\tau,m} q^{L_0/2} b z^{-m-\tau-n}}_C^2 \,.
  \end{align}
  Each term $\norm{q^{L_0/2} \iy(a)_{\tau,m} q^{L_0/2} b z^{-m-\tau-n}}_C^2$ for $m \in \Z$ is positive, and hence we make the expression only larger if we drop the restriction $|m|\leq N$. But then
  \begin{align}\label{eq:boundtruncWbyW}
    \norm{\Wtrq{q}{N}(a,z)b}_C^2 \leq \sum_{\tau\in I, m\in \mathbb{Z}} \norm{q^{L_0/2} \iy(a)_{\tau,m} q^{L_0/2} b z^{-m-\tau-n}}_C^2 = \norm{\W_q(a,z)b}_C^2 \,,
  \end{align}
  by the same arguments as before. 

 Let us now consider a general~$b\in B$, and decompose it into homogeneous elements belonging to fixed levels as 
\begin{align*}
b&=\sum_{\substack{n\in\mathbb{N}_0\\
\beta\in I_B}} \mu_{n,\beta} b_{n,\beta}\qquad\textrm{ where }\mu_{n,\beta}\in\mathbb{C}\textrm{ and } b_{n,\beta}\in B_{n,\beta}\ ,\quad \|b_{n,\beta}\|_B=1\ ,
\end{align*}
implying $\|b\|_B^2=\sum_{\substack{n\in\mathbb{N}_0\\
\beta\in I_B}}|\mu_{n,\beta}|^2$. By~\eqref{eq:scaledintertwinerfactorization} we have for $a\in A$, and $q = q_1 \cdot q_2$ that
  \begin{align}
    \W_q(a,z)b = q^{L_0/2}_1 \, \W_{q_2}(q_1^{L_0/2}a,z) q^{L_0/2}_1 b\ .
  \end{align}
  Combining these observations with~\eqref{eq:boundtruncWbyW} we find using the Cauchy-Schwarz inequality
  \begin{align}
    \norm{\Wtrq{q}{N}(a,z)b}_C &\leq \sum_{\substack{n \in \Nl_0\\
\beta\in I_B}} |\mu_{n,\beta}| \cdot \norm{\Wtrq{q}{N}(a,z)b_{n,\beta}}_C\\
& \leq \norm{b}_B \cdot \left(\sum_{\substack{n\in \Nl_0\\
\beta\in I_B}} \norm{q_1^{L_0/2} \W_{q_2}(q_{1}^{L_0/2}a,z) q_1^{L_0/2} b_{n,\beta}}_C^2\right)^{\half}\,.
  \end{align}
  Applying the fact that $\W_{q_2}$ is a scaled intertwiner and hence bounded by Corollary~\ref{cor:boundednessscaledintertwiner}, as well as using $q_1^{L_0/2}b_{n,\beta} = q_1^{\beta/2}q_1^{n/2} b_{n,\beta}$ and $\norm{b_{n,\beta}}_B =1$ leads to
  \begin{align}
    \norm{\Wtrq{q}{N}(a,z)b}_C \leq \bnd(q_2,z) \norm{b}_B  \norm{a}_A\cdot \left(\sum_{\substack{n \in \Nl_0\\
\beta\in I_B}}  q_1^{\beta+n} \right)^\half \leq \sqrt{|I_B|}\cdot \norm{b}_B  \norm{a}_A\,\frac{\bnd(q_2,z)}{\sqrt{1-q_1}}\,,
  \end{align}
  since $q_1^\beta \leq 1$  (since $0<q_1<1$ and $I_B\subset \mathbb{R}_{\geq 0}$ for unitary modules) and the operator norm of~$q_1^{L_0/2}$ is bounded by one.
\end{proof}

Our main motivation for the introduction of  truncated scaled intertwiners is the fact that these operators do not increase the weight too much, and hence the image of a finite direct sum of weight spaces is still of that form, with adjusted parameters. As we will see below, this allows us to approximately restrict to a finite bond dimension when considering genus-$0$ or genus-$1$ correlation functions. The next step towards this goal is to establish an explicit error bound on the approximation of $\W_q$ by $\W_q^{[N]}$ (see  Theorem~\ref{thm:truncationestimate} below). The following lemma will be needed for this purpose. It formalizes the fact that elements of high weight are scaled down by operating with $q^{L_0}$.

\begin{lemma}\label{lem:weightspaceintermediateproj}
Consider an irreducible unitary module $A$, and let $P_{[n-N,n+N]}$ be the weight space projection onto \begin{align*}
\bigoplus_{n-N\leq m\leq n+N}A_m\ .
\end{align*} 
For $0<q<1$,  $N\in\mathbb{N}$, we have
\begin{align*}
\sum_{n\in \mathbb{N}_0}q^n\big\|(\mathbf{I}-P_{[n-N,n+N]})q^{L_0/2}\big\|^2\leq \,\,q^{N}\,\frac{3}{(1-q)^2}\,.
\end{align*}
\end{lemma} 

\begin{proof}
Consider an arbitrary element $a\in A$ of unit norm $\|a\|_A=1$ decomposed as
\begin{align*}
a&=\sum_{m\in \mathbb{N}_0}a_m\qquad\textrm{ where }a_m\in A_m\ .
\end{align*}
Then (since $L_0$ is compatible with the grading), we have 
\begin{align*}
(\mathbf{I}-P_{[n-N,n+N]})q^{L_0/2}a&=\sum_{\substack{m\in\mathbb{N}_0\\
m<n-N}}q^{L_0/2}a_m+\sum_{\substack{m\in\mathbb{N}_0\\
m>n+N}}q^{L_0/2}a_m
\end{align*}
and thus
\begin{align*}
\big\|(\mathbf{I}-P_{[n-N,n+N]})q^{L_0/2}a\big\|^2_A&=
 \sum_{\substack{m\in\mathbb{N}_0\\
m<n-N}}\|q^{L_0/2}a_m\|^2_A+\sum_{\substack{m\in\mathbb{N}_0\\
m>n+N}}\|q^{L_0/2}a_m\|^2_A\,.
\end{align*}
Here we used the fact that the projector $(\mathbf{I}-P_{[n-N,n+N]})$ leaves the grading invariant as well as the assumption of irreducibility to conclude that elements \(a_m \in A_m\), \(a_{m^\prime} \in A_{m^\prime}\) have different weights for \(m \neq m^\prime\) and hence are orthogonal. These weights are then of the form \(h+m\), \(h\geq0\) and \(m \in \Nl_0\). Using that $\|a_m\|_A\leq  \|a\|_A=1$, we have 
\begin{align*}
\sum_{\substack{m\in\mathbb{N}_0\\
m<n-N}}\|q^{L_0/2}a_m\|^2_A&=\begin{cases}
\sum_{m=0}^{n-N-1}\|q^{L_0/2}a_m\|^2_A\leq q^{h} \cdot \sum_{m=0}^{n-N-1}q^{m}
\qquad & \textrm{ if }n> N\\
0 & \textrm{otherwise}
\end{cases}\ ,
\end{align*}
that is, with~$\sum_{m=0}^{n-N-1}q^{m}\leq\frac{1}{1-q}$ 
\begin{align*}
\sum_{\substack{m\in\mathbb{N}_0\\
m<n-N}}\|q^{L_0/2}a_m\|^2_A&\leq \begin{cases}
\frac{q^h}{1-q}\qquad & \textrm{ if }n> N\\
0 & \textrm{otherwise}
\end{cases}\,.
\end{align*}
Similarly, we can bound
\begin{align*}
\sum_{\substack{m\in\mathbb{N}_0\\
m>n+N}}\|q^{L_0/2}a_m\|^2_A&\leq q^h \sum_{\substack{m\in\mathbb{N}_0\\m>n+N}} q^{m}\leq q^h \sum_{m=n+N+1}^\infty q^{m}\ .
\end{align*} 
But since we have $\sum_{m=n+N+1}^\infty q^{m}= q^{n+N+1}\sum_{m=0}^\infty q^{m}=\frac{q^{n+N+1}}{1-q}$, we obtain
\begin{align*}
  \big\|(\mathbf{I}-P_{[n-N,n+N]})q^{L_0/2}a\big\|^2_A &\leq \,\frac{q^h}{1-q}\,
  \begin{cases} 
    (1+q^{n+N+1}) \, \qquad&\textrm{ if }n> N\\
    q^{n+N+1} \, &\textrm{ otherwise}\ .
  \end{cases}
\end{align*}
From this expression, we get the bounds
\begin{align*}
\sum_{n\in\mathbb{N}_0}q^n\big\|(\mathbf{I}-P_{[n-N,n+N]})q^{L_0/2}a\big\|_A^2
&\leq \,\frac{q^h}{1-q} \,\left[ \sum_{n> N} (1+q^{n+N+1}) q^n + \sum_{n=0}^{N-1} q^{2n+N+1}\right] \nonumber\\
&= \,\frac{q^h}{1-q} \,\left[\frac{q^{N+1}}{1-q} + \frac{q^{3(N+1)}}{1-q} + q^{N+1} \frac{1-q^{2N}}{1-q^2}\right]\\
&\leq\,q^{N+1}\,\frac{3}{(1-q)^2}
\end{align*}
where we used that \(q^h<1\) since $q<1$. The claim follows since $a\in A$ with $\|a\|_A=1$ was arbitrary. 
\end{proof}

Our central result is the following bound on the approximation error (in operator norm). 

\begin{theorem}\label{thm:truncationestimate}
Let  $\Wtrq{q}{N}$ be the $N$-th level truncation of a scaled intertwiner~$\W_q$ of type $\itype{C}{A}{B}$.
Let~$\bnd$ be the boundedness parameter of $\W_q$ introduced in Definition~\ref{def:boundedintertwiner}.  Fix
$z\in\mathbb{C}\backslash\{0\}$,
$0<q<\min\{|z|^2,1/|z|^2\}$ and $0<q_1,q_2<1$ satisfying $q_1\cdot q_2 =q$.
Then there is a constant $\kappa>0$, depending only on the three modules $A,B,C$ such that we have
\begin{align}\label{eq:normboundtrct}
   \|\W_{q}(a,z) - \Wtrq{q}{N}(a,z)\| \leq \norm{a}_A \cdot\kappa\cdot \bnd(q_2,z) \,q_1^{N/2}\,\frac{1}{1-q_1}
\end{align}
for all $a\in \Shom$ and $b\in B$. In particular, setting $q_1=q_2=\sqrt{q}$, we get the bound 
\begin{align}
\|\W_{q}(a,z) - \Wtrq{q}{N}(a,z)\| &\leq 
\|a\|_A\cdot \kappa\cdot\bnd(\sqrt{q},z) \,q^{N/4}\,\frac{1}{1-\sqrt{q}} \,.\label{eq:truncationwwnbound}
\end{align}
\end{theorem}
\noindent Before presenting the proof, let us briefly discuss the bound. The function
\begin{align}
  \apperr(q,z):=\frac{\kappa\cdot\bnd(\sqrt{q},z)}{1-\sqrt{q}}
\end{align}
is finite for all $0<q<1$ as well as independent of the truncation parameter $N$. Hence we get an exponentially fast convergence of the truncated scaled intertwiner to its non-truncated version with respect to the operator norm as we increase the truncation parameter. Recalling from Observations~\ref{obs:conformalmapcorrfuncplane} and~\ref{obs:conformalmapcorrfunctorus} that $q=e^{-\mathsf{d}}$, with $\mathsf{d}$ being the minimal distance between insertion points (or the ultraviolett cutoff), we see that \(\mathsf{d}\) determines the speed of the exponential convergence.

\begin{proof} 
  We first restrict to the case where $A,B,C$ are irreducible modules. As we will argue below, the general case follows since any module can be decomposed into finitely many irreducible ones, as the VOA is rational. We first establish a relationship (see~\eqref{eq:wqnprojckt} below) between
$W_q(a,z)b$ and $W_q^{[N]}(a,z)b$: we will show that, for homogeneous $a\in \Shom$ and $b\in B_n$ belonging to a fixed level~$n$, the latter is obtained by applying a weight projection~$P_{[n-N,n+N]}$ to the former. 
To do so, we  use the definition of scaled intertwiners.
For homogeneous $a\in S\cap A_r$ at level $r$ we have 
$q^{L_0/2}a\in S \cap A_r$ and thus, writing out the definition of scaled intertwiners
\begin{align*}
\W_q(a,z)b&=\sum_{m\in \mathbb{Z}}q^{L_0/2}\iy(q^{L_0/2}a)_{\tau,m} z^{-m-\tau-r}q^{L_0/2}b\,,
\end{align*}
where $\tau=h_A+h_B-h_C$ is a fixed real number depending on the irreducible modules $A,B,C$, see Sect.~\ref{sec:intertwiners}. Each term in the sum has the form
\begin{align*}
z^{-m-\tau-r} q^{L_0/2}\iy(q^{L_0/2}a)_{\tau,m}q^{L_0/2}b&=q^{(\wt a)/2}z^{-m-\tau-r} q^{L_0/2}\iy(a)_{\tau,m}q^{L_0/2}b\ ,
\end{align*}
that is, we have
\begin{align}
\W_q(a,z)b&=q^{(\wt a)/2} \sum_{m\in \mathbb{Z}}  q^{L_0/2}\iy(a)_{\tau,m}q^{L_0/2}b \,z^{-m-\tau-r}\\
\W^{[N]}_q(a,z)b&=q^{(\wt a)/2} \sum_{\substack{m\in \mathbb{Z}\\
|m|\leq N}}q^{L_0/2}\iy(a)_{\tau,m}q^{L_0/2}b \,z^{-m-\tau-r}
\label{eq:partialsumvbq}
\end{align}
where we applied the same reasoning to~$\W^{[N]}_q(a,z)b$. Since $L_0$ leaves the levels invariant, we have
\begin{align*}
 q^{L_0/2}\iy(a)_{\tau,m}q^{L_0/2}B_n\subset B_{n-m}\qquad\textrm{ for all }n\in\mathbb{N}_0\textrm{ and }m\in \mathbb{Z}\ ,
\end{align*}
according to~\eqref{eq:containementinequalitylevelsintertw} with the convention that \(B_{n^\prime} = 0\) for \(n^\prime<0\).
 In particular, this means that for every $b_n\in B_n$, we have 
\begin{align*}
P_{[n-N,n+N]} q^{L_0/2}\iy(a)_{\tau,m}q^{L_0/2}b_n=\begin{cases}
q^{L_0/2}\iy(a)_{\tau,m}q^{L_0/2}b_n\qquad& \textrm{ if } |m|\leq N\\
0 & \textrm{ otherwise}\ .
\end{cases}
\end{align*}
Thus applying $P_{[n-N,n+N]}$ to~\eqref{eq:partialsumvbq} gives the identity
\begin{align}
  W_q^{[N]}(a,z) b_n=P_{[n-N,n+N]}\W_q(a,z) b_n\qquad\textrm{ for all }b_n\in B_n\ , \label{eq:wqnprojckt}
\end{align}
where $n\in\mathbb{N}_0$ is arbitrary. Since this identity is independent of the level $r$ of $a$, and the scaled intertwiner as well as its truncated version are linear in $a$, it extends to the whole of $\Shom$. 

Consider now a general $b\in B$, and decompose it into elements belonging to fixed levels as 
\begin{align*}
b&=\sum_{n\in\mathbb{N}_0} \mu_n b_n\qquad\textrm{ where }\mu_n\in\mathbb{C}\textrm{ and } b_n\in B_n\ ,\quad \|b_n\|_B=1\ ,
\end{align*}
that is, $\|b\|_B^2=\sum_{n\in\mathbb{N}_0}|\mu_n|^2$. 
Then~\eqref{eq:wqnprojckt} gives
\begin{align}
\big\|\W_q(a,z)b-\W_q^{[N]}(a,z)b\big\|_C &=
\left\|\sum_{n\in\mathbb{N}_0}\mu_n (\mathbf{I}-P_{[n-N,n+N]})\W_q(a,z)b_n\right\|_C\nonumber\\
&\leq\sum_{n\in\mathbb{N}_0}|\mu_n|\cdot \big\|(\mathbf{I}-P_{[n-N,n+N]})\W_q(a,z)b_n\big\|_C\nonumber\\
&\leq \|b\|_B \cdot \left(\sum_{n\in\mathbb{N}_0}
\big\|(\mathbf{I}-P_{[n-N,n+N]})\W_q(a,z) b_n\big\|_C^2\right)^{1/2} \label{eq:normdifferenceboundxx}
\end{align}
where we used the Cauchy-Schwarz inequality in the last step. 
For $a\in A$, we have   for $q = q_1 \cdot q_2$  that
with~\eqref{eq:scaledintertwinerfactorization}
\begin{align*}
(\mathbf{I}-P_{[n-N,n+N]})\W_q(a,z)b_n
&=(\mathbf{I}-P_{[n-N,n+N]})q^{L_0/2}_1 \, \W_{q_2}(q_1^{L_0/2}a,z) q^{L_0/2}_1 b_n\ .
\end{align*}
With the norm bound $\|Xv\|_C\leq \|X\|\cdot \|v\|_C$, applied to $X=(\mathbf{I}-P_{[n-N,n+N]})q^{L_0/2}_1 $ and $v=\W_{q_2}(q_1^{L_0/2}a,z) q^{L_0/2}_1 b_n$, we  get
\begin{align}
  &\big\|(\mathbf{I}-P_{[n-N,n+N]})\W_{q_2}(a,z)b_n\big\|_C\leq 
\big\|(\mathbf{I}-P_{[n-N,n+N]})q_1^{L_0/2}\big\|\cdot \big\|\W_{q_2}(q_1^{L_0/2}a,z) q^{L_0/2}_1 b_n\big\|_C
\end{align}
The second norm factor can be upper bounded using the fact that $\W_{q_2}$ is a scaled intertwiner, which is bounded. This gives
\begin{align}
  \norm{(\mathbf{I}-P_{[n-N,n+N]})\W_{q_2}(q_1^{L_0/2} a,z)b_n}^2_C \leq \bnd(q_2,z)^2 \norm{a}_A^2 \,q_1^n \big\|(\mathbf{I}-P_{[n-N,n+N]})q_1^{L_0/2}\big\|^2
\end{align}
where we used that $\norm{q_1^{L_0/2}}$ is bounded by \(1\) as well as that $q^{L_0/2}_1 b_n = q_1^{(h_B+n)/2} b_n$ and $\norm{b_n}_B=1$. Evaluating the sum over $n \in \Nl_0$ as needed by~\eqref{eq:normdifferenceboundxx} and using Lemma~\ref{lem:weightspaceintermediateproj} gives 
\begin{align}
  \big\|\W_q(a,z)b-\W_q^{[N]}(a,z)b\big\|_C \leq  \norm{a}_A \norm{b}_B \,\bnd(q_2,z) \,q_1^{N/2}\,\frac{\sqrt{3}\,q_1^{h_B/2}}{1-q_1} \,,
\end{align}
for irreducible modules $A,B,C$. Since \(q_1<1\), this proves the claim for irreducible modules of highest weights \(h_A, h_B, h_C\). 

Consider now three arbitrary modules, which we again denote by $A,B,C$. Each module can be decomposed into a finite direct sum of irreducible ones,
\begin{align}
  A = \bigoplus_{\alpha} A[\alpha]\,,\quad B = \bigoplus_\beta B[\beta]\,,\quad C=\bigoplus_\gamma C[\gamma] \,.
\end{align}
Let $|I|$ denote the maximal number of irreducible modules appearing in these decompositions. This is well-defined since the VOA is rational. Denoting by $Q_{D[\delta]}$ the projector onto the irreducible component, $D \in \{A,B,C\}$, $\delta \in \{\alpha,\beta,\gamma\} $, we find since the irreducible modules are orthogonal
\begin{align}
  &\big\|\W_q(a,z)b-\W_q^{[N]}(a,z)b\big\|^2_C = \sum_\gamma \big\| \sum_{\alpha,\beta} Q_{C[\gamma]} \left(\W_q(Q_{A[\alpha]}a,z)-\W_q^{[N]}(Q_{A[\alpha]}a,z)\right)\,Q_{B[\beta]}b\big\|^2_C \nonumber\\
  &\qquad\quad\leq \sum_\gamma \left(\sum_{\alpha,\beta} \big\|\left(Q_{C[\gamma]}\W_q(Q_{A[\alpha]}a,z)Q_{B[\beta]} -Q_{C[\gamma]}\W_q^{[N]}(Q_{A[\alpha]}a,z)Q_{B[\beta]} \right)\,Q_{B[\beta]}b\big\|_C \right)^2\,.
\end{align}
Now each of the operators $Q_{C[\gamma]}\W_q(Q_{A[\alpha]}a,z)Q_{B[\beta]}$ (resp. $Q_{C[\gamma]}\W_q^{[N]}(Q_{A[\alpha]}a,z)Q_{B[\beta]}$) is a scaled intertwiner between irreducible modules (resp. its truncated version) and hence our previous bound applies. Moreover, we have
\begin{align}
  \norm{Q_{C[\gamma]}\W_q(Q_{A[\alpha]}a,z)Q_{B[\beta]}} \leq \norm{\W_q(Q_{A[\alpha]}a,z)} \leq \bnd(q,z) \norm{Q_{A[\alpha]}a}_A \leq \bnd(q,z) \norm{a}_A \,,
\end{align}
since projectors have operator norm equal to one. Inserting this and using Cauchy-Schwarz once for the sum over the indices $\alpha,\beta$ leads to 
\begin{align}
  \big\|\W_q(a,z)b-\W_q^{[N]}(a,z)b\big\|^2_C \leq 3 |I|^2 \sum_{\alpha,\beta,\gamma} \norm{Q_{A,\alpha}a}_A^2 \norm{Q_{B,\beta}b}_B^2  \,\bnd(q_2,z)^2\,q_1^{N}\,\frac{q_1^{h_{\gamma}}}{(1-q_1)^2} \,.
\end{align}
But since $q_1^{h_{\gamma}} \leq 1$ and $\sum_\alpha \norm{Q_{A,\alpha} a}_A^2 = \norm{a}_A^2$ and similarily for $b$, the result follows with $\kappa = \sqrt{3}\,|I|^{3/2}$. 
\end{proof}

\subsection{Approximating correlation functions}

In this section, we use the notion of bounded intertwiners to establish error bounds on the approximation of correlation functions.  We proceed in two steps: in Section~\ref{sec:errorboundstruncintertwiner}, we establish bounds on the approximation accuracy when replacing intertwiners by their truncated versions in the definition of the transfer operator $\Top$. In Section~\ref{sec:finitebonddimensionerror}, we additionally project onto a finite-dimensional subspace to obtain bounds for the approximation by MPS (or finitely correlated states) with finite bond dimension. 

As we are only interested in approximation statements in this section, we only consider the case where the formal variables are replaced by complex numbers, or more precisely a single one, denoted again by $z$ and assumed to satisfy $\min\{|z|^2,|z|^{-2}\} > q >0$. Of course, we could also first define the truncated transfer operator as a formal polynomial, from which an operator is then obtained by fixing the values of the indeterminates. However, we felt that this approach is of limited added value.

\subsubsection{Truncated transfer operators and error bounds \label{sec:errorboundstruncintertwiner}}
In analogy to Definition~\ref{def:transferoperators}, we can define transfer operators which are truncated at some level~$N$. The corresponding definition is the following. 
\begin{definition}[Truncated transfer operators]\label{def:transferoperatorstruncated}
  Let $A^{(i)}$, $i=1,\ldots,n$ and $B^{(i)}$, $i=0,\ldots,n$ be unitary modules of a VOA~$\cV$. Fix some $z\in \C \setminus \{0\}$ and some $0<q<\min\{|z|^2,|z|^{-2}\}$. For $i=1,\ldots,n$, let $\iY_i$ an $S^{(i)}$-bounded intertwiner of  type $\itype{B^{(i-1)}}{A^{(i)}}{B^{(i)}}$, and let 
\begin{align*}
  \W_i^{[N]}=\W^{[N]}_{i,q}(\cdot,z):S^{(i)}\rightarrow \End(B^{(i)},B^{(i-1)})
\end{align*} be the associated truncated scaled intertwiner.  For $a_i\in \Shom^{(i)}$, $i=1,\ldots,n$, define $\Top^{[N]}:B^{(n)}\rightarrow B^{(0)}$ by
\begin{align}
  \Top^{[N]}&=\W^{[N]}_1(a_1,z) \circ \W^{[N]}_{2}(a_{2},z) \circ \cdots \circ  \W^{[N]}_n(a_n,z)\,. \label{eq:operatorz1def}
\end{align}
Then $\Top^{[N]}$ is called the {\em transfer operator with insertions $\{a_i\}_{i=1}^n$, truncated at level $N$} or simply the {\em truncated transfer operator}. 
\end{definition}
We stress that $\Top^{[N]}$ is by itself not an operator acting on  finite-dimensional spaces; further steps will be needed to arrive at a finite-dimensional MPS. The main question we address in this section is how the truncated operator $\Top^{[N]}$ can be interpreted as an approximation to the transfer operator~$\Top$ introduced in Definition~\ref{def:transferoperators}. The following is the key technical result. 

\begin{lemma}[Approximation by truncated transfer operators]\label{lem:normboundtransferop}
  Let $\Top$ be the transfer operator with insertions $a_i \in \Shom^{(i)}$, \(i=1,\ldots,n\)  and parameters $(q,z)$ as in Lemma~\ref{lem:boundedtransferop}, and let $\Top^{[N]}$ be the associated truncated transfer operator as in Definition~\ref{def:transferoperatorstruncated}. Then we have
\begin{align*}
\|\Top-\Top^{[N]}\|\leq  
q^{N/4} \cdot \left[n\cdot \kappa\cdot \left(\max_{1\leq j \leq n}\frac{\sqrt{|I_{B^{(j)}}|}\cdot\bnd_j(\sqrt{q},z)}{1-\sqrt{q}}\right)^n\right] \cdot \left(\prod_{i=1}^n \|a_i\|_{A_i}\right) \,,
\end{align*}
where $\|\cdot\|$ is the operator norm. 
\end{lemma}

Before we proceed to the proof, let us discuss the dependence of the bound on the parameters. We see that the expression in the square brackets
\begin{align}
  \Delta(n,q,z) := n\cdot \kappa\cdot \left(\max_{1\leq j \leq n}\frac{\sqrt{|I_{B^{(j)}}|}\cdot\bnd_j(\sqrt{q},z)}{1-\sqrt{q}}\right)^n\label{eq:deltanqzdefa}
\end{align}
is independent of the truncation parameter $N$, and thus constant for a fixed number of insertions~$n$ and any $0<q<1$. Hence we  get exponentially fast convergence (in the truncation parameter $N$) of the truncated transfer operator to the original one. The speed of convergence is again governed by the value of~$q$, or equivalently, by the minimal distance between insertion points on the plane or the periodic strip. 

\begin{proof}
Defining
\begin{align*}
  \Omega^{(\ell)}_n&=\hat{\W}^{(\ell)}_n(a_n,z)\in \End(B^{(n)},B^{(n-1)})\\
  \Omega^{(\ell)}_k&=\hat{\W}^{(\ell)}_k(a_k,z) \circ \Omega_{n+1}^{(\ell)}\in \End(B^{(k)},B^{(k-1)})\qquad\textrm{ for } k=n-1,\ldots,1\ ,
\end{align*}
where
\begin{align*}
\hat{\W}^{(\ell)}_n&=\begin{cases}
\W_n \qquad &\textrm{ for }n\leq \ell\\
\W_n^{(N)} &\textrm{ for }n>\ell\ ,
\end{cases}
\end{align*}
we have $\Omega^{(0)}_1=\Top^{[N]}$ and $\Omega^{(n)}_1=\Top$ (and more generally, $\Omega^{(\ell)}_1$ contains $\ell$~non-truncated intertwiners and $n-\ell$~truncated ones).
Using the telescoping sum, we get the upper bound
\begin{align}
\|\Omega_1^{(n)}-\Omega_1^{(0)}\|&\leq \sum_{\ell=1}^{n}\|\Omega_1^{(\ell)}-\Omega_1^{(\ell-1)}\|\label{eq:telescopingsum}
\end{align}
For notational simplicity, let us first introduce the following set of abbreviations
\begin{align}
  \apperr_j(q,z) := \frac{\kappa\,\bnd_j(\sqrt{q},z)}{1-\sqrt{q}} \,,\qquad
  \bndt_j(q,z) := \frac{\sqrt{|I_{B^{(j)}}|}\cdot\bnd_j(\sqrt{q},z)}{\sqrt{1-\sqrt{q}}}\,,
\end{align}
where $\bnd_j(q,z)$ is the boundedness parameter of the $j$-th scaled intertwiner in the definition of the transfer operator $\Top$. We proceed to bound the differences $\|\Omega_1^{(\ell)}-\Omega_1^{(\ell-1)}\|$. For $\ell=1$, we have
\begin{align*}
  \Omega_1^{(1)}-\Omega_1^{(0)}&=\W_1(a_1,z) \circ \W_2^{[N]}(a_2,z) \circ \cdots \circ \W_n^{(N)}(a_n,z)\\
  &\quad-\W^{[N]}_1(a_1,z)\circ \W_2^{[N]}(a_2,z)\circ \cdots \circ \W_n^{(N)}(a_n,z)\ ,
\end{align*}
hence applying
Theorem~\ref{thm:truncationestimate} and Lemma~\ref{lem:boundednessparametertruncated} (the latter $n-1$-times) yields
\begin{align}
  \|\Omega_1^{(1)}-\Omega_1^{(0)}\|&\leq \|a_1\|_{A_1}\cdot \|\W_2^{[N]}(a_2,z) \circ \W_3^{(N)}(a_3,z) \circ \cdots \circ \W_n^{(N)}(a_n,z)\|\cdot \apperr_1(q,z) q^{N/4} \nonumber\\
  &\leq \|a_1\|_{A^{(1)}}\cdot \|a_2\|_{A^{(2)}}\cdot \|\W_3^{(N)}(a_3,z) \circ \cdots \circ \W_n^{(N)}(a_n,z)\|\cdot  \apperr_1(q,z) q^{N/4} \bndt_2(q,z) \nonumber\\
&\leq 
\left(\prod_{i=1}^n \|a_i\|_{A^{(i)}}\right)\cdot \apperr_1(q,z) q^{N/4} \prod_{j=2}^n \bndt_j(q,z) \ .\label{eq:trigfirst}
\end{align}
Similar reasoning for $\ell=2$ gives
\begin{align}
\|\Omega_1^{(2)}-\Omega_1^{(1)}\|&\leq 
\left(\prod_{i=1}^n \|a_i\|_{A^{(i)}}\right)\cdot  
 \apperr_2(q,z) q^{N/4} \bndt_1(q,z) \prod_{j=3}^n \bndt_j(q,z)\ .\label{eq:trigsecond}
\end{align}
More generally, for $\ell\geq 3$, we have
\begin{align*}
  \Omega_1^{(\ell)}-\Omega_1^{(\ell-1)}&=\W_1(a_1,z) \circ \W_2(a_2,z) \circ \cdots\circ \W_{\ell-2}(a_{\ell-2},z)\delta\qquad\textrm{ where }\\
  \delta&=(\W_{\ell-1}(a_\ell,z)-\W_{\ell-1}^{[N]}(a_\ell,z))\eta\\
  \eta&=\W_{\ell}^{(N)}(a_\ell,z) \circ \W_{\ell+1}^{[N]}(a_{\ell+1},z) \circ \cdots \W_{n}^{(N)}(a_n,z) \ .
\end{align*}
Theorem~\ref{thm:truncationestimate} and Lemma~\ref{lem:boundednessparametertruncated} then yield
\begin{align}\label{eq:trigthird}
\|\Omega_1^{(\ell)}-\Omega_1^{(\ell-1)}\|&\leq\left(\prod_{i=1}^n \|a_i\|_{A_i}\right)\cdot  
\apperr_{l-1}(q,z) q^{N/4} \prod_{j=1}^{\ell-2} \bnd_j(q,z) \prod_{j=\ell}^n \bndt_j(q,z) \nonumber\\
&=q^{N/4}\,\left(\prod_{i=1}^n \|a_i\|_{A_i}\right) \frac{\kappa\,\bnd_{l-1}(\sqrt{q},z)}{1-\sqrt{q}}  \prod_{j=1}^{\ell-2} \bnd_j(q,z) \prod_{j=\ell}^n \frac{\sqrt{|I_{B^{(j)}}|}\cdot\bnd_j(\sqrt{q},z)}{\sqrt{1-\sqrt{q}}}
\end{align}
for $\ell\geq 3$. We then have for $0<q<1$ that $\sqrt{1-\sqrt{q}} \geq 1-\sqrt{q}$, as well as $\bnd_j(q,z) \leq \frac{\bnd_j(\sqrt{q},z)}{\sqrt{1-\sqrt{q}}}$ due to the last assertion of Lemma~\ref{lem:boundednessparametertruncated}. Combining these facts with~\eqref{eq:trigfirst},~\eqref{eq:trigsecond} and~\eqref{eq:trigthird} with~\eqref{eq:telescopingsum} gives 
\begin{align*}
\|\Top-\Top^{[N]}\|\leq  \left(\prod_{i=1}^n \|a_i\|_{A^{(i)}}\right)\cdot  
q^{N/4} \,n\,\kappa\, \prod_{j=1}^n \frac{\sqrt{|I_{B^{(j)}}|}\cdot\bnd_j(\sqrt{q},z)}{1-\sqrt{q}}\,.
\end{align*}
\end{proof}

An immediate consequence of Corollaries~\ref{cor:corrfctexactzero} and~\ref{cor:corrfctexactrepr} together with Lemma~\ref{lem:normboundtransferop} are the following estimates of the correlation functions in terms of truncated intertwiners. These are the ``truncated'' counterparts to those results: indeed, in the limit $N\rightarrow\infty$ of no truncation, the exact statement is reproduced. In order to simplify the expressions encountered, we henceforth assume that the insertion vectors~$\{a_i\}_{i=1}^n$ appearing in the definition of the transfer operator as well as its truncated version are of unit norm, that is, 
\begin{align}
  \norm{a_i}_{A^{(i)}} = 1 \quad \textrm{ for all } i=1,\ldots,n\,.
\end{align}
All of our approximation statements are now subject to this assumption. The general case may be obtained by multiplying the error by the product $\prod^n_{i=1} \norm{a_i}_{A^{(i)}}$. 

\begin{corollary}[Approximate reproduction of genus-$0$ correlation functions]\label{cor:approxcorfuncplane}
  For $a_i\in \Shom^{(i)}$, $i=1,\ldots,n$, let $\Top^{[N]}:B^{(n)}\rightarrow B^{(0)}$
be the truncated transfer operator with normalized insertions~$\{a_i\}_{i=1}^n$ and parameters $(q,z)$ (cf.~Definition~\ref{def:transferoperatorstruncated}).
Suppose $v^{(0)} \in B^{(0)}_{\ell_0}$ is at level~$\ell_0$ and 
$v^{(n)} \in B^{(n)}_{\ell_n}$ is at level~$\ell_n$ and both vectors are of unit norm. Then the genus-$0$ correlation function $F^{(0)}_{v^{(0)},v^{(n)}}$ is approximated as 
\begin{align*}
  \big|\langle v^{(0)},\Top^{[N]}v^{(n)}\rangle- q^{(n+1/2)\ell_n+\ell_0/2+\sum_{j=1}^n j\,\wt a_j}\cdot F^{(0)}_{v^{(0)},v^{(n)}}((\tilde{a}_1,\zetap_1),\ldots,(\tilde{a}_n,\zetap_n))\big| \leq q^{N/4} \,\Delta(n,q,z)\,,  
\end{align*}
where
\begin{align}
\tilde{a}_j=q^{L_0/2}a_j\qquad\textrm{ and }\qquad \zetap_j=z q^j\qquad\textrm{ for }j=1,\ldots,n\, .
\end{align}
\end{corollary}

\noindent 
Note that the dependence on $N$ is exponential. However, the bound turns out to be only interesting for $N$ large compared to the difference $|\ell_0-\ell_n|$ between the levels: as we will show below in~\eqref{eq:finiteprojcvdf}, the expression $\Scp{v^{(0)}}{\Top^{[N]}v^{(n)}}$ vanishes if $|\ell_0-\ell_n|>n N$.

\begin{proof}
  Inserting the expression of Corollary~\ref{cor:corrfctexactzero} for the correlation function in terms of the transfer operator $\Top$, we find
  \begin{align}
    &\left|\langle v^{(0)},\Top^{[N]}v^{(n)}\rangle - q^{\left[(n+1/2)\ell_n+\ell_0/2+\sum_{j=1}^n j\,\wt a_j\right]}F^{(0)}_{v^{(0)},v^{(n)}}((a_1,\zetap_1),\ldots,(a_n,\zetap_n)) \right| \nonumber\\
    &\quad= \left|  \Scp{v^{(0)}}{\left(\Top^{[N]} - \Top\right) v^{(n)}} \right| \leq \norm{\Top^{[N]} - \Top} \norm{v^{(0)}}_{B^{(0)}} \norm{v^{(n)}}_{B^{(n)}} \,.
  \end{align}
  Inserting the bound from Lemma~\ref{lem:normboundtransferop} finishes the proof.
\end{proof}

\begin{corollary}[Approximate reproduction of genus-$1$ correlation functions]\label{cor:correlationfunctionsapprox}
Assume periodic boundary conditions $B^{(0)}=B^{(n)}=B$. Fix normalized $a_i\in \Shom^{(i)}$, $i=1,\ldots,n$ and  let $\Top^{[N]}:B^{(n)}\rightarrow B^{(0)}$
be the truncated transfer operator with insertions~$\{a_i\}_{i=1}^n$ and parameters $(q,z)$ (cf.~Definition~\ref{def:transferoperatorstruncated}).
Let $0<r<1$, and let 
\begin{align}
\tilde{a}_j=q^{L_0/2}a_j\qquad\textrm{ and }\qquad \zetap_j=zq^j\qquad\textrm{ for }j=1,\ldots,n\ ,
\torusq&=r q^{n}\ .\label{eq:ajzjdef}
\end{align} Then
the genus-$1$ correlation function (cf.~Eq.~\eqref{eq:genusoneintertwinercorrelationfct})
\begin{align}
F^{(1)}_{\torusq}\left((\tilde{a}_1,\zetap_1),\ldots,(\tilde{a}_n\zetap_n)\right)\label{eq:genusonecorrelationfunctionapproximationexpression}
\end{align}
 is approximated as
\begin{align}\label{eq:correlationfunctionapproxexpr}
  \big|\tr_B\Top^{[N]}r^{L_0}-\torusq^{c/24} F^{(1)}_{\torusq}\left((\tilde{a}_1,z_1),\ldots,(\tilde{a}_n,z_n)\right) \big|\leq q^{N/4} \,\Delta(n,q,z)\cdot  r^{c/24} Z_B(r)
\end{align}
where $Z_B(r)$ is the partition function (or character) of the module $B$,
\begin{align}
Z_B(r) = \tr_B[r^{L_0-c/24}]\ .\label{eq:errortermlambdadef}
\end{align}
\end{corollary}

Consider for example the (typical) translation-invariant case,
where the modules $A^{(j)}=A$, $B^{(j)}=B$ and intertwiners
$\iY_j=\iY$, $j=1,\ldots,n$, are all identical.
Then the error bound  becomes
\begin{align}
  \big|\tr_B\Top^{[N]}r^{L_0}-\torusq^{c/24}F \big|\leq q^{N/4} \,n\,\kappa\,\left[\frac{\sqrt{|I_B|}\cdot\bnd(\sqrt{q},z)}{1-\sqrt{q}}\right]^n \cdot  r^{c/24} Z_B(r)
\end{align}
The expression in the square brackets should be regarded as normalization factor determining the overall scale of the
approximation error. In contrast, the error $q^{N/4} \,n\,\kappa$ decays exponentially with~$N$. Finally, the third factor in~\eqref{eq:correlationfunctionapproxexpr} only depends on the regularization parameter~$r$ and should be considered as a constant. To summarize, Eq.~\eqref{eq:correlationfunctionapproxexpr} implies that  the error is exponentially small in the truncation level~$N$  independently of the number~$n$  of insertions.    

\begin{proof}
  The proof of the equivalent result in the genus-$1$ case follows again by using Lemma~\ref{lem:normboundtransferop}, as well as by using the H\"older inequality for the Schatten norms. Explicitly, denoting the correlation function by
  \begin{align}
    F = F^{(1)}_{\torusq}\left((\tilde{a}_1,\zetap_1),\ldots,(\tilde{a}_n,\zetap_n)\right) \,,
  \end{align}
  we have by Corollary~\ref{cor:corrfctexactrepr}
  \begin{align}
    |\tr_B\Top^{[N]}r^{L_0}-\torusq^{c/24} F| = |\tr_B\left(\Top^{[N]} - \Top\right) r^{L_0}| \leq \norm{\Top^{[N]} - \Top} \,\tr_B r^{L_0}\,,
  \end{align}
  since the operator $r^{L_0}$ is positive. Inserting the bound from Lemma~\ref{lem:normboundtransferop} as before completes the argument.
\end{proof}

\subsubsection{Approximation using finite bond dimension\label{sec:finitebonddimensionerror}}
Up to this point, we have  merely truncated and scaled the intertwiner~$\iY$ in such a way that it does not change the level in one application by more that~$N$ in absolute value, see Eq.~\eqref{eq:containmenttruncatedscaledintertwiner}. However, this does not mean  that the spaces involved in the evaluation of e.g., approximate correlation functions (as in Lemma~\ref{cor:correlationfunctionsapprox}) are finite-dimensional. To obtain a representation by an MPS with finite bond dimension, we generally have to further restrict the truncation space (unless we are considering genus-$0$ vacuum-to-vacuum correlation functions, as discussed below).

More precisely, the trace in expression~\eqref{eq:correlationfunctionapproxexpr}
is taken over the infinite-dimensional module $B^{(0)}=B^{(n)}=B$.
 To obtain an approximation using a finite bond dimension, consider the $\mathbb{N}_0$-grading $B=\bigoplus_{n\in \mathbb{N}_0}B_n$ and the projection~$P^{[M]}$ onto $\bigoplus_{n\leq M} B_n$. That is, the operator~$P^{[M]}$ projects onto the highest weights (up to and including the level~$M$).  Here $M\in\mathbb{N}$ is an additional truncation parameter which ultimately determines the bond dimension of the resulting MPS (see Lemma~\ref{lem:projectedcorrelationfct} below) according to
\begin{align}
d_B(M):=\dim P^{[M]} B&=\sum_{0\leq n\leq M} \dim B_n\ \qquad\textrm{ for }M\geq 0\ .\label{eq:dbmdef}
\end{align}
To distinguish it from the `truncation level'~$N$ used in the definition of truncated intertwiners, let us refer to $M$~simply as the {\em cutoff} or {\em cutoff parameter}.
 
We are ultimately interested in the approximation error when approximating correlation functions by MPS with finite bond dimension.   For this purpose,  we need a refinement of Corollary~\ref{cor:correlationfunctionsapprox}, which provides a bound on the accuracy when approximating a correlation function by a truncated trace~$\tr_B P_B^{[M]} \Top^{[N]}r^{L_0}$. In Section~\ref{sec:recoveringthemps}, we will argue that this expression is given by matrix elements of an MPS. 

\begin{lemma}\label{lem:projectedcorrelationfct}
Assume periodic boundary conditions $B^{(0)}=B^{(n)}=B$. Fix normalized $a_i\in \Shom^{(i)}$, $i=1,\ldots,n$ and  let $\Top^{[N]}:B^{(n)}\rightarrow B^{(0)}$
be the truncated transfer operator with insertions~$\{a_i\}_{i=1}^n$ and parameters $(q,z)$ (cf.~Definition~\ref{def:transferoperatorstruncated}).
For $M>0$, let $P_B^{[M]}$ be the projection onto~$\bigoplus_{0\leq m\leq M}B_m$.  Let $F$ be the genus-one correlation function defined by~\eqref{eq:genusonecorrelationfunctionapproximationexpression}.
Then for any $0<r<1$, we have
\begin{align}
\big|
\tr_BP_B^{[M]}\Top^{[N]}r^{L_0}-\torusq^{c/24}F \big|\leq   \left(n\,\kappa\,q^{N/4} \, r^{c/24}Z_B(r) +  r^{M/2} r^{c/12}Z_B(\sqrt{r}) \right)\\
\qquad \cdot \left(\max_{1\leq j \leq n}\frac{\sqrt{|I_{B^{(j)}}|}\bnd_j(\sqrt{q},z)}{1-\sqrt{q}}\right)^n
\end{align}
where $Z_B(r)$ is the character of the module~$B$ (cf.~\eqref{eq:partitionfunction}). 
\end{lemma}

Here we should again think of the multiplicative factor involving the boundedness parameters~$\vartheta_j$ as a normalization factor, as it determines the overall scale of the correlation function. The partition function appearing in both terms plays the same role. Hence this yields an exponentially small (in the truncation parameters $N,M$) error estimate  when approximating a correlation function by the expression~$\tr_BP_B^{[M]}\Top^{[N]}r^{L_0}$. It remains to connect this expression to a finite-dimensional MPS. In the next section, we show that for a certain bond dimension~$D$ and a certain choice of  matrices~$\{\mpsA_a\}_a\in\mathsf{Mat}(\mathbb{C}^D)$ defining an MPS, the quantity~$\tr_BP_B^{[M]}\Top^{[N]}r^{L_0}$ is equal to the  sum of certain matrix elements of the MPS.

\begin{proof}
Recall from Corollary~\ref{cor:corrfctexactrepr} that
$\tr \Top r^{L_0}=\torusq^{c/24}F$ is the exact correlation function.
Using the triangle inequality, we thus have
\begin{align*}
\big|\tr_BP_B^{[M]}\Top^{[N]}r^{L_0}-\torusq^{c/24}F \big|&\leq \big|\tr_BP_B^{[M]}\Top^{[N]}r^{L_0} - \tr_B \Top^{[N]}r^{L_0}\big| + \big|\tr_B \Top^{[N]}r^{L_0} - \torusq^{c/24}F \big|\\
&=|\tr_B\left(\idty-P_B^{[M]}\right)\Top^{[N]}r^{L_0}\big| + \big|\tr_B \Top^{[N]}r^{L_0} - \torusq^{c/24}F \big|\\
&\leq |\tr_B\left(\idty-P_B^{[M]}\right)r^{L_0}\big| \norm{\Top^{[N]}} + q^{N/4} \,\Delta(n,q,z)\cdot  r^{c/24} Z_B(r)\,,
\end{align*}
where we used the result of Corollary~\ref{cor:correlationfunctionsapprox} in the last step as well as the H\"older inequality for Schatten norms. We proceed in two steps: First, following exactly the same arguments of Lemma~\ref{lem:boundedtransferop} and using Lemma~\ref{lem:boundednessparametertruncated}, we find
\begin{align}
  \norm{\Top^{[N]}} \leq \prod_{j=1}^n \bndt_j(q,z) \,, 
\end{align}
since all vectors $a_i$ are supposed to be normalized. Second, since $\idty - P_B^{[M]}$ is the projector onto all weight subspaces of weight lower bounded by $M$, we have
\begin{align}
  \tr_B \left(\idty - P_B^{[M]}\right) r^{L_0} = \sum_{n\geq M} \dim(B_n) r^{n} \leq r^{M/2} \sum_{n\geq0} \dim(B_n) r^{n/2} = r^{(h_B+M)/2} \tr_B r^{L_0/2}\,.
\end{align}
The claim follows since $\tr(r^{L_0/2})=r^{c/12}Z_B(\sqrt{r})$, \(r^{h_B/2}<1\), and by inserting the definition of $\bndt_j(q,z)$ as well as using $\sqrt{1-\sqrt{q}} \geq 1-\sqrt{q}$. 
\end{proof}

\subsection{Structure of approximation and the bond dimension\label{sec:bondimension}}
In the following, we combine our results and give the proof of Theorem~\ref{thm:main1}. On a structural level, we first have to connect matrix elements of the truncated transfer operator to an MPS. In Section~\ref{sec:propsofapprox}, we will then analyze the bond dimension of the resulting MPS.

\subsubsection{Recovering the MPS\label{sec:recoveringthemps}}

To relate expressions such as $\langle v^{(0)},\Top^{[N]}v^{(n)}\rangle$ and~$\tr_BP_B^{[M]}\Top^{[N]}r^{L_0}$  to a finite-dimensional MPS, we 
simply use Eq.~\eqref{eq:containmenttruncatedscaledintertwiner}, i.e., the fact that application of a truncated scaled intertwiner $\W_q^{[N]}$ does not change the level of a vector by more than $N$ in absolute value. As an immediate consequence,
any sequence of  truncated intertwiners~$\W^{[N]}_i(a_i,z)$ applied to vector of bounded weight is still contained in a finite-dimensional subspace.  This fact is made use of  in the proof of the following result, which gives an upper bound on the bond dimension.
 (Recall that $d_B(M)$ is 
the dimension of the subspace obtained by keeping all levels up to and including~$M$ in the module~$B$, see~\eqref{eq:dbmdef}). 
We will also show (as mentioned earlier) that 
the expression $\langle v^{(0)},\Top^{[N]}v^{(n)}\rangle$ vanishes unless $N$ is large compared to the difference between the levels of $v^{(0)}$ and $v^{(n)}$.

Recall from Eq.~\eqref{eq:dbmdef} that $d_B(M)$ denotes the total dimension of all levels up to and including~$B_M$ in the module~$B$.
\begin{lemma}[Genus-$0$ MPS]\label{lem:mpsformapproxgenuszero}
  For $a_i\in \Shom^{(i)}$, $i=1,\ldots,n$, and \(z\in \C \setminus \{0\}\), \(0<q<\min\{|z|^2,|z|^{-2}\}\), let $\Top^{[N]}:B^{(n)}\rightarrow B^{(0)}$ 
be the truncated transfer operator defined in Lemma~\ref{lem:normboundtransferop}.
Suppose $v^{(0)} \in B^{(0)}_{\ell^{(0)}}$ is at level~$\ell^{(0)}$ and 
$v^{(n)} \in B^{(n)}_{\ell^{(n)}}$ is at level~$\ell^{(n)}$. 
Then
\begin{align}
\Scp{v^{(0)}}{\Top^{[N]}v^{(n)}}&=0\qquad\textrm{ if }\qquad |\ell^{(0)}-\ell^{(n)}|>n N\ . \label{eq:finiteprojcvdf}
\end{align}
Fix a cutoff $M>0$. Then there exist operators $\tilde{\mpsA}_{a_1}, \tilde{\mpsA}_{a_2}, \dots, \tilde{\mpsA}_{a_n}$ on $\C^D$, linearly depending on the vectors~$a_j$, with 
\begin{align*}
D&=\max_{0\leq j\leq n }d_{B^{(j)}}(M+nN)\ ,
\end{align*} 
and linear embeddings
\begin{align*}
\begin{matrix}
\bigoplus_{0\leq m\leq M} B^{(n)}_m&\rightarrow  &\mathbb{C}^{D^{(n)}}\subset \mathbb{C}^D\\
v^{(n)} & \mapsto &\tilde{v}^{(n)}
\end{matrix}\qquad\textrm{ and }\qquad 
\begin{matrix}
\bigoplus_{0\leq m\leq M} B^{(0)}_m&\rightarrow & \mathbb{C}^{D^{(0)}}\subset \mathbb{C}^D\\
v^{(0)} & \mapsto& \tilde{v}^{(0)}
\end{matrix}
\end{align*}
with $D^{(n)}= d_{B^{(n)}}(M)$, $D^{(0)}= d_{B^{(0)}}(M)$, such that
\begin{align}
\Scp{v^{(0)}}{\Top^{[N]}v^{(n)}}_{B^{(0)}}&=\langle \tilde{v}^{(0)},\tilde{\mpsA}_{a_1}\cdots \tilde{\mpsA}_{a_n}\tilde{v}^{(n)}\rangle_{\mathbb{C}^D}\ .
\end{align}
for all $v^{(0)}\in\bigoplus_{0\leq m\leq M}B^{(0)}_m$  and $v^{(n)}\in\bigoplus_{0\leq m\leq M}B^{(n)}_m$. 
\end{lemma} 
\noindent 

Note that Lemma~\ref{lem:mpsformapproxgenuszero}
merely expresses the fact that $\Scp{v^{(0)}}{\Top^{[N]}v^{(n)}}$ can be related to a finite-dimensional MPS; we defer the discussion of the 
relationship between the approximation accuracy of correlation functions and the bond dimension to  Section~\ref{sec:bondimension}. 

We also point out that in the special case of vacuum-to-vacuum correlation functions, we will set $v^{(0)}=v^{(n)}=\1$ and can thus take $M=0$.
\begin{proof}
The expression $\Top^{[N]}v^{(n)}$  is defined recursively with $\Top^{[N]}v^{(n)}=\Top^{[N]}_1v^{(n)}$ as 
\begin{align*}
  \Top^{[N]}_n v^{(n)}&=\W^{[N]}_n(a_n,z)(v^{(n)})\\
  \Top^{[N]}_{k}v^{(n)}&=\W^{[N]}_{k}(a_{k},z)\circ \Top^{[N]}_{k+1}(v^{(n)})\qquad\textrm{ for }k=n-1,\ldots,1\ .
\end{align*}
In particular, if we define operators
\begin{align*}
\mpsA_{a_j} :  B^{(j)}&\rightarrow B^{(j-1)}\\
b &\mapsto \W^{[N]}_{j}(a_j,z)(b) 
\end{align*}
we have
\begin{align}
\Top^{[N]}v^{(n)} &=\mpsA_{a_1}\cdots \mpsA_{a_{n}} v^{(n)}\ .\label{eq:zonenv1}
\end{align}
Furthermore, by~\eqref{eq:containmenttruncatedscaledintertwiner}, we have the containment
\begin{align}
\mpsA_{a_j}B^{(j)}_m \subset \bigoplus_{|m'-m|\leq N}B^{(j-1)}_{m'}\ \label{eq:bjmbprimetwosided}
\end{align}
 for any level $m\in\mathbb{N}$. In particular, if $v^{(n)}\in B^{(n)}_{\ell^{(n)}}$, then 
we conclude that 
\begin{align*}
\Top^{[N]}v^{(n)}&\in \bigoplus_{|m'-\ell^{(n)}|\leq nN}B^{(0)}_{m'}\ .
\end{align*}
by iteratively applying~\eqref{eq:bjmbprimetwosided} to~\eqref{eq:zonenv1}.
Since different levels are orthogonal, Eq.~\eqref{eq:finiteprojcvdf} follows immediately. Eq.~\eqref{eq:bjmbprimetwosided} also implies
\begin{align}
\mpsA_{a_j}\bigoplus_{0\leq m\leq M}B^{(j)}_m \subset \bigoplus_{0\leq m\leq M+N}B^{(j-1)}_{m}\ \label{eq:bjmbprime}
\end{align}
  Iteratively applying~\eqref{eq:bjmbprime}, we conclude that for $v^{(n)}\in \bigoplus_{0\leq m\leq M}B^{(n)}_{m}=P_{B^{(n)}}^{[M]}B^{(n)}$ and all $r=1,\ldots,n$, we have 
\begin{align}
\mpsA_{a_{r}}\mpsA_{a_{r+1}}\cdots \mpsA_{a_n}v^{(n)}\in \bigoplus_{0\leq m\leq M+(n-r+1)N} B^{(r-1)}_{m}\subset \bigoplus_{0\leq m\leq M+nN}B_{m}^{(r-1)}=P^{[M+nN]}_{B^{(r-1)}}B^{(r-1)}\ .\label{eq:productaaj}
\end{align}
For any
\begin{align*}
v^{(n)}\in \bigoplus_{0\leq m\leq M} B_m^{(n)}\qquad\textrm{ and }\qquad v^{(0)}\in \bigoplus_{0\leq m\leq M} B_m^{(0)}\ ,
\end{align*}
projecting the vectors, i.e., setting 
$\tilde{v}^{(0)}=P^{(0)}_{B^{(0)}}v^{(0)}$ and $\tilde{v}^{(n)}=P^{(n)}_{B^{(n)}}v^{(n)}$
leaves them invariant, i.e., 
\begin{align}
v^{(0)}=\tilde{v}^{(0)}\qquad\textrm{ and }\qquad \tilde{v}^{(n)}=\tilde{v}^{(n)}\ .\label{eq:projectedvectorsinvariance}
\end{align}
Combining~\eqref{eq:zonenv1},~\eqref{eq:productaaj} and~\eqref{eq:projectedvectorsinvariance} therefore gives 
\begin{align*}
\langle v^{(0)},\Top^{[N]}v^{(n)}\rangle&=\langle\tilde{v}^{(0)},\tilde{\mpsA}_{a_1}\cdots \tilde{\mpsA}_{a_n} \tilde{v}^{(n)}\rangle\ ,
\end{align*} 
where 
\begin{align*}
  \tilde{\mpsA}_{a_j}&=P^{[M+nN]}_{B^{(j-1)}}\mpsA_{a_j}P^{[M+nN]}_{B^{(j)}}
\end{align*}
The claim follows since all  operators and vectors involved in the expression~\eqref{eq:pbmprod} are
supported on spaces of the form $\bigoplus_{0\leq m\leq M+nN}B^{(j)}_m$. These spaces can be embedded in~$\mathbb{C}^D$ with $D=\max_{0\leq j\leq n} d_{B^{(j)}}(M+nN)$, their maximal dimension. 
\end{proof}
Note that in the second half of this proof, we only made use of~\eqref{eq:bjmbprime} instead of the stronger two-sided bound~\eqref{eq:bjmbprimetwosided}. This is motivated by the fact that we are mostly interested in correlation functions associated with vectors of low weight, as for example for vacuum-to-vacuum correlation functions. Other correlation functions may be extracted using the Ward identities~\cite{francesco2012conformal}.

To make a similar statement about genus-$1$-correlation functions, we again introduce a cutoff by projecting onto the subspace of levels less than $M$ of the module $B$.  As before (see Section~\ref{sec:finitebonddimensionerror}), we denote the corresponding projection by~$P_B^{[M]}$. The argument proceeds by reduction to the genus-$0$-case.

\begin{lemma}[Genus-$1$-MPS]\label{lem:mpsformapproxtorus}
  Let $A^{(i)}$ and $B^{(i)}$, $i=1,\ldots,n$ be unitary modules of a VOA~$\cV$, with $B^{(0)}=B^{(n)}=:B$. Let $S^{(i)}\subset A^{(i)}$ and  $a_i\in \Shom^{(i)}$ be as in  Lemma~\ref{lem:normboundtransferop}. Let  $\Top^{[N]}:B\rightarrow B$ be the associated truncated transfer operator. Then there exist operators $\tilde{\mpsA}_{a_1}, \tilde{\mpsA}_{a_2}, \dots, \tilde{\mpsA}_{a_n}$ on $\C^D$, linearly depending on the vectors $a_j$, with $D = \max_{1\leq j\leq N}d_{B^{(j)}}(M+ nN)$,  
and an operator~$\tilde{\mathsf{X}}$ on $\mathbb{C}^D$ of rank~$D^{(0)}=d_B(M)\leq D$ such that 
  \begin{align}\label{eq:lemmpsformapprox}
    \tr_BP_B^{[M]}\Top^{[N]}r^{L_0} = \tr_{\mathbb{C}^D}\left[\tilde{\mpsA}_{a_1} \tilde{\mpsA}_{a_2} \cdots \tilde{\mpsA}_{a_n} \tilde{\mathsf{X}}\right] \,.
  \end{align}
\end{lemma}

\noindent 
We remark that the dependence of the bond dimension~$D$ on $n$ can be slightly improved from $n$ to $\lfloor \frac{n}{2}\rfloor$ using the fact that we are taking the trace, though this does not affect our conclusions in a significant way. 

\begin{proof}
Choose an orthonormal basis $\{b_i\}_{i=1}^{d_B(M)}$ of the space $P_B^{[M]}B=\bigoplus_{0\leq m\leq M }B_{m}$ consisting of homogeneous vectors. 
Then 
\begin{align}
  \tr_BP_B^{[M]}\Top^{[N]}r^{L_0} = \sum_{i=1}^{d_B(M)} r^{\wt b_i} \langle b_i, \Top^{[N]}b_i\rangle\label{eq:firstidentitypbm}
\end{align} 
By Lemma~\ref{lem:mpsformapproxgenuszero}, there are operators 
$\mpsA_{a_j}:B^{(j)}\rightarrow B^{(j-1)}$ for $j=1,\ldots,n$ 
such that the projected operators $\tilde{\mpsA}_{a_j}=P^{[M+nM]}_{B^{(j-1)}}\mpsA_{a_j}P^{[M+nM]}_{B^{(j)}}$ satisfy
\begin{align}
 \langle P^{[M]}_{B^{(0)}}b_i, \tilde{\mpsA}_{a_1}\cdots\tilde{\mpsA}_{a_n}P^{[M]}_{B^{(n)}}b_i\rangle&=\langle b_i, \Top^{[N]}b_i\rangle\qquad\textrm{ for all }i=1,\ldots,d_B(M)\ .\label{eq:secondidentityvm}
\end{align}
Observe that because $B^{(0)}=B^{(n)}=B$ in the genus-$1$-case, and since $\{b_i\}_i$ is a basis of $\bigoplus_{0\leq m\leq M}B_m$, we have $P^{[M]}_{B^{(n)}}b_i=P^{[M]}_{B^{(0)}}b_i=b_i$. Because $L_0$ is compatible with the grading, we further have $r^{L_0}P_B^{[M]}b_i=r^{\wt b_i}b_i$. Defining 
\begin{align*}
\tilde{\mathsf{X}}&=P^{[M]}_{B}r^{L_0}P^{[M]}_{B}\qquad\textrm{ and }\qquad \tilde{b}_i=P_B^{[M]}b_i\ ,
\end{align*}
we hence get
\begin{align}
\tr_B P_B^{[M]}\Top^{[N]}r^{L_0}&=\sum_{i=1}^{d_B(M)}\langle \tilde{b}_i,\tilde{\mpsA}_{a_1}\cdots \tilde{\mpsA}_{a_n}\tilde{X}\tilde{b}_i\rangle\label{eq:pbmprod}
\end{align}
 by combining~\eqref{eq:firstidentitypbm} and~\eqref{eq:secondidentityvm}. 
The claim follows since all  operators and vectors involved in the expression~\eqref{eq:pbmprod} are
supported on spaces of the form $\bigoplus_{0\leq m\leq M+nN}B^{(j)}_m$ which can all be embedded in $\mathbb{C}^D$. 
\end{proof} 

Let us now for simplicity assume that $d=\dim S^{(j)}$ is identical for all $1\leq j \leq n$, and let $\{\ket{k}\}_{1\leq k \leq d}$ denote an orthonormal basis of $\C^d$. Since the operators $\tilde{\mpsA}_{a_j}$ depend linearly on the vectors $a_j$, we can define a linear form on $(\C^d)^{\otimes n}$ as
\begin{align}\label{eq:mpsformapprox}
  \functionalstate_{MPS_1} : \ket{k_1} \otimes \ket{k_2} \otimes \cdots \otimes \ket{k_n} \mapsto \tr_{\C^D}\left[\tilde{\mpsA}_{k_1} \tilde{\mpsA}_{k_2} \cdots \tilde{\mpsA}_{k_n} \tilde{\mathsf{X}}\right] \,,
\end{align}
where we denoted $\tilde{\mpsA}_{k_2} = \tilde{\mpsA}_{\ket{k_2}}$. We can thus finally collect the pieces of provide a proof of our first main result.

\begin{proof}[Proof of Theorem~\ref{thm:main1}]
In the genus-0 case, we use Lemma~\ref{lem:mpsformapproxgenuszero}
(in the form~\eqref{eq:mpsformapprox}) 
to write the vacuum expectation value of the truncated transfer operator as
 the matrix element of an MPS (cf.~\eqref{eq:MPSstatex}), that is,
\begin{align*}
\functionalstate_{MPS_0}(\ket{k_1}\otimes\cdots\otimes\ket{k_n})&=
 \langle \1, \Top^{[N]}\1\rangle_{B}\ .
\end{align*}
According to Corollary~\ref{cor:approxcorfuncplane}, we have (since the level of the vacuum~$\1$ is $\ell_0=\ell_n=0$) 
\begin{align}
\big| \langle \1, \Top^{[N]}\1\rangle_{B}-q^{\sum_{j=1}^n j \wt a_j}\cdot F^{(0)}((\tilde{a}_1,\zeta_1),\ldots,(\tilde{a}_n,\zeta_n))\big| &\leq q^{N/4} \,\Delta(n,q,z)\\
&=:\epsilon_{MPS_0}(n,q,z,N)
\end{align}
This means that for primary fields $\{a_j\}_{j=1}^n$, the MPS approximates 
a (scalar multiple)  of the genus-$0$ $n$-point vacuum-to-vacuum correlation  function (see Observation~\ref{obs:conformalmapcorrfuncplane}) with error  
(cf.~\eqref{eq:deltanqzdefa} for the definition of $\Delta(n,q,z)$)
\begin{align}
\epsilon_{MPS_0}(n,q,z,N)=q^{N/4} \cdot  n\cdot \kappa\cdot \left(\max_{1\leq j \leq n}\frac{\sqrt{|I_{B^{(j)}}|}\cdot \bnd_j(\sqrt{q},z)}{1-\sqrt{q}}\right)^n\quad\textrm{ where } z=e^{-d_0}e^{i\theta}\ , q=e^{-\mathsf{d}}\ .\label{eq:mpserror}
\end{align}
Consider now the  genus-$1$ case. Here we have to additionally cut off the Hilbert space dimension at some level~$M\in\mathbb{N}$.
Comparing the rhs.~of Eq.~\eqref{eq:mpsformapprox} with Eq.~\eqref{eq:mpsdef}, we see that they agree. Correspondingly, we have (with Lemma~\ref{lem:mpsformapproxtorus}) an MPS $\functionalstate_{MPS_1}$ such that
\begin{align}
  \functionalstate_{MPS_1}(\ket{k_1}\otimes\ket{k_2}\otimes\cdots\otimes\ket{k_n})&=\tr_BP_B^{[M]}\Top^{[N]}r^{L_0}\ .
\end{align}
But according to Lemma~\ref{lem:projectedcorrelationfct}, the latter expression approximates a correlation function (see Observation~\ref{obs:conformalmapcorrfunctorus}) 
\begin{align*}
\big|\tr_BP_B^{[M]}\Top^{[N]}r^{L_0}
-\torusq^{c/24}F^{(1)}_{\torusq}\left((\tilde{a}_1,\zeta_1),\ldots,(\tilde{a}_n,\zeta_n)\right)\big| \leq \epsilon_{MPS_1}(n,q,z,r,N,M)
\end{align*}
with error 
\begin{align}
\epsilon_{MPS_1}(n,q,z,r,N,M)=\left(n\,\kappa\,q^{N/4} \, r^{c/24}Z_B(r) +  r^{M/2} r^{c/12}Z_B(\sqrt{r}) \right)\\\
\qquad\cdot\left(\max_{1\leq j \leq n}\frac{\sqrt{|I_{B^{(j)}}|}\cdot\bnd_j(\sqrt{q},z)}{1-\sqrt{q}}\right)^n\ .\label{eq:mpsoneerror}
\end{align}
Thus the MPS~$\functionalstate_{MPS_1}$ approximates the correlation genus-$1$-correlation function, as claimed. 
\end{proof}

\subsubsection{Scaling of the required bond dimension}\label{sec:propsofapprox}
Having constructed an MPS for approximating CFT correlation functions, we will next discuss  the relationship between the bond dimension, the approximation accuracy, and parameters of the CFT.  Recall that the finite-dimensional MPS approximates
the exact correlation function (with prefactors) with error 
$\epsilon_{MPS_0}$ (cf.~\eqref{eq:mpserror}) in the genus-$0$ case and with error~$\epsilon_{MPS_1}$ (cf.~\eqref{eq:mpsoneerror}) in the genus-$1$ case. 
Both quantities depend on
\begin{enumerate}[(i)]
\item\label{it:numinsertions}
the number~$n$ of insertions. (In the case of a full CFT, the number~$n$  can be replaced by the number of non-trivial insertions, i.e., those not equal to the vacuum vector -- see the discussion in Section~\ref{sec:finitelycorrfull}.)
\item\label{it:minimumdistancerl}
the minimal distance~$\mathsf{d}$ between insertion points (via $q=e^{-\mathsf{d}}$)  and the offset~$\mathsf{d}_0$ (see Fig.~\ref{fig:torusvm}).
(Alternatively, we may use the length~$L$ of the interval considered: according to Observation~\ref{obs:conformalmapcorrfuncplane}, the number of insertions is equal to~$n=L/\mathsf{d}$, where $\mathsf{d}$ is again the minimal distance between two insertion points.) The quantity $\mathsf{d}$ (respectively $1/\mathsf{d}$) may be identified as the ultraviolett (UV) cutoff; the continuum limit corresponds to~$\mathsf{d}\rightarrow 0$.
 In addition, in the genus-$1$ case, we have a regularization parameter $0< r\leq 1$ (which determines the diameter of the torus via $\mathsf{p}=re^{-n\mathsf{d}}$). In the limit $\mathsf{d}\rightarrow 0$, it is natural to set this to an exponentially small value in~$\mathsf{d}$. For concreteness, we will choose $r=q^{1/2}=e^{-\mathsf{d}/2}$. 
\item\label{it:boundednessparamtheta}
the CFT under consideration, in terms of the boundedness parameters~$\vartheta_j$,
\item
the truncation parameter~$N\in\mathbb{N}_0$, and, in the genus-$1$-case, the cutoff parameter~$M\in\mathbb{N}_0$. We will treat both on the same footing by setting $M=0$ for genus-$0$ vacuum-to-vacuum correlation functions  (see remark after Lemma~\ref{lem:mpsformapproxgenuszero}).
\label{it:replaceditem}
\end{enumerate}
The quantities~\eqref{it:numinsertions} and~\eqref{it:minimumdistancerl} determine the correlation function under consideration.  In a first step, we will substitute~\eqref{it:replaceditem}
by
\begin{enumerate}[(i')] \setcounter{enumi}{3}
\item\label{it:dimensioncft}
the bond dimension~$D$ of the MPS, determined by 
\begin{align*}
D= d_{B}(M+nN)\ .
\end{align*}
Here we assume translation-invariance for simplicity, i.e., the modules $B^{(j)}=B$ are all identical for $j=1,\ldots,n$. 
According to Lemma~\ref{lem:mpsformapproxgenuszero} and~\ref{lem:mpsformapproxtorus}, this bond dimension is 
sufficient to yield an MPS with error~$\epsilon_{MPS}$. 
\end{enumerate}
We stress that~(\ref{it:dimensioncft}') depends on the dimension of the weight spaces, which in turn is determined by the CFT under consideration. Indeed, the following analysis shows that~(\ref{it:dimensioncft}') is the most significant dependence on the CFT (rather than~\eqref{it:boundednessparamtheta}). In the following, we will use a simple upper bound on the growth of the weight spaces of modules (i.e., the function $d_B(\cdot)$). Its dependence on the VOA~$\cV$ in question is captured  by the quantity
\begin{align}\label{eq:dimctwodefv}
\dimctwo = \dim \cV/C_2 \,,\qquad\,C_2=\mathsf{span}\{u_{-2}v\ |\ u,v\in\cV\}\ .
\end{align} 
We then have the following statements:
\begin{corollary}
\begin{enumerate}[(i)]
\item\label{it:firstclaimbonddim}
The constructed MPS  requires a bond dimension~$D$ scaling polynomially in~$1/\eps$, where $\eps>0$ is the approximation accuracy (i.e., $\eps_{MPS}\leq \eps$). 
The degree of the polynomial depends linearly on each of the parameters~$\dimctwo$, the number of insertions~$n$, and the inverse distance~$1/d$ between insertions. That is, there are constants~$\Omega_1,\Omega_2>0$, such that a bond dimension of 
\begin{align*}
D\approx \Omega_1 \cdot\left(1/\epsilon\right)^{\Omega_2\,\dimctwo \,\frac{n}{d}}\ 
\end{align*}
is sufficient to yield accuracy~$\eps$. 
\item\label{it:secondclaimbonddim}
For a fixed approximation accuracy~$\eps>0$, the required bond dimension~$D$ grows sub-exponentially with the number~$n$ of insertions (or the number of non-trivial insertions in the full CFT case, see Section~\ref{sec:finitelycorrfull}), and the rate of growth is determined by~$\dimctwo$.
More precisely, there is a constant~$\Omega=\Omega(\eps)>0$ such that 
\begin{align}
D\approx \Omega \cdot e^{\frac{\pi}{3}\sqrt{\dimctwo n}}\label{eq:logarithmicbonddimbnd}
\end{align}
is sufficient to yield the desired accuracy. 
\end{enumerate}
\noindent In particular, for a fixed interval length~$L$, the dependence of the bond dimension on the UV cutoff~$1/\mathsf{d}=n/L$ is sub-exponential. 
\end{corollary}
Remarkably, the dependence of the required bond dimension on the VOA~$\cV$ is reduced to a single parameter~$\dimctwo$. In some sense, this quantity determines the difficulty of approximation; because of its direct link to the bond dimension, it appears to represent a measure of the amount of correlations and/or entanglement encoded in correlation functions.  Following Gaberdiel and Neitzke~\cite[p. 324]{GaberdielNeitzke}, we can interpret the parameter~$\dimctwo$ roughly as the number of degrees of freedom of our theory if counted using free fermions for comparison. This matches the assumption that a theory with more degrees of freedom should require a larger bond dimension. 

If the quantity~$\dimctwo$ could be replaced by the central charge~$c$  of the corresponding representation of the Virasoro algebra acting on the module $B$, these bounds would be closer to what may be  conjectured based on results found in the literature: a bound similar  to~\eqref{eq:logarithmicbonddimbnd} would imply that  the bond dimension scales as the density of states of the CFT (with the UV cutoff $1/\mathsf{d}$ being the corresponding energy). The latter  is given by the famous formula of Cardy in the context of BTZ black holes~\cite{cardy1986operator}. Indeed, this formula was also employed by Zaletel and Mong~\cite{zaletel2012exact} for a more heuristic argument to get an estimate on the bond dimension in the context of quantum Hall physics.  Unfortunately, though, only partial results are known that connect  $\dimctwo$ to  the central charge~$c$. Gaberdiel and Neitzke~\cite{GaberdielNeitzke} have shown that it is an upper bound on the effective central charge, 
\begin{align}
  c - 24 \min (\wt b \,:\,b \in B) \leq \frac{\dimctwo}{2} \,. 
\end{align}
However, we do not know how tight this bound is, nor can we explain appearance of the factor of one half. 

Let us also mention that intuitively, we could also expect that the scaling of the bond dimension is related to the well-known entropy formula by Holzhey et.~al. and Calabrese and Cardy~\cite{Holzhey1994443,CalabreseCardy04}. This states that the von Neumann entropy of the reduced density matrix of a line of lenght \(L\) of a full conformal field theory should scale as \(c/3 \log L\), with \(c\) the central charge of the theory. However, as noted in~\cite{Schuchetalentropymps2008}, a bound on the von Neumann entropy is usually not sufficient to get a rigorous estimate on the required bond dimension. Moreover, here our aim is not to approximate the path integral description of the reduced density matrix, but rather to get approximate expressions for correlation functions. 

\begin{proof}
 Let us first consider the bond dimension~$D$ of the MPS. Eq.~\eqref{eq:dimweightspaces} implies the bound
\begin{align}
  D = d_{B_j}(M+ nN) \leq (\dim B_0)\cdot \sum_{1\leq l \leq M + nN} P(l,\dimctwo) \,,
\end{align}
where $\dim B_0$ is  the dimension of the top level of the module~$B$. Inserting the (pretty rough) upper bound on the multi-partition function $P(n,\dimctwo)$ obtained in Lemma~\ref{lem:app:boundpartitionfunction} in the appendix, we can further bound this by
\begin{align}
  D \leq (\dim B_0) (M+nN) e^{2\pi \sqrt{\frac{\dimctwo}{6} (M+nN)}}\, .\label{eq:dupperboundsubex}
\end{align}
The bound~\eqref{eq:dupperboundsubex}  implies a sub-exponential growth of the bond dimension with the truncation parameters $M,N$. However, the approximation errors~$\eps_{MPS_0},\eps_{MPS_1}$  (cf.~\eqref{eq:mpserror},~\eqref{eq:mpsoneerror}) scale  exponentially in these parameters.
This amounts to saying that the bond dimension scales at most polynomially within the approximation error.

More precisely, recall from~\eqref{eq:mpserror} that $\eps_{MPS_0}$ is proportional to $q^{N/4}=e^{-\mathsf{d}N/4}$, and thus decays exponentially
with a rate determined by the minimal distance~$d$ between the insertions, see Fig.~\ref{fig:plane}~(b). For the genus-$1$ correlation function, we choose the cutoff level to be $M=N$, and the regularization parameter $r=q^{1/2}=e^{-\mathsf{d}/2}$.  Then both in the genus-$0$ and genus-$1$-case, we have 
\begin{align}
  \eps_{MPS} \leq (\textrm{constant})\cdot e^{-\mathsf{d} N/4}\, .\label{eq:truncationpfixedeps}
\end{align}
In particular, combining this with~\eqref{eq:dupperboundsubex}, we find that for a fixed approximation guarantee~$\eps$, a bond dimension of  roughly
\begin{align}
  D \lessapprox (\textrm{constant})\, \frac{n+1}{\mathsf{d}} \,\log\left(\frac{1}{\eps}\right) e^{2\pi \sqrt{\dimctwo \frac{2(n+1)}{3\mathsf{d}}  \log\left(\frac{1}{\eps}\right)}} \,\lessapprox\, (\textrm{constant})\,\left(\frac{1}{\eps}\right)^{\dimctwo\,\frac{4 \pi  (n+1)}{3\mathsf{d}}} \,.
\end{align}
is sufficient (for a fixed number of insertions~$n$). This is the claim~\eqref{it:firstclaimbonddim}. Similarly, for a fixed approximation accuracy~$\epsilon>0$, Eq.~\eqref{eq:truncationpfixedeps} fixes~$N$ and thus the bond dimension~$D$ by~\eqref{eq:dupperboundsubex}. The claim~\eqref{it:secondclaimbonddim} then follows (neglecting logarithmic terms).
\end{proof}

\subsection{Intertwiners and $G$-invariant MPS for WZW models\label{sec:intertwinginvariant}}
We end this section with a discussion of the MPS approximation in the case of our main example, the WZW models. We are interested in the case where~$S$
is the subspace of top-level vectors in an irreducible module. The following statement characterizes the degrees of freedom that 
scaled intertwiners have in this case. It also applies to truncated scaled intertwiners, and thus characterizes MPS for WZW-models. 

\begin{proposition}\label{prop:wzwintertwiner}
Consider a WZW model with internal symmetry group $G$ (a compact Lie group), specified by the Lie algebra~$\g$. For $j=1,2,3$, fix integral dominant weights~$\lambda_j$ such that~$\lambda_j(\theta) \leq k$.
 Let $\modulekflambda{\lambda_j}$ be the corresponding $\voalkzero$-module, and recall  (cf.~Section~\ref{sec:introVOAmodules}) that the top-level~$V_j=\modulekflambda{\lambda_j}(0)$ is an irreducible $\g$-module. Let  $U_j:G\rightarrow\mathsf{GL}(V_j)$  be the associated unitary representation of~$G$. 
Let   $S=\modulekflambda{\lambda_3}(0)$ be the top level in the module $\modulekflambda{\lambda_3}$. Then there is a one-to-one correspondence between
\begin{enumerate}[(i)]
\item
$G$-intertwining maps $\cW:V_3\otimes V_2\rightarrow V_1$
\item
$S$-bounded intertwiners $\W(\cdot,z):S\rightarrow\mathsf{End}(\modulekflambda{\lambda_2},\modulekflambda{\lambda_1})\{\{z,z^{-1}\}\}$.
\end{enumerate}
More precisely, the correspondence is given by the identity
\begin{align}
\cW(\varphi_3\otimes\varphi_2)=z^{\tau}q^{-\frac{1}{2}\sum_{j=1}^3 h_{\lambda_j}}\cdot P_1^{[0]}\W_q(\varphi_3,z)P_2^{[0]}\varphi_2\label{eq:correspondencerestriction}
\end{align}
for  $\varphi_2\in V_2=\modulekflambda{\lambda_2}(0)$ and $\varphi_3\in S=V_3=\modulekflambda{\lambda_3}(0)$, where $\tau=h_{\lambda_3}+h_{\lambda_2}-h_{\lambda_1}$. 
\end{proposition}
Observe that the expression on the rhs.~of Eq.~\eqref{eq:correspondencerestriction} can be written in terms of the truncated scaled intertwiner~$\W_q^{[N]}$ with truncation parameter $N=0$ if the vectors belong to the top levels, i.e.,   
\begin{align}
P_1^{[0]}\W_q(\varphi_3,z)P_2^{[0]}\varphi_2=\W_q^{[0]}(\varphi_3,z)\varphi_2 \quad \textrm{ for }\varphi_j\in \modulekflambda{\lambda_j}(0)\ . \label{eq:ponezerowqvarphi}
\end{align}   In other words, a scaled intertwiner~$\W_q^{[N]}$ truncated at $N=0$ determines a $G$-intertwining map and vice versa. As discussed in Section~\ref{sec:ginvariantmpsdescription}, the latter describe $G$-invariant MPS. That is, part of the approximation procedure for WZW models involves the construction of a group covariant MPS. Conversely, for every such MPS --- or rather for every group-covariant isometry associated with it --- we can uniquely construct an intertwiner for a WZW model. In Appendix~\ref{app:algorithmwzw}, we describe a constructive algorithm for this purpose.

We remark that in fact, it is possible to extend Proposition~\ref{prop:wzwintertwiner}: a $G$-intertwining map~$\cW:V_3\otimes V_2\rightarrow V_1$ actually determines the whole intertwiner $\iY$ (defined beyond the subspace~$S$) of type $\itype{\modulekflambda{\lambda_1}}{\modulekflambda{\lambda_3}}{\modulekflambda{\lambda_2}}$. This follows by extending the algorithm in Appendix~\ref{app:algorithmwzw} along the lines of~\cite{Huang:1992kw}. Here we restrict to $S$-bounded scaled intertwiners for concreteness and since this is our main object of interest.

The proof of Proposition~\ref{prop:wzwintertwiner} involves standard arguments: the constructive algorithm is an explicit application of Zhu's theory~\cite{Zhu:1996cp,Frenkel:1992jt}, and can also be seen as the generalization of Huang's arguments for minimal models~\cite{Huang:1992kw}. It is based on the observation that the matrix elements of the zero-mode fixes all other matrix elements by the local symmetry properties as defined by the affine Lie algebra~$\ga$. The argument is presented in Appendix~\ref{app:algorithmwzw}. 
For completeness, here we argue one direction of Proposition~\ref{prop:wzwintertwiner}: every scaled intertwiner~$\W$ defines a $G$-intertwining map~$\cW$ by restriction 
of the zero mode to the top levels.

\begin{proof}
According to Proposition~\ref{prop:Sboundedfinitedimensional}, an intertwiner $\iY$ of type $\itype{\modulekflambda{\lambda_1}}{\modulekflambda{\lambda_3}}{\modulekflambda{\lambda_2}}$  is $S$-bounded, and we obtain a scaled $S$-bounded intertwiner $\W(\cdot,z)$.  It will be convenient to work with the intertwiner $\iY$ instead of its scaled version. Since scaling amounts to the introduction of additional factors when considering homogeneous vectors (which is sufficient for this purpose), we can relate the rhs.~of Eq.~\eqref{eq:correspondencerestriction} to the zero-mode~$\iy(\cdot)_{\tau,0}$ of $\iY(a,z)=\sum_{m\in\mathbb{Z}}\iy(a)_{\tau,0}z^{-m-\tau-n}$: we have 
with~\eqref{eq:ponezerowqvarphi}
\begin{align*}
P_1^{[0]} \W_q(\vphi_3,z) P_2^{[0]}\vphi_2 &=P_1^{[0]} \W^{[0]}(P_3^{[0]} \vphi_3,z) P_2^{[0]}\vphi_2\\
&= q^{(h_{\lambda_1} + h_{\lambda_2}+h_{\lambda_3})/2} z^{-\tau} \iy(\vphi_3)_{\tau,0} \vphi_2\, .
\end{align*}
Here we used that elements of the top level~$\modulekflambda{\lambda_j}(0)$ have weight $h_{\lambda_j}$, and $\iy(\varphi_3)_{\tau,0}\varphi_2$ also belongs to the top level for $\varphi_2\in \modulekflambda{\lambda_2}(0)$. 
 Statement~\eqref{eq:correspondencerestriction} then expresses the fact that
a $G$-intertwiner~$\cW:V_3\otimes V_2\rightarrow V_1$ is uniquely determined by
\begin{align}
P_1^{[0]}\iy(P_3^{[0]}\cdot)_{\tau,0}P_2^{[0]}:
\modulekflambda{\lambda_3}(0)\rightarrow \mathsf{End}(\modulekflambda{\lambda_2}(0),\modulekflambda{\lambda_1}(0))\ ,\label{eq:restrictiontozeromodeunscaled}
\end{align}
and vice versa. We call~\eqref{eq:restrictiontozeromodeunscaled} the {\em restriction of the zero-mode of $\iY$ to the top levels}. In this language, Eq.~\eqref{eq:correspondencerestriction} takes the form
\begin{align}
\cW(\varphi_3\otimes\varphi_2)&=P_1^{[0]}\iy(P_3^{[0]}\varphi_3)_{\tau,0}P_2^{[0]}\varphi_2=\iy(\varphi_3)_{\tau,0}\varphi_2\ ,\label{eq:correspondencerestrictionnew}
\end{align}
for  vectors~$\varphi_j\in \modulekflambda{\lambda_j}(0)$, $j=2,3$ belonging to the top level. This is because the projections~$P^{[0]}_j$, $j=2,3$  leave these vectors invariant and application of the zero-mode does not change the weight, resulting in a vector again belonging to the top level.

Given an intertwiner $\iY(a,z)=\sum_{m\in\mathbb{Z}}\iy(a)_{\tau,0}z^{-m-\tau-n}$ of type $\itype{\modulekflambda{\lambda_1}}{\modulekflambda{\lambda_3}}{\modulekflambda{\lambda_2}}$, define a linear map~$\cW:V_3\otimes V_2\rightarrow V_1$  by~\eqref{eq:correspondencerestrictionnew}. We claim that $\cW$ is a $G$-intertwining map. (The converse direction will be shown in Appendix~\ref{app:algorithmwzw}.)

It follows from Eq.~\eqref{eq:wzwintertwinerprop} that for $\fa \in \g $ we have
\begin{align}
  \cW(\fa\vphi_3 \otimes \vphi_2) = \fa \cW(\vphi_3 \otimes \vphi_2) -  \cW(\vphi_3 \otimes \fa\vphi_2) \,,\qquad\textrm{ for all }\varphi_j\in V_j=\modulekflambda{\lambda_j}(0)\ ,
\end{align}
where we denote the action of $\fa \in \g\subset\ga $ on the irreducible module $\modulekflambda{\lambda_i}$, $i=1,2,3$ by the same letter as the corresponding Lie algebra element (Note that these operators also do not change the weight). Exponentiating, we conclude that the operator~$\cW$ satisfies
\begin{align}
  \cW\, (U_3(g) \otimes U_2(g)) = U_1(g) \,\cW\,, \qquad\textrm{ for all }g\in G\ ,
\end{align}
where $U_i\,:\,G \to \mathsf{GL}(V_i)$, $i=1,2,3$, are the irreducible unitary representations of the compact Lie group $G$. Hence the operator $\cW$ intertwines irreducible group representations of $G$, as claimed. 
\end{proof}

\section{Finitely correlated states for full CFTs\label{sec:fullcft}}

Using VOAs, we have established a faithful MPS representation of chiral CFTs. These describe fields depending holomorphically on the coordinates; antiholomorphic fields are not incorporated into this framework. In the following, we describe how to extend these results to a full CFT containing both holomorphic and antiholomorphic fields (also referred to left- and right-movers).

\subsection{Background: VOAs and conformal full field algebras\label{sec:voafullfieldalgebras}}

We again restrict our attention to the complex plane as well as the torus, where a rigorous construction of CFTs has been given by  Huang and Kong~\cite{HuangKong,Huang:2010fv} in terms an algebraic object called a {\em conformal full field algebra}. In contrast, the algebraic construction of CFTs on higher-genus surfaces remains a topic of ongoing research, but see the work of Fuchs, Runkel and Schweigert~\cite{fuchs2005tft}.

We only briefly sketch the pertinent ingredients in this construction and refer
the interested reader to~\cite{HuangKong,Huang:2010fv}. Following the seminal work of Belavin, Polyakov and Zamolodchikov~\cite{BPZ84}, one starts with a
symmetry algebra being a product of two chiral factors (associated with left- and right-movers) and studies its representations. The factors
each contain a copy of the Virasoro algebra, whose generators $\{L^L_n\}_n$ respectively $\{L^R_n\}_n$ are the Laurent modes
of a chiral respectively antichiral field (i.e., the energy-momentum tensor). The whole set of fields consists of conformal families that are obtained as descendants of some number of primary fields. As discussed in the introduction
 (see Section~\ref{sec:cftintro}), correlation functions  (or `insertions') depend on two a priori independent variables~$z,\bar{z}\in\mathbb{C}$. In the statistical mechanics (Euclidean) case we will choose~$\bar{z}$ to be the complex conjugate of~$z$, whereas for relativistic theories,  both parameters stay independent, but real and imaginary part are equal to the lightcone variables before compactifying the space-time.

\subsubsection{Conformal full field algebras}
In the terminology of VOAs, a CFT on the plane (or more precisely, a conformal full field algebra as defined in~\cite{HuangKong}) can be constructed starting from two VOAs~$\cV^L=(\cV^L,\Y^L,\1^L,\omega^L)$ and $\cV^R=(\cV^R,\Y^R,\1^R,\omega^R)$ satisfying our main technical assumptions: they are rational and $C_2$-cofinite, and we assume that they are unitary. Given two such VOAs, one obtains a VOA~$\cV^L\otimes \cV^R$ by setting
\begin{align*}
\1=\1_{\cV^L}\otimes\1_{\cV^R}\ ,\qquad \omega=\omega_{\cV^L}\otimes\1_{\cV^R}+\1_{\cV^L}\otimes\omega_{\cV^R}
\end{align*}
and
\begin{align*}
\Y(v^L\otimes v^R,z)&=\Y^L(v^L,z)\otimes\Y^R(v^R,z)\ .
\end{align*}
Observe also that if $A^L$ and $B^R$ are modules of $\cV^L$ and $\cV^R$, respectively, then $A^L\otimes B^L$ defines a module for the VOA~$\cV^L\otimes\cV^R$.

\subsubsection{The full CFT space}
A CFT based on~$\cV^L\otimes\cV^R$ is  determined by
a module~$\cH$ of the VOA~$\cV^L\otimes\cV^R$ and an intertwiner~$\iY$ of type~$\itype{\cH}{\cH}{\cH}$. Let us first discuss properties of the module~$\cH$. Under the given technical assumptions,~$\cH$ decomposes (see~\cite[Corollary 2.2]{HuangKong}) as
\begin{align}
\cH\cong \bigoplus_{j}(A[j]\otimes B[j])\ \label{eq:fullcfthilbertspace}
\end{align}
into tensor products of irreducible modules $A[j]$ of $\cV^L$ and $B[j]$ of $\cV^R$ (in contrast to~\cite{HuangKong}, we include multiplicities in this sum). Importantly, this is a finite sum for the VOAs of interest, since rational VOAs only have finitely many non-isomorphic irreducible modules~\cite{Frenkel:1993fv}. As for the VOAs $\cV^L$ and $\cV^R$, we assume that the modules~$A[j]$ of~$\cV^L$ and~$B[j]$ of~$\cV^R$ are unitary modules. Hence they are equipped with a positive-definite sesquilinear form, as well as with an anti-unitary involution. Upon completion with respect to the positive-definite form, they turn into Hilbert spaces with inner product defined by the extension of the positive definite form to the completion. As in the chiral case, in the following we do not differentiate between the module $A$ and its Hilbert space completion.

 We denote the images of 
  $\bigoplus_j (L_{A[j],n}\otimes \mathsf{id}_{B[j]})$ and $\bigoplus_j (\mathsf{id}_{A[j]}\otimes L_{B[j],n})$
under the isomorphism~\eqref{eq:fullcfthilbertspace}
by~$L^L_n:\cH\rightarrow\cH $ and $L^R_n:\cH\rightarrow\cH$, respectively.
The operators $L^L_0$ and $L^R_0$ give an~$\mathbb{N}_0\times\mathbb{N}_0$-grading of the Hilbert space
\begin{align}
\cH=\bigoplus_{(n^L,n^R)\in \mathbb{N}_0\times \mathbb{N}_0}\cH_{(n^L,n^R)}\qquad\textrm{ where }\qquad 
\cH_{(n^L,n^R)}=\bigoplus_{j}\cH_{(n^L,n^R),j}\ . \label{eq:hgradingfullcft}
\end{align}
A vector $\psi\in \cH_{(n^L,n^R),j}$ satisfies 
\begin{align*}
L^L_0 \psi =(h_{A[j]}+n^L)\psi\qquad L^R_0 \psi =(h_{B[j]}+n^R)\psi\ ,
\end{align*}
where $h_{A[j]}$ and $h_{B[j]}$ denote the highest weights of $A[j]$ and $B[j]$, respectively. The spaces $\cH_{(n^L,n^R),j}$ are given by
\begin{align}
\cH_{(n^L,n^R),j}=A[j]_{n^L}\otimes B[j]_{n^R}
\end{align}
   We will set
\begin{align}
\mathbb{L}_n\equiv L^L_n+L^R_n\qquad\textrm{ for } n\in\mathbb{N}_0\ .\label{eq:fullcftlmodeoperators}
\end{align}
As before, a vector~$\psi\in\cH$  is called homogeneous if it is an eigenvector of the operator~$\mathbb{L}_0$. The eigenvalue is called the weight of $\psi$ and denoted~$\wt \psi$, i.e.,  
\begin{align*}
\mathbb{L}_0\psi &= (\wt \psi)\psi\ .
\end{align*}
It follows that  weights are of the form $h_{A[j]}+h_{B[j]}+n$ for some $j$ and some integer $n\in\mathbb{N}_0$. Finally, we mention that 
primary vectors~$\psi\in\cH$ are defined (as before) by the property that 
\begin{align*}
\mathbb{L}_n \psi=0\qquad\textrm{ for all }n>0\ . 
\end{align*}

\paragraph{Diagonal CFTs and uniqueness of vacuum.}
In fact, the decomposition~\eqref{eq:fullcfthilbertspace} of the full CFT module~$\cH$ can be refined. Denoting by $\{A[\alpha]\}_{\alpha\in\widehat{\cV}^L}$ the set of equivalence classes of irreducible modules of the VOA~$\cV^L$,
and by $\{B[\beta]\}_{\beta\in\widehat{\cV}^R}$ those of the VOA~$\cV^R$,  we  have
\begin{align}
\cH\cong \bigoplus_{\alpha\in\widehat{\cV^L},\beta\in\widehat{\cV^R}} m_{\alpha,\beta} A[\alpha]\otimes B[\beta]\ \label{eq:multiplicitydecompositiondirectsum}
\end{align}
for some  integer multiplicities $m_{\alpha,\beta}\in\mathbb{N}_0$.  Physical considerations~\cite{Moore:1989cm} then further demand that there exists an injective map $\sigma:\widehat{\cV^L}\rightarrow\widehat{\cV^R}$  such that
only pairs of the form~$(\alpha,\beta)=(\alpha,\sigma(\alpha))$ appear in~\eqref{eq:multiplicitydecompositiondirectsum}, that is, we have
\begin{align}\label{eq:directsumdecompfullcft}
  \cH \cong  \bigoplus_{\alpha\in\widehat{\cV^L}} m_{\alpha} A[\alpha]\otimes B[\sigma(\alpha)]\
\end{align}
for some integers~$m_\alpha$. In the following, we will focus (following~\cite[Section 3]{HuangKong}) on CFTs of this form, where $\cV^L=\cV^R=\cV$, and $B[\sigma(\alpha)]=A[\alpha]^\prime$ is the contragredient module of the module~$A[\alpha]$. The decomposition~\eqref{eq:fullcfthilbertspace} then takes the form
\begin{align}\label{eq:fullcftdecomp}
\cH\cong \bigoplus_{j}(A[j]\otimes A[j]^\prime)\,,
\end{align}
where the index $j$ now runs over the set of irreducible modules as well as counting multiplicities. 
 Such CFTs are often  referred to as diagonal CFTs, and they have the feature that the corresponding partition function on the torus is modular invariant~\cite{Huang:2010fv}.
 
Finally, we also add the following assumption on our full CFT: We require that the direct sum decomposition of Eq.~\eqref{eq:directsumdecompfullcft} contains one and only one summand equal to the tensor product of the adjoint module $\cV$ with its corresponding contragredient module $\cV^\prime$, and all other appearing modules are not equal to the adjoint module $\cV$ or $\cV^\prime$. This amounts to saying that there is a unique vacuum state~$\1\otimes \1'\in\cH$. This vector plays a particular role, see~\eqref{eq:vacuumfullcftidentity}. 
  
\subsubsection{Full CFT intertwiner}
Given the intertwiner $\iY$ of ~$\itype{\cH}{\cH}{\cH}$, where $\cH$ is a module of $\cV^L\otimes\cV^R$, the  full (CFT-)intertwiner is given by a formal Laurent-like series in two indeterminates~$z$ and $\bar{z}$. It is given by the expression 
\begin{align}
\cftY(\psi,(z,\bar{z}))&=z^{L^{L}_0}\bar{z}^{L^{R}_0}\cY(\psi,1)z^{-L^{L}_0}\bar{z}^{-L^{R}_0}\qquad\textrm{ for }\psi\in\cH\ .\label{eq:fullcftintertwiner}
\end{align}
This object obeys similar properties as usual intertwiners. Most importantly for our purposes, the following analog of~\eqref{eq:dilationintegrated} holds for dilations:
\begin{align}
q^{\mathbb{L}_0}\mathbb{Y}(\psi,(z,\overline{z}))q^{-\mathbb{L}_0}&=\mathbb{Y}(q^{\mathbb{L}_0}\psi,(qz,q\overline{z}))\ .\label{eq:dilationequationfullcft}
\end{align}
In addition, the full intertwiner~$\cftY$ has the property that on the vacuum vector~$\1\otimes \1'$, it evaluates to the identity (irrespective of the formal parameters $(z,\bar{z})$), that is (see~\cite{HuangKong})
  \begin{align}
    \cftY(\1 \otimes \1',(z,\bar{z})) = \idty_\cH \, .\label{eq:vacuumfullcftidentity}
  \end{align}
One implication of~\eqref{eq:vacuumfullcftidentity} is that insertions of the vector~$\1\otimes\1'$ in a correlation function (see Sec.~\ref{sec:fullcftcorr}) do not affect the value of the correlation function and can be removed. 

The full CFT intertwiner~$\cftY$ can be decomposed into constituent intertwiners, as follows.  Denoting the isomorphism~\eqref{eq:fullcfthilbertspace} by $\gamma:\bigoplus_j (A[j]\otimes B[j])\rightarrow\cH$ and omitting superscripts $\{L,R\}$ for notational convenience, it can be shown (see~\cite[Proposition~2.3]{HuangKong}) that the intertwiner~\eqref{eq:fullcftintertwiner} has the form
\begin{align*}
\cftY(\gamma(\tilde{a}_j\otimes\tilde{b}_j),(z,\bar{z}))\gamma(a_k\otimes b_k)&=\sum_\ell \gamma\left(\itw{A[\ell]}{A[j]}{A[k]}(\tilde{a}_j,z)a_k\otimes \itw{B[\ell]}{B[j]}{B[k]}(\tilde{b}_j,\bar{z})b_k\right)\ ,
\end{align*}
for $\tilde{a}_j\in A[j]$, $\tilde{b}_j\in B[j]$  and
$a_k\in A[k]$, $b_k\in B[k]$, for some intertwiners $\itw{A[\ell]}{A[j]}{A[k]}$ of type $\itype{A[\ell]}{A[j]}{A[k]}$,
and intertwiners $\itw{B[\ell]}{B[j]}{B[k]}$ of type $\itype{B[\ell]}{B[j]}{B[k]}$.
In the following, we also omit explicity writing the isomorphism, i.e., we write this decomposition as
\begin{align}
\cftY(\tilde{a}_j\otimes\tilde{b}_j,(z,\bar{z}))a_k\otimes b_k&=\sum_\ell\itw{A[\ell]}{A[j]}{A[k]}(\tilde{a}_j,z)a_k\otimes \itw{B[\ell]}{B[j]}{B[k]}(\tilde{b}_j,\bar{z})b_k\ .\label{eq:generalintertwinerdecomposed}
\end{align} 
In the special case of a diagonal CFT,  intertwiners have the form
\begin{align}\label{eq:decompfullcftintdiagoal}
  \cftY(\cdot,(z,\bar{z})): \cH  &\to \mathsf{End}(\cH)\{\{z,z^{-1},\bar{z},\bar{z}^{-1}\}\} 
  \end{align}
  where 
  \begin{align}  
  \cftY(a_j\otimes a'_j,(z,\bar{z}))b_k\otimes b'_k &=\sum_\ell\itw{A[\ell]}{A[j]}{A[k]}(a_j,z)b_k\otimes \itw{A[\ell]'}{A[j]'}{A[k]'}(a'_j,\bar{z})b'_k\  \label{eq:fullintertwinercft}
\end{align}
for $a_j\in A[j]$, $a'_j\in A[j]'$, and $b_k\in A[k], b_k'\in A[k]'$. We call the VOA intertwiners appearing in this decomposition the \emph{constituents} of the CFT intertwiner $\cftY$.

\subsubsection{Correlation functions of full CFTs\label{sec:fullcftcorr}}

As in the chiral case, correlation functions of full CFTs are obtained by evaluating matrix elements of formal Laurent series with coefficients in the endomorphisms on $\cH$, and then setting the indeterminates to values given by complex numbers. The question of convergence has to be addressed again, but due to the decomposition~\eqref{eq:decompfullcftintdiagoal}, this can be reduced to convergence properties of the chiral and anti-chiral part (see~\cite{HuangKong} as well as~\cite{Huang:2010fv} for a detailed discussion). As in the chiral case, we will use the same letters for the indeterminates and their complex valued counterparts, with the understanding that whenever matrix elements or operators on Hilbert spaces are considered, then complex numbers are used. 

Correlation functions of the full CFT are specified by a sequence of the full intertwiner, evaluated at different values of complex numbers $(z_1,\bar{z}_1),\dots,(z_n,\bar{z}_n) $ and at different elements $\psi_1, \dots, \psi_n \in \cH$. As for the chiral case, we consider  genus-0 or genus-1 correlation functions. For the genus-0 case, the vacuum-to-vacuum correlation function is
  \begin{align}
    F^{(0)}_{\cH}((\psi_1,z_1,\bar{z}_1),\ldots,(\psi_n,z_n,\bar{z}_n)) = \Scp{\1\otimes\1^\prime}{\cftY(\psi_1,z_1,\bar{z}_1) \cdots \cftY(\psi_n,z_n,\bar{z}_n) \1\otimes\1^\prime}\,.
  \end{align}
  Correspondingly, the genus-1 case is given as the trace of such a sequence of intertwiners, again scaled by a factor $\torusq$ using the dilation operator defined by~\eqref{eq:fullcftlmodeoperators}, i.e., 
  \begin{align}
  F^{(1)}_{\torusq,\cH}((\psi_1,z_1,\bar{z}_1),\ldots,(\psi_n,z_n,\bar{z}_n))&= \tr_{\cH}\left(\cftY(z_1^{L^L_0}\bar{z}_1^{L^R_0}\psi_1,z_1,\bar{z}_1) \cdots \cftY(z_n^{L^L_0}\bar{z}_n^{L^R_0}\psi_n,z_n,\bar{z}_n) \torusq^{\mathbb{L}_0- 2c/24}  \right)\,.
\end{align}
We again assume that $\psi_1, \ldots \psi_n \in \cS$, where $\cS \subset \cH$ is a finite-dimensional subspace of the full module $\cH$, consisting of primary vectors.

As in the chiral case, we assume that our insertion points are separated by a minimal distance. More precisely, we consider the setup  illustrated in Fig.~\ref{fig:plane} and~\ref{fig:torusvm}, respectively. That is, the variables $\zeta_i$ lie on a line of fixed imaginary part~$\theta$ --- either in the plane or on the periodic strip --- and are separated by a minimal distance~$\mathsf{d}$. As explained before, the variable $\bar{z}$ is set to be equal to the complex conjugate of~$z$. Applying again either the conformal mapping $z \mapsto e^{-z}$ (in the case of the plane) or the principal branch of the complex logarithm (for the periodic strip) to both the variables $\zeta$ and $\bar{\zeta}$ leads to the following analogue of Observations~\ref{obs:conformalmapcorrfuncplane} and~\ref{obs:conformalmapcorrfunctorus}.

\begin{observation}\label{obs:fullcft}
  Let $\cS \subset \cH$ be a finite-dimensional linear subspace of primary vectors of the Hilbert space of the full CFT, and $\psi_1,\ldots,\psi_n \in \cS$ homogeneous elements. Then the vacuum-to-vacuum correlation functions of $n$ equally spaced points $\zeta_1,\ldots,\zeta_n$ on a line  with offset $\mathsf{d}_0$ and minimal distance $\mathsf{d}$ can be expressed as
  \begin{align}
    &F^{(0)}_{\cH}((\psi_1,\zeta_1,\zeta_1^*),\ldots,(\psi_n,\zeta_n,\zeta_n^*)) = \\
    &\qquad\qquad=|z|^{2\sum_j \wt\psi_j} q^{\sum_j j \wt\psi_j}\,\Scp{\1 \otimes \1^\prime}{\cftY(\psi_1,\zetap_1,\bar{\zetap}_1) \cdots \cftY(\psi_n,\zetap_n,\bar{\zetap}_n)\1 \otimes \1^\prime}
  \end{align}
  with the identification $0< q = e^{-\mathsf{d}} < 1$, $z = e^{-\mathsf{d}_0} e^{i \theta }$, $\zetap_j = zq^j$, $\bar{\zetap}_j = (\zetap_j)^*$. Similarily, the correlation function of $n$ equispaced points $\zeta_1,\ldots,\zeta_n$ within the periodic strip of height $2\pi$ and length $\log(1/\torusq)$ for these elements $\psi_1,\ldots,\psi_n$ is given by
  \begin{align}
    &F^{(1)}_{\torusq,\cH}((\psi_1,zq,\bar{z}q),\ldots,(\psi_n,zq^n,\bar{z}q^n)) = \\
    &\qquad\qquad=|z|^{2\sum_j \wt\psi_j} q^{\sum_j j \wt\psi_j}\,\tr_{\cH}\left(\cftY(\psi_1,\zetap_1,\bar{\zetap}_1)\cdots \cftY(\psi_n,\zetap_n,\bar{\zetap}_n)\torusq^{\mathbb{L}_0-c/12}\right)
  \end{align}
  under the identifications $z=e^{\mathsf{d}_0 + i \theta}$, $\zetap_j = zq^j$, $\bar{\zetap}_j = (\zetap_j)^*$ and the assumption $|z|>1$, $q=e^{-\mathsf{d}} < 1$.
\end{observation}
The convergence of the formal Laurent series with indeterminates replaced by complex numbers as above follows from the results on chiral VOAs, see~\cite{HuangKong} as well as~\cite{Huang:2010fv}. 

As in the chiral case,  the transformed correlation functions motivate the definition of scaled intertwiners. Correlations functions can then again be expressed as sequences of these scaled intertwiners, which are combined to the definition of appropriate transfer operators. Our analysis then continues as before, by first approximating the scaled intertwiner and the corresponding transfer operator, and finally truncating the Hilbert space.

\subsection{Matrix product tensor networks for full CFTs}
We proceed to construct a matrix product tensor network approximating the correlation functions of a full CFT. In Section~\ref{sec:fullcftsasfinitely} we then show how this tensor network can be interpreted as a finitely correlated state. Once again, the expression ``state'' should not be taken literally, as the normalization depends on the correlation function. 

The definitions and results for the chiral case can be generalized with minimal adaptations to full CFTs.  A main difference is that all expressions depend both on~$(z,\bar{z})$, and truncations have to be defined with respect to the $\mathbb{N}_0\times\mathbb{N}_0$-grading of~$\cH$. Nevertheless, most proofs directly generalize to this setting. For conciseness, we only present
the arguments that deviate from the chiral case and omit derivations that are completely analogous. 

In analogy to the definitions for the chiral case, we will first introduce the concept of a scaled intertwiner, as well as the corresponding transfer operator. As before, we will usually assume that the subspace~$\cS$ in the following definitions is spanned by primary vectors and contains the vacuum vector $\1 \otimes \1^\prime$. Moreover, we assume that $\cS$ has the form
\begin{align}
  \cS \cong \bigoplus_j S[j] \otimes S[j]^\prime\ ,\label{eq:Sdirectsumdecomp}
\end{align}
where $S[j]\subset A[j]$ is a finite-dimensional subspace of primary vectors,
and $S[j]'\subset A[j]'$ is the image of $S[j]$ under the isomorphism~$\tilde{\eta}:A[j]\rightarrow A[j]'$ introduced in~\eqref{eq:isodualcomplconj}. Using the definitions of Section~\ref{sec:unitarity}, it is straightforward to check that $S[j]'$ also consists of primary vectors.
  
\subsubsection{Scaled intertwiners and correlation functions}

The scaled intertwiner as well as the associated transfer operator are introduced as formal Laurent series, which turn into well-defined operators on $\cH$ once the formal indeterminates are replaced by complex numbers.
   
\begin{definition}\label{def:cftscaledintertwiner}
  Consider a full CFT of diagonal type.  For $0<q<1$ and a linear subspace~$\cS\subset\cH$ of the form~\eqref{eq:Sdirectsumdecomp}, define the scaled intertwiner~$\cftW_q$ as  the linear map from $\cS$ to formal Laurent-like series with coefficients in the endomorphism of $\cH$,
   \begin{align*}
      \begin{matrix}
        \cftW_q(\cdot,(z,\bar{z})):  &\cS& \rightarrow &\Endint{\cH}{z,z^{-1},\bar{z},\bar{z}^{-1}}\\
        &\psi& \mapsto  &   \cftW_q(\psi, (z,\bar{z}))  = q^{\mathbb{L}_0/2} \mathbb{Y}(q^{\mathbb{L}_0/2}\psi,(z,\bar{z})) q^{\mathbb{L}_0/2}\ .
      \end{matrix}
   \end{align*}
\end{definition}

Using this scaled intertwiner, we can then proceed to introduce transfer operators as in Definition~\ref{def:transferoperators}. In contrast to Definition~\ref{def:transferoperators}, this is based only on one type of intertwiner, which is the full CFT intertwiner $\mathbb{Y}$. 

\begin{definition}[Transfer operators for the full CFT]\label{def:cfttransferoperators}
Let $\cS\subset\cH$ be a subspace spanned by primary vectors. Let $0<q<1$ and let $\cftW_q$ be the scaled intertwiner of a CFT on a Hilbert space~$\cH$. For any $\psi_i\in \cS, i=1,\ldots,n$, define the {\em transfer operator} $\mathbb{T}$  with insertions $\{\psi_i\}_{i=1}^n$ as the element 
\begin{align}
  \mathbb{T}\in \Endint{\cH}{z_1,z_1^{-1},\bar{z}_1,\bar{z}_1^{-1},\ldots,z_n,z_n^{-1},\bar{z}_n,\bar{z}^{-1}_n}
\end{align}
given by 
\begin{align*}
  \mathbb{T} =\cftW_q(\psi_1,(z_1,\bar{z}_1))\circ\cdots\circ\cftW_q(\psi_n,(z_n,\bar{z}_n))\ .
\end{align*}
\end{definition}

Then, as before, the genus-$0$ and genus-$1$ correlation functions are encoded in the transfer operator in the following sense (compare to Corollary~\ref{cor:corrfctexactzero} and Corollary~\ref{cor:corrfctexactrepr}). 

\begin{corollary}[Exact reproduction of CFT correlation functions]
  Let $z,\bar{z}\in\mathbb{C}\backslash\{0\}$ two complex numbers and $0<q<1$ such that 
 $0<q<\min \{1/|z|^2, 1/|\bar{z}|^2\}$. Let $\mathbb{T}$ be the full CFT transfer operator with insertions $\{\psi_i\}_{i=1}^n$ as in Definition~\ref{def:cfttransferoperators}, with the indeterminates replaced by $z$ and $\bar{z}$ respectively. Then we have for the genus-0 case
  \begin{align}\label{eq:corrfunctransferopfullcftgenus0}
  \langle \1\otimes \1', \mathbb{T}\1\otimes \1'\rangle = q^{\sum_j j \wt \psi_j} F^{(0)}_{\cH}((\tilde{\psi}_1,\zetap_1,\bar{\zetap}_1),\ldots,(\tilde{\psi}_n,\zetap_n,\bar{\zetap}_n)) \,,
\end{align}
and for the genus-1 case with $0<r\leq 1$
\begin{align}\label{eq:corrfunctransferopfullcftgenus1}
  \tr_\cH \mathbb{T}r^{\mathbb{L}_0}&=\torusq^{c/12} F^{(1)}_{\torusq,\cH}\left((\tilde{\psi}_1,\zetap_1,\bar{\zetap}_1),\ldots,(\tilde{\psi}_n,\zetap_n,\bar{\zetap}_n)\right)
\end{align}
with the identifications
\begin{align}\label{eq:corrfunctransferopfullcftvar}
  \tilde{\psi}_i = q^{\mathbb{L}_0/2} \psi\,\qquad\textrm{ and }\qquad \zetap_j=z q^j\,,\; \bar{\zetap}_j = \bar{z} q^j\qquad\textrm{ for }j=1,\ldots,n\ ,\torusq&=rq^{n} \,.
\end{align}
\end{corollary}

\begin{proof}
The proof of this statement is analogous to the proof of  Lemma~\ref{lem:trsfopexplicit} and Corollaries~\ref{cor:corrfctexactzero} and~\ref{cor:corrfctexactrepr}; it is a direct consequence of~\eqref{eq:dilationequationfullcft} and will be omitted here. 
\end{proof}

We proceed to argue that the concept of  bounded intertwiners also generalizes to full CFTs.

\subsubsection{Bounded intertwiners for full CFTs}

As in the chiral case, in order to establish bounds for the approximation of full CFT correlation functions by finite-dimensional MPS, we need bounds on the norms of  operators appearing in the definition of the transfer operator. Accordingly, we first introduce an analogous notion of a bounded intertwiner.
\begin{definition}\label{def:sboundedcft}
Let $\mathbb{Y}$ be the full CFT intertwiner of a CFT on a Hilbert space~$\cH$, and let $\cS\subset\cH$ be a linear subspace. Let $\cftW_q$ be the associated scaled intertwiner (Definition~\ref{def:cftscaledintertwiner}). We call $\mathbb{Y}$ $\cS$-bounded if 
for all $z,\bar{z}\in\mathbb{C}\backslash\{0\}$ and $0<q<\min \{|z|^2,1/|z|^2,|\bar{z}|^2,1/|\bar{z}|^2\}$, there is a constant $\cftbnd(q,(z,\bar{z}))$  such that 
\begin{align}
  \norm{\cftW_q(\psi,(z,\bar{z}))} \leq \cftbnd(q,(z,\bar{z}))\cdot \norm{\psi}_{\cH}\qquad\textrm{ for all }\psi\in \cS\ .
\end{align}   
\end{definition}

The following result implies that $\cS$-boundedness can be established by considering only intertwiners associated with a chiral part of the CFT (cf.~\eqref{eq:fullintertwinercft}).

\begin{lemma}\label{lem:cftsboundednesstochiral}
  Let $\cftY$ be the full CFT intertwiner of a diagonal CFT on a Hilbert space~$\cH$, and let $\cS\subset\cH$ be a linear subspace of the form~\eqref{eq:Sdirectsumdecomp}. Then $\cftY$ is $\cS$-bounded if and only if each constituent is either $S[j]$- or $S[j]^\prime$-bounded, depending on its type.
\end{lemma}
 
Assume that the VOA $\cV=\cV^L=\cV^R$ satisfies all our technical assumptions.
Then Lemma~\eqref{lem:cftsboundednesstochiral}, when combined with  Corollary~\ref{cor:boundednessscaledintertwiner}, implies that $\cftY$ is $\cS$-bounded for  any finite-dimensional subspace of the form~\eqref{eq:Sdirectsumdecomp}.

\begin{proof}
  Inserting the definitions, we find the following expressions for the full scaled intertwiner $\cftW_q$,
  \begin{align}
    &\cftW_q(a_j\otimes a'_j,(z,\bar{z}))(b_k\otimes b'_k)=\\
    &\qquad=\sum_\ell q^{L_0/2}\itw{A[\ell]}{A[j]}{A[k]}(q^{L_0/2}a_j,z)q^{L_0/2}b_k\otimes q^{L'_0/2} \itw{A[\ell]'}{A[j]'}{A[k]'}(q^{L'_0/2}a'_j,\bar{z})q^{L'_0/2}b_k' \nonumber\\
&\qquad=\sum_\ell \W\itw{\ell}{j}{k}_q(a_j,z)b_k\otimes  \W'\itw{\ell}{j}{k}_q(a'_j,z)b_k'\,,
  \end{align}
  for $a_j \in S[j] \subset A[j]$, $a'_j \in S[j]^\prime \subset A[j]^\prime$, $b_k\in A[k]$ and $b'_k\in A[k]'$.  In this identity,
$\W\itw{\ell}{j}{k}_q$ is the $q$-scaled version of the intertwiner~$\itw{A[\ell]}{A[j]}{A[k]}$, and, analogously, $\W'\itw{\ell}{j}{k}_q$
is the scaled version of~$\itw{A[\ell]'}{A[j]'}{A[k]'}$.  We conclude that 
  \begin{align}
    \cftW_q(a_j\otimes a_j',(z,\bar{z})) = \sum_{k,\ell} \W\itw{\ell}{j}{k}_q(a_j,z) \otimes \W'\itw{\ell}{j}{k}_q(a_j',z)\,\label{eq:scaledintertwinerprojectedout}
  \end{align}
where each of the summands is an operator supported on $A[k]\otimes A[k]'$. By the triangle inequality and the assumption that the constituents are $S[j]$- or $S[j]'$-bounded, we get  for normalized vectors $a_j\in S[j]$ and $a'_j\in S[j]'$ the inequality
  \begin{align}
    \norm{\cftW_q(a_j\otimes a_j',(z,\bar{z}))} &\leq \sum_{k,\ell} \big\|\W\itw{\ell}{j}{k}_q(a_j,z)\big\| \cdot \big\|\W'\itw{\ell}{j}{k}_q(\tilde{a}_j',z)\big\| \\
    &\leq \sum_{k,\ell} \bnd_{jk\ell}(q,z) \,\bnd_{j k \ell}^\prime(q,\bar{z})\,,
  \end{align}
  where $\bnd_{j k \ell}(q,z)$ and $\bnd_{j k \ell}^\prime(q,\bar{z}) $ are the boundedness parameters of the constituents. 

Consider  a general element $\psi \in \cS$, which can be decomposed into 
  \begin{align}\label{eq:decomppsifullcft}
    \psi = \sum_{jrs} c_{jrs} a_{jr} \otimes a_{js}' \,,\quad \textrm{ with } a_{jr} \in S[j]\,,\; a_{js}'\in S[j]'\ ,
  \end{align}
  where we assume that for each fixed $j$, the vectors $a_{jr}$ (resp. $a'_{js}$), $r,s =1,\ldots, \dim S[j]$, form a orthonormal basis of the space $S[j]$ (resp. of $S[j]^\prime$) and thus 
  \begin{align}
    \norm{\psi}_{\cH}^2 = \sum_{jrs} |c_{jrs}|^2 \,.
  \end{align}
  Inserting the decomposition into the expression for the full scaled intertwiner, we find 
  \begin{align}\label{eq:boundenessparamfulcft}
    \norm{\cftW_q(\psi,(z,\bar{z}))}_{\cH} &\leq \sum_{jrs}  |c_{j rs}|\cdot \norm{\cftW_q(a_{jr}\otimes a'_{js},(z,\bar{z}))}\\ 
    &\leq \norm{\psi}_{\cH}\cdot \sqrt{\sum_{jrs}\norm{\cftW_q(a_{jr}\otimes a_{js}',(z,\bar{z}))}^2}\\
    &\leq \norm{\psi}_{\cH} \sqrt{\sum_j \left(\dim S[j]\,\sum_{k,\ell}\bnd_{j k \ell}(q,z)\bnd'_{j k \ell}(q,\bar{z})\right)^2  }=:\|\psi\|_{\cH}\cdot \Theta(q,z,\bar{z})
 \end{align}
 by the Cauchy-Schwarz inequality. Since the sum  under the square root only involves a finite number of terms by our assumptions on $S[j]$ and the rationality of the VOA~$\cV$, the expression $\Theta(q,z,\bar{z})$ is  finite. This proves that if the constitutents are $S[j]$- or $S[j]'$-bounded, respectively, then the full scaled intertwiner~$\cftW_q$ is $\cS$-bounded.

For the converse, assume that the full scaled intertwiner is bounded with boundedness parameter~$\Theta(q,(z,\bar{z}))$. Observe that according to~\eqref{eq:scaledintertwinerprojectedout}, constituents may be otained by applying suitable projections $Q_j$ onto the subspaces $A[j]\otimes A[j]'\subset \cH$, i.e., we have
\begin{align*} 
Q_\ell\cftW_q(a_j\otimes a_j',(z,\bar{z})) Q_k =  \W\itw{\ell}{j}{k}_q(a_j,z) \otimes \W'\itw{\ell}{j}{k}_q(a_j',\bar{z})\, 
\end{align*}
Since projections do not increase the norm, we obtain
\begin{align}
 \Big\| \W\itw{\ell}{j}{k}_q(a_j,z) \otimes \W'\itw{\ell}{j}{k}_q(a_j',\bar{z})\Big\|
\leq \Theta(q,(z,\bar{z})) \|a_j\otimes a_j'\|_{\cH} \ ,
\end{align}
that is,
\begin{align*}
\frac{\Big\| \W\itw{\ell}{j}{k}_q(a_j,z)\Big\|}{\|a_j\|_{A[j]}}\cdot 
\frac{\Big\| \W'\itw{\ell}{j}{k}_q(a_j',\bar{z})\Big\|}{\|a_j'\|_{A[j]'}}\leq \Theta(q,(z,\bar{z}))\qquad\textrm{ for all }a_j\in S[j]\textrm{ and }a_j'\in S[j]'\ .
\end{align*}
Taking the supremum over $a_j$ and $a_j'$, we conclude that the constituents~$\W\itw{\ell}{j}{k}_q$ 
and $\W'\itw{\ell}{j}{k}_q$
are $S[j]$-bounded and  $S[j]'$-bounded, respectively, with boundedness parameters satisfying 
\begin{align}
\vartheta_{jk\ell}(q,z)\vartheta'_{jk\ell}(q,\bar{z})\leq \Theta(q,(z,\bar{z}))\ .\label{eq:thetaproductupperboundp}
\end{align}
\end{proof}

The preceding proof exemplifies how statements about the full CFT case can be established by expressing the scaled intertwiner as a sum of tensor products of scaled intertwiners of the original VOA, and then applying the results for the chiral case. Since most statements in this section are obtained following this straightforward strategy, we will merely sketch the arguments.

As in the chiral case (see~Eq.~\eqref{eq:zonenormboundv}), the boundedness of scaled intertwiners (Definition~\ref{def:sboundedcft}) immediately implies the boundedness of the associated transfer operator. More precisely, we have for the operator norm of~$\mathbb{T}$ (cf. Definition~\ref{def:cfttransferoperators}),  with insertions $\psi_i\in \cS, i=1,\ldots,n$ and the formal variables replaced by two complex numbers $z$, $\bar{z}$ such that $z,\bar{z}\in\mathbb{C}\backslash\{0\}$ and $0<q<\min \{|z|^2,1/|z|^2,|\bar{z}|^2,1/|\bar{z}|^2\}$ the bound 
\begin{align}
  \|\mathbb{T}\|\leq \left(\cftbnd(q,z,\bar{z})\right)^n \prod_{j=1}^n\|\psi_j\|_{\cH}\ . \label{eq:cftzonenormboundv}
\end{align}
We will use~\eqref{eq:cftzonenormboundv} below to estimate the errors when truncating operators.

\subsubsection{Approximation results for truncated intertwiners}
We next introduce a truncated version of the scaled intertwiner, now for the full CFT. It is again defined in such a way that it does not change the grading by more than the truncation parameter~$N$. However, the latter notion has to be defined with respect to the $\mathbb{N}_0\times\mathbb{N}_0$-grading of the Hilbert space~$\cH$ given in~\eqref{eq:hgradingfullcft}. Recall (cf.~\eqref{eq:fullintertwinercft}) that the full CFT intertwiner~$\mathbb{Y}$ is given by a linear combination of tensor products of intertwiners~$\itw{A[\ell]}{A[j]}{A[k]}$ and $\itw{A[\ell]'}{A[j]'}{A[k]'}$. To define a truncated version, we  simply truncate each of these constituent intertwiners. The definition of a truncated scaled intertwiner follows accordingly. As in the chiral case, we have to ensure that the subspace $\cS$ is spanned by homogeneous elements, which requires the introduction of a version of $\Shom$ for the full CFT case. In complete analogy to Eq.~\eqref{eq:defShom}, we define
\begin{align}
  \cShom:=\bigoplus_{j}\bigoplus_{n^L,n^R\in\mathbb{N}_0} \left(\cS\cap (A[j]_{n^L}\otimes A[j]'_{n^R})\right)\,. \label{eq:defShomfullcft}
\end{align}

\begin{definition}[Truncated intertwiner for full CFT]
  Let $\mathbb{Y}$ be the full intertwiner of a diagonal CFT with associated Hilbert space $\cH$, with action given as in~\eqref{eq:fullintertwinercft}. For a linear subspace~$\cS\subset \cH$, let $\cShom$ be defined by Eq.~\eqref{eq:defShomfullcft} and suppose that  $\mathbb{Y}$ is~$\cS$-bounded. Then for~$N>0$, the truncated intertwiner
  \begin{align*}
    \mathbb{Y}^{[N]}(\cdot, (z,\bar{z})):\cH\rightarrow\Endint{\cH}{z,z^{-1},\bar{z},\bar{z}^{-1}}    
  \end{align*}
  is defined by
  \begin{align}  
    \cftY^{[N]}(a_j\otimes a'_j,(z,\bar{z}))b_k\otimes b'_k &=\sum_\ell\itw{A[\ell]}{A[j]}{A[k]}^{[N]}(a_j,z)b_k\otimes \itw{A[\ell]'}{A[j]'}{A[k]'}^{[N]}(a'_j,\bar{z})b'_k\  \label{eq:truncatedfullintertwinercft}
\end{align}
  for $a_j\otimes a'_j\in A[j]\otimes A[j]'$ and  $b_k\otimes b'_k\in A[k]\otimes A[k]'$. The associated scaled truncated intertwiner (cf.~Definition~\ref{def:scaledintertwiner}) is defined by 
   \begin{align*}
     \cftW_q^{[N]}:\cShom & \rightarrow  \Endint{\cH}{z,z^{-1},\bar{z},\bar{z}^{-1}}\\
     \psi  & \mapsto  \cftW_q^{[N]}(\psi,(z,\bar{z})):=q^{\mathbb{L}_0/2} \mathbb{Y}^{[N]}(q^{\mathbb{L}_0/2}\psi,(z,\bar{z}))q^{\mathbb{L}_0/2}\ ,
    \end{align*}
  for all homogeneous $\psi\in \cS$, and linearly extended to $\cShom$. We call the family $\{\cftW_q^{[N]}\}_{0<q<1}$  the {\em $N$-th level truncation of~$\cftW$}, or simply a {\em truncated scaled intertwiner}.
\end{definition}

As  with intertwiners, we can decompose truncated scaled intertwiners into their constituents; those happen to be made up of truncated scaled intertwiners in the chiral sense. More precisely, by inserting the definition of~$\mathbb{Y}^{[N]}$, we have (using~\eqref{eq:fullcftlmodeoperators})
\begin{align}\label{eq:truncatedscaledvb} 
\cftW_q^{[N]}(a_j\otimes a_j',(z,\bar{z}))(b_k\otimes b_k')&=
\sum_\ell \W\itw{\ell}{j}{k}_q^{[N]}(a_j,z)b_k\otimes \W'\itw{\ell}{j}{k}_q^{[N]}(a_j',z)b_k'\ ,
\end{align}
where each $\W\itw{\ell}{j}{k}_q^{[N]}$ is a truncated $S[j]$-bounded scaled intertwiner of type $\itype{A[\ell]}{A[j]}{A[k]}$, and similarly, each $\W'\itw{\ell}{j}{k}_q^{[N]}$ is a truncated $S[j]'$-bounded scaled intertwiner of type $\itype{A[\ell]'}{A[j]'}{A[k]'}$. Again by reduction to the chiral case (cf.~Lemma~\ref{lem:boundednessparametertruncated}), we find that the truncated scaled intertwiner is itself bounded.

\begin{corollary}
  Let $\cftW_q^{[N]}$ be the truncated scaled intertwiner of a full CFT, with indeterminates replaced by complex numbers $z,\bar{z}\in\mathbb{C}\backslash\{0\}$ such that $0<q<\min \{|z|^2,1/|z|^2,|\bar{z}|^2,1/|\bar{z}|^2\}$. Then we have for $\psi\in \cShom$
  \begin{align}
    \norm{\cftW^{[N]}_q(\psi,(z,\bar{z}))} \leq \norm{\psi}_\cH \frac{\cftbnd(\sqrt{q},(z,\bar{z}))}{1-\sqrt{q}} \,,
  \end{align}
  where the bound holds for all $N \in \Nl \cup \{\infty\}$, so in particular also for the non-truncated scaled intertwiner.
\end{corollary}

Starting from the truncated scaled intertwiner, we immediately get the truncated version of the transfer operator. That is, the {\em truncated transfer operator} $\mathbb{T}^{[N]}:\cH\rightarrow\cH$ with insertions $\psi_i\in \cS, i=1,\ldots,n$, is defined by the composition
\begin{align}\label{eq:deftruncatedtransferfullcft}
  \mathbb{T}^{[N]} &=\cftW_q^{[N]}(\psi_1,(z,\bar{z})) \circ \cftW^{[N]}(\psi_{2},(z,\bar{z})) \circ \cdots  \circ \cftW^{[N]}(\psi_n,(z,\bar{z})) \,,
\end{align} where the formal variables have been replaced by complex numbers $z,\bar{z}\in\mathbb{C}\backslash\{0\}$ satisfying $0<q<\min \{|z|^2,1/|z|^2,|\bar{z}|^2,1/|\bar{z}|^2\}$.

In order to establish error bounds on the approximation of CFT correlation functions, we need an analogue of Theorem~\ref{thm:truncationestimate}. In fact, we again can obtain a norm estimate on the difference between $\cftW_q$ and $\cftW_q^{[N]}$ by a reduction to that theorem. Due to its importance for our main result, we include a proof of this fact.
\begin{theorem}\label{thm:truncationestimatefullcft}
Let $z,\overline{z}\in\mathbb{C}\backslash\{0\}$ be such that $0<q<\min\{|z|^2,1/|z|^2,|\overline{z}|^2,1/|\overline{z}|^2\}$ and let~$\cftW_q^{[N]}$ be the $N$-th level truncation of a scaled intertwiner~$\cftW_q$
with boundedness parameter $\Theta$. 
 Then the norm of the difference between the scaled intertwiner and its truncated version is bounded by
\begin{align}
  \big\|\cftW_{q}(\psi,(z,\overline{z})) - \cftW_q^{[N]}(\psi,(z,\overline{z}))\big\| &\leq q^{N/4}\cdot\Gamma\cdot\frac{\cftbnd(\sqrt{q},(z,\bar{z}))}{(1-\sqrt{q})^2} \cdot \|\psi\|_\cH
\end{align}
for all $\psi\in\cShom$, where $\Gamma$ is a constant only depending on the dimension of $\cS$ and the  number of irreducible modules of the VOA~$\cV$ appearing in the decomposition~\eqref{eq:fullcftdecomp}  of $\cH$.
\end{theorem}

\begin{proof}
  Fix some index~$j$ in  the decomposition~\eqref{eq:fullcftdecomp} of $\cH$ and let $a_j\otimes a'_j\in \Shom \cap (A[j]\otimes A[j]')$
  be arbitrary but normalized. By the triangle inequality, we find
  \begin{align}\label{eq:sumwqnwq}
   &\big\| \cftW_q^{[N]}(a_j\otimes a'_j,(z,\bar{z})) -\cftW_q(a_j\otimes a'_j,(z,\bar{z}))\big\| \leq\nonumber\\
   &\qquad\leq\sum_{k,\ell} \Big\| \W\itw{\ell}{j}{k}_q^{[N]}(a_j,z)\otimes \W'\itw{\ell}{j}{k}_q^{[N]}(a'_j,\bar{z})-\W\itw{\ell}{j}{k}_q(a_j,z)\otimes \W'\itw{\ell}{j}{k}_q(a'_j,\bar{z})\Big\|  
  \end{align}
  Every term in this sum has the form
  \begin{align*}
    \big\| \W^{[N]}_q(a_j,z)\otimes (\W')^{[N]}_q(a_j',\bar{z})- \W_q(a_j,z)\otimes \W'_q(a_j',\bar{z})\big\|\ ,
  \end{align*}
  where $\W_q$ and $\W'_q$ are $S[j]$ and $S[j]'$-bounded scaled intertwiners with parameters $\bnd_{jk\ell}(q,z)$ and $\bnd'_{jk\ell}(q,\bar{z})$, respectively (cf.~Lemma~\ref{lem:cftsboundednesstochiral}), and $\W_q^{[N]}$ and $(\W')_q^{[N]}$ are their truncated versions (cf.~\eqref{eq:truncatedscaledvb}). In particular,
  \begin{align*}
    &\big\| \W^{[N]}_q(a_j,z)\otimes (\W')^{[N]}_q(a_j',\bar{z})- \W_q(a_j,z)\otimes \W'_q(a_j',\bar{z})\big\|\\
    &\qquad\leq \big\|\W^{[N]}_q(a_j,z)-\W_q(a_j,z)\big\|\cdot\big\|(\W')^{[N]}_q(a_j',\bar{z})\big\|+\|\W_q(a_j,z)\|\cdot \big\|(\W')^{[N]}_q(a_j',\bar{z})-\W'_q(a_j',\bar{z})\big\| \\
    &\qquad\leq \frac{q^{N/4}}{\sqrt{1-\sqrt{q}}} \,\left(\apperr_{jk\ell}(q,z) \bnd'_{jk\ell}(\sqrt{q},\bar{z}) +  \apperr'_{jk\ell}(q,\bar{z})\bnd_{jk\ell}(\sqrt{q},z) \right)
  \end{align*}
  where we used Theorem~\ref{thm:truncationestimate} and Lemma~\ref{lem:boundednessparametertruncated}. Taking the sum over $k,\ell$ and inserting the expressions for the error bounds $\apperr_{jk\ell}(q,z)$, $\apperr'_{jk\ell}(q,\bar{z})$ from Theorem~\ref{thm:truncationestimate} then leads to the bound
  \begin{align}
    \big\| \cftW_q^{[N]}(a_j\otimes a'_j,(z,\bar{z})) -\cftW_q(a_j\otimes a'_j,(z,\bar{z}))\big\|  &\leq 2\kappa\, \frac{q^{N/4}}{(1-\sqrt{q})^2} \,\sum_{k,\ell} \,\left(\bnd'_{jk\ell}(\sqrt{q},\bar{z})\bnd_{jk\ell}(\sqrt{q},z)\right)
  \end{align}
  The case of a general element $\psi\in \cS$ then follows as in Lemma~\ref{lem:cftsboundednesstochiral}, following Eq.~\eqref{eq:decomppsifullcft} and the remainder of the proof of Lemma~\ref{lem:cftsboundednesstochiral}. Analogous to Eq.~\eqref{eq:boundenessparamfulcft}, we obtain the estimate
\begin{align}
    \big\| \cftW_q^{[N]}(\psi,(z,\bar{z})) -\cftW_q(\psi,(z,\bar{z}))\big\|  &\leq 2\kappa\, \frac{q^{N/4}}{(1-\sqrt{q})^2} \,
    \sqrt{\sum_j  \left(\dim S[j] \,\sum_{k,\ell}\bnd_{j k \ell}(\sqrt{q},z)\bnd'_{j k \ell}(\sqrt{q},\bar{z})\right)^2  }\ 
  \end{align}
 for normalized $\psi\in\cS$. 
But using~\eqref{eq:thetaproductupperboundp}, we have for all $j,k,\ell$
\begin{align*}
\bnd_{j k \ell}(\sqrt{q},z)\bnd'_{j k \ell}(\sqrt{q},\bar{z})\leq \Theta(\sqrt{q},(z,\bar{z}))\ .
\end{align*}
The claim follows, since
\begin{align*}
  \sum_j \left(\dim S[j]\,\sum_{k,\ell}\Theta(\sqrt{q},(z,\bar{z}))\right)^2&= \Theta(\sqrt{q},(z,\bar{z}))^2\sum_j \left(\dim S[j]\, \sum_{k,\ell}1\right)^2\\
&=\Theta(\sqrt{q},(z,\bar{z}))^2\cdot \dim \cS \cdot |J|(|J|^2)^2\ ,
\end{align*}
where $|J|$ is the number of summands in the decomposition~\eqref{eq:fullcftdecomp} of~$\cH$.
\end{proof}

By again exploiting the telescoping sum technique, we obtain the following bound on the norm difference between the truncated transfer operator and its original.  This statement is similar to Lemma~\ref{lem:normboundtransferop} with one important improvement: instead of the number~$n$ of insertions, the bound depends on the number~$m$ of {\em non-trivial} insertions, that is, insertions that are not equal to the vacuum~$\1\otimes\1'$. 
This stems from the additional property~\eqref{eq:vacuumfullcftidentity} of
full CFT intertwiners: the intertwining maps associated with~$\1\otimes\1'$ are the identity.

\begin{corollary}
  Let $\mathbb{T}$ be the transfer operator of the full CFT (cf. Definition~\ref{def:cfttransferoperators}), with 
each insertion $\psi_i$, $i=1,\ldots,n$ satisfying either
\begin{enumerate}[(i)]
\item
$\psi_i\in\cShom\backslash\{1\otimes 1'\}$, $\|\psi_i\|_{\cH}=1$ or 
\item
$\psi_i=\1\otimes \1'$,
\end{enumerate}
and the formal variables replaced by two complex numbers $z,\bar{z}\in\mathbb{C}\backslash\{0\}$ such that $0<q<\min \{|z|^2,1/|z|^2,|\bar{z}|^2,1/|\bar{z}|^2\}$. Let $\mathbb{T}^{[N]}$ be its truncated version (cf Eq.\eqref{eq:deftruncatedtransferfullcft}). Then the norm of their difference is bounded by 
  \begin{align}
    \norm{\mathbb{T} - \mathbb{T}^{[N]}} \leq q^{N/4} \cdot m\cdot \Gamma \cdot\left(\frac{\cftbnd(\sqrt{q},(z,\bar{z}))}{(1-\sqrt{q})^2}\right)^m \,,
  \end{align}
   where $m$ is the number of insertions which are not identical to the vacuum vector, $m = \#\{i\in\{1,\ldots,n\}\,:\,\psi_i \neq \1 \otimes \1'\}$.
\end{corollary}

\begin{proof}
  In the case where all insertions $\psi_i$, $i=1,\ldots,n$ are not equal to the vacuum, we have $n=m$ and the proof is completely analogous to Lemma~\ref{lem:normboundtransferop} and thus omitted. 

Now consider the case where~$\psi_j = \1 \otimes \1'$ for some~$j \in \{1,\ldots,n\}$.  We employ the fact that intertwiners of full CFTs map the vacuum vector to the identity map on $\cH$ independently of the formal variables $z,\bar{z}$, 
cf.~\eqref{eq:vacuumfullcftidentity}.  This implies immediately that the truncated intertwiner satisfies
  \begin{align}
    \cftY^{[N]}(\1 \otimes \1',(z,\bar{z})) = \cftY(\1 \otimes \1',(z,\bar{z}))\,,
  \end{align}
  independently of the truncation parameter $N$, and hence also
  \begin{align}
    \cftW_q^{[N]}(\1 \otimes \1',(z,\bar{z})) = \cftW_q(\1 \otimes \1',(z,\bar{z}))\,.
  \end{align}
  Hence within the telescoping sum argument of Lemma~\ref{lem:normboundtransferop} (or rather its straightforward adaption to the full CFT case), the replacement of $\cftW_q(\1 \otimes \1',(z,\bar{z}))$ with its truncated version $\cftW_q^{[N]}(\1 \otimes \1',(z,\bar{z}))$ can be done without picking up an error term resulting from the approximation. This proves the assertion.
\end{proof}

Finally, we need to project the Hilbert space~$\cH$ onto a finite-dimensional subspace. For this purpose, we use the $\mathbb{N}_0\times\mathbb{N}_0$-grading of the Hilbert space~$\cH$ given in~\eqref{eq:hgradingfullcft}. For $M>0$, we define the projection~$P^{[M]}$ on $\cH$ as the orthogonal projection onto the subspace
\begin{align}
  \mathbb{P}^{[M]}\cH=\bigoplus_{n^L\leq M, n^R\leq M}\cH_{(n^L,n^R)}\ .\label{eq:cutoffspacefullcft}
\end{align}
Recall that the Hilbert space~$\cH$ of a diagonal CFT decomposes as in~\eqref{eq:fullcftdecomp} into a direct sum of tensor products~$A[j]\otimes A[j]'$, where each  $A[j]$ is  irreducible  and $A[j]'$ is the associated contragredient module. Using the~$\mathbb{N}$-grading on these modules, the space~\eqref{eq:cutoffspacefullcft} can be written as
\begin{align*}
  \mathbb{P}^{[M]}\cH=\bigoplus_{j}  P_{A[j]}^{[M]}A[j]\otimes P_{A[j]'}^{[M]}A[j]'\ ,
\end{align*} 
where $P^{[M]}_A$ for a module $A$ is the projection introduced in Section~\ref{sec:finitebonddimensionerror}. That is, the projection associated with the full CFT has the form
\begin{align*}
  \mathbb{P}^{[M]}=\bigoplus_j P_{A[j]}^{[M]}\otimes P_{A[j]'}^{[M]}\ ,  
\end{align*} 
i.e., it cuts off each module independently. 

Combining the bound between the transfer operator and its truncated version with the expressions for the correlation functions leads directly to the analogue of Corollary~\ref{cor:approxcorfuncplane} for the vacuum-to-vacuum full CFT correlation function. Incorporating the projector $\mathbb{P}^{[M]}$ and repeating the steps in the proof of Lemma~\ref{lem:projectedcorrelationfct} leads to the corresponding analogue in the genus-1 case. We collect these approximation statements into a Corollary, but omit the proof (which as stated is exactly as in the chiral case).

\begin{corollary}[Matrix product tensor network for full CFT]\label{cor:approxfulfcftmps}
  Let $M,N$ be positive natural numbers, the truncation parameters, and let $\mathbb{T}^{[N]}$ be the truncated transfer operator (cf. Eq.~\eqref{eq:deftruncatedtransferfullcft}) of a full CFT with Hilbert space $\cH$, with 
each insertion $\psi_i$, $i=1,\ldots,n$ satisfying either
\begin{enumerate}[(i)]
\item
$\psi_i\in\cShom\backslash\{1\otimes 1'\}$, $\|\psi_i\|_{\cH}=1$ or 
\item
$\psi_i=1\otimes 1'$.
\end{enumerate}
Let $F^{(0)}_\cH$ and $F^{(1)}_{\torusq,\cH}$ be the genus-0 and genus-1 correlation functions defined by the right hand sides of Eqs.~\eqref{eq:corrfunctransferopfullcftgenus0} and~\eqref{eq:corrfunctransferopfullcftgenus1} with the identification of variables as in~\eqref{eq:corrfunctransferopfullcftvar}. Then we have the following approximation statement for the vacuum-to-vacuum correlation function,
  \begin{align}
    \left|\Scp{\1 \otimes \1'}{\mathbb{T}^{[N]}\,\1 \otimes \1'} - q^{\sum_j j \wt \psi_j} F^{(0)}_\cH \right| \leq q^{N/4} \cdot m \cdot \Gamma \cdot \left(\frac{\cftbnd(\sqrt{q},(z,\bar{z}))}{(1-\sqrt{q})^2}\right)^m \,,
  \end{align}
  where $m$ is the number of insertions which are not identical to the vacuum vector, $m = \#\{i\in\{1,\ldots,n\}\,:\,\psi_i \neq \1 \otimes \1'\}$. Similarly, the following approximation statement for the genus-1 correlation function holds for $0<r<1$,
  \begin{align}
    \left|\tr_\cH \mathbb{P}^{[M]} \mathbb{T}^{[N]} r^{\mathbb{L}_0} - \torusq^{c/12} F_{\torusq,\cH}^{(1)} \right| \leq \left(q^{N/4} \,m\, \Gamma\, r^{c/12} \mathbb{Z}(r) + r^M r^{c/6} \mathbb{Z}(\sqrt{r})\right) \cdot\left(\frac{\cftbnd(\sqrt{q},(z,\bar{z}))}{(1-\sqrt{q})^2}\right)^m\,,
  \end{align}
  where $\mathbb{Z}(r)$ is the partition function of the full CFT,
  \begin{align}
    \mathbb{Z}(r) = \sum_j Z_{A[j]}(r)^2 \,.
  \end{align}
\end{corollary}
We have thus found  approximating expressions for the correlation functions of full CFTs in the genus-0 and the genus-1 case. It remains to show that these expressions can be represented using finitely correlated functionals on matrix algebras. This is the topic of Section~\ref{sec:fullcftsasfinitely}. 
     
\subsection{Full CFTs as finitely correlated functionals\label{sec:fullcftsasfinitely}}
To relate the expressions involving the truncated transfer operator~$\mathbb{T}^{[N]}$ in Corollary~\ref{cor:approxfulfcftmps} to finitely correlated functionals, we proceed in two steps. We first show that elements in the full CFT Hilbert space $\cH$ can be interpreted as Hilbert-Schmidt operators. This identification lifts to a norm-preserving isomorphism. We then show that the truncated scaled intertwiner is equivalent to a linear map on this space of Hilbert-Schmidt operators. Using this observation, we can construct certain completely positive maps defining an FCS.

\subsubsection{Elements of the full CFT Hilbert space as Hilbert-Schmidt operators}
In order to argue that $\cH$ can be seen as a space of Hilbert-Schmidt operators, consider first the case where the direct sum in~\eqref{eq:fullcftdecomp} only consists of a single term. This implies that~$\cH$ is the algebraic tensor product $\cH = A \otimes A^\prime$, with $A'$ being the restricted dual space\footnote{If $A=\bigoplus_{n \in \Nl} A_n$ is the decomposition of $A$ into levels, then  $A^\prime = \bigoplus_{n \in \Nl} A_n^\prime$.}. The space~$\cH$ can alternatively be represented as a space of operators  on $A$. In fact, the equation
\begin{align}
  \left(\sum_i a_i \otimes a_i^\prime \right)(b) = \sum_i a_i \,a_i^\prime(b) \,,\quad a_i \in A\,,\; a_i^\prime \in A^\prime\,,
\end{align}
defines a linear operator of finite rank on the Hilbert space $A$. Conversely, a
finite rank operator $\cO$ on $A$ can be decomposed as
\begin{align}
  \cO(\cdot ) = \sum_i a_i \,\Scp{\tilde{a}_i}{\cdot}\,,\quad a_i,\,\tilde{a}_i \,\in A\,.
\end{align}
In fact, setting $a_i^\prime = \tilde{\eta}(\tilde{a}_i)$ (see~\eqref{eq:isodualcomplconj})  gives rise to the well-known linear isomorphism~$\upsilon$ between the algebraic tensor product space $A \otimes A^\prime$ and the space of linear finite-rank operators on $A$. If we complete $A \otimes A^\prime$ as a tensor product Hilbert space (there is a unique scalar product), $\upsilon$ extends to a norm-continuous Banach space isomorphism with image being the Hilbert-Schmidt operators on $A$. This isomorphism also naturally extends to direct sums. We summarize the discussion into the following lemma.

\begin{lemma}\label{lem:identifycftops}
  Let $\cH$ be a rational diagonal CFT with unique vacuum associated with a VOA~$\cV$. Then $\cH$ is linearly isomorphic to the direct sum of Hilbert spaces of Hilbert-Schmidt operators $\mathbb{H}(A[j])$ acting on the Hilbert space $A[j]$,
  \begin{align}
    \cH \cong \mathbb{H} =  \bigoplus_j \mathbb{H}(A[j]) \cong \bigoplus_{\alpha\in\widehat{\cV}} \mathbb{H}(A_\alpha) \otimes \Complex^{m_\alpha}\,.
  \end{align}
  We denote this isomorphism of Hilbert spaces by $\upsilon:\cH\rightarrow\mathbb{H}$.
\end{lemma}
By definition, the isomorphism~$\upsilon$ is
compatible with the $\mathbb{N}_0\times\mathbb{N}_0$-grading of the Hilbert space~$\cH$: the images of $\cH_{(n^L,n^R)}$ under $\upsilon$ give an $\mathbb{N}_0\times\mathbb{N}_0$-grading on $\mathbb{H}$. In particular,  the image of
 $\mathbb{P}^{[M]}\cH$ (cf.~\eqref{eq:cutoffspacefullcft}) 
under $\upsilon$ is the space of Hilbert-Schmidt operators  with levels $(n^L,n^R)$ satisfying  $n^L,n^R\leq M$. We will write $\mathbb{P}^{[M]}\mathbb{H}$ for this space. By definition, the space $\mathbb{P}^{[M]}\mathbb{H}$ is equal to the direct sum
\begin{align}
  \mathbb{B}^{[M]} = \bigoplus_j \mathsf{Mat}(\C^{d_{A[j]}(M)})\,\label{eq:bmmatrixalgebra}
\end{align}
of matrix algebras equipped with the Hilbert-Schmidt norm.

Similarly, let us denote the images of $\cS$ and $\cShom$ under $\upsilon$ by $\mathbb{S}$ and $\mathbb{S}_{\mathsf{hom}}$, respectively. Observe that the isomorphism~$\upsilon$ is also compatible with $\cS$ if the latter has the form~\eqref{eq:defShomfullcft}: in this case, its restriction to $S[j]\otimes S[j]'$ has image $\mathbb{H}(S[j])\subset\mathbb{H}(A[j])$.
Note also that since the spaces $S[j]$ are assumed to be finite-dimensional, it follows from Lemma~\ref{lem:identifycftops} that $\mathbb{S}$ is equal to a direct sum of matrix algebras, equipped with the Hilbert-Schmidt norm,
\begin{align}
  \mathbb{S} = \bigoplus_j \mathsf{Mat}(\C^{\dim S[j]}) \,.\label{eq:smatrixalgebradef}
\end{align}
\subsubsection{Intertwiners as maps on Hilbert-Schmidt operators}
 Because
of the compatibility of $\upsilon$ with $\cS$, we can compose the isomorphism $\upsilon$ with an instance of a scaled intertwiner $\cftW_q :\cShom \mapsto \End(\cH)$ to obtain a bounded map on Hilbert-Schmidt operators. We denote this map by 
\begin{align}\label{eq:deftransferfullcft}
\begin{matrix}
  \cftE_q: &\mathbb{S}_{\mathsf{hom}} \times \C \times \C &\rightarrow &\End(\mathbb{H})\\
&(\Psi,z,\bar{z})&\mapsto & \cftE_q(\Psi,(z,\bar{z})) = \upsilon \circ \cftW_q(\upsilon^{-1}(\Psi),(z,\bar{z})) \circ \upsilon^{-1}\ .
\end{matrix}
\end{align}
Again, the formal variables are replaced by complex numbers $z$, $\bar{z}$ such that $z,\bar{z}\in\mathbb{C}\backslash\{0\}$ and $0<q<\min \{|z|^2,1/|z|^2,|\bar{z}|^2,1/|\bar{z}|^2\}$. Since the map $\upsilon$ is an isomorphism, any sequence of scaled intertwiners $\cftW_q$ on $\cH$ can be reexpressed as a sequence of maps $\cftE_q$ on Hilbert-Schmidt operators. 

We can apply the same arguments to truncated scaled intertwiners, obtaining linear maps on the space of Hilbert-Schmidt operators. This map will change  the $\Nl_0 \times \Nl_0$-grading of  $\mathbb{H}$ by at most $N$ (in each argument). Furthermore, this procedures also applies if we additionally project onto weight spaces. 
That is, for truncation- and cutoff-parameters $M,N\in\mathbb{N}_0$, we define a map
\begin{align}
\begin{matrix}
  \cftE^{[M,N]}_q: &\mathbb{S}_{\mathsf{hom}} \times \C \times \C &\rightarrow &\End(\mathbb{P}^{[M]}\mathbb{H})\,,\\
&(\Psi,z,\bar{z}) & \mapsto &  \cftE^{[M,N]}_q(\Psi,(z,\bar{z})) = \upsilon \circ \left(\mathbb{P}^{[M]}\, \cftW^{[N]}_q(\upsilon^{-1}(\Psi),(z,\bar{z})) \mathbb{P}^{[M]} \right) \circ \upsilon^{-1}\ .
\end{matrix}
\end{align}
(Contrary to our treatment of the chiral case, we choose to incorporate the cutoff, i.e., projection using $\mathbb{P}^{[M]}$,  into the definition of the truncated scaled intertwiner.)

\subsubsection{Finitely correlated states for full CFTs\label{sec:finitelycorrfull}}
Using the isomorphism $\upsilon:\cH\rightarrow\mathbb{H}$ and following the arguments in the proofs of Lemma~\ref{lem:mpsformapproxgenuszero} and Lemma~\ref{lem:mpsformapproxtorus}, we conclude that expressions
such as $\Scp{\1 \otimes \1^\prime}{\mathbb{T}^{[N]}\,\1 \otimes \1^\prime}$ in Corollary~\ref{cor:approxfulfcftmps} can be expressed in terms of a composition of the maps~$\cftE^{[M+nN,N]}_q$ (see the proof of Lemma~\ref{lem:fcsformtrunctransfer} below for explicit expressions). We will show that  the maps 
 $\cftE^{[M+nN,N]}_q$ can be replaced by  completely positive maps on matrix algebras. This  turns the expressions of interest into  finitely correlated functionals. The involved matrix algebras are built from the images~$\mathbb{B}^{[M+nN]}$
of weight spaces (cf.~\eqref{eq:bmmatrixalgebra}), as well as the matrix algebra~$\mathbb{S}$ associated with the subspace~$\cS$ (cf.~\eqref{eq:smatrixalgebradef}).

\begin{lemma}[FCS for full CFTs]\label{lem:fcsformtrunctransfer}
  Let $M,N\in\mathbb{N}_0$ be the cutoff and truncation parameters. 
There exist a bounded linear embedding 
\begin{align}
\iota: \mathbb{B}^{[M+nN]} \to \mathsf{Mat}(\C^2) \otimes \mathbb{B}^{[M+nN]}
\end{align}
 as well as a completely positive map
  \begin{align}
\begin{matrix}
    \cftEcp_q : \mathbb{S} \otimes \mathsf{Mat}(\C^2) \otimes \mathbb{B}^{[M+nN]} \to \mathsf{Mat}(\C^2) \otimes \mathbb{B}^{[M+nN]}\,, 
\end{matrix}
  \end{align}
  such that the following holds for
the maps ($\Psi\in\mathbb{S}$)
\begin{align*}
\cftEcp_{q,\Psi}(Y)=\cftEcp_q(\Psi\otimes Y)\qquad Y\in\mathsf{Mat}(\C^2) \otimes \mathbb{B}^{[M+nN]}\ .
\end{align*}
For all $\Psi_i\in \cShom$, $i=1,\ldots,n$, parameters
 $z,\bar{z} \in \C \setminus \{0\}$, $0<q<\min\{|z|^2,|z|^{-2},|\bar{z}|^2,|\bar{z}|^{-2}\}$, $0<r\leq 1$ and $\phi_1,\phi_2 \in \mathbb{P}^{[M+nN]}\mathbb{H}$,
the (projected) matrix elements of the truncated transfer operator~$\mathbb{T}^{[N]}$  (cf. Eq.~\eqref{eq:deftruncatedtransferfullcft}) with insertions $\Psi_i \in \cShom$, $i=1,\ldots,n$
are given by a composition of the maps~$\cftEcp_{q,\Psi_j}$ as  
  \begin{align}
    \Scp{\phi_1}{\mathbb{P}^{[M+nN]} \mathbb{T}^{[N]} r^{\mathbb{L}_0} \mathbb{P}^{[M+nN]} \phi_2}  = \tr_{\C^D}\left[\iota(\Phi_1)^*\,\cftEcp_{q,\Psi_1} \circ  \cdots \circ \cftEcp_{q,\Psi_n} \circ \cftEcp_{r,\ket{\1}\bra{\1}}(\iota(\Phi_2))\right] \,. 
  \end{align}
 Furthermore, 
the map 
 $\cftEcp_{q=1,\ket{\1}\bra{\1}}$ satisfies
 \begin{align}\label{eq:fullcfttransferopidentiycond}
\cftEcp_{1,\ket{\1}\bra{\1}}(\iota(\Phi_2)) = \iota(\Phi_2)\qquad\textrm{ for all }\qquad\Phi_2\in  \mathsf{Mat}(\C^2) \otimes \mathbb{B}^{[M+nN]}\ .
\end{align}
\end{lemma}

\begin{proof}
By repeating the arguments of Lemma~\ref{lem:mpsformapproxgenuszero} and exploiting the isomorphism $\upsilon$, we get the following representation of matrix elements of the projected and truncated scaled intertwiner,
  \begin{align}
    &\Scp{\phi_1}{\mathbb{P}^{[M+nN]} \mathbb{T}^{[N]} r^{\mathbb{L}_0} \mathbb{P}^{[M+nN]} \phi_2} = \\
    &\qquad=\;\tr_{\C^{D/2}}\left[\upsilon(\phi_1)^* \cdot\mathbb{D}^{(1)} \circ \cdots \circ \mathbb{D}^{(n)} \circ \mathbb{D}^{(n+1)}(\upsilon(\phi_2))\right]\,,
  \end{align}
  where we abbreviated
  \begin{align}
    \mathbb{D}^{(i)} = \cftE^{[M+nN,N]}_{q}(\Psi_i,(z,\bar{z}))\,,\quad \mathbb{D}^{(n+1)} = \cftE^{[M+nN,N]}_{r}(\ket{\1} \bra{\1},(z,\bar{z}))\,.
  \end{align}
  Here, we also used that the $\upsilon$ is a Hilbert space isomorphism and that the scalar product for Hilbert-Schmidt operators is given by the trace. The map $\cftE_q$ is a bounded mapping from a finite-dimensional space of Hilbert-Schmidt operators $\mathbb{S}$ into the endomorphism of a space of Hilbert-Schmidt operators $\mathbb{P}^{[M+nN]}\mathbb{H}$. It can therefore be linearily extended to a bounded map defined on the tensor product,
  \begin{align}
    \mathbb{F}_q\,:\, \mathbb{S} \otimes \mathbb{B}^{[M+nN]} \to \mathbb{B}^{[M+nN]}\,,\quad \textrm{by setting } \;\mathbb{F}_q(\Psi \otimes \Phi) = \cftE^{[M+nN,N]}_q(\Psi,(z,\bar{z}))(\Phi)\,.
  \end{align}
  for $\Psi \in \mathbb{S}$ and $\Phi \in \mathbb{B}^{[M+nN]}$. Here, we suppressed the dependence on the parameters $z,\bar{z}$ in $\mathbb{F}_q$ for notational convenience. 

Since $\mathbb{F}_q$ is a bounded operator between two spaces of Hilbert-Schmidt operators, it is also bounded as a map from matrices to matrices equipped with the operator norm. As the matrix algebras are finite-dimensional, it is completely bounded (explicit bounds will of course include various factors depending on the dimensions). By the extension of Stinespring's theorem to completely bounded mappings~\cite[Theorem 8.4]{paulsen2002completely}, there exists linear embeddings $V_{1,q}$, $V_{2,q}$ of $\C^{D/2} $ into $\C^{d} \otimes \C^{D/2} \otimes \C^{D^\prime}$, $D^\prime \leq (dD)^2$ such that
  \begin{align}
    \mathbb{F}_q(\Psi \otimes \Phi) = V_{1,q}^* \,\left(\Psi \otimes \Phi \otimes \idty_{\C^{D^\prime}} \right)\, V_{2,q} \,.
  \end{align}
  This statement is immediate for each direct summand of $\mathbb{B}^{[M+nN]}$ and can be extended to the whole space since the sum is finite. We then define the bounded embedding
  \begin{align}
    \iota\,:\,\mathbb{B}^{[M+nN]} \to \mathsf{Mat}(\C^2) \otimes \mathbb{B}^{[M+nN]}\,,\quad\iota\,:\,\Phi \mapsto \ket{0}\bra{1} \otimes \Phi \,,
  \end{align}
  where $\ket{0}$,$\ket{1}$ is the standard basis of $\C^2$. We now define 
  \begin{align}
    V_q \,:\,\C^2 \otimes \C^{D/2} \rightarrow \C^2 \otimes \C^{d} \otimes \C^{D/2} \otimes \C^{D^\prime}\,,\qquad V_q = V_{1,q} \otimes \ket{0}\bra{0} + V_{2,q} \otimes \ket{1}\bra{1}\,.
  \end{align}
This   gives rise to a completely positive map
  \begin{align}
    &\cftEcp_q\,:\,\mathbb{S} \otimes \mathsf{Mat}(\C^2) \otimes \mathbb{B}^{[M+nN]}  \to \mathsf{Mat}(\C^2) \otimes \mathbb{B}^{[M+nN]}\\
    &\cftEcp_q\,:\,X_1 \otimes X_2 \mapsto V_q^* \,\left(X_1 \otimes X_2 \otimes \idty_{\C^{D^\prime}}\right) \,V_q \,.
  \end{align}
  This map is of the form
  \begin{align}
    \cftEcp_q = \begin{pmatrix*}[l] \mathbb{F}_{1,q} & \mathbb{F}_q \\ \mathbb{F}_q^* & \mathbb{F}_{2,q} \end{pmatrix*} \,,
  \end{align}
  where $\mathbb{F}_{q,i}$, $i=1,2$ are completely positive maps, and $\mathbb{F}_q^*$ is the Hermitian conjugate of the map $\mathbb{F}$. It follows from this structure that
  \begin{align}
    \cftEcp_q(\Psi \otimes \iota(\Phi)) = \iota(\mathbb{F}_q(\Psi \otimes \Phi))\,,
  \end{align}
  and hence we find by composing this identity 
  \begin{align}\label{eq:seqcftEcpintermsF}
    \cftEcp_{q,\Psi_1}\circ\cftEcp_{q,\Psi_2}\circ \cdots \circ \cftEcp_{q,\Psi_n}\circ\iota(\Phi) = \iota\circ\mathbb{F}_{q,\Psi_1}\circ\mathbb{F}_{q,\Psi_2}\circ \cdots \circ \mathbb{F}_{q,\Psi_n}(\Phi)
  \end{align}
  where we abbreviated $\mathbb{F}_{q,\Psi}(\Phi) := \mathbb{F}_q(\Psi \otimes \Phi)$. Furthermore, we have
  \begin{align}\label{eq:iotaiotaadjointaction}
    \tr_{\C^D}\left[\iota(\Phi)^* \iota(\Phi^\prime)\right] = \tr_{\C^D}\left[\kettbra{1} \otimes (\Phi^* \Phi^\prime)\right] = \tr_{\C^{D/2}}\left[\Phi^* \Phi^\prime\right]
  \end{align}
  and thus combining Eqs.~\eqref{eq:seqcftEcpintermsF} and~\eqref{eq:iotaiotaadjointaction} gives
  \begin{align}
    \Scp{\phi_1}{\mathbb{P}^{[M+nN]} \mathbb{T}^{[N]} r^{\mathbb{L}_0} \mathbb{P}^{[M+nN]} \phi_2} = \tr_{\C^D}\left[\iota(\Phi_1)^*\,\cftEcp_{q,\Psi_1} \circ  \cdots \circ \cftEcp_{q,\Psi_n} \circ \cftEcp_{r,\ket{\1}\bra{\1}}(\iota(\Phi_2))\right] \,.
  \end{align}
  In order to verify the statement of Eq.~\eqref{eq:fullcfttransferopidentiycond}, note that the pre-image of $\kettbra{\1}$ under the isomorphism $\upsilon$ equals the vacuum vector $\1 \otimes \1^\prime \in \cV \otimes \cV \subset \cS \subset \cH$. And since intertwiners of a full CFT are required to map the vacuum vector $\1 \otimes \1^\prime \in \cV \otimes \cV^\prime$ to the identity operator (cf. Eq.~\eqref{eq:vacuumfullcftidentity}), we have
\begin{align}
  \cftE_{q}(\ket{\1}\bra{\1},(z,\bar{z}))(X) = q^{\mathfrak{L}_0}(X)\,,
\end{align}
where we denoted by $\mathfrak{L}_0 = \upsilon \circ \mathbb{L}_0 \circ \upsilon^{-1}$ the corresponding grading map on the space of Hilbert-Schmidt operators. This implies that $\cftE_{1}(\ket{\1}\bra{\1},(z,\bar{z}))(X) = X$.

\end{proof}

With this Lemma at hand, we can also complete the proof of our main theorem.

\begin{proof}[Proof of Theorem~\ref{thm:main1}]
  The chiral case was already proven in Section~\ref{sec:recoveringthemps}. The full CFT case follows by combining Observation~\ref{obs:fullcft} with Corollary~\ref{cor:approxfulfcftmps} and  Lemma~\ref{lem:fcsformtrunctransfer}.
\end{proof}
  
We remark on one further feature of this approximation. If we are interested in $n$-point correlation functions with non-equispaced insertion points, but all distances are a multiple of a smallest distance (which we may call the UV cutoff), then we do not need to increase the bound dimension to achieve a desired level of approximation. Instead, we only have to apply the map~$\cftEcp_{q,\ket{\1}\bra{\1}}$ a certain number of times. Let us illustrate this with the example of a three-point correlation function. 

Consider the correlation function of the full CFT, with three insertion points on the real line, $\zeta_1 = \zeta^*_1$, $\zeta_2 = \zeta^*_2$, $\zeta_3 = \zeta^*_3$,  separated by distances which are multiples of $a>0$ and evaluated for three primary vectors $\psi_1$, $\psi_2$, $\psi_3$. Without loss of generality, due to translational invariance, we can assume that $\zeta_1=0$, $\zeta_2 = t\cdot a$, $\zeta_3 = (t+s+1)\cdot a$ with $t,s \in \Nl$. 
Since the image of the vacuum element under the action of $\cftY$ is required to be the identity operator, we find for the corresponding correlation function
\begin{align}
  &F^{(0)}((\psi_1,\zeta_1,\zeta_1),(\psi_2,\zeta_2,\zeta_2),(\psi_3,\zeta_3,\zeta_3))\\
  &\qquad=F^{(0)}((\psi_1,\zeta_1,\zeta_1),(\1\otimes \1^\prime,a,a),\ldots,(\psi_2,\zeta_2,\zeta_2),\ldots,(\psi_3,\zeta_3,\zeta_3))\,,
\end{align}
with dots $\ldots$ representing insertions of the vacuum. Using Corollary~\ref{cor:approxfulfcftmps} as well as Lemma~\ref{lem:fcsformtrunctransfer}, we find
\begin{align}\label{eq:twopointfunctiontransfermap}
  &q^{\wt \psi_1+t\cdot\wt \psi_2 + (t+s+1)\wt \psi_3}F^{(0)}((\psi_1,\zeta_1,\zeta_1),(\psi_2,\zeta_2,\zeta_2),(\psi_3,\zeta_3,\zeta_3)\\
  &\qquad\approx \tr_{\C^D}\left[\iota(\ket{\1}\bra{\1})\cdot \cftEcp_{q,\Psi_1} \circ \left(\cftEcp_{q,\ket{\1}\bra{\1}}\right)^t \circ \cftEcp_{q,\Psi_2} \circ \left(\cftEcp_{q,\ket{\1}\bra{\1}}\right)^{s} \circ \cftEcp_{q,\Psi_3}(\iota(\ket{\1}\bra{\1}))\right]\,.
\end{align}
This may motivate the term \emph{transfer map} for the mapping $\cftEcp_{q,\ket{\1}\bra{\1}}$, since it plays the role of the transfer operator in the usual setting of MPS and FCS.

\section{Summary and Outlook\label{sec:outlook}}

We have obtained rigorous error bounds for approximating  correlation functions of conformal field theories in terms of finite-dimensional matrix product tensor networks. These bounds apply both to chiral as well as full conformal field theories.  Prior to our work, various numerical and theoretical studies (see Section~\ref{sec:intro} for  references) suggested that such approximations are possible, yet explicit error bounds were missing. 

Our construction, which is phrased in terms of vertex operator algebras, their modules and intertwiners, exploits and elucidates the underlying representation-theoretic structures of a CFT. The power of this approach becomes especially apparent in the case of WZW models, where we exhibit a close connection between our  construction of transfer operators and group-covariant MPS. 

An intermediate step in our arguments shows that correlation functions are {\em exactly} represented by certain (infinite-dimensional) MPS or finitely correlated states. This property is ultimately a consequence of the ``gluing axiom'' (see e.g.,~\cite{gawedzki1997lectures}) of CFT correlation functions, a feature which  has a counterpart in the axiomatization 
of topological quantum field theories (see e.g., \cite{Witten89,Atiyah89,Walker91,MooreSeiberg98}). For the latter setting, various tensor network techniques have been developed~\cite{KoeBil10,Pfeiferetal10,singh14}, which in some cases~\cite{AguadoVidal08,KoeReichardtVidal09} also turn out to give exact descriptions. However, contrary to the setting of topological quantum field theories, CFTs are inherently infinite-dimensional. From the point of view of computational and variational physics, it is therefore imperative to quantitatively understand the approximability by  finite-dimensional tensor networks. Here our analysis establishes relationships between the accuracy of approximation, MPS parameters (in particular, the bond dimension), and parameters of the CFT.

We expect that the tools developed here can also be applied to obtain similar quantitative statements for other classes of tensor network functionals approximating quantum field theories. An example of such a class of finite-dimensional approximations are continuous matrix product states. Another class of tensor network Ansatz states,  the so-called MERA (for \emph{multi-scale renormalization Ansatz}), has empirically been  shown to provide accurate numerical descriptions of quantum critical systems~\cite{PhysRevLett.99.220405,PhysRevLett.102.180406,PhysRevB.79.144108}. Its application is also supported by entanglement entropy considerations~\cite{PhysRevLett.99.220405}. This method has been extremely successful at numerically identifying the CFT resulting from the continuum limit of a critical spin system. Our work is taking a complementary approach:  given a CFT, we seek suitable finite-dimensional tensor networks encoding the correlation functions. Having addressed this problem for MPS, we expect that  our methods constitute the first step in a similar program for MERA. Ultimately, this should  provide a convincing theoretical explanation for the numerical success of MERA.

Beyond alternative classes of tensor networks, several fundamental open problems remain to be addressed. Our work is restricted to genus-$0$ and genus-$1$ correlation functions; analogous results for higher genus surfaces would be desirable. A potential strategy could be to follow the work~\cite{fuchs2005tft} and try to find approximate versions of their arguments. 
Finally, whether such rigorous statements about the existence of \emph{useful} finite-dimensional approximations can also be made for more general quantum field theories, that is, not necessarily conformal ones, remains completely open.

\section*{Acknowledgements}

We thank Tobias Osborne, Patrick Hayden, John Preskill and Michael Walter for valuable discussions. 
 VBS expresses his gratitude to the Institute of Quantum Computing at the University of Waterloo as well as to the Institute for Advanced Study and the Department of Mathematics at the Technische Universit\"at M\"unchen for several visits, during which this project was put forward. 
RK is supported by the Technische Universit\"at M\"unchen -- Institute for Advanced Study, funded by the German Excellence Initiative and the European Union Seventh Framework Programme under grant agreement no.~291763. He gratefully acknowledges partial support by NSERC during the initial stages of this project.  
VBS acknowledges partial support by  the NCCR QSIT. 

\appendix

\section{Bounds on partition functions and polylogarithms\label{app:bounds}}

In this appendix we summarize various estimates used in the main paper on the asymptotic behavior and upper bounds for partition functions. We fix some notation before we start. Let $P(n,k)$ be the number of (multi-)partitions of $n$ into integers of \(k\) colors. Its generating function can be found to be~\cite{murty},
\begin{align}
  \sum_{n=0}^\infty q^n  P(n,k) = \prod_{n\geq1}(1-q^n)^{-k}\,,
\end{align}
and we denote the right hand side by \(g_k(q)\). Throughout this appendix, we assume that \(q\) is a real number with \(0<q<1\). 

The next Lemma gives a very rough upper bound on the asymptotic growth of \(P(n,k)\). Its proof is analogues to a well-known argument by Siegel, see~\cite[pp. 316-318]{apostol2013introduction} for an exposition. 

\begin{lemma}\label{lem:app:boundpartitionfunction}
 We have
  \begin{align}
    P(n,k) < e^{2\pi \sqrt{\frac{k n}{6}}} \,.
  \end{align}
\end{lemma}

\begin{proof}
  Let $g_k(q)=\sum_{n=0}^\infty q^n  P(n,k) = \prod_{n\geq1}(1-q^n)^{-k}$. It follows that \(P(n,k) q^n < g_k(q)\) and thus
  \begin{align}
    \log P(n,k) < \log g_k(q) + n \log \frac{1}{q} \,.
  \end{align}
  We estimate both terms on the rhs.~separately, and begin with the first.
  \begin{align}
    \log g_k(q) &= - \log \prod_{n\geq1}(1-q^n)^{-k} = -k \sum_{n\geq 1} \log(1-q^n) = k \sum_{n\geq 1} \sum_{m \geq 1} \frac{q^{mn}}{m}\\
    &=k \sum_{m \geq 1} \frac{1}{m} \sum_{n \geq 1} (q^m)^n = k \sum_{m \geq 1} \frac{1}{m} \frac{q^m}{1-q^m} \,.
  \end{align}
  However, an easy argument shows that
  \begin{align}
    \frac{1}{m}\frac{q^m}{1-q^m} \leq \frac{1}{m^2} \frac{q}{1-q} \,,
  \end{align}
  from which we get that
  \begin{align}
    \log g_k(q) \leq k \frac{q}{1-q} \,\sum_{m\geq 1} \frac{1}{m^2} = k \frac{q}{1-q} \frac{\pi^2}{6} \,.
  \end{align}
  Moreover, we have that
  \begin{align}
    \log \frac{1}{q} < \frac{1-q}{q} \,,
  \end{align}
  from which it follows that
  \begin{align}
    \log P(n,k) < \log g_k(q) + n \log \frac{1}{q} \leq k \frac{q}{1-q} \frac{\pi^2}{6} + n \frac{1-q}{q} \,.
  \end{align}
  Choosing the value of \(q\) as
  \begin{align}
    q \mapsto \frac{1}{1+\sqrt{\frac{k \pi^2}{6 n}}}
  \end{align}
  then leads to the bound
  \begin{align}
    \log P(n,k) < 2 n \sqrt{\frac{k \pi^2}{6 n}} = 2 \pi \sqrt{\frac{k n}{6}} \,.
  \end{align}
\end{proof}

Apart from partition functions, in the discussion of WZW models, certain sums appeared in the norm bounds of module and intertwiner operators. In general, these sums had the form
\begin{align}
  \sum_{n\geq 0} q^{a n} (n+1)^{b}
\end{align}
for positive integers \(a,b\) and \(0<q<1\). However, slight rewriting leads to
\begin{align}
  \sum_{n\geq 0} q^{a n} (n+1)^{b} = q^{-a} \sum_{n\geq 1} \frac{(q^a)^n}{n^{-b}}
\end{align}
which is a multiple of the defining equation for the poly-logarithm. At negative integers, it can be expressed as
\begin{align}
  \sum_{n\geq 1} \frac{(q^a)^n}{n^{-b}} = \left( x \frac{\partial}{\partial x} \right)^b \frac{x}{1-x} \,\left|_{x=q^a} \right.
\end{align}
For the convenience of the reader, we provide a more explicit bound on the dependence in \(q\).

\begin{lemma}\label{lem:estimategammasum}
  Let \(a,b\) be positive integers, and \(0<q<1\). Then we have the estimate
  \begin{align}
    \sum_{n\geq 0} q^{a n} (n+1)^{b} \,\leq\, a^{-b-1} b!\, q^{-a} \log\left(\frac{1}{q}\right)^{-b-1} \,.
  \end{align}
\end{lemma}

\begin{proof}
  We use the first order expression of the integral given by the Euler-MacLaurin formula, see~\cite{apostolEuler}, and find
  \begin{align}
    \sum_{n\geq 0} q^{a n} (n+1)^{b} &= \int_{0}^{\infty} dt q^{a t} (t+1)^{b} + \int_{0}^{\infty} dt (t-\floor{t}) \left[b q^{a t} (t+1)^{b-1} + a \log(q) q^{a t} (t+1)^{b}\right] \\
    &\leq \int_{0}^{\infty} dt q^{a t} (t+1)^{b} + b \int_{0}^{\infty} dt q^{a t} (t+1)^{b-1} \,,
  \end{align}
  since \(\log(q)<0\). By subsequent variable substitution, we can rewrite the integral
  \begin{align}
    \int_{0}^{\infty} dt q^{a t} (t+1)^{b} = q^{-a} a^{-b-1} \log(1/q)^{-b-1} \int_{a \log(\frac{1}{q})}^\infty dt e^{-t} \,t^{b} \leq q^{-a} a^{-b-1} \log(1/q)^{-b-1} b! \,,
  \end{align}
  where we used that the gamma function at integral values reduces to the factorial function. 
\end{proof}

\section{Algorithm for WZW models\label{app:algorithmwzw}}

In this appendix, we present an algorithm for computing matrix elements of an intertwiner between three irreducible modules (for WZW models). 

\subsection{Computing normal forms of vectors}
For this purpose, we need to introduce suitable bases, both of the
Lie algebra~$\g$ as well as of the Hilbert spaces of the corresponding $\voalkzero$-modules.
The relationship~\eqref{eq:involutionvbgmodule}  between the involution~$\hat{\eta}:\g\rightarrow\g$ and adjoints suggests that it is especially convenient to work with elements $\fe_{1},\ldots,\fe_{\ell}\in \g$ with the property that 
 \begin{align}\hat{\eta}(\fe_j) = \fe_j \label{eq:involutioninvariancev}
\end{align}
 for $j=1,\ldots,\ell$.
 By polarization, we can assume that  elements~$\{\fe_j\}_{j}$ with the property~\eqref{eq:involutioninvariancev} span~$\g $. We will choose a basis $\{\fe_i\}_{i=1}^{\dim\g}$ of the Lie algebra~$\g$ of this form. In the following,  we will often pick families
$\{\fa_k\}_k\subset \{\fe_i\}_{i=1}^{\dim\g}$ or 
$\{\fb_k\}_k\subset \{\fe_i\}_{i=1}^{\dim\g}$ of basis elements of~$\g$, i.e., $\fa_k=\fe_{i_k}$ for some function $i:k\mapsto i_k\in \{1,\ldots,\dim\g\}$ and similiarly for $\fb_k$ (but it will be convenient not to use the latter notation). 

Let us next consider an irreducible module $\modulekflambda{\lambda}$ of the VOA~$\voalkzero$
corresponding to an integral  dominant weight~$\lambda$ such that $\lambda(\theta) \leq k$.  The Hilbert space~$\modulekflambda{\lambda}$ is the space spanned by elements of the form
\begin{align}
  \fe_{i_1}(-n_1) \fe_{i_2}(-n_2) \cdots \fe_{i_l}(-n_l) \vphi^{i}\label{eq:elementsinWZWtrunc}
\end{align}
such that the following conditions hold:
\begin{enumerate}[(i)]
\item
The vectors 
 $\{\vphi^i\}_{i=1}^{\dim \modulekflambda{\lambda}(0)}$ are a basis of the irreducible $\g$-module~$\modulekflambda{\lambda}(0)$ (which coincides with the top level of the $\voalkzero$-module, see Section~\ref{sec:introVOAmodules}).
  \item The $n_i$ are positive integers. A vector of the form~\eqref{eq:elementsinWZWtrunc} belongs to the  level $\modulekflambda{\lambda}(N)$ of the module, where $N=\sum_i n_i$. This also implies that $l$ is bounded by the level, $l \leq N$.
\end{enumerate}
That fact that the space~$\modulekflambda{\lambda}$ has a basis of the form~\eqref{eq:elementsinWZWtrunc} can be seen formally from Claim~\ref{eq:firstclaimv}, which provides an algorithm for expanding arbitrary vectors; an additional property of this expansion is that the integers $\{n_i\}_i$ are ordered.

We remark that these vectors are not orthonormal, but this is not necessary for our purposes (if an orthonormal basis is required, this can be obtained by applying Gram-Schmidt within each level, see e.g.,~\cite{Estienne:2013vf}). 
The representation of vectors as linear combinations of
vectors of the form~\eqref{eq:elementsinWZWtrunc} provides a natural way to index elements of the (truncated) Hilbert space: indeed,~\eqref{eq:elementsinWZWtrunc} is homogeneous and of weight $h_\lambda+\sum_{j=1}^\ell n_j$ (cf.~\eqref{eq:vkgradingdef}).

Recall that the top level $\modulekflambda{\lambda}(0)$ is an irreducible $\g$-module of highest weight~$\lambda$.  As explained in Section~\ref{sec:unitarity}, property~\eqref{eq:involutioninvariancev}  implies that 
\begin{align}\label{eq:adjointnegativecond}
  \Scp{\fe_j(m) \vphi}{\vphi'}_{\modulekflambda{\lambda}} &=\Scp{\vphi}{\fe_j(-m)\vphi'}_{\modulekflambda{\lambda}} 
\end{align} 
for any two elements 
 $\vphi,\vphi' \in \modulekflambda{\lambda}(0)$ belonging to the top level of~$\modulekflambda{\lambda}$.  This is a key property we will use extensively in the following.

Elements $\fa(n) = \fa \otimes t^n$ of the affine Lie algebra~$\ga$ act on the vectors~\eqref{eq:elementsinWZWtrunc} in a natural way.  Here and below, we can identify the zero modes $\fa(0)$ with the action of $\fa \in \g $ on the irreducible $\g $-module $\modulekflambda{\lambda_j}(0)$, and thus we henceforth write $\fa(0)\vphi_j = \fa \vphi_j$, for $\vphi_j \in \modulekflambda{\lambda_j}(0)$. 
If~$n<0$, then, after decomposing $\fa$ into a linear combination of the $\fe_i$s, the vector is again a linear combination of elements of the form~\eqref{eq:elementsinWZWtrunc}. However, if $n\geq0$, then we have to recursively apply the commutation relation (cf.~\eqref{eq:wzwcommutator})
\begin{align}
  [\fa(n),\fb(m)] = [\fa,\fb](n+m) + n \delta_{n+m,0} (\fa,\fb) k \idty\label{eq:wzwcommutatornew}
\end{align}
in order to obtain a linear combination of our basis elements~\eqref{eq:elementsinWZWtrunc}. Let us denote this linear transformation by $\mathsf{L}$. We note that it only depends on the structure of the affine Lie algebra~$\ga$ in question. Pseudo-code for this linear transformation is given in Algorithm~\ref{algo:WZWL}. It is straightforward to show the following.
\begin{claim}\label{eq:firstclaimv}
For a given input specifying a product $\fb_r(n_r)\cdots \fb_1(n_1)\in \ga$, where $\fb_j\in \{\fe_i\}_{i=1}^{\dim\g}$,
the Algorithm~\ref{algo:WZWL} produces
a list 
$\mathcal{L}=\bigl\{(m,\alpha_m,\substack{i_{s_m},\ldots,i_1\\
k_{s_m},\ldots,k_{1}})\}_{m}$
with the property that 
\begin{align}
\fb_r(n_r)\cdots \fb_1(n_1)\varphi&=\sum_{m\in \cL}\alpha_m \fe_{i_{s_m}}(k_{s_m})\cdots\fe_{i_1}(k_1)\varphi \label{eq:linearcombinationbproductv}
\end{align}
for all $\varphi\in \modulekflambda{\lambda}(0)$ belonging to the top level of (any) irreducible $\voalkzero$-module $\modulekflambda{\lambda}$. (In the sum on the right, an additional index~$m$ associated with the tuple~$\substack{i_{s_m},\ldots,i_1\\
k_{s_m},\ldots,k_{1}})$ is left implicit.
 Furthermore, for every entry in the list, the integers $k_j$ satisfy
\begin{align}
k_{s_{m}}<k_{s_{m}-1}<\cdots <k_{1}\leq 0\ .\label{eq:normalorderingcondition}
\end{align} 
\end{claim}
In other words, $\mathsf{L}$ performs a kind of normal ordering: it converts any product of elements in a~$\ga$-module to a linear combination of ``normally ordered'' terms  as in  Eq.~\eqref{eq:elementsinWZWtrunc} 
(observe that $\fe_{i_1}(0)$ preserves the top level).

\begin{algorithm}[h!]
\KwIn{
$(\substack{
\fb_r,\ldots,\fb_1\\
n_r,\ldots,n_1
})$, where\\
\qquad\qquad\begin{tabular}{l}
elements $\fb_1,\ldots,\fb_r\in\{\fe_{i}\}_{i=1}^{\dim\g}\subset \g$\\
integers $n_1,\ldots,n_r$
\end{tabular}}
\KwOut{ 
Finite list $\mathcal{L}=\bigl\{(m,\alpha_m,\substack{i_{s_m},\ldots,i_1\\
k_{s_m},\ldots,k_{1}})\bigr\}_{m}$, where for each $m$, we have $\alpha_m\in\mathbb{C}$ and $-k_j\in\mathbb{N}_0, 1\leq i_j\leq \dim\g$ for $j=1,\ldots,s_m$. 
This list satisfies
$\fb_r(n_r)\cdots \fb_1(n_1)\varphi=\sum_{m} \alpha_m \fe_{i_{s_m}}(k_{s_m})\cdots \fe_{i_1}(k_1)\varphi$ for all $\varphi\in \modulekflambda{\lambda}(0)$, and $k_{s_m}<\cdots <k_2<k_1\leq 0$.}
Set $\cL$ equal to the one-element list with entry representing the product~$\fb_r(n_r)\cdots \fb_1(n_1)$, i.e., 
$\cL=\{(1,\alpha_1=1,\substack{i_{r},\ldots,i_1\\
n_{r},\ldots,n_{1}})\}$ where $\fb_{j}=\fe_{i_j}$. \\
$\mathsf{changed}\longleftarrow \mathsf{true}$\\
\While{$\mathsf{changed}=\mathsf{true}$}{
  remove all entries  $(m,\alpha_m,\substack{i_{s_m},\ldots,i_1\\
k_{s_m},\ldots,k_{1}})\in\cL$ with $k_1>0$\\
  $\mathsf{changed}\longleftarrow\mathsf{false}$\\
  try to find an entry $(m,\alpha_m,\substack{i_{s_m},\ldots,i_1\\
k_{s_m},\ldots,k_{1}})\in\cL$ with $k_{r}\geq k_{r-1}$ for some $r\in\{2,\ldots,s_m\}$. \\
  \If{found}{
      $\mathsf{changed}\longleftarrow\mathsf{true}$\\
      Assume that $r\in \{2,\ldots,s_{m}\}$ is the maximal integer such that $k_r\geq k_{r-1}$\\
      remove the entry $(m,\alpha_m,\substack{i_{s_m},\ldots,i_1\\
k_{s_m},\ldots,k_{1}})$ from $\cL$\\
      create  new entries in $\cL$ by applying the commutation relations~\eqref{eq:wzwcommutatornew}. That is, compute
the commutator $[\fe_{i_r},\fe_{i_{r-1}}]=\sum_{\ell=1}^{\dim\g} \beta_{\ell}\fe_{\ell}$ and create
entries of the form $(\alpha_m \beta_\ell,\substack{i_{s_m},\ldots,i_{r+1},\ell,i_{r-2}\ldots,i_1\\
k_{s_m},\ldots,k_{r+1},k_{r}+k_{r-1},k_{r-2},\ldots,k_{1}})$ for $\ell=1,\ldots,\dim\g$, \\
as well as an additional entry of the form
$(\alpha_m k_r\delta_{k_r+k_{r-1},0}(e_{i_r},e_{i_{r-1}})k,\substack{i_{s_m},\ldots,i_{r+1},i_{r-2}\ldots,i_1\\
k_{s_m},\ldots,k_{r+1},k_{r-2},\ldots,k_{1}})$ \\
and an entry of the form
$(\alpha_m ,\substack{i_{s_m},\ldots,i_{r+1},i_{r-1},i_{r},i_{r-2}\ldots,i_1\\
k_{s_m},\ldots,k_{r+1},k_{r-1},k_{r},k_{r-2},\ldots,k_{1}})$ in $\cL$.
   }
} 
return $\cL$
\caption{The routine $\mathsf{L}$
expands a product $\fb_r(n_r)\cdots \fb_1(n_1)\in \ga$, where $\fb_j\in \{\fe_i\}_{i=1}^{\dim\g}$, into a linear combination of terms
$\fe_{i_{s_m}}(k_{s_m})\cdots \fe_{i_1}(k_1)$ such that $k_{s_m}<\cdots <k_2<k_1\leq 0$ and the linear combination has the same action on elements~$\varphi\in \modulekflambda{\lambda}(0)$ in the top level of a $\voalkzero$-module $\modulekflambda{\lambda}$. It simply applies commutators for achieving this. There is  no dependence on the module. \label{algo:WZWL}}
\end{algorithm}

\begin{proof}
 It is easy to check
that
Algorithm~\ref{algo:WZWL} maintains property~\eqref{eq:linearcombinationbproductv}  throughout. Indeed, any replacement described in lines~$7-12$ simply constitutes an application of the commutation relations~\eqref{eq:wzwcommutatornew} in the form
\begin{align}
\fe_{i_{s_m}}(k_{s_m})\cdots \fe_{i_1}(k_1)&=R_1 (\fe_{i_r}(k_r)\fe_{i_{r-1}}(k_{r-1}))R_2\\
&=R_1 \left([\fe_{r}(k_{r}),\fe_{r-1}(k_{r-1})]+\fe_{r-1}(k_{r-1})\fe_{i_r}(k_{r})\right)R_2\\
&=R_1\left([\fe_{r},\fe_{r-1}](k_{r}+k_{r-1})+k_{r}\delta_{k_r+k_{r-1},0}(\fe_{i_r},\fe_{i_{r-1}})k\mathbf{I})\right)R_2
\end{align}
to a single term in the sum, where
\begin{align*}
R_1&=\fe_{i_{s_m}}(k_{s_m})\cdots \fe_{i_{r+1}}(k_{r+1})\\
R_2&=\fe_{i_{r-2}}(k_{r-2})\cdots \fe_{i_1}(k_1)\ .
\end{align*}
Also, step~$4$ does not change the action on vectors $\varphi\in \modulekflambda{\lambda}(0)$ belonging to a top level since for such vectors, we have
\begin{align*}
  \fe_{i}(k)\varphi=0\qquad\textrm{ for all }i=1,\ldots,\dim\g\qquad \textrm{ and }\qquad k>0\ .
\end{align*}
The claim follows since `incorrectly ordered' terms are moved to the right by successive application of the commutation relations and 
eventually disappear at step~$4$ (thus ensuring property~\eqref{eq:normalorderingcondition}). 
\end{proof}

The algorithm~\ref{algo:WZWL} is a key ingredient in the following algorithm for computing matrix elements of intertwiners.  

\subsection{An algorithm for extending $G$-invariant maps to intertwiners}
Having introduced the basis~\eqref{eq:elementsinWZWtrunc} for modules and discussed the action of module mode operators, we proceed to consider intertwiners. For $j=1,2,3$, consider  irreducible modules~$\modulekflambda{\lambda_j}$ of the VOA~$\voalkzero$. To fully specify the intertwiner~$\iY$, we need to provide all matrix elements between vectors of  the form~\eqref{eq:elementsinWZWtrunc}, i.e., expressions of the form
\begin{align}
\Scp{\fa_\ell(m_\ell)\cdots \fa_1(m_1)\vphi^p_1}{\iY(\fc_s(k_s)\cdots \fc_1(k_1)\vphi^r_3,z)\fb_r(n_r)\cdots \fb_1(n_1)\vphi^q_2}_{\modulekflambda{\lambda_1}}\ .
\end{align}
where $\{\varphi_j^p\}_{p}$ is a basis of the top level $\modulekflambda{\lambda_j}(0)$ of module $\modulekflambda{\lambda_j}$ for $j=1,2,3$, and $\fa_p,\fc_r,\fb_q\in \{\fe_i\}_{i=1}^{\dim\g}$. In principle, this can be done following the techniques of~\cite{Huang:1992kw}. However, our setting is somewhat simpler  because we are interested in $S$-bounded intertwiners with the subspace $S=\modulekflambda{\lambda_3}(0)$ equal to the top level of the module~$\modulekflambda{\lambda_3}$. Here it suffices to consider the case where the first argument of $\iY$ is of the form $\varphi^r_3$ instead of~$\fc_s(k_s)\cdots \fc_1(k_1)\vphi^r_3$.

 In the following, we give an algorithm which  computes the matrix elements of~ $\iY(\varphi_3,z)$  with respect to vectors of the form~\eqref{eq:elementsinWZWtrunc}.
The algorithm requires that an intertwining map 
\begin{align}
\cW: V_3 \otimes V_2 \to V_1
\end{align} between $\g $-modules $V_i$, $i=1,2,3$ is given (equivalently, this is a $G$-intertwining map between the corresponding unitary representations of the group $G$), and extends this to an intertwiner of the VOA modules (or more precisely, to operators $\iY(\varphi_3,z)$ for $\varphi_3\in \modulekflambda{\lambda_3}(0)$ in the top level): The resulting interwiner~$\iY$ has the property that
its restriction to the top levels is the intertwining map~$\cW$. This establishes the converse direction in Proposition~\ref{prop:wzwintertwiner}, i.e., the correspondence between group-covariant MPS and the intertwiner of the VOA. 

We call the algorithm~$F$ (Algorithm~\ref{algo:IntWZW} in the following pseudocode). It takes as an argument two natural numbers $\ell,r$, as well as three vectors $\vphi_i \in \modulekflambda{\lambda_i}(0)=V_i$ together with elements $\fa_1,\ldots,\fa_\ell,\fb_1,\ldots,\fb_1 \in \{\fe_i\}_{i=1}^{\dim\g}\subset \g$ and integers~$m_1,\ldots,m_\ell,n_1,\ldots,n_r$.
It outputs the matrix element
\begin{align*}
\Scp{\fa_\ell(m_\ell)\cdots \fa_1(m_1)\vphi_1}{\iY(\vphi_3,z)\fb_r(n_r)\cdots \fb_1(n_1)\vphi_2}_{\modulekflambda{\lambda_1}}\ .
\end{align*}
The evaluation proceeds recursively, and invokes an auxiliary function~$\algoG$ (Algorithm~\ref{algo:WZWG}). The latter function evaluates matrix elements of a specific form by reduction to the zero mode; it depends itself on the subroutine $\mathsf{L}$ (Algorithm~\ref{algo:WZWL}) performing the normal ordering.

\begin{algorithm}[t]
\KwIn{
$(z,\ell,r,\vphi_3,\substack{\vphi_1\\\vphi_2}|\substack{\fa_\ell,\ldots,\fa_1\\ m_\ell,\ldots, m_1}|\substack{\fb_r,\ldots,\fb_1\\ n_r,\ldots, n_1})$, where
\qquad\qquad\begin{tabular}{l}
$z\in\mathbb{C}$\\
vectors $\varphi_j\in V_j$ belonging to irreducible $\g$-modules $V_j$, $j=1,2,3.$\\
$\ell,r\in\mathbb{N}_0$\\
 elements  $\fa_1,\ldots,\fa_\ell,\fb_1,\ldots,\fb_r \in \{\fe_i\}_{i=1}^{\dim \g}\subset \g$\\
integers $m_1,\ldots,m_\ell$, $n_1,\ldots,n_r$
\end{tabular}}
\KwOut{  matrix element $\Scp{\fa_\ell(m_\ell)\cdots \fa_1(m_1)\vphi_1}{\iY(\vphi_3,z)\fb_r(n_r)\cdots \fb_1(n_1)\vphi_2}_{\modulekflambda{\lambda_1}}$ }
\eIf{$\ell > 0$}{\Return{
$\begin{matrix}
F(z,\ell-1,r+1,\vphi_3,\substack{\vphi_1\\\vphi_2}|\substack{\fa_{\ell-1},\ldots,\fa_1\\ m_{\ell-1},\ldots, m_1}|\substack{\fa_\ell,\fb_r,\ldots,\fb_1\\ -m_\ell,n_r,\ldots, n_1})\\
+z^{m_\ell}\,F(z,\ell-1,r,\fa_\ell\vphi_3,\substack{\vphi_1\\\vphi_2}|\substack{\fa_{\ell-1},\ldots,\fa_1\\ m_{\ell-1},\ldots, m_1}|\substack{\fb_r,\ldots,\fb_1\\ n_r,\ldots, n_1})
\end{matrix}$}}{\Return{$\algoG(z,r,\vphi_3,\substack{\vphi_1\\\vphi_2}|\substack{\fb_r,\ldots,\fb_1\\n_r,\ldots,n_1})$}}
\caption{{\sc $F$} computes the matrix element of an intertwiner, using the subroutine~$\algoG$}
\label{algo:IntWZW}
\end{algorithm}

\begin{algorithm}[t]
\KwIn{
$(z,r,\vphi_3,\substack{\vphi_1\\\vphi_2}|\substack{\fb_r,\ldots,\fb_1\\n_r,\ldots,n_1})$, where
\qquad\qquad\begin{tabular}{l}
$z\in\mathbb{C}$\\
$r\in\mathbb{N}_0$\\
vectors $\varphi_j\in V_j$ belonging to irreducible $\g$-modules $V_j$, $j=1,2,3.$\\
elements $\fb_1,\ldots,\fb_r\in \{\fe_i\}_{i=1}^{\dim \g}\subset \g$\\
integers $n_1,\ldots,n_r$ satisfying $n_r<n_{r-1}<\cdots<n_1\leq 0$
\end{tabular}}
\KwOut{  matrix element $\Scp{\vphi_1}{\iY(\vphi_3,z)\fb_r(n_r)\cdots \fb_1(n_1)\vphi_2}_{\modulekflambda{\lambda_1}}$ }
\eIf{$r=0$}{
\Return{$z^{-\tau}\langle\varphi_1,\cW(\vphi_3\otimes\varphi_2)\rangle_{V_1}$}
}
{\eIf{$n_1=0$}{
compute $\varphi'_2\leftarrow \fb_1\varphi_2\in V_2$\\
\Return{$\algoG(z,r-1,\varphi_3,\substack{\vphi_1\\\vphi'_2}|\substack{\fb_r,\ldots,\fb_2\\n_r,\ldots,n_2})$}
}{
$\mathcal{L}=\bigl\{(m,\alpha_m,\substack{i_{s_m},\ldots,i_1\\
k_{s_m},\ldots,k_{1}})\bigr\}_{m}\longleftarrow \mathsf{L}(\substack{\fb_r,\ldots,\fb_1\\
n_r,\ldots,n_1},\varphi_2) $\\
\Return{$\sum_{m}\alpha_m\cdot
(-1)^{s_m} z^{\sum_{j=1}^{s_m} k_j} z^{-\tau}\Scp{\vphi_1}{\cW(\fe_{i_1} \cdots\fe_{i_{s_m}} \vphi_3\otimes \vphi_2)}_{V_1}$}
}
}
\caption{$\algoG$ computes certain matrix elements of an intertwiner 
from the restriction of the zero mode to the top levels. The latter is given by an $G$-intertwining map~$\cW$. The implementation of $\algoG$ relies on a subroutine~$\mathsf{L}$. \label{algo:WZWG}}
\end{algorithm}

\begin{claim}
Suppose we are given
an intertwining map
$\cW: V_3 \otimes V_2 \to V_1$
between the tensor product $\g$-module $V_3\otimes V_2$ and the $\g$-module $V_1$, where 
 $V_j$, $j=1,\ldots,3$ are irreducible $\g$-modules of highest weight~$\lambda_j$. We assume that $\cW$ is  specified in terms of  (an algorithm for computing) the matrix elements 
\begin{align*}
\Scp{\vphi_1}{\cW(\fe_{i_1} \cdots\fe_{i_{s_m}} \vphi_3 \otimes \vphi_2)}_{V_1}\ .
\end{align*}
for $\varphi_j\in V_j$. 
Set $\tau=h_{\lambda_3}+h_{\lambda_2}-h_{\lambda_1}$ and let $z\in\mathbb{C}$. 
Then algorithm~\ref{algo:IntWZW}  computes matrix elements of  the form
\begin{align}
\Scp{\fa_\ell(m_\ell)\cdots \fa_1(m_1)\vphi_1}{\iY(\vphi_3,z)\fb_r(n_r)\cdots \fb_1(n_1)\vphi_2}_{\modulekflambda{\lambda_1}}\qquad\textrm{ where }\varphi_j\in V_j=\modulekflambda{\lambda_j}(0)\ ,
\end{align} and where $\iY$ is an intertwiner 
of type $\itype{\modulekflambda{\lambda_1}}{\modulekflambda{\lambda_3}}{\modulekflambda{\lambda_2}}$, with the property that the restriction of its zero mode to the top levels (see Proposition~\ref{prop:wzwintertwiner}) coincides with~$\cW$.
\end{claim}

\begin{proof}
We show how to construct the matrix element of the intertwiner~$\iY$  from the map~$\cW$. Evaluating Eq.~\eqref{eq:commutationmoduleintertwiner} for modes of the form $\fa(m)$, $\fa \in \fg$, and intertwiners evaluated for $\vphi_3\in \modulekflambda{\lambda_3}(0)$, we find
\begin{align}\label{eq:basecasesforrecursive}
  \fa(m) \iY(\vphi_3,z)&=\iY(\vphi_3,z)\fa(m)+ \iY(\fa(0) \vphi_3,z)z^{m}\,,
\end{align}
since the action of positive modes on top level vectors vanishes, $\fa(m)\vphi_3 = 0$, for $m>0$.  Applying property~\eqref{eq:adjointnegativecond} and Eq.~\eqref{eq:basecasesforrecursive} once, we find
\begin{align}
  &\Scp{\fa_\ell(m_\ell)\cdots \fa_1(m_1)\vphi_1}{\iY(\vphi_3,z)\fb_r(n_r)\cdots \fb_1(n_1)\vphi_2}_{\modulekflambda{\lambda_1}}=\\
  &\qquad=\Scp{\fa_{\ell-1}(m_{\ell-1})\cdots \fa_1(m_1)\vphi_1}{\fa_\ell(-m_\ell)\iY(\vphi_3,z)\fb_r(n_r)\cdots \fb_1(n_1)\vphi_2}_{\modulekflambda{\lambda_1}}\\
  &\qquad=\Scp{\fa_{\ell-1}(m_{\ell-1})\cdots \fa_1(m_1)\vphi_1}{\iY(\vphi_3,z)\fa_\ell(-m_\ell) \fb_r(n_r)\cdots \fb_1(n_1)\vphi_2}_{\modulekflambda{\lambda_1}}\\
  &\qquad\qquad+\Scp{\fa_{\ell-1}(m_{\ell-1})\cdots \fa_1(m_1)\vphi_1}{\iY(\fa_\ell\vphi_3,z) \fb_r(n_r)\cdots \fb_1(n_1)\vphi_2}_{\modulekflambda{\lambda_1}}\,z^{m_\ell}\,,
\end{align}
for $\vphi_i \in \modulekflambda{\lambda_i}(0)$, $i=1,2,3$. Let us define the function
\begin{align*}
  f_{\ell,r}(z,\vphi_3,\substack{\vphi_1\\\vphi_2}|\substack{\fa_\ell,\ldots,\fa_1\\ m_\ell,\ldots, m_1}|\substack{\fb_r,\ldots,\fb_1\\ n_r,\ldots, n_1}):=\Scp{\fa_\ell(m_\ell)\cdots \fa_1(m_1)\vphi_1}{\iY(\vphi_3,z)\fb_r(n_r)\cdots \fb_1(n_1)\vphi_2}_{\modulekflambda{\lambda_1}}\,,
\end{align*}
for $\vphi_i \in \modulekflambda{\lambda_i}(0)$, elements $\fa_1,\ldots,\fa_\ell$ and $\fb_1,\ldots,\fb_r$ in the Lie algebra $\g $, and integer numbers $m_1,\ldots,m_\ell$, $n_1,\ldots,n_r$. Then we have established the recursion relation
\begin{align}
  f_{\ell,r}(z,\vphi_3,\substack{\vphi_1\\\vphi_2}|\substack{\fa_\ell,\ldots,\fa_1\\ m_\ell,\ldots, m_1}|\substack{\fb_r,\ldots,\fb_1\\ n_r,\ldots, n_1}) &= f_{\ell-1,r+1}(z,\vphi_3,\substack{\vphi_1\\\vphi_2}|\substack{\fa_{\ell-1},\ldots,\fa_1\\ m_{\ell-1},\ldots, m_1}|\substack{\fa_\ell,\fb_r,\ldots,\fb_1\\ -m_\ell,n_r,\ldots, n_1}) \\
  &\quad+ z^{m_\ell}\,f_{\ell-1,r}(z,\fa_\ell\vphi_3,\substack{\vphi_1\\\vphi_2}|\substack{\fa_{\ell-1},\ldots,\fa_1\\ m_{\ell-1},\ldots, m_1}|\substack{\fb_r,\ldots,\fb_1\\ n_r,\ldots, n_1})\,.
\end{align}
Applying this function recursively, we can decrease the index $\ell$ step by step to zero. Then we have to consider the function
\begin{align}
  \anag_r(z,\vphi_3,\substack{\vphi_1\\\vphi_2}|\substack{\fb_r,\ldots,\fb_1\\n_r,\ldots,n_1}):=f_{0,r}(z,\vphi_3,\substack{\vphi_1\\\vphi_2}|\ |  \substack{\fb_r,\cdots,\fb_1\\ n_r,\cdots, n_1})\,,
\end{align}
which is more explicitly given by
\begin{align}
  \anag_r(z,\vphi_3,\substack{\vphi_1\\\vphi_2}|\substack{\fb_r,\ldots,\fb_1\\n_r,\ldots,n_1}) = \Scp{\vphi_1}{\iY(\vphi_3,z)\fb_r(n_r)\cdots \fb_1(n_1)\vphi_2}_{\modulekflambda{\lambda_1}}\,.
\end{align}
The function $\anag_r$ can easily be computed if $\vphi_j\in \modulekflambda{\lambda_j}(0)$, $j=1,2,3$ belong to the top level, and all $n_j<0$  are negative for $j=1,\ldots,r$. Indeed, in this case we can again apply~\eqref{eq:basecasesforrecursive}
and the adjoint condition~\eqref{eq:adjointnegativecond} to get
\begin{align}
  \anag_r(z,\vphi_3,\substack{\vphi_1\\\vphi_2}|\substack{\fb_r,\ldots,\fb_1\\n_r,\ldots,n_1}) &= \Scp{\fb_r(-n_r)\vphi_1}{\iY(\vphi_3,z)\fb_{r-1}(n_{r-1})\cdots \fb_1(n_1)\vphi_2}_{\modulekflambda{\lambda_1}}\\
  &\quad - z^{n_r}\, \Scp{\vphi_1}{\iY(\fb_r \vphi_3,z)\fb_{r-1}(n_{r-1})\cdots \fb_1(n_1)\vphi_2}_{\modulekflambda{\lambda_1}}\,,
\end{align}
and the first term vanishes since $\fb_r(-n_r)\vphi_1 = 0$ for $n_r<0$. Applying this reasoning recursively and inserting the mode expansion of the intertwiner we arrive at the explicit expression
\begin{align}
  \anag_r(z,\vphi_3,\substack{\vphi_1\\\vphi_2}|\substack{\fb_r,\ldots,\fb_1\\n_r,\ldots,n_1}) = (-1)^r z^{\sum_{j=1}^r n_j} \Scp{\vphi_1}{\iY(\fb_1\cdots\fb_r \vphi_3,z)\vphi_2}_{\modulekflambda{\lambda_1}}\,.
\end{align}
This expression only depends on the zero mode of the intertwiner
(since all vectors belong to their respective top level, and only the zero mode does not change the weight). We conclude that in order to 
obtain an extension of the intertwining map~$\cW$, we must set
\begin{align}
  \anag_r(z,\vphi_3,\substack{\vphi_1\\\vphi_2}|\substack{\fb_r,\ldots,\fb_1\\n_r,\ldots,n_1}) = (-1)^r z^{\sum_{j=1}^r n_j} z^{-\tau}\Scp{\vphi_1}{\cW(\fb_1,\ldots,\fb_r \vphi_3\otimes \vphi_2)}_{V_1}\,.\label{eq:grzvphiextension}
\end{align}
 for $\vphi_j\in \modulekflambda{\lambda_j}(0)$, $j=1,2,3$ belonging to the top level (or equivalently the irreducible $\g$-modules~$V_j$), and all $n_j<0$  negative for $j=1,\ldots,r$. 

Similarly, we can treat the case where $n_1=0$ and $n_j<0$ for $j=2,\ldots,r$: here we observe that $\varphi'_2=\fb_1\varphi_2$ is also an element of the top level $\modulekflambda{\lambda_2}(0)$, and we obtain
\begin{align*}
\anag_r(z,\vphi_3,\substack{\vphi_1\\\vphi_2}|\substack{\fb_r,\ldots,\fb_1\\n_r,\ldots,n_1})&=\anag_r(z,\vphi_3,\substack{\vphi_1\\\vphi'_2}|\substack{\fb_r,\ldots,\fb_2\\n_r,\ldots,n_2})\ .
\end{align*}
Also, if $r=0$, 
then  we must set
\begin{align*}
\anag_r(z,\varphi_3,\substack{\varphi_1\\\vphi_2}| )&=\langle\varphi_1,\iY(\vphi_3,z)\vphi_2\rangle_{\modulekflambda{\lambda_1}}\\
&=z^{-\tau}\langle\varphi_1,\cW(\varphi_3\otimes\varphi_2)\rangle_{V_1}\ .
\end{align*}

Consider now the case where some of the integers~$n_j$, $j=1,\ldots,r$ are positive and again $\varphi_2\in \modulekflambda{\lambda_2}(0)$. In this case, we can apply the subroutine $\mathsf{L}$ to get a linear combination of vectors of the form~\eqref{eq:elementsinWZWtrunc}, i.e., 
\begin{align*}
\fb_r(n_r)\cdots\fb_1(n_1)\varphi_2&=\sum_{m} \alpha_m \fe_{i_{s_m}}(k_{s_{m}})\cdots \fe_{i_1}(k_1)\varphi_2 ,
\end{align*}
where $k_{s_m}<\cdots <k_1\leq 0$ is non-positive (we again
abuse notation and suppress an additional index~$m$ for $i$ and $k$). By linearity of the inner product, we conclude that 
\begin{align*}
\anag_r(z,\vphi_3,\substack{\vphi_1\\\vphi_2}|\substack{\fb_r,\ldots,\fb_1\\n_r,\ldots ,n_r}) = \sum_{m}\alpha_m \anag_r(z,\vphi_3,\substack{\vphi_1\\\vphi_2}|\substack{\fe_r,\ldots,\fe_1\\k_r,\ldots ,k_1})\ .
\end{align*}
At this stage, we use the explicit form~\eqref{eq:grzvphiextension} of the function $g$ for every element of this linear combination. 
\end{proof}
Algorithms~\ref{algo:WZWL},~\ref{algo:IntWZW} and~\ref{algo:WZWG}  summarize these arguments as pseudocode. In this code, we use the notation
\begin{align*}
F(z,\ell,r,\vphi_3,\substack{\vphi_1\\\vphi_2}|\substack{\fa_\ell,\ldots,\fa_1\\ m_\ell,\ldots, m_1}|\substack{\fb_r,\ldots,\fb_1\\ n_r,\ldots, n_1})&=
  f_{\ell,r}(z,\vphi_3,\substack{\vphi_1\\\vphi_2}|\substack{\fa_\ell,\ldots,\fa_1\\ m_\ell,\ldots, m_1}|\substack{\fb_r,\ldots,\fb_1\\ n_r,\ldots, n_1})\\
\algoG(z,r,\vphi_3,\substack{\vphi_1\\\vphi_2}|\substack{\fb_r,\ldots,\fb_1\\n_r,\ldots,n_1})&=\anag_r(z,\vphi_3,\substack{\vphi_1\\\vphi_2}|\substack{\fb_r,\ldots,\fb_1\\n_r,\ldots,n_1})\ .
\end{align*}

\newpage
\providecommand{\href}[2]{#2}\begingroup\raggedright\endgroup


\end{document}